\documentclass[11pt,reqno]{amsart}

\newcommand{\RemoveComments}{}
\usepackage{hyperref}

\usepackage{amsmath}
\usepackage{amsthm}
\usepackage{amsfonts}
\usepackage{amssymb}

\usepackage{dsfont}
\usepackage{latexsym}
\usepackage{geometry}
\usepackage{fixmath}
\usepackage{thmtools}
\usepackage{thm-restate}
\usepackage{array}

\usepackage{cleveref}

\usepackage{upgreek}
\usepackage{bm}
\usepackage{mathrsfs}

\usepackage[dvips]{graphicx}
\usepackage{graphicx}
\usepackage{fullpage}
\usepackage{appendix}

\usepackage{color} 
\usepackage{times}

\usepackage{leftidx}
\usepackage{scrextend}
\usepackage{float}
\usepackage[symbol]{footmisc}
\usepackage[shortlabels]{enumitem}

\numberwithin{equation}{section}

\newtheorem{theorem}{Theorem}[section]
\newtheorem{lemma}[theorem]{Lemma}

\newtheorem{proposition}[theorem]{Proposition}
\newtheorem{claim}[theorem]{Claim}

\newtheorem{definition}[theorem]{Definition}

\usepackage{tikz}
\usetikzlibrary{shapes.geometric, arrows}

\tikzstyle{startstop} = [rectangle, rounded corners, minimum width=3cm, minimum height=1cm,text centered, draw=black, fill=green!10]
\tikzstyle{arrow} = [thick,->,>=stealth]

\hypersetup{
citebordercolor={0.0 0.42 0.24}, 
linkbordercolor=red, 
colorlinks=true,
citecolor=cadmiumgreen,
linkcolor=cadmiumgreen,
 urlcolor=cyan,
}

\Crefname{lem}{Lemma}{Lemmas}
\Crefname{section}{Section}{Sections}
\Crefname{lemma}{Lemma}{Lemmas}
\Crefname{thm}{Theorem}{Theorems}
\Crefname{corollary}{Corollary}{Corollaries}
\Crefname{theorem}{Theorem}{Theorems}
\Crefname{defn}{Definition}{Definitions}
\Crefname{definition}{Definition}{Definitions}
\Crefname{fact}{Fact}{Facts}
\Crefname{figure}{Fig.}{Figures}
\Crefname{clm}{Claim}{Claims}
\Crefname{claim}{Claim}{Claims}
\Crefname{prop}{Proposition}{Propositions}
\Crefname{proposition}{Proposition}{Propositions}
\Crefname{algocf}{Algorithm}{Algorithms}

\newcommand{\cC}{{\mathcal{C}}}

\newcommand{\cE}{{\mathcal{E}}}
\newcommand{\cF}{{\mathcal{F}}}

\newcommand{\cI}{{\mathcal{I}}}
\newcommand{\cJ}{{\mathcal{J}}}

\newcommand{\cL}{{\mathcal{L}}}
\newcommand{\cM}{{\mathcal{M}}}
\newcommand{\cN}{{\mathcal{N}}}
\newcommand{\cO}{{\mathcal{O}}}
\newcommand{\cP}{{\mathcal{P}}}

\newcommand{\cR}{{\mathcal{R}}}

\newcommand{\cW}{{\mathcal{W}}}

\DeclareSymbolFont{EuclidLetter}{U}{eur}{m}{n}
\DeclareMathSymbol{\UpA}{\mathord}{EuclidLetter}{65}
\DeclareMathSymbol{\Upa}{\mathord}{EuclidLetter}{97}

\DeclareMathSymbol{\UpB}{\mathord}{EuclidLetter}{66}
\DeclareMathSymbol{\Upb}{\mathord}{EuclidLetter}{98}

\DeclareMathSymbol{\UpC}{\mathord}{EuclidLetter}{67}
\DeclareMathSymbol{\upc}{\mathord}{EuclidLetter}{99}

\DeclareMathSymbol{\UpD}{\mathord}{EuclidLetter}{68}
\DeclareMathSymbol{\upd}{\mathord}{EuclidLetter}{100}

\DeclareMathSymbol{\UpE}{\mathord}{EuclidLetter}{69}
\DeclareMathSymbol{\upe}{\mathord}{EuclidLetter}{101}

\DeclareMathSymbol{\UpF}{\mathord}{EuclidLetter}{70}
\DeclareMathSymbol{\upf}{\mathord}{EuclidLetter}{102}

\DeclareMathSymbol{\UpG}{\mathord}{EuclidLetter}{71}
\DeclareMathSymbol{\upg}{\mathord}{EuclidLetter}{103}

\DeclareMathSymbol{\UpH}{\mathord}{EuclidLetter}{72}
\DeclareMathSymbol{\uph}{\mathord}{EuclidLetter}{104}

\DeclareMathSymbol{\UpI}{\mathord}{EuclidLetter}{73}

\DeclareMathSymbol{\UpJ}{\mathord}{EuclidLetter}{74}

\DeclareMathSymbol{\UpK}{\mathord}{EuclidLetter}{75}
\DeclareMathSymbol{\upk}{\mathord}{EuclidLetter}{107}

\DeclareMathSymbol{\UpL}{\mathord}{EuclidLetter}{76}

\DeclareMathSymbol{\UpM}{\mathord}{EuclidLetter}{77}

\DeclareMathSymbol{\UpN}{\mathord}{EuclidLetter}{78}
\DeclareMathSymbol{\UpP}{\mathord}{EuclidLetter}{80}
\DeclareMathSymbol{\UpQ}{\mathord}{EuclidLetter}{81}
\DeclareMathSymbol{\UpR}{\mathord}{EuclidLetter}{82}

\DeclareMathSymbol{\UpS}{\mathord}{EuclidLetter}{83}
\DeclareMathSymbol{\ups}{\mathord}{EuclidLetter}{115}

\DeclareMathSymbol{\UpT}{\mathord}{EuclidLetter}{84}
\DeclareMathSymbol{\upt}{\mathord}{EuclidLetter}{116}
\DeclareMathSymbol{\UpW}{\mathord}{EuclidLetter}{87}
\DeclareMathSymbol{\upw}{\mathord}{EuclidLetter}{119}
\DeclareMathSymbol{\UpX}{\mathord}{EuclidLetter}{88}
\DeclareMathSymbol{\UpY}{\mathord}{EuclidLetter}{89}
\DeclareMathSymbol{\UpZ}{\mathord}{EuclidLetter}{90}
\DeclareMathSymbol{\upz}{\mathord}{EuclidLetter}{122}

\DeclareSymbolFont{EuclidLetter}{U}{eur}{m}{n}

\definecolor{ao(english)}{rgb}{0.0, 0.5, 0.0}
\definecolor{airforceblue}{rgb}{0.36, 0.54, 0.66}
\definecolor{amber}{rgb}{1.0, 0.75, 0.0}
\definecolor{amber(sae/ece)}{rgb}{1.0, 0.49, 0.0}
\definecolor{amethyst}{rgb}{0.6, 0.4, 0.8}
\definecolor{applegreen}{rgb}{0.55, 0.71, 0.0}
\definecolor{azure(colorwheel)}{rgb}{0.0, 0.5, 1.0}
\definecolor{arylideyellow}{rgb}{0.91, 0.84, 0.42}

\definecolor{bostonuniversityred}{rgb}{0.8, 0.0, 0.0}
\definecolor{bittersweet}{rgb}{1.0, 0.44, 0.37}
\definecolor{bleudefrance}{rgb}{0.19, 0.55, 0.91}
\definecolor{blue(pigment)}{rgb}{0.2, 0.2, 0.6}
\definecolor{blue-violet}{rgb}{0.54, 0.17, 0.89}
\definecolor{britishracinggreen}{rgb}{0.0, 0.26, 0.15}
\definecolor{brilliantrose}{rgb}{1.0, 0.33, 0.64}
\definecolor{byzantine}{rgb}{0.74, 0.2, 0.64}
\definecolor{byzantium}{rgb}{0.44, 0.16, 0.39}

\definecolor{charcoal}{rgb}{0.21, 0.27, 0.31}
\definecolor{cadmiumgreen}{rgb}{0.0, 0.42, 0.24}
\definecolor{cadmiumorange}{rgb}{0.93, 0.53, 0.18}
\definecolor{coquelicot}{rgb}{1.0, 0.22, 0.0}
\definecolor{capri}{rgb}{0.0, 0.75, 1.0}

\definecolor{deeppink}{rgb}{1.0, 0.08, 0.58}
\definecolor{dollarbill}{rgb}{0.52, 0.73, 0.4}

\definecolor{darkmagenta}{rgb}{0.55, 0.0, 0.55}
\definecolor{darkmidnightblue}{rgb}{0.0, 0.2, 0.4}
\definecolor{darkpastelpurple}{rgb}{0.59, 0.44, 0.84}
\definecolor{darkspringgreen}{rgb}{0.09, 0.45, 0.27}
\definecolor{darktangerine}{rgb}{1.0, 0.66, 0.07}
\definecolor{darkgoldenrod}{rgb}{0.72, 0.53, 0.04}

\definecolor{eggplant}{rgb}{0.38, 0.25, 0.32}
\definecolor{electricviolet}{rgb}{0.56, 0.0, 1.0}
\definecolor{ferrarired}{rgb}{1.0, 0.11, 0.0}
\definecolor{forestgreen(traditional)}{rgb}{0.0, 0.27, 0.13}

\definecolor{green(pigment)}{rgb}{0.0, 0.65, 0.31}
\definecolor{glaucous}{rgb}{0.38, 0.51, 0.71} 
\definecolor{goldenbrown}{rgb}{0.6, 0.4, 0.08}
\definecolor{gold(metallic)}{rgb}{0.83, 0.69, 0.22}
\definecolor{gold(web)(golden)}{rgb}{1.0, 0.84, 0.0}
\definecolor{goldenpoppy}{rgb}{0.99, 0.76, 0.0}
\definecolor{goldenyellow}{rgb}{1.0, 0.87, 0.0}
\definecolor{goldenrod}{rgb}{0.85, 0.65, 0.13}

\definecolor{harvestgold}{rgb}{0.85, 0.57, 0.0}
\definecolor{heartgold}{rgb}{0.5, 0.5, 0.0}

\definecolor{heliotrope}{rgb}{0.87, 0.45, 1.0}

\definecolor{iris}{rgb}{0.35, 0.31, 0.81}

\definecolor{oldgold}{rgb}{0.81, 0.71, 0.23}

\definecolor{palegold}{rgb}{0.9, 0.75, 0.54}

\definecolor{rosegold}{rgb}{0.72, 0.43, 0.47}
\definecolor{red(pigment)}{rgb}{0.93, 0.11, 0.14}

\definecolor{sapphire}{rgb}{0.03, 0.15, 0.4}
\definecolor{satinsheengold}{rgb}{0.8, 0.63, 0.21}

\definecolor{uclagold}{rgb}{1.0, 0.7, 0.0}

\definecolor{vegasgold}{rgb}{0.77, 0.7, 0.35}

\newcommand{\applegreen}[1]{ {\color{applegreen}#1}}

\newcommand{\blue}[1]{{\color{blue}#1}}
\newcommand{\bluepigment}[1]{{\color{blue(pigment)}#1}}
\newcommand{\byzantine}[1]{{\color{byzantine}#1}}

\newcommand{\capri}[1]{{\color{capri}#1}} 
\newcommand{\cadmiumgreen}[1]{{\color{cadmiumgreen}#1}} 
\newcommand{\cadmiumorange}[1]{{\color{cadmiumorange}#1}}

\newcommand{\darkmagenta}[1]{{\color{darkmagenta}#1}}
\newcommand{\deeppink}[1]{{\color{deeppink}#1}}

\newcommand{\goldenbrown}[1]{{\color{goldenbrown}#1}}

\newcommand{\magenta}[1]{{\color{magenta}#1}}
\newcommand{\iris}[1]{{\color{iris}#1}}
\newcommand{\racinggreen}[1]{{\color{britishracinggreen}#1}} 
\newcommand{\red}[1]{{\color{red}#1}}
\newcommand{\rosegold}[1]{{\color{rosegold}#1}}

\newcommand{\Ind}{\mathds{1}}

\newcommand{\dist}{{\rm dist}}

\newcommand{\maxDeg}{{\Delta}}

\newcommand{\ExtV}{\iris{\cE}}

\newcommand{\cconnective}{\goldenbrown{\upd}}

\newcommand{\PartFun}{\mathcal{Z}}

\newcommand{\dcritical}{{\Delta_c}}
\newcommand{\lcritical}{{\lambda_c}}

\newcommand{\UnIsing}{{\mathbb{U}_{\rm Ising}}}

\newcommand{\potF}{{\Psi}}
\newcommand{\dpotF}{{\psi}}
\newcommand{\xdpotF}{\chi}

\newcommand{\pfs}{\darkmagenta{s}}

\newcommand{\gratio}{{R}}
\newcommand{\trecur}{F}
\newcommand{\logtrecur}{H}
\newcommand{\dlogtrecur}{h}
\newcommand{\ratiorange}{J}

\newcommand{\Glauber}{{\{X_{\kt}\}_{\kt\geq 0}}}

\newcommand{\Tsaw}{T_{\rm SAW}}

\newcommand{\cp}{{\tt A}}
\newcommand{\scp}{\widehat{\tt A}}

\newcommand{\infmatrix}{\cI}

\newcommand{\spradius}{{\rho}}

\newcommand{\eigenval}{\upxi}

\newcommand{\enorm}[1]{ \goldenbrown{\|} #1 \goldenbrown{\|}}
\newcommand{\norm}[2]{\goldenbrown{\|} #1 \goldenbrown{\|}_{#2}}
\newcommand{\nnorm}[3]{\goldenbrown{\|} #1\goldenbrown{\|}^{#2}_{#3}}
\newcommand{\tv}{tv}
\newcommand{\abs}[1]{\blue{|}#1\blue{|}}

\newcommand{\spectrum}{{\Upphi}}

\newcommand{\PTV}[1]{\check{#1}}

\newcommand{\MTR}[1]{\overline{#1}}

\newcommand{\Adjacency}{{\UpA}}

\newcommand{\powadj}{{{\UpB}}}

\newcommand{\aspradius}{{\varrho}}

\newcommand{\maxeigenv}{\upphi_1} 

\newcommand{\eigenv}{\upphi}

\newcommand{\NBMatrix}{\cadmiumgreen{\UpH}_{G,\kk}}
\newcommand{\NBMatrixE}{\byzantine{\UpH}}

\newcommand{\VToEdge}{\UpD}
\newcommand{\EdgeToV}{\UpC}

\newcommand{\Invol}{\UpR}

\newcommand{\spreadpoint}{\newpage}
\newcommand{\LastReview}[1]{\rosegold{\hspace{.2cm} [Rev: \textrm{#1}] \hspace{.2cm}}}
\newcommand{\LastReviewG}[1]{{\color{magenta}\hspace{.2cm} [Rev: \textrm{#1} - \bluepigment{G} ] \hspace{.2cm}}}
\newcommand{\charis}[1]{--\bluepigment{\rosegold{Charis:} #1}}

\ifdefined \RemoveComments

\renewcommand{\LastReview}[1]{}
\renewcommand{\LastReviewG}[1]{}
\renewcommand{\charis}[1]{}
\renewcommand{\spreadpoint}{}

\renewcommand{\applegreen}[1]{#1}

\renewcommand{\blue}[1]{#1}
\renewcommand{\bluepigment}[1]{#1}
\renewcommand{\byzantine}[1]{#1}

\renewcommand{\capri}[1]{#1}
\renewcommand{\cadmiumgreen}[1]{#1}
\renewcommand{\cadmiumorange}[1]{#1}

\renewcommand{\darkmagenta}[1]{#1}
\renewcommand{\deeppink}[1]{#1}

\renewcommand{\goldenbrown}[1]{#1}

\renewcommand{\magenta}[1]{#1}
\renewcommand{\iris}[1]{#1}
\renewcommand{\racinggreen}[1]{#1}
\renewcommand{\red}[1]{#1}
\renewcommand{\rosegold}[1]{#1}

\fi

\newcommand{\DBounded}{\mathbold{q}}

\newcommand{\SpGMatrix}{\Upxi}

\newcommand{\saw}{{\rm SAW}}

\newcommand{\ka}{\electricviolet{a}}

\newcommand{\kb}{\deeppink{b}}

\newcommand{\kC}{\deeppink{C}}
\newcommand{\kc}{\deeppink{c}}

\newcommand{\kd}{\goldenbrown{d}}

\newcommand{\ke}{\blue{e}}

\newcommand{\kf}{\deeppink{f}}

\newcommand{\kh}{\cadmiumorange{h}}

\newcommand{\ki}{\capri{i}}

\newcommand{\kj}{\coquelicot{j}}

\newcommand{\kK}{\bluepigment{K}}
\newcommand{\kk}{\bluepigment{k}}

\newcommand{\kL}{\cadmiumorange{L}}

\newcommand{\kell}{{\color{green(pigment)}\ell}}

\newcommand{\kM}{\cadmiumorange{M}}

\newcommand{\kP}{\goldenbrown{P}}

\newcommand{\kQ}{\amethyst{Q}}
\newcommand{\kq}{\amethyst{q}}

\newcommand{\kR}{\bostonuniversityred{R}}
\newcommand{\kr}{\bostonuniversityred{r}}

\newcommand{\ks}{\cadmiumgreen{s}}

\newcommand{\kt}{\cadmiumorange{t}}

\newcommand{\ku}{\red{u}}

\newcommand{\kv}{\electricviolet{v}}

\newcommand{\kW}{\blue{W}}
\newcommand{\kw}{\blue{w}}

\newcommand{\kx}{\heliotrope{x}}

\newcommand{\ky}{\cadmiumorange{y}}

\newcommand{\kz}{\cadmiumgreen{z}}

\ifdefined \RemoveComments

\renewcommand{\ka}{a}

\renewcommand{\kb}{b}

\renewcommand{\kC}{C}
\renewcommand{\kc}{x}

\renewcommand{\kd}{d}

\renewcommand{\ke}{e}

\renewcommand{\kf}{f}

\renewcommand{\kh}{h}

\renewcommand{\ki}{i}

\renewcommand{\kj}{j}

\renewcommand{\kK}{K}
\renewcommand{\kk}{k}

\renewcommand{\kL}{L}

\renewcommand{\kell}{\ell}

\renewcommand{\kM}{M}

\renewcommand{\kP}{P}

\renewcommand{\kQ}{Q}
\renewcommand{\kq}{q}

\renewcommand{\kR}{R}
\renewcommand{\kr}{r}

\renewcommand{\ks}{s}

\renewcommand{\kt}{t}

\renewcommand{\ku}{u}

\renewcommand{\kv}{v}

\renewcommand{\kW}{W}
\renewcommand{\kw}{w}

\renewcommand{\kx}{x}

\renewcommand{\ky}{\cadmiumgreen{y}}

\renewcommand{\kz}{z}

\fi

\date{\today}

\newcommand{\infmatrixB}{\mathcal{J}}

\newcommand{\ExtdInfMatrix}{\rosegold{\cL}}
\newcommand{\ExtdInfMatrixF}{\applegreen{\cL}^{\Lambda,\tau}_{G,\kk}}

\newcommand{\CovMatrix}{\Upsigma}

\newcommand{\SemExtdInfMatrix}{\cF}
\newcommand{\spset}{\ExtV}
\newcommand{\SymWeightM}{\cN}
\newcommand{\SymWeightMA}{\cN_A}
\newcommand{\SymWeightMB}{\cN_B}

\newcommand{\infweight}{\upalpha}

\newcommand{\hsingular}{\upsigma}
\newcommand{\SingBound}{\uprho}

\begin{document}

\title{On sampling two spin models using \\the local connective constant}
\author{
Charilaos Efthymiou
}
 
\thanks{ 
{\em Funding:} Research supported by EPSRC New Investigator Award, grant EP/V050842/1, and 
Centre of Discrete Mathematics and Applications (DIMAP), University of Warwick, UK. \\ 
\hspace*{.725cm}{\em Acknowledgement}: The author would like to thank Daniel \v Stefankovi\v c, Eric Vigoda and Kostas Zampetakis for the comments and the fruitful discussions. 
}
\address{Charilaos Efthymiou, {\tt charilaos.efthymiou@warwick.ac.uk}, University of Warwick, Coventry, CV4 7AL, UK.}

\maketitle

\thispagestyle{empty}

\begin{abstract}
This work establishes {\em novel} optimum mixing bounds for the Glauber dynamics  on the Hard-core and 
Ising models.  These bounds are expressed in terms of the {\em local connective constant}  of  the underlying 
graph $G$. This is a notion of effective degree for $G$.

Our results have some interesting consequences for bounded degree graphs:
\begin{enumerate}[(a)]
\item They include the max-degree bounds as a special case
\item They improve on the running time of the FPTAS considered in [Sinclair, Srivastava, \v Stefankoni\v c and Yin: PTRF 2017] for general graphs
\item They allow us to obtain  mixing bounds in terms of the spectral radius of the adjacency matrix and
 improve on  [Hayes: FOCS 2006].
\end{enumerate}
We obtain our results  using  tools from the theory of high-dimensional expanders and, in particular, the 
Spectral Independence method  [Anari, Liu, Oveis-Gharan: FOCS 2020]. 
We explore a new direction by utilising  the notion of the  {\em $k$-non-backtracking matrix}  $\NBMatrix$ in our analysis with 
the Spectral Independence.  
The results with $\NBMatrix$ are interesting in their own right. 
\end{abstract}

\spreadpoint

\thispagestyle{empty}

\spreadpoint

%

\newpage
\setcounter{page}{1}

\newcommand{\SingleEdgeM}{\UpC}

\newcommand{\chiofh}{\uppsi}

\newcommand{\gext}[2]{\racinggreen{ {#1}_{#2}} }
\newcommand{\sbsplit}{\byzantine{\partial \UpS}}
\newcommand{\ssplit}{\cadmiumgreen{\UpS}}

\newcommand{\subcont}{{\cR}^{(\kell)}}

\newcommand{\inflPotB}{\mathcal{H}}

\section{Introduction \LastReviewG{2025-03-28} }\label{sec:Introduction}
This work focuses on counting and sampling problems that arise from the study of Gibbs distributions.
We  typically consider  Gibbs distributions $\mu$ which are defined with respect to an underlying 
graph $G=(V,E)$. Unless otherwise specified,  graph $G$ is finite, undirected and connected.

We parametrise $\mu$ by using the numbers $\beta \geq 0$ and $\gamma, \lambda >0$, while each configuration $\sigma\in \{\pm 1\}^V$ 
gets probability mass 
\begin{align}\label{def:GibbDistr}
\mu(\sigma) &=\frac{1}{\PartFun} \times \lambda^{\# \textrm{assignments ``1" in $\sigma$ } }
\times \beta^{\# \textrm{edges with both ends ``1" in $\sigma$ } }
\times \gamma^{\# \textrm{edges with both ends ``-1" in $\sigma$} } \enspace,
\end{align}
where $\PartFun=\PartFun(G,\beta,\gamma, \lambda)$ is the normalising quantity called the {\em partition function}. 

A natural problem in the study of Gibbs distributions  is  computing  the partition function $\PartFun$.  
In many cases,  this is a  computationally hard problem, e.g., see \cite{Valiant79,JerSincIsing93}. Typically,  
the focus  is on efficiently  obtaining good approximations of $\PartFun$.

To keep our discussion concise in this introduction, we restrict  our  attention to the technically most interesting 
case,  the  {\em Hard-core model}.  The reader can find further results in the subsequent sections. 
The Hard-core  model is a probability distribution over  the  {\em independent sets} $\sigma$ of an underlying graph $G$ such that each 
 $\sigma$  gets probability mass $\mu(\sigma)$,  which is proportional to $\lambda^{\abs{\sigma}}$, where $\abs{\sigma}$ 
is the cardinality of the independent  set $\sigma$. 
Using the  formulation in \eqref{def:GibbDistr}, this distribution corresponds to  having $\beta=0$, $\gamma=1$ and $\lambda>0$. 
Typically,  the parameter $\lambda>0$ is called  {\em fugacity}.

The seminal works of Weitz in \cite{Weitz} and Sly in \cite{Sly10} (together with the improvements on \cite{Sly10} obtained in
\cite{SS14,GSV16Hardness}) establish a beautiful connection between {\em phase transitions} and the performance 
of counting algorithms for the Hard-core  model.    This connection implies that for graphs $G$ of maximum degree $\maxDeg$ 
and the  Hard-core model  with fugacity $\lambda$ on such graphs,  
there  exists a critical value $\lcritical(\maxDeg-1)$ such that following is true: for any fugacity $\lambda<\lcritical(\maxDeg-1)$, 
there exists an FPTAS for estimating the partition function, whereas, for 
$\lambda>\lcritical(\maxDeg-1)$ it is NP-hard to estimate the  partition function even within an exponentially large factor.  
Notably,  the critical value depends on the {\em maximum degree} $\maxDeg$ of  $G$.

Recall that, for integer $\kk>1$,   the critical value $\lcritical(\kk)=\frac{\kk^{\kk}}{(\kk-1)^{(\kk+1)}}$  signifies the  uniqueness/non-uniqueness
phase transition for the Hard-core model on  the $\kk$-ary  tree, established by Kelly in   \cite{Kelly85}.
Also, recall that FPTAS stands for  Fully Polynomial   Time Approximation Scheme.
 An FPTAS  computes an  estimation     which is within a factor  $(1\pm \upepsilon)$ from $\PartFun$
in time  which is polynomial   $n$ and $\log \frac{1}{\upepsilon}$.

An important improvement on the aforementioned connection between phase transitions and hardness comes with the seminal 
works of Sinclair, Srivastava and  Yin in \cite{ConnectiveConstFirst} and the subsequent sharper analysis  provided by Sinclair, Srivastava, \v Stefankovi\v c 
and Yin in \cite{ConnectiveConst}. These works  propose an FPTAS for estimating the partition function  of the Hard-core 
model for any fugacity  $\lambda<\lcritical(\cconnective)$, where $\cconnective$ is the {\em connective constant} of the 
underlying graph $G$. The connective constant is a natural measure of  ``effective degree" for   $G$.  Recall that 
it  is defined by $\cconnective = \lim_{\kk\to \infty}\sup_{\kv\in V} \pi(\kv,\kk)^{1/\kk}$, where 
$\pi(\kv,\kk)$ is the number of self-avoiding walks of length $\kk$ emanating from vertex $\kv$.  The connective constant is  
a well-known quantity, especially  studied in mathematics and physics (with the breakthrough result in \cite{HoneycombCC}).
\blue{For many  families of graphs, the connective constant $\cconnective$ is significantly smaller than 
the maximum degree $\maxDeg$. For such families,  \cite{ConnectiveConstFirst,ConnectiveConst} 
provide a tighter analysis than  \cite{Weitz}. }

Interestingly, the results in \cite{ConnectiveConstFirst,ConnectiveConst}  refine the aforementioned connection between phase transitions 
and hardness,  implying  that the dependence of the critical value should be  on $\cconnective$ instead of  $\maxDeg$. 
In that respect, the hardness results in \cite{SS14,GSV16Hardness} are sufficient to imply that it is NP-hard to estimate the partition function for 
$\lambda>\lcritical(\cconnective)$ (see also discussion later in this introduction).

The  bound $\cconnective\leq \maxDeg-1$ implies that the results with respect to the maximum degree  
correspond to considering the worst-case graph instances with regard to the relation between $\maxDeg$ and 
$\cconnective$, i.e., assume $\cconnective=\maxDeg-1$. On the other hand, the use of the connective constant $\cconnective$ 
provides a {\em finer} classification of the hard instances of the problem, i.e.,  since  graphs with the same maximum degree 
can have different connective constants.

Impressive as they may be,  the aforementioned algorithmic results are primarily of theoretical importance rather than 
practical. Their running times, even though polynomial in $n$, have a rather {\em heavy} dependence on the maximum 
degree $\maxDeg$. Typically, the exponent of the polynomial bound is an increasing function of $\maxDeg$, e.g., we  
have $O(n^{C \sqrt{\maxDeg}})$  running time.  
This motivates the  use of the Markov Chain Monte Carlo (MCMC) approach to the problem which, typically, 
yields {\em significantly} faster  algorithms.
For the distributions we consider here, the aim is to use the MCMC approach and get a Fully Polynomial  Randomized Time 
Approximation Scheme  (FPRAS)   for the partition function  $\PartFun$.   An FPRAS  computes with probability $1-\updelta$ an 
estimation   which is within a factor  $(1\pm \upepsilon)$ from $\PartFun$, in time which is polynomial 
in  $n$, $\log \frac{1}{\upepsilon}$  and $\log\frac{1}{\updelta}$.

To this end,  we use {\em Glauber dynamics}. This  is a popular Markov chain utilised for sampling from high-dimensional 
Gibbs distributions  such as the Hard-core model. In this setting,  the measure of the efficiency  is the {\em mixing time}, 
i.e., the rate at which the Markov chain converges to the equilibrium.

In recent years, there have been significant developments with regard to analysing the mixing time of
Glauber dynamics using ideas from the theory of {\em high dimensional expanders}  e.g., see Alev and Lau in \cite{FirstSpInd}. 
Advances in the area, such as the work of Anari, Liu and Oveis-Gharan in \cite{OptMCMCIS}, which 
introduces the well-known {\em Spectral Independence} method, as well as the subsequent 
work of Chen, Liu and Vigoda in \cite{VigodaSpectralIndB}, establish {\em optimum} mixing 
for Glauber  dynamics on the Hard-core model for any $\lambda<\lcritical(\maxDeg-1)$.  Recall that  Glauber 
dynamics exhibits  optimum mixing when the mixing time attains its (asymptotically) minimum value, which 
is $O(n\log n)$; see  Hayes and Sinclair \cite{OptMixingTime}.   
Optimum mixing for Glauber dynamics  implies an $O^*(n^2)$ time FPRAS for the  partition function (where $O^*(\cdot)$ 
omits the poly-logarithmic terms).

It is also worth mentioning the further improvements in \cite{WeimingCoUbboundedDelta,FastMCMCLocalisation},  
which obtain mixing bounds for graphs of unbounded  maximum degree.

In light of all the aforementioned  results which establish mixing bounds  expressed in terms of the maximum 
degree $\maxDeg$, the  works in \cite{ConnectiveConstFirst,ConnectiveConst}  motivate us to ask  whether 
it is possible to obtain  similar bounds expressed in terms of the connective 
constant $\cconnective$, i.e., establish optimum mixing for Glauber dynamics  for $\lambda<\lcritical(\cconnective)$. 
Our aim is to obtain such results using  tools from the theory of high-dimensional expanders and, in particular, the 
Spectral Independence method. 
To the extent that we are aware, such an approach has not been considered before for {\em general graphs}.

Obtaining   mixing bounds of Glauber dynamics in terms of  the connective constant is a well-studied problem in the area
and  results  are already known in the  literature for  {\em special}  families of graphs, e.g. see the seminal works in 
\cite{HCZ2PaperA,HCZ2PaperB} for the Hard-core model on the integer lattice $\mathbb{Z}^2$.
Using standard arguments like those in   \cite{CesiAmenable,MixingTimeSpaceRSA}, one can show that 
\cite{ConnectiveConstFirst,ConnectiveConst} imply $O(n\log n)$ mixing for Glauber dynamics  on 
the Hard-core model with  $\lambda<\lcritical(\cconnective)$ on  {\em amenable graphs}
\footnote{Roughly speaking,  an infinite graph $G=(V,E)$ is amenable if $\inf\frac{\abs{ \partial W}}{\abs{W}}=0$, 
where the infimum ranges over  all finite $W\subset  V$, while $\partial W$ is the boundary of $W$. If $G$ is finite, 
we typically have $\abs{ \partial W}\ll \abs{W}$, for sufficiently large blocks $W\subset V$.}.

Another well-studied case where the connective constant arises naturally in the mixing bounds is that
of the Hard-core model  on  the sparse random  graph $G(n,d/n)$. The state-of-the-art  bound for 
this distribution  is from  the work of Efthymiou and Feng in  \cite{EfthFeng23} which obtains $O(n^{1+\frac{C}{\log\log n}})$ 
mixing of Glauber dynamics for any $\lambda<\lcritical(d)$, improving on the works in \cite{SIGnp,Efth19}.  

As opposed to the previous works  which focus on {\em special} families of graphs, our endeavour  here is to obtain 
optimum mixing bounds    {\em without any assumptions} about the underlying graph, i.e., we don't necessarily assume that
the  graph is amenable or a typical instance of $G(n,d/n)$. 

For our algorithmic applications, it makes sense to use  a  finite version of the connective constant, i.e., 
rather than the limiting object. For integer $\kk\geq 1$,  we introduce the notion of the 
{\em radius-$\kk$  connective constant}  $\cconnective_{\kk}$ defined by  $\cconnective_{\kk}=\sup_{\kv\in V}\pi(\kv,\kk)^{1/\kk}$.
The connective constant $\cconnective$ corresponds to   $\cconnective_{\kk}$ for $\kk\to\infty$.

One of the main results in this paper is the following theorem.

\begin{theorem}\label{thrm:HC4CCK}
For any $\varepsilon \in (0,1)$, $\maxDeg>1$, $\kk\geq 1$ and $\cconnective_{\kk}>1$ consider graph $G=(V,E)$ 
of maximum degree $\maxDeg$ such that the radius-$k$ connective constant is $\cconnective_{\kk}$.
Also, let $\mu$ be the Hard-core model on $G$ with fugacity $\lambda\leq (1-\varepsilon)\lcritical(\cconnective_{\kk})$.

There is a constant $C=C(\maxDeg, \kk,  \cconnective_{\kk}, \varepsilon)$ such that the mixing time of Glauber dynamics on $\mu$ 
is at most $C n\log n$.
\end{theorem}

Handling directly $\cconnective_{\kk}$ for a general graph $G$ is an extremely challenging task 
(at least for the authors of this work). We obtain our results  indirectly by utilising  the notion of the
$\kk$-non-backtracking  matrix $\NBMatrix$. We establish rapid mixing  bounds  in terms of 
$\nnorm{ (\NBMatrix)^N}{1/N}{2}$,  where $N$ is a large number.

\subsubsection*{The $k$-non-backtracking matrix $\NBMatrix$:}
Perhaps, the best known ``incarnation" of the $k$-non-backtracking matrix $\NBMatrix$ is  the  
{\em Hashimoto non-backtracking  matrix} $\NBMatrixE_G$,  introduced in \cite{Hash89}.   Matrix 
$\NBMatrixE_G$ is typically  studied in (mathematical) physics. Recently it has found  application 
in   spectral algorithms for network 
inference problems, e.g. see \cite{CKMZNBM,AbbeSandon14} (with the breakthrough result in \cite{NBMGnp}).

Recall that $\NBMatrixE_G$ is  a 0/1 matrix indexed by the {\em ordered pairs} of the adjacent vertices
in $G$  such that for any $\ke=\ku\kw$ and $\kf=\kz \ky$, we have 
\begin{align}\nonumber 
\NBMatrixE_G(\ke, \kf) =\Ind\{\kw=\kz \}\times \Ind\{\ku\neq \ky\} \enspace.
\end{align}
That is, $\NBMatrixE_G(\ke, \kf)$ is equal to $1$, if $\kf$ follows the edge $\ke$ without creating a loop; 
otherwise,  it  is equal to zero.

 $\NBMatrix$ is a generalisation of $\NBMatrixE_G$.  It is a    0/1 matrix indexed by the $\kk$-non-backtracking walks 
 in $G$. For    $\ke$ and $\kf$, two $\kk$-non-backtracking walks, 
we have $\NBMatrix(\ke,\kf)=1$ if walk $\kf$ extends walk $\ke$ 
by one vertex without creating a loop. Otherwise, the entry is zero.   
For a formal definition of $\NBMatrix$, see 
\Cref{sec:MatricesOnPaths}. 
Note that $\NBMatrixE_G$ corresponds to $\NBMatrix$ for $\kk=1$.

We obtain the following optimum mixing results using $\NBMatrix$.
 
\begin{theorem}\label{thrm:HC4SPRadiusHash}
For $\varepsilon \in (0,1)$, for $\SingBound>1$ and integers $k,N>0$, $\maxDeg>1$, consider graph 
$G=(V,E)$ of maximum degree $\maxDeg$ such that $\nnorm{ (\NBMatrix)^N}{1/N}{2}=\SingBound$.
Also, let $\mu$ be the Hard-core model on $G$ with fugacity $\lambda\leq (1-\varepsilon)\lcritical(\SingBound)$.

There is a constant $C={C(\maxDeg,\SingBound,\kk,\varepsilon)}$ such that the mixing time of Glauber dynamics on $\mu$ 
is at most $C n\log n$.
\end{theorem}

Related to \Cref{thrm:HC4SPRadiusHash}, are the rapid mixing bounds in the literature that are expressed in terms of norms of   $\Adjacency_G$, 
the {\em adjacency matrix} of $G$.  
In his seminal work  in \cite{Hayes06}, Hayes obtains rapid mixing bounds in terms of $\norm{\Adjacency_G}{2}$.
Further results with respect to $\norm{\Adjacency_G}{2}$  can be obtained using other approaches, e.g.,
using  the  {\em stochastic localisation} method to analyse the mixing of Glauber dynamics \cite{FastMCMCLocalisation}.  
Using $\norm{\Adjacency_G}{2}$ in our setting has major drawbacks compared to 
$\nnorm{ (\NBMatrix)^N}{1/N}{2}$.   E.g., $\norm{\Adjacency_G}{2}$ tends to be  strongly correlated 
with the maximum degree  $\maxDeg$. 
On the other hand,  it is well-known that  for small values of $N,\kk$,  the norm  $\nnorm{ (\NBMatrix)^N}{1/N}{2}$ captures 
information related to the degree sequence, whereas,   for large values of these two parameters, the norm expresses 
{\em structural properties} of graph $G$. 
This makes it highly desirable to work with $\NBMatrix$.  It turns out that it is  more challenging to handle matrix $\NBMatrix$ 
than $\Adjacency_G$  in the analysis.  This is because $\NBMatrix$  does not possess  standard symmetries, e.g.,  it is {\em not} a normal matrix etc.
To our  knowledge, this is the first rapid mixing result  obtained in terms of norms of the  $\kk$-non-backtracking 
matrix $\NBMatrix$.

\Cref{thrm:HC4CCK} follows easily from  \Cref{thrm:HC4SPRadiusHash}. Specifically, it suffices to use the fact that 
 for large $N$ we essentially have  $\nnorm{ (\NBMatrix)^N}{1/N}{2}\leq \cconnective_{\kk}$ (see \Cref{prop:HSPRadCCLCC} for
the exact relation) and note that the function $\lcritical(\kz)$ is decreasing in $\kz$.

We further show that the above  mixing bound in terms of  $\nnorm{ (\NBMatrix)^N}{1/N}{2}$ is unlikely to
be improved.

\begin{theorem}\label{thrm:HarndessTransHCNBM}
Unless ${\rm NP} = {\rm RP}$, for $\varepsilon\in (0,1)$, $\maxDeg\geq 3$, $N,\kk\geq 1$ and $\SingBound > 1$, 
the following is true:

For any $\lambda>(1+\varepsilon)\lcritical(\lceil \SingBound \rceil)$, there is no FPRAS for estimating the partition function 
of the Hard-core model for graphs $G$ of maximum degree at most $\maxDeg$ such that $\nnorm{ (\NBMatrix)^N}{1/N}{2}=\SingBound$.

For any $\lambda<(1-\varepsilon)\lcritical(\SingBound)$, there is an FPRAS for estimating the partition function of the Hard-core model
for graphs $G$ of maximum degree at most $\maxDeg$ such that $\nnorm{ (\NBMatrix)^N}{1/N}{2}=\SingBound$.
\end{theorem}

The results in  \cite{SS14,GSV16Hardness}  that we utilise to show \Cref{thrm:HarndessTransHCNBM}, require the argument 
of $\lcritical(\cdot)$ to be an integer. 
For this reason, we use $\lcritical(\lceil  \SingBound \rceil)$ for the  lower bound in the theorem above. 

\subsection{Further results for the Ising model}
Using the formalism in \eqref{def:GibbDistr}, the Ising model corresponds to taking $\beta=\gamma$. Usually, we refer to the 
parameters $\beta$ as the {\em inverse temperature} and $\lambda$ as the {\em external field}. 
Recall that  when $\beta>1$, the Ising model is called {\em ferromagnetic}, while for 
$\beta<1$,  it is called {\em antiferromagnetic}.

For  integer $\kk\geq 2$, it is  well-known  that  the Ising model  on the infinite $\kk$-ary tree, with external field  $\lambda>0$ 
and inverse temperature $\beta\geq 0$,  exhibits uniqueness for 
\begin{align}\nonumber
\frac{\kk-1}{\kk+1} <\beta< \frac{\kk+1}{\kk-1} \enspace .
\end{align}
For $\kz>1$ and $\delta\in (0,1)$, we let the interval
\begin{align}\label{eq:DefOfIsingUniquReg}
\UnIsing(\kz,\delta) =\textstyle \left [ \frac{\kz-1+\delta}{\kz+1-\delta},   \frac{\kz+1-\delta}{\kz-1+\delta}\right] \enspace.
\end{align}
We obtain the following result for the Ising model. 
\begin{theorem}\label{thrm:Ising4CCK}
For $\varepsilon\in (0,1)$, $\maxDeg >1$, $\kk\geq 1$, $\cconnective_{\kk}>1$ and $\lambda>0$ 
consider graph $G=(V,E)$ of maximum degree $\maxDeg$ such that the radius-$\kk$ connective constant is $\cconnective_{\kk}$.
Let $\mu$ be the Ising model on $G$ with external field $\lambda$ and inverse temperature $\beta\in \UnIsing(\cconnective_{\kk},\varepsilon)$. 

There is a constant $C={C( \maxDeg,\cconnective_{\kk}, \varepsilon)}$ such that the mixing time of Glauber dynamics 
on $\mu$ is at most $C n\log n$. 
\end{theorem}

It is not hard to obtain  an FPTAS for the Ising model using parameterizations  in 
terms of  the connective constant, i.e., like the Hard-core model  in 
\cite{ConnectiveConstFirst,ConnectiveConst}. 
This  readily implies $O(n\log n)$ mixing for Glauber dynamics on amenable graphs. 
A related result in this area (and  one of the best understood cases) is  the Ising model on 
the  $\kr$-dimensional  square lattice $\mathbb{Z}^{\kr}$, 
for $\kr>0$.  Optimum bounds for the mixing time of Glauber dynamics for the Ising on $\mathbb{Z}^{\kr}$ have  been known since the 1990s from 
the influential work of Martinelli  and Olivieri in \cite{MartinelliOlivieriA,MartinelliOlivieriB}. These  results are 
better than those one obtains with the connective constant.

The mixing time of  Glauber dynamics for the Ising model on  $G(n,\kd/n)$ is also well-studied. 
It is known that  Glauber dynamics on typical instances of this distribution is rapidly mixing  for  any 
$\beta\in \UnIsing(\cconnective,\varepsilon)$ and any $\varepsilon>0$.   The ferromagnetic case follows from  the 
work of Mossel and Sly in   \cite{ISINGGNPMS}. The bounds for the antiferromagnetic Ising are implied by  the work of 
Efthymiou, Hayes, \v Stefankovi\v c and Vigoda in \cite{Efth19}, while  explicit bounds appear in   \cite{Efth24}.

As opposed to the previous works  which focus on  special families of graphs, our endeavour  here is to obtain 
optimum mixing bounds    {\em without} any assumptions about the underlying graph, i.e., we don't necessarily assume that
the underlying graph is amenable or a typical instance of $G(n,d/n)$.

As in the case of \Cref{thrm:HC4CCK}, we obtain \Cref{thrm:Ising4CCK} by establishing  
optimum mixing bounds expressed in terms of  $\nnorm{ (\NBMatrix)^N}{1/N}{2}$. In particular, we establish the following result.

\begin{theorem}\label{thrm:Ising4SPRadiusHash}
For $\varepsilon\in (0,1)$, $\lambda>0$, $\SingBound>1$ and integers $N>0$, $\maxDeg>1$, 
consider graph $G=(V,E)$ of maximum degree $\maxDeg$ such that $\nnorm{ (\NBMatrix)^N}{1/N}{2} =\SingBound$.
Let $\mu$ be the Ising model on $G$ with external field $\lambda$ and inverse temperature $\beta\in \UnIsing({\SingBound},\varepsilon)$. 

There is a constant {$C={C(\maxDeg, {\SingBound}, \varepsilon)}$} such that the mixing time of Glauber dynamics 
on $\mu$ is at most $C n\log n$. 
\end{theorem}

There is no hardness region for the ferromagnetic Ising model. This is due to the algorithm of 
Jerrum and Sinclair in \cite{JerSincIsing93}. However, we get the following hardness results   for the {\em antiferromagnetic} Ising model.

\begin{theorem}\label{thrm:HarndessTransIsingNBM}
Unless ${\rm NP} = {\rm RP}$, for $\varepsilon\in (0,1)$, $\maxDeg\geq 3$, $N, k\geq 1$ and $\SingBound > 1$ 
the following is true:

For any $0<\beta \leq \frac{\SingBound-1+\varepsilon}{\SingBound+1-\varepsilon}$ and $\lambda$ such that
we are not in the tree uniqueness region, there is no FPRAS for estimating 
the partition function of the {\em antiferromagnetic} Ising model
for graphs $G$ of maximum degree at most $\maxDeg$ such that $\nnorm{ (\NBMatrix)^N}{1/N}{2}=\SingBound$.

For any $\frac{\SingBound-1+\varepsilon}{\SingBound+1-\varepsilon}\leq \beta\leq 1 $ and any $\lambda>0$,
there is an FPRAS for estimating the partition function of the {\em antiferromagnetic} Ising model
for graphs $G$ of maximum degree at most $\maxDeg$ such that $\nnorm{ (\NBMatrix)^N}{1/N}{2}=\SingBound$.
\end{theorem}

\subsection{Results with backtracking walks}

We    explore the direction of using   {\em backtracking} walks to establish rapid mixing bounds, i.e., 
rather than  the  $k$-non-backtracking ones.  This approach naturally formalises in terms of the {\em adjacency matrix}  
$\Adjacency_G$.  
Recall that  $\Adjacency_G$ is a 0/1 matrix indexed by   the vertices of $G$ such that   for any  $\ku,\kw$,  we have 
\begin{align}\nonumber 
\Adjacency_G( \ku,\kw)= \Ind\{\textrm{$\ku, \kw$ are adjacent in $G$}\} \enspace.
\end{align} 
The rapid mixing bounds we obtain are  expressed in terms of $\| \Adjacency_G \|_2$, 
which is equal to the {\em spectral radius}  of the matrix.
Utilising  matrix $\Adjacency_G$ and its norms to obtain rapid mixing bounds has been considered before in the literature.   The seminal  work 
of Hayes in \cite{Hayes06} first introduces this approach by considering $\|\Adjacency_G\|_2$. Subsequently, further extensions 
were introduced   in \cite{DGJ09,Hayes07}.

 Compared to $\NBMatrix$, matrix $\Adjacency_G$ is very well-behaved, e.g., it is symmetric, irreducible, etc. This nice behaviour of
 $\Adjacency_G$ gives rise to a  more elegant analysis  than what we have with $\NBMatrix$.
In that regard, we obtain the following  result for the Hard-core model.

\begin{theorem}\label{thrm:HC4SPRadiusAdj}
For any $\varepsilon \in (0,1)$, for $\maxDeg, \aspradius>1$, consider graph $G=(V,E)$ 
of maximum degree $\maxDeg$ such that $\norm{ \Adjacency_G}{2}=\aspradius$. 
Also, let $\mu$ be the Hard-core model on $G$ with fugacity $\lambda\leq (1-\varepsilon)\lcritical(\aspradius)$.

There is a constant $C={C(\maxDeg, \aspradius, \varepsilon)}$ such that the mixing time of Glauber dynamics on $\mu$ 
is at most $C n\log n$.
\end{theorem}

The maximum degree bounds that apply in our setting are those from \cite{OptMCMCIS,VigodaSpectralIndB}.
 \Cref{thrm:HC4SPRadiusAdj} improves on these bounds  when  $\spradius(\Adjacency_G)<\maxDeg-1$. 
\Cref{tab:SpectralRadiusResults} shows  some notable families of graphs where \Cref{thrm:HC4SPRadiusAdj} 
 gives significant improvements over the maximum degree bounds. 
\Cref{thrm:HC4SPRadiusAdj} also  improves   on the  results  in   \cite{Hayes06}, which 
establishes  optimum mixing for the  Hard-core model with fugacity  $\lambda<\frac{1}{\aspradius}$. The bound 
we obtain here  is, roughly,   $\lambda<\frac{e}{\aspradius}$. That is, we get an improvement of  factor $e$ 
for the bounded degree graphs. 
Unlike our  approach that  utilises  the Spectral Independence method,  the  results    in \cite{DGJ09,Hayes06,Hayes07} 
rely on the {\em path coupling} technique  by Bubley and Dyer in  \cite{PathCoupling}.  

\begin{table}
    \centering
    \begin{tabular}{ c cc }
        \hline \\ [-10pt]
        \begin{tabular}{c}Graph with \\ \hline max Degree $\maxDeg$ \end{tabular}& & \begin{tabular}{c} Spectral radius of \\ \hline the Adjacency matrix  \end{tabular}  
	\\ [8pt] \hline \\  [-10pt]
        Planar & $\quad$&$\sqrt{8(\maxDeg-2)}+2\sqrt{3}$ 
        \\ [3pt]
        $r$-degenarate & & $ 2\sqrt{r(\maxDeg-r)} $ 
        \\ [3pt]
        Euler genus $g$ & & $\sqrt{8(\maxDeg-2g+4 )}+2g+4$   
        \\ [3pt]
        \hline  \\ [-8pt]
    \end{tabular}
    \caption{Bounds on the spectral radius of the Adjacency matrix, from  \cite{SPRadiusPlannar,Hayes06}}
            \label{tab:SpectralRadiusResults}
\end{table}

As far as the Ising model and the adjacency matrix are concerned, we obtain the following
result. 
\begin{theorem}\label{thrm:Ising4SPRadiusAdj}
For $\varepsilon\in (0,1)$, $\lambda>0$, $\aspradius>1$ and integer $\maxDeg >1$, 
consider graph $G=(V,E)$ of maximum degree $\maxDeg$ such that $\norm{\Adjacency_G}{2}=\aspradius$.
Let $\mu$ be the Ising model on $G$ with external field $\lambda$ and inverse temperature $\beta\in \UnIsing(\aspradius,\varepsilon)$. 

There is a constant $C={C( \maxDeg,\aspradius, \varepsilon)}$ such that the mixing time of Glauber dynamics 
on $\mu$ is at most $C n\log n$. 
\end{theorem}
The results we obtain in \Cref{thrm:Ising4SPRadiusAdj} for the Ising model match those in \cite{Hayes06}.

\subsubsection*{Notation} Consider  graph $G=(V,E)$, we a Gibbs distribution $\mu_G$ on 
$\{\pm 1\}^V$. 
For $\Lambda\subset V$,  we let $\mu_{\Lambda}$ denote the marginal of $\mu$ at set $\Lambda$.
We also let $\mu_G(\cdot \ |\ M, \sigma)$ and $\mu^{M,\sigma}_G$ denote the distribution $\mu$ conditional on 
the configuration at $M\subset V$ being $\sigma$. 
Sometimes, we use the term ``pinning" to refer to $(M,\sigma)$.  Also, we interpret the conditional
marginal $\mu_{\Lambda}(\cdot \ |\ M, \sigma)$ in the natural way. Similarly, for $\mu^{M,\sigma}_{\Lambda}$.

For vertex $\kw\in V$, we let $N_G(\kw)$ be the set of vertices which are adjacent to $\kw$ in $G$.


 \begin{figure}
 \begin{minipage}{.45\textwidth}
 \centering
		\includegraphics[width=.41\textwidth]{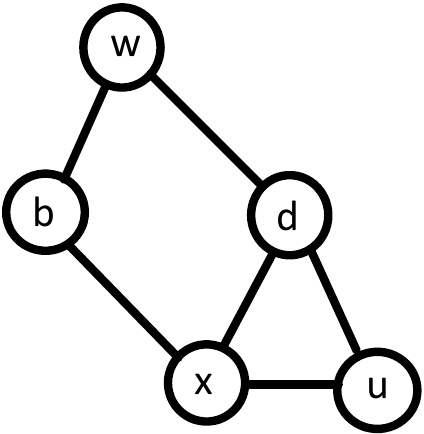}
		\caption{Initial graph $G$}
	\label{fig:Base4TSAWIntro}
\end{minipage}
 \begin{minipage}{.4\textwidth}
 \centering
	\includegraphics[height=3.1cm]{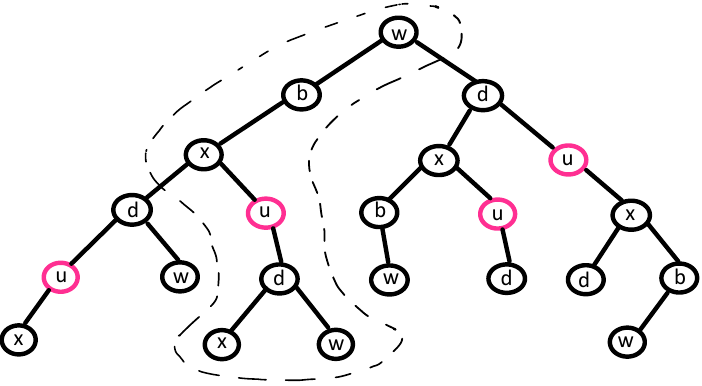}
		\caption{$\Tsaw(G, \kw)$}
	\label{fig:TSAWExampleIntro}
\end{minipage}
\end{figure}

\section{Approach \LastReviewG{2025-04-02} }\label{sec:Approach}

We utilise the Spectral Independence method to show our rapid mixing results.  
Central to this method is the notion of the   {\em pairwise influence matrix} $\infmatrix^{\Lambda,\tau}_{G}$.
Let us briefly visit this notion.   

Let graph $G=(V,E)$ and a Gibbs distribution $\mu$ on $\{\pm 1\}^V$.
For  $\Lambda\subset V$ and a configuration $\tau$ at $\Lambda$, 
the  $\infmatrix^{\Lambda,\tau}_{G}$ is a matrix  indexed by the vertices in 
$V\setminus \Lambda$ such that  for any $\kw,\ku\in V\setminus \Lambda$, we have
\begin{align}\label{def:InfluenceMatrixHighLevelA}
\infmatrix^{\Lambda,\tau}_{G}(\kw,\ku) &= \mu_{\ku }(+1\ |\ (\Lambda, \tau), (\kw, +1))- \mu_{\ku }(+1\ |\ (\Lambda, \tau), (\kw, -1)) 
\enspace.
\end{align} 
The Gibbs marginal $\mu_{\ku }(+1 \ |\ (\Lambda, \tau), (\kw, +1))$ indicates the probability that vertex $\ku$ gets $+1$, 
conditional on the configuration at $\Lambda$ being $\tau$ and the configuration at $\kw$ being $+1$. 
We have the analogous for the marginal  $\mu_{\ku }(+1 \ |\ (\Lambda, \tau), (\kw, -1))$.

For the cases we consider here if the maximum eigenvalue of  $\infmatrix^{\Lambda,\tau}_{G}$ is bounded, 
then the corresponding Glauber dynamics on $\mu$ has
 $O(n\log n)$  mixing time. This follows from  \cite{VigodaSpectralIndB}. 
The rapid mixing results  in \Cref{sec:Introduction} are obtained by establishing  such bounds on 
the maximum eigenvalue of  $\infmatrix^{\Lambda,\tau}_{G}$. In particular, we establish that the spectral radius 
$\spradius( \infmatrix^{\Lambda,\tau}_{G})$ is $O(1)$. 

Before showing how we obtain such bounds for $\spradius( \infmatrix^{\Lambda,\tau}_{G})$, let us visit 
the notion of the tree of self-avoiding walks that we use in the analysis with $\infmatrix^{\Lambda,\tau}_{G}$.

\subsubsection*{The $\Tsaw$ construction:}
It is a well-known fact that each entry $\infmatrix^{\Lambda,\tau}_{G}(\kw,\ku)$ can 
be expressed in terms of a topological construction called the {\em tree of self-avoiding walks} 
together with weights on the paths of this tree that we call {\em influences}. 
It is worth giving a high-level (hence imprecise) description of this relation. For a formal description,  
see \Cref{sec:TsawConstruction}.

Recall that a walk is {\em self-avoiding} if it does not repeat vertices. For vertex $\kw$ in $G$, we 
define $T=\Tsaw(G,\kw)$, the tree of self-avoiding walks starting from $\kw$, as follows: Consider 
the set of walks $\kv_0, \ldots, \kv_{\kr}$ in graph $G$ that emanates from vertex $\kw$, i.e., 
$\kv_0=\kw$, while one of the following two holds
\begin{enumerate}
\item $\kv_0, \ldots, \kv_{\kr}$ is a self-avoiding walk,
\item $\kv_0, \ldots, \kv_{\kr-1}$ is a self-avoiding walk, while there is $\kj\leq \kr-3$ such that $\kv_{\kr}=\kv_{\kj}$.
\end{enumerate}
Each of these walks  corresponds to a vertex in  $T$. Two vertices in this tree are adjacent  if the 
corresponding walks are adjacent, i.e., one extends the other by one vertex.  The vertex $\kz$ in 
$T$ corresponding to walk $\kv_0, \ldots, \kv_{\kr}$ in $G$ is called a {\em copy} of vertex $\kv_{\kr}$ 
in  tree $T$.

\Cref{fig:TSAWExampleIntro} shows an application of the above construction on the  graph shown in 
\Cref{fig:Base4TSAWIntro}.  The tree of self-avoiding walks in this example starts from vertex $\kw$. 
The vertices with label $\kv$ in the tree  correspond to the copies of  vertex $\kv$ of the initial 
graph $G$,  e.g., all the highlighted vertices in \Cref{fig:TSAWExampleIntro} are copies of vertex $\ku$.

Each path in $T=\Tsaw(G,\kw)$ is associated with a real number that we call {\em influence}. 
The influences depend on the parameters of the Gibbs distribution $\mu^{\Lambda,\tau}$. 
For the path of length $1$, which corresponds to  edge $\ke$ in $T$, we let $\infweight(\ke)$ denote 
its influence. For a path $\kP$ of length $>1$, the influence $\infweight(\kP)$ is given by
\begin{align}\nonumber
\infweight(\kP) &= \prod\nolimits_{\ke\in \kP}\infweight(\ke) \enspace,
\end{align}
i.e., $\infweight(\kP)$ is equal to the product  of influences of the edges in this path\footnote{In the 
related literature, influences are defined w.r.t. the vertices of the tree, not the edges. Our formalism 
here is equivalent.}.

The entry $\infmatrix^{\Lambda,\tau}_{G}(\kw,\ku)$ can be expressed
as a sum of influences over paths in $T$, i.e., 
\begin{align}\nonumber %
\infmatrix^{\Lambda,\tau}_{G}(\kw,\ku) & =\sum\nolimits_{\kP} \infweight(\kP) \enspace,
\end{align}
where $\kP$ varies over the paths from the root to the copies of vertex $\ku$ in $T$. 
The $\Tsaw$-construction is quite helpful for the analysis as it allows us to   compute 
the entries of $\infmatrix^{\Lambda,\tau}_{G}$ recursively.

\subsubsection*{Bounds from matrix norms for $\spradius(\infmatrix^{\Lambda,\tau}_{G})$:}
We use  matrix norms\footnote{Unless otherwise specified, the norm $\norm{\UpM}{p}$ for matrix $\UpM$ corresponds to the induced $\ell_p$ to $\ell_p$ matrix norm, i.e., 
$\norm{M}{P}=\max_{\norm{{\bf x}}{p}=1}\norm{\UpM{\bf x}}{p}$. } to bound the spectral radius $\spradius(\infmatrix^{\Lambda,\tau}_{G})$. 

Let us start by considering  the results for the Hard-core model,  expressed in terms of the adjacency matrix $\Adjacency_G$.  
From  standard  linear algebra, we have 
\begin{align*}
\spradius\left( \infmatrix^{\Lambda,\tau}_{G} \right) &\leq \norm{ \UpD^{-1} \cdot \infmatrix^{\Lambda,\tau}_{G} \cdot \UpD }{\infty} \enspace,
\end{align*}
for any non-singular matrix $\UpD$ such that $\UpD$ and $\infmatrix^{\Lambda,\tau}_{G}$ are conformable 
for multiplication.

Choosing $\UpD$ to be the identity matrix, one recovers the  norm used to obtain the max-degree 
bounds in  \cite{VigodaSpectralInd}. Here, we make a non-trivial choice for this matrix. We set $\UpD$ 
to be diagonal, such that $\UpD(\kv,\kv)=\left( \eigenv_1(\kv)\right)^{1/\pfs}$,  with $\eigenv_1$ being the 
principal eigenvector of $\Adjacency_G$. The quantity $\pfs$ is a parameter of the potential function 
we use in the analysis.  Typically, we have $\pfs>1$.

Since we assume that $G$ is connected, matrix  $\Adjacency_G$ is {\em irreducible}. 
The Perron-Frobenius theorem implies that all entries in $\eigenv_1$ are strictly positive; 
hence, $\UpD$ is non-singular.

The estimation of $\norm{\UpD^{-1} \cdot \infmatrix^{\Lambda,\tau}_{G} \cdot \UpD }{\infty}$ is carried 
out by exploiting the $\Tsaw$-constructions for $\infmatrix^{\Lambda,\tau}_{G}$. Specifically, for each 
vertex $\kw$ of  graph $G$, we estimate the weighted influence of the root of $\Tsaw(G,\kw)$ to the 
vertices at each level $\kh$ of the tree. Note that the influences are weighted according to the entries of matrix $\UpD$. 
To illustrate this,  consider the example in  \Cref{fig:WeightedInfluences}.  Instead of just calculating the total influence of the
root to the vertices at level $4$ in  $\Tsaw(G,\kw)$, in this example, we have the extra weight 
$\frac{\UpD(\ku,\ku)}{\UpD(\kw,\kw)}$ on the  influence from the root to each  copy of vertex $\ku$ at 
level $4$ in the tree. 
This calculation gives rise to some  natural tree recursions.

Utilising potential functions, we subsequently  obtain 
\begin{align}
 \norm{ \UpD^{-1} \cdot \infmatrix^{\Lambda,\tau}_{G} \cdot \UpD }{\infty} \leq C\cdot \sum\nolimits_{\kell\geq 0} 
 \left( \delta^{\kell}\cdot \norm{ \Adjacency^{\kell}_G}{2} \right)^{1/\pfs}\enspace, \nonumber
\end{align}
where $C>0$ is a bounded number, while  $\delta$ depends on the parameters of 
the Gibbs distribution. 
Specifically, for any $\lambda<\lcritical(\norm{\Adjacency_G}{2})$, we show that $\delta\leq \frac{1-\epsilon}{\norm{\Adjacency_G}{2}}$, 
where $\epsilon$ increases with the  distance of $\lambda$ from $\lcritical(\norm{\Adjacency_G}{2})$. 
For such $C,\delta$ and $\pfs$,  we  obtain that $\spradius(\infmatrix^{\Lambda,\tau}_{G})$ is bounded.

We would like to use a similar idea  for our results with $\NBMatrix$. It turns out that the weighting we
need to use in this case is more involved. As illustrated in  the example for $\Adjacency_G$  in
\Cref{fig:WeightedInfluences},  all copies of vertex $\ku$ in tree $\Tsaw(G,\kw)$ get the same weight
$\frac{\UpD(\ku,\ku)}{\UpD(\kw,\kw)}$.  Working with $\NBMatrix$,  the weight of each copy of vertex
$\ku$ needs to  also depend on the path that connects this vertex   with the root of the tree. This means
that the various copies of  vertex $\ku$ in $\Tsaw(G,\kw)$ get different weights. This somehow precludes 
the use of norm $\norm{ \UpD^{-1} \cdot \infmatrix^{\Lambda,\tau}_{G} \cdot \UpD }{\infty}$ (or a similar norm) for our
endeavours. 
To this end, rather than using $\infmatrix^{\Lambda,\tau}_{G}$, we introduce a new influence matrix to work with that
we call the {\em extended  influence matrix}  $\ExtdInfMatrixF$.

Let us take things from the beginning. 
For brevity, set  $\nnorm{(\NBMatrix)^N}{1/N}{2}=\SingBound$, for some choice 
of $\kk,N\geq 1$, while assume  that  $\lambda<\lcritical(\SingBound)$. 
We obtain our results for the Hard-core model  in terms of $\NBMatrix$
by showing 
\begin{align}\label{eq:2SNormCIOverviewH}
\spradius( \infmatrix^{\Lambda,\tau}_{G} ) & \leq C\cdot \sum\nolimits_{\kell\geq 0} \left( \delta^{\kell}
\cdot \norm{( \NBMatrix)^{\kell}}{2} \right)^{1/\pfs}\enspace,
\end{align}
where  $C>0$ is a constant,  $\pfs> 1$ is a parameter of the potential functions we are using  in the analysis 
and  $\delta\leq \frac{1-\epsilon}{\SingBound}$,  where $\epsilon>0$ increases with the  distance of 
$\lambda$ from $\lcritical(\SingBound)$. 
For such $C,\delta$ and $\pfs$,  it is not hard to show that the above inequality implies
that $\spradius( \infmatrix^{\Lambda,\tau}_{G})$ is $O(1)$. We have to note, though, that this is not 
entirely straightforward as  the value of  the quantity $\nnorm{( \NBMatrix)^{\kell}}{1/\kell}{2}$  varies with $\kell$.

We introduce the notion of  the extended  influence matrix $\ExtdInfMatrixF$ to establish \eqref{eq:2SNormCIOverviewH}.  
This  is a matrix induced by $G$ and the Gibbs distribution $\mu^{\Lambda,\tau}_G$, while it 
is indexed by the self-avoiding walks of length $\kk$ in $G$ that  do not  intersect with $\Lambda$.

We think of $\ExtdInfMatrixF$ as  a matrix that expresses influences with respect to the Gibbs distribution 
$\mu^{\Lambda,\tau}_G$, but in a more refined way than $\infmatrix^{\Lambda,\tau}_{G}$.
Intuitively,  $\infmatrix^{\Lambda,\tau}_{G}(\kw,\ku)$ accounts for   the amount of information that propagates from 
 vertex $\kw$ to vertex $\ku$ over the paths of graph $G$. 
 In that respect,   
for $\kP$ and $\kQ$, two self-avoiding walks of length $\kk$ that emanate from $\kw$ and $\ku$, respectively, 
the entry $\ExtdInfMatrixF(\kP,\kQ)$ expresses  the amount of information propagating  from $\kw$  to $\ku$ 
over the paths of  graph $G$  that overlap with $\kP$ and $\kQ$. 

In what follows,  we give a more concrete description of the extended influence matrix.

 \begin{figure}
 \begin{minipage}{.65\textwidth}
 \centering
		\includegraphics[height=4cm]{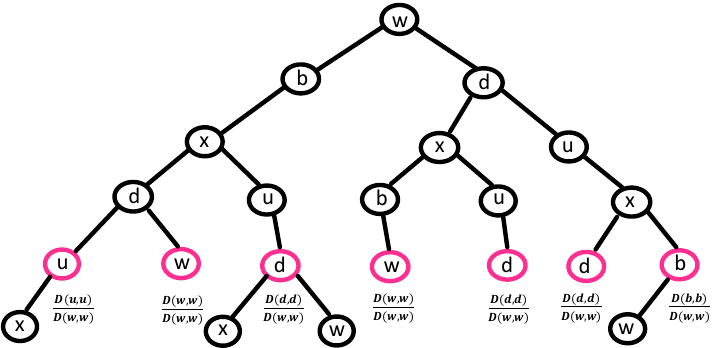}
		\caption{Weighted Influences}
	\label{fig:WeightedInfluences}
\end{minipage}
 \begin{minipage}{.4\textwidth}
\end{minipage}
\end{figure}

\subsubsection*{The extended influence matrix $\ExtdInfMatrixF$:} 
The definition of $\ExtdInfMatrixF$ can get  too technical for this high-level exposition. 
For this reason,  we focus on the simplest case by considering the parameter $\kk=1$.
For a formal definition of $\ExtdInfMatrixF$,  see \Cref{sec:ExtInfluenceMatrixNew}.

To define the extended influence matrix, we first  introduce the notion of the {\em extension}.
 An extension is an operation that we typically apply to a  graph $G$ or a Gibbs distribution on this graph. 
We start with the description of the graph 
extension. To keep the description simple, assume that we are given the  graph $G$ shown in  
\Cref{fig:Base4TSAWIntro}. W.l.o.g. assume that there is a total ordering of the vertices in $G$.

Every extension of graph $G$ is  specified with respect to a self-avoiding walk $\kP$  of length $\kk$ 
in this graph (here $\kk=1$). The resulting graph is called the $\kP$-extension of $G$ and is 
denoted as $\gext{G}{\kP}$. For our example, let $\kP=\kw\kb$. Then, $\gext{G}{\kP}$ is a new 
graph obtained by applying the following operations on $G$:
\begin{enumerate}[(a)]
\item remove all the edges that are incident to vertex $\kw$, apart from the one that connects 
$\kw$  and $\kb$,
\item for each neighbour $\kz$ of vertex $\kw$ that has been disconnected, insert a new vertex 
$\kw\kz$ and connect it only to vertex $\kz$.
\end{enumerate}
\Cref{fig:ADExtGNew} shows the resulting graph after applying the above operations to graph $G$ in 
\Cref{fig:Base4TSAWIntro}.

Compared to $G$, graph $\gext{G}{\kP}$ has an extra vertex which we denote with two letters, i.e.,
vertex $\kw\kd$. This new vertex is called  {\em split-vertex}. Split-vertices are important when  
we consider the extensions of Gibbs distributions.

Let us now describe the extensions of Gibbs distribution $\mu^{\Lambda,\tau}_G$.  As in the case
of graph $G$, the  extension of  $\mu^{\Lambda,\tau}_G$ is always specified with respect to a 
self-avoiding walk $\kP$ of the underlying graph $G$, while  walk $\kP$ is not allowed to intersect
with $\Lambda$, the set of pinned vertices. 

The $\kP$-extension of  $\mu^{\Lambda,\tau}_G$ is a new Gibbs distribution with underlying graph $\gext{G}{\kP}$, i.e.
the $\kP$-extension of $G$, and it has the same specification as   $\mu^{\Lambda,\tau}_G$, i.e., the same parameters
$\beta,\gamma$ and $\lambda$.  

Letting $\mu^{M,\sigma}_{G_{\kP}}$ be the $\kP$-extension of  $\mu^{\Lambda,\tau}_G$, 
the set of pinned vertices $M$ corresponds to the union of the vertices in $\Lambda$ and the split-vertices of graph $G_{\kP}$. 
The pinning $\sigma$ is specified as follows: for any $\kv\in \Lambda$, we have
$\sigma(\kv)=\tau(\kv)$, while for each  split-vertex $\kv\kz$, we have

\begin{align}
\sigma(\kv \kz) &= \left\{
\begin{array}{lcl}
+1& \quad & \textrm{if $\kx>\kz$}\enspace,\\ 
-1& \quad &\textrm{if $\kx<\kz$} \enspace,
\end{array}
\right. 
\end{align}
where $\kx$ is the vertex after $\kv$ in path $\kP$. The comparison between $\kx$ and $\kz$ is with respect to the
total ordering of the vertices in $G$. 

Before concluding the  description of the extension of  $\mu^{\Lambda,\tau}_G$, we note  that the pinning of 
the split-vertices  is reminiscent of the pinning we have in Weitz's $\Tsaw$-construction in \cite{Weitz}.

For two self-avoiding paths $\kP,\kQ$ in $G$, each of length $\kk= 1$, which do not intersect with each
other, we  also consider the $\{\kP,\kQ\}$-extension  of graph $G$. We obtain this graph by first getting
$\gext{G}{\kP}$,  i.e., the $\kP$-extension of $G$, and then we take the $\kQ$-extension of graph $G_{\kP}$.
Similarly, we obtain the $\{\kP,\kQ\}$-extension of the Gibbs distribution $\mu^{\Lambda,\tau}_G$. That is, 
we first take the $\kP$-extension of $\mu^{\Lambda,\tau}_G$ and then, apply the $\kQ$-extension to the
resulting Gibbs distribution.

 The order in which we consider $\kP$ and $\kQ$ to obtain the $\{\kP,\kQ\}$-extension of graph $G$  does not matter. 
As long as $\kP$ and $\kQ$ do not intersect, one obtains the same  graph $\gext{G}{\kP,\kQ}$ 
 by either taking the $\kP$-extension of $\gext{G}{\kQ}$ or the $\kQ$-extension of $\gext{G}{\kP}$. Similarly
 for $\mu^{\Lambda,\tau}_G$.

 \begin{figure}
 \begin{minipage}{.35\textwidth}
 \centering
		\includegraphics[width=.38\textwidth]{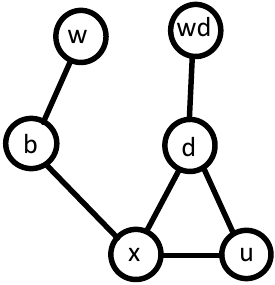}
		\caption{ $\kw\kb$-extension}
	\label{fig:ADExtGNew}
\end{minipage}
 \begin{minipage}{.41\textwidth}
 \centering
		\includegraphics[width=.42\textwidth]{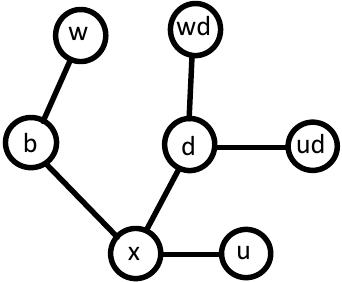}
		\caption{$\{\kw\kb,\ku\kx\}$-extension}
	\label{fig:ABFEExtGNew}
\end{minipage}
\end{figure}

We now proceed with the definition of the extended influence matrix $\ExtdInfMatrixF$.  
Setting $\kk=1$,   matrix $\ExtdInfMatrixF$ is indexed by  self-avoiding walks of length $1$
(ordered pairs of adjacent vertices) in $V\setminus \Lambda$.

For $\kP$ and $\kQ$ self-avoiding walks in graph $G$, each of length $\kk=1$, which do not intersect with each other 
and do not intersect with $\Lambda$,  for matrix $\ExtdInfMatrix=\ExtdInfMatrixF$, the following holds:
 letting $\zeta^{\kP,\kQ}$ be the $\{\kP,\kQ\}$-extension 
of $\mu^{\Lambda,\tau}_G$,  we have
\begin{align}
\ExtdInfMatrix(\kP,\kQ)&=\zeta^{\kP,\kQ}_{\ku}(+1\ |\  (\kw,+1))-\zeta^{\kP,\kQ}_{\ku}(+1\ | \ (\kw,-1))\enspace,
\end{align}
where $\kw$ and $\ku$ are starting  vertices of walks $\kP,\kQ$, respectively.

Letting $\infmatrix^{\kP,\kQ}$ the influence matrix of $\zeta^{\kP,\kQ}$, the $\{\kP,\kQ\}$-extension of $\mu^{\Lambda,\tau}_G$, 
 the above implies that
\begin{align}
\ExtdInfMatrix(\kP,\kQ) &=\infmatrix^{\kP,\kQ}(\kw,\ku)\enspace,
\end{align}
where, as mentioned earlier,  $\kw$ and $\ku$ are the starting  vertices of walks $\kP,\kQ$, respectively.

The above definition implies that each entry of $\ExtdInfMatrix$   is a standard influence, i.e.,  in 
the sense of the influences we have with the entries in the pairwise influence matrix. This allows us to, e.g., utilise the $\Tsaw$-construction in   the analysis with 
$\ExtdInfMatrix$ and exploit  nice algebraic properties that emerge for this  matrix.

\subsubsection*{Connections between $\ExtdInfMatrixF$ with $\infmatrix^{\Lambda,\tau}_{G}$:}
It is  interesting to compare the entries of matrix $\ExtdInfMatrix=\ExtdInfMatrixF$ with those in $\infmatrix=\infmatrix^{\Lambda,\tau}_{G}$.

Consider a Gibbs distribution $\mu_G$ on the graph $G$ shown in \Cref{fig:Base4TSAWIntro}. As before, consider 
$\kk=1$, while for simplicity, let $\Lambda=\emptyset$.  Consider walks $\kP=\kw\kb$ and $\kQ=\ku\kx$ in $G$.
It is natural to compare the entry of the extended influence matrix $\ExtdInfMatrix(\kP,\kQ)$ with   entry  
$\infmatrix(\kw,\ku)$. Note that the entry in $\infmatrix$ is specified by the starting vertices of $\kP$ and $\kQ$, respectively. 

Consider the $\Tsaw$-construction for $\infmatrix(\kw,\ku)$ shown in \Cref{fig:TSAWExampleIntro}.  Then,  
as argued earlier, $\infmatrix(\kw,\ku)$ corresponds to  the sum of influences  over the paths from the root to the copies 
of vertex $\ku$ in this tree, i.e., the highlighted vertices.

Recalling that $\kP=\kw\kb$ and $\kQ=\ku\kx$,  the entry $\ExtdInfMatrix(\kP,\kQ)$ can be described in terms 
of the same tree too. It is  {\em approximately equal} to the sum of  influences  over the paths from the 
root to the copies of $\ku$,  with the {\em extra  constraint} that  the second vertex of each path
is a copy of vertex $\kb$, while the vertex prior  to the last one is a copy of vertex $\kx$. 
This corresponds to  the influence from the root to the highlighted  vertex in the subtree encircled by the dotted line in \Cref{fig:TSAWExampleIntro}. 

Using a line of arguments similar to what we describe  above, we   show that there exists a constant $C_0>0$ such that 
\begin{align}\nonumber
\norm{ \infmatrix}{2}  \leq C_0 \cdot \norm{\abs{\ExtdInfMatrix}}{2} \enspace.
\end{align}
where $\abs{\ExtdInfMatrix}$ is the matrix with entries $\abs{\ExtdInfMatrix(\kP,\kQ)}$. 

From the above, we get \eqref{eq:2SNormCIOverviewH} by showing that  there exists a constant $C_1>0$
such that
\begin{align}\nonumber
\norm{\abs{\ExtdInfMatrix}}{2} &\leq C_1\cdot \sum\nolimits_{\kell\geq 0} \left( \delta^{\kell}
\cdot \norm{( \NBMatrix)^{\kell}}{2} \right)^{1/\pfs}\enspace.
\end{align}
where the quantities $\delta,\pfs$ were specified when we introduce \eqref{eq:2SNormCIOverviewH}.

\subsubsection*{Perturbed influences in $\Tsaw$:}
The analysis we use to calculate norms of $\ExtdInfMatrix=\ExtdInfMatrixF$ is not too different from what we had 
for the case of $\Adjacency_G$ and the influence matrix $\infmatrix^{\Lambda,\tau}_G$. 
Even though we do not use norms like $\norm{\UpD^{-1} \cdot \abs{\ExtdInfMatrix}\cdot \UpD}{\infty}$ in the analysis, 
we set up tree 
recursions based on the $\Tsaw$-construction with appropriate weighting of the influences.

A point that requires some care is that the definition of $\ExtdInfMatrix$ implies that  the Gibbs distribution we 
consider at each entry  $\ExtdInfMatrix(\kP,\kQ)$ depends on the choice of $\kP$ and $\kQ$. Hence, different 
entries might consider different Gibbs distributions.  This implies that when we calculate the influence of the root 
to the vertices at  level $\kh$ of the tree of self-avoiding walks, the influence of each path comes from a different 
Gibbs distribution.  This observation is   important in the analysis and could cause problems when  we use amortisation 
 and potential  functions.  

We circumvent the problem by establishing that the corresponding influences that arise from the different Gibbs
distributions can be viewed as one being a small perturbation of the  other. We establish this by utilising a  {\em correlation decay} 
argument.

\spreadpoint

\section{Spectral Bounds for $\infmatrix^{\Lambda,\tau}_{G}$ from likelihood  ratios \LastReviewG{2025-03-05}}\label{sec:RecursionVsSpectralIneq}

Consider tree $T=(V,E)$ of maximum degree at most $\maxDeg$, while let vertex $\kr$ be the root. 
Let the Gibbs distribution $\mu$ on $\{\pm 1\}^{V}$ be specified as in \eqref{def:GibbDistr} with respect to  parameters $\beta, \gamma$ 
and $\lambda$.

For $\kK \subseteq V\setminus\{\kr\}$ and $\tau\in \{\pm 1\}^{\kK}$, let the {\em ratio of  marginals} at the root 
$\gratio^{\kK, \tau}_{\kr}$ be defined by 
\begin{align}\label{eq:DefOfR}
\gratio^{\kK, \tau}_{\kr}=\frac{\mu^{\kK,\tau}_{\kr}(+1)}{\mu^{\kK,\tau}_{\kr}(-1)} \enspace.
\end{align}
Recall that $\mu^{\kK,\tau}_{\kr}(\cdot)$ denotes the marginal of $\mu^{\kK,\tau}(\cdot)$ 
at the root $\kr$. The above allows for $\gratio^{\kK, \tau}_{\kr}=\infty$ when 
$\mu^{\kK,\tau}_{\kr}(-1)=0$.

For vertex $\ku\in V$, we let $T_{\ku}$ be the subtree of $T$ that includes $\ku$ and all its descendants. 
We always assume that the root of $T_{\ku}$ is vertex $\ku$. 
With a slight abuse of notation, we let $\gratio^{\kK, \tau}_{\ku}$ denote the ratio of marginals at the root for 
the subtree $T_{\ku}$, where the Gibbs distribution is with respect to $T_{\ku}$, while we impose 
pinning $\tau(\kK\cap T_{\ku})$.

Suppose that  the vertices $\kv_1, \ldots, \kv_{\kd}$ are the children of the root $\kr$, for an integer $\kd>0$. 
It is standard to express $\gratio^{\kK, \tau}_{\kr}$ in terms of $\gratio^{\kK, \tau}_{\kv_{\ki}}$'s 
by having $\gratio^{\kK, \tau}_{\kr}=\trecur_{\kd} (\gratio^{\kK, \tau}_{\kv_1},  \ldots, \gratio^{\kK, \tau}_{\kv_{\kd}} )$, where
\begin{align}\label{eq:BPRecursion}
\trecur_{\kd}:[0, +\infty]^{\kd}\to [0, +\infty] & &\textrm{s.t.}& &
 (\kx_1, \ldots, \kx_{\kd})\mapsto \lambda \prod^{\kd}_{\ki=1}\frac{\beta {\kx}_{\ki}+1}{{\kx}_{\ki}+\gamma} \enspace.
\end{align}
To get cleaner results in the analysis, we work with log-ratios of 
Gibbs marginals. Let $\logtrecur_{\kd}=\log \circ F_{\kd} \circ \exp$, which means that 
\begin{align}\label{eq:DefOfH}
\logtrecur_{\kd}:[-\infty, +\infty]^{\kd}\to [-\infty, +\infty] &&\textrm{s.t.}&& (\kx_1, \ldots, \kx_{\kd})\mapsto \log \lambda+\sum^{\kd}_{\ki=1}
\log\left( \frac{\beta \exp(\kx_{\ki})+1}{\exp(\kx_{\ki})+\gamma} \right) \enspace.
\end{align}
From \eqref{eq:BPRecursion}, it is elementary to verify that $\log \gratio^{\kK, \tau}_{\kr}=\logtrecur_{\kd}(\log \gratio^{\kK, \tau}_{\kv_1}, 
\ldots, \log \gratio^{\kK, \tau}_{\kv_{\kd}})$.

\newcommand{\mybrO}{[}
\newcommand{\mybrC}{]}

We let the function 
\begin{align}\label{eq:DerivOfLogRatio}
\dlogtrecur:[-\infty,+\infty]\to \mathbb{R}&& \textrm{s.t. }&& 
\kx \mapsto -\frac{(1-\beta\gamma)\cdot \exp(\kx)}{(\beta \exp(\kx)+1)(\exp(\kx)+\gamma)} \enspace.
\end{align}
For any $\ki\in [\kd]$, we have $\frac{\partial}{\partial \kx_{\ki}}\logtrecur_{\kd}(\kx_1, \ldots, \kx_{\kd})=\dlogtrecur(\kx_{\ki})$, 
where recall that set $[\kd]=\{1,\ldots, \kd\}$.
Furthermore, let the interval $\ratiorange_{\kd} \subseteq \mathbb{R}$ be defined by 
\begin{align} \nonumber 
\ratiorange_{\kd} &= \left \{
\begin{array}{lcl}
\mybrO \log (\lambda\beta^{\kd}), \log(\lambda/\gamma^{\kd}) \mybrC & \quad & \textrm{if $\beta\gamma<1$} \enspace, \\ \vspace{-.3cm } \\
\mybrO \log (\lambda/\gamma^{\kd}), \log(\lambda\beta^{\kd}) \mybrC & \quad & \textrm{if $\beta\gamma>1$}\enspace. 
\end{array}
\right . 
\end{align}
Standard algebra implies that $\ratiorange_{\kd}$ contains all the log-ratios for a vertex with $\kd$ children. 
Also, let 
\begin{align}\label{eq:DefOfRatiorange}
 \ratiorange &= \bigcup\nolimits_{1\leq \kd<\maxDeg}\ratiorange_{\kd} \enspace.
\end{align}
The set $\ratiorange$ contains all log-ratios in the tree $T$.

\subsection{First attempt}\label{sec:SIFirstAttempt}
We present the first set of our results, which we use to establish spectral independence. These results rely on the notion of 
$\delta$-contraction.

\begin{definition}[$ \delta$-contraction]\label{def:HContraction}
Let $\delta\geq 0$ and the integer $\maxDeg\geq 1$. Also, let $\beta,\gamma, \lambda\in \mathbb{R}$ be
such that $0\leq \beta\leq \gamma$, 
$\gamma >0$ and $\lambda>0$. 
The set of functions $\{ \logtrecur_d\}_{1\leq \kd <\maxDeg}$, defined as in \eqref{eq:DefOfH}
with respect to the parameters $\beta,\gamma,\lambda$,
exhibits {\em $\delta$-contraction } if it satisfies the following condition: 

For any $1\leq \kd < \maxDeg$ and any $({\bf \ky}_1, \ldots, {\bf \ky}_{\kd})\in [-\infty,+\infty]^{\kd}$, we have 
$\norm{\nabla \logtrecur_{\kd} ({\bf \ky}_1, \ldots, {\bf \ky}_{\kd})}{\infty} \leq \delta.$
\end{definition}

The $\delta$-contraction condition is equivalent to having $\dlogtrecur(\kz)\leq \delta$, for any $\kz\in [-\infty,+\infty]$, where 
$\dlogtrecur(\kz)$ is defined in \eqref{eq:DerivOfLogRatio}.

\begin{theorem}[Adjacency Matrix] \label{thrm:AdjacencyInfinityMixing}
Let $\varepsilon\in (0,1)$, $\aspradius > 1$ and the integer $\maxDeg > 1$. Also, let $\beta,\gamma, \lambda\in \mathbb{R}$ 
be such that $0\leq \beta\leq \gamma$, $\gamma >0$ and $\lambda>0$. 

Consider graph $G=(V,E)$ of maximum degree $\maxDeg$ such that   $\norm{\Adjacency_G}{2}=\aspradius$. 
 Also, consider $\mu$ the Gibbs distribution on $G$, specified by the parameters $(\beta, \gamma, \lambda)$.

For $\delta =\frac{1-\varepsilon}{\aspradius}$, suppose  the set of functions  $\{ \logtrecur_{\kd}\}_{1\leq \kd< \maxDeg}$ 
specified by $(\beta,\gamma,\lambda)$ exhibits $\delta$-contraction. For any  $\Lambda\subset V$ and 
$\tau\in \{\pm 1\}^{\Lambda}$, the influence matrix  $\infmatrix^{\Lambda,\tau}_{G}$ induced by $\mu$ satisfies 
\begin{align}\nonumber 
\textstyle \spradius \left( \infmatrix^{\Lambda,\tau}_{G} \right) &\leq \varepsilon^{-1} \enspace.
\end{align}
\end{theorem}

The proof of \Cref{thrm:AdjacencyInfinityMixing} appears in \Cref{sec:thrm:AdjacencyInfinityMixing}.

\begin{definition}[$b$-marginal boundedness]
For a number $b\geq 0$, we say that a distribution $\mu$ over $\{\pm 1\}^V$ is $b$-marginally bounded if  for every 
$\Lambda\subset V$ and any $\tau\in \{\pm 1\}^{\Lambda}$,  we have the following: for any  $\ku\in V\setminus \Lambda$ 
and for any $\kx\in\{\pm 1\} $ which is in the support of $\mu_{\ku}(\cdot \ |\ \Lambda, \tau)$, we have
$\mu_{\ku}(\kx\ |\ \Lambda, \tau)\geq b$. 
\end{definition}

For matrix $\NBMatrix$ we have the following result.

\begin{theorem}[$k$-non-backtracking matrix] \label{thrm:NonBacktrackingInfinityMixing}
Let $\varepsilon, b\in (0,1)$, $\SingBound>1$ and the integers $k,N \geq 1$, $\maxDeg > 1$. 
Also, let $\beta,\gamma, \lambda\in \mathbb{R}$ be such that $0\leq \beta\leq \gamma$, 
$\gamma >0$ and $\lambda>0$.

Consider graph $G=(V,E)$ of maximum degree $\maxDeg$ such that  $\nnorm{ (\NBMatrix)^N}{1/N}{2} =\SingBound$.
Assume that $\mu$, the Gibbs distribution on $G$ specified by the parameters  
$(\beta, \gamma, \lambda)$, is $b$-marginally bounded.

For $\delta =\frac{1-\varepsilon}{\SingBound}$, suppose  the set of functions  
$\{\logtrecur_d\}_{1\leq \kd< \maxDeg}$  specified by $(\beta,\gamma,\lambda)$ exhibits 
$\delta$-contraction. Then,  there exists a bounded number $\widehat{C}=\widehat{C}(\maxDeg, \SingBound, N, b)$ s
uch that for any $\Lambda\subset V$ and $\tau\in \{\pm 1\}^{\Lambda}$, 
the influence matrix  $\infmatrix^{\Lambda,\tau}_{G}$ induced by $\mu$ satisfies
\begin{align}\nonumber 
\textstyle \spradius \left( \infmatrix^{\Lambda,\tau}_{G} \right) &\leq \widehat{C} \cdot \varepsilon^{-1} \enspace.
\end{align}
\end{theorem}

The proof of \Cref{thrm:NonBacktrackingInfinityMixing} appears in \Cref{sec:thrm:NonBacktrackingInfinityMixing}.

We use \Cref{thrm:AdjacencyInfinityMixing,thrm:NonBacktrackingInfinityMixing} to show our results for the Ising model
in \Cref{thrm:Ising4SPRadiusAdj} and \Cref{thrm:Ising4SPRadiusHash}, respectively. 

Using \Cref{thrm:AdjacencyInfinityMixing}    one retrieves the rapid mixing  results for the Hard-core model in \cite{Hayes06}. 
To get  improved results for the  Hard-core model, we utilise potential functions. The approach with the potential functions is
described in the subsequent section.

\subsection{Second Attempt}\label{sec:SISecondAttempt}

We introduce a second set of results by utilising the notion of the potential functions in the analysis.
It is worth mentioning that we use the sharp results for potential functions from \cite{ConnectiveConst}.

\begin{definition}[$(\pfs,\delta, c)$-potential]\label{def:GoodPotential}
Let $s\geq 1$, $\delta,c>0$ and let the integer $\maxDeg> 1$. 
Also, let $\beta,\gamma, \lambda\in \mathbb{R}$ be such that $0\leq \beta\leq \gamma$, 
$\gamma >0$ and $\lambda>0$.

Consider $\{ \logtrecur_{\kd}\}_{1\leq \kd< \Delta}$, defined in \eqref{eq:DefOfH} with respect to 
$(\beta,\gamma,\lambda)$. The differentiable, strictly increasing function $\potF$, with image $S_{\potF}$, is called $(\pfs, \delta, c)$-potential 
if it satisfies the following two conditions: 
\begin{description}
\item[Contraction]
For $1\leq \kd< \maxDeg$, for $(\widehat{\bf \ky}_1, \ldots, \widehat{\bf \ky}_{\kd})\in (S_{\potF})^{\kd}$, 
 and ${\bf m}=({\bf m}_1, \ldots, {\bf m}_{\kd})\in \mathbb{R}^{\kd}_{ \geq 0}$ we have that 
\begin{align} \label{eq:contractionRelationPF}
 \xdpotF( \logtrecur_d({\bf \ky}_1, \ldots, {\bf \ky}_{\kd}) ) \cdot \sum^{\kd}_{j=1}
\frac{ \abs{ \dlogtrecur\left({\bf \ky}_j \right)}}{ \xdpotF\left( {\bf \ky}_j \right)} \cdot {\bf m}_j \leq \delta^{\frac{1}{\pfs}} \cdot \norm{ {\bf m}}{\pfs} 
\enspace ,
\end{align}
where $\xdpotF=\potF'$, ${\bf \ky}_j=\potF^{-1}(\widehat{\bf \ky}_j)$, while $\dlogtrecur(\cdot)$ is the function defined in \eqref{eq:DerivOfLogRatio}.
\item[Boundedness] We have that
\begin{align}\label{eq:BoundednessPF}
  \max_{\kz,\kx\in \ratiorange}\left\{ \xdpotF(\kz) \cdot \frac{\abs{\dlogtrecur(\kx)}}{\xdpotF(\kx)} \right\} & \leq c \enspace.
\end{align}
\end{description}
\end{definition}

The notion of the $(\pfs,\delta, c)$-potential function  is a generalisation of the 
so-called ``$(\alpha,c)$-potential function" that is introduced in \cite{VigodaSpectralInd}. Note that 
the notion of $(\alpha,c)$-potential function implies the use of the $\kell_1$-norm in the analysis. 
The setting we consider here is more general. The condition in \eqref{eq:contractionRelationPF}
 implies that we are using the $\kell_{\kr}$-norm, where $\kr$ is the H\"older conjugate of the
parameter $\pfs$ in the $(\pfs,\delta, c)$-potential function\footnote{ $\kr^{-1}+\pfs^{-1}=1$. }.

 \begin{theorem}[Adjacency Matrix]\label{thrm:AdjacencyPotentialSpIn}
 Let $\varepsilon\in (0,1)$, $\aspradius > 1$, $s\geq 1$, $\zeta>0$ and 
 the integer $\maxDeg > 1$.
 Also, let $\beta,\gamma, \lambda\in \mathbb{R}$ be such that $\gamma >0$,
$0\leq \beta\leq \gamma$ and $\lambda>0$. 

Consider graph $G=(V,E)$ of maximum degree $\maxDeg$
such that $\norm{\Adjacency_G}{2}=\aspradius$. Also, consider $\mu$ the Gibbs
distribution on $G$ specified by the parameters $(\beta, \gamma, \lambda)$.

For $\delta= \frac{1-\varepsilon}{\aspradius}$ and $c=\frac{\zeta}{\aspradius}$, suppose that there
is a $(\pfs,\delta,c)$-potential function $\potF$ with respect to $(\beta,\gamma,\lambda)$. Then, for any
$\Lambda\subset V$ and any $\tau\in \{\pm 1\}^{\Lambda}$, the influence matrix
$\infmatrix^{\Lambda,\tau}_{G}$ induced by $\mu$ satisfies 
\begin{align}\label{eq:thrm:AdjacencyPotentialSpIn}
\spradius \left( \infmatrix^{\Lambda,\tau}_{G} \right) &\leq
1+ \zeta \cdot (1-(1-\varepsilon)^{\pfs})^{-1} \cdot (\maxDeg/\aspradius)^{1-(1/\pfs)} \enspace.
\end{align}
 \end{theorem}
The proof of \Cref{thrm:AdjacencyPotentialSpIn} appears in \Cref{sec:thrm:AdjacencyPotentialSpIn}.

For the $\kk$-non-backtracking matrix we get the following result.

\begin{theorem}[$k$-non-backtracking matrix]\label{thrm:NonBacktrackingPotentialSpIn}
 Let $b,\varepsilon\in (0,1)$, $\SingBound> 1$, $\pfs \geq 1$, $c>0$ and 
 the integers $\kk,N \geq 1, \maxDeg > 1$.
Also, let $\beta,\gamma, \lambda\in \mathbb{R}$ be such that $\gamma >0$,
$0\leq \beta\leq \gamma$ and $\lambda>0$.

Consider graph $G=(V,E)$ of maximum degree $\maxDeg$ such that 
$\nnorm{ (\NBMatrix)^N}{1/N}{2}= \SingBound$.
Assume that $\mu$, the Gibbs distribution on $G$ specified by the parameters $(\beta, \gamma, \lambda)$,
is $b$-marginally bounded.

For $\delta= \frac{1-\varepsilon}{\SingBound}$, suppose  there
is a $(\pfs,\delta,c)$-potential function $\potF$ with respect to $(\beta,\gamma,\lambda)$. 
There is a bounded number $C>1$ such that for any
$\Lambda\subset V$ and $\tau\in \{\pm 1\}^{\Lambda}$, the influence matrix
$\infmatrix^{\Lambda,\tau}_{G}$ induced by $\mu$ satisfies 
{\small
\begin{align}\nonumber 
 \spradius \left( \infmatrix^{\Lambda,\tau}_{G} \right) &\leq 
 \frac{C}{1-(1-\varepsilon/4)^{1/\pfs}} \enspace. 
\end{align}
}
\end{theorem}
The proof of \Cref{thrm:NonBacktrackingPotentialSpIn} appears in \Cref{sec:thrm:NonBacktrackingPotentialSpIn}.

We use \Cref{thrm:AdjacencyPotentialSpIn,thrm:NonBacktrackingPotentialSpIn} to show our
main results for the Hard-core model in \Cref{thrm:HC4SPRadiusAdj} and \Cref{thrm:HC4SPRadiusHash}, respectively.

\begin{figure}
 \centering
\begin{tikzpicture}[node distance=1.8cm]
\node (HCNBack) [startstop, thick, fill=black!20] {\Cref{thrm:HC4SPRadiusHash}};
\node (HCNBackPotential) [startstop, below of=HCNBack,  thick, fill=black!20] {\Cref{thrm:GoodPotentialHC}};
\node (HCNBackIS) [startstop, right of=HCNBackPotential, xshift=2cm, thick, fill=black!20] {\Cref{thrm:NonBacktrackingPotentialSpIn}};
\node (CINorm2VsCLNorm2) [startstop, right of=HCNBackIS, xshift=2cm,thick, fill=black!20] {\Cref{proposition:CINorm2VsCLNorm2}};
\node (HCNBackIVsExtI) [startstop, above of=CINorm2VsCLNorm2,  thick, fill=black!20] {\Cref{thrm:InflVsExtInfl}};
\node (HCNBackNormRecurence) [startstop, below of=CINorm2VsCLNorm2,thick, fill=black!20] {\Cref{thrm:InfNormBoundHConj}};
\node (HCNBackMultiRecur) [startstop, below of=HCNBackIS, thick, fill=black!20] {\Cref{thrm:HC-VEntryBound}};
\node (HCNBackSSM) [startstop, below of=HCNBackPotential,thick, fill=black!20] {\Cref{thrm:ProdWideHatBetaVsBeta}};
\draw [arrow, thick] (HCNBackIS) -- (HCNBack);
\draw [arrow, thick] (HCNBackPotential) -- (HCNBack);
\draw [arrow, thick] (HCNBackIVsExtI) -- (CINorm2VsCLNorm2);
\draw [arrow, thick] (CINorm2VsCLNorm2) -- (HCNBackIS);
\draw [arrow, thick] (HCNBackNormRecurence) -- (HCNBackIS);
\draw [arrow, thick] (HCNBackMultiRecur) -- (HCNBackNormRecurence);
\draw [arrow, thick] (HCNBackSSM) -- (HCNBackMultiRecur);
\end{tikzpicture}
		\caption{Proof structure for \Cref{thrm:HC4SPRadiusHash} - Hard-core model and $\NBMatrix$.}
	\label{fig:StrctHCNBacktracking}
\end{figure}

 \begin{figure}
 \centering
\begin{tikzpicture}[node distance=1.8cm]
\node (HCAdjc) [startstop, thick, fill=black!20] {\Cref{thrm:HC4SPRadiusAdj}};
\node (HCAdjcSI) [startstop, right of=HCAdjc, xshift=2cm, thick, fill=black!20] {\Cref{thrm:AdjacencyPotentialSpIn}};
\node (HCAdjcNormRecurence) [startstop, right of=HCAdjcSI, xshift=2cm, thick, fill=black!20] {\Cref{thrm:InflNormBound4GeneralD}};
\node (HCAdjcPotential) [startstop, below of=HCAdjcSI, thick, fill=black!20] {\Cref{thrm:GoodPotentialHC}};
\draw [arrow] (HCAdjcSI) -- (HCAdjc);
\draw [arrow] (HCAdjcNormRecurence) -- (HCAdjcSI);
\draw [arrow] (HCAdjcPotential) -- (HCAdjc);
\end{tikzpicture}
		\caption{Proof structure for \Cref{thrm:HC4SPRadiusAdj} - Hard-core model and $\Adjacency_G$.}
	\label{fig:StrctHCAdjacency}
\end{figure}

  \spreadpoint
\section{Structure of the paper. \LastReviewG{2025-03-31}}

For the Hard-core model, the main technical results we prove in this paper are \Cref{thrm:HC4SPRadiusAdj,thrm:HC4SPRadiusHash}.
\Cref{fig:StrctHCNBacktracking} shows  the basic structure for proving  \Cref{thrm:HC4SPRadiusAdj},
while \Cref{fig:StrctHCAdjacency} shows the proof structure  for   \Cref{thrm:HC4SPRadiusHash}.

Similarly, for the Ising model,   \Cref{fig:StrctISINGNBacktracking,fig:StrctISINGAdjacency} show the basic structure 
of the proofs of  \Cref{thrm:Ising4SPRadiusHash,thrm:Ising4SPRadiusAdj}, respectively.

The proofs of the results are presented in  order of increasing difficulty.  We first present the proofs of the results 
related to  the adjacency matrix $\Adjacency_G$. Subsequently, we formally introduce the notion of extended 
influence matrix $\ExtdInfMatrixF$ together with some technical results about this matrix. Then, we   present the results related to $\NBMatrix$. 

In the appendix, there are some standard or easy to prove results which we include for the sake of the paper being self-contained.

\begin{figure}
 \centering
\begin{tikzpicture}[node distance=1.8cm]
\node (ISINGNBack) [startstop,thick, fill=black!20] {\Cref{thrm:Ising4SPRadiusHash}};
\node (ISINGNBackIS) [startstop, right of=ISINGNBack, xshift=2cm,thick, fill=black!20] {\Cref{thrm:NonBacktrackingInfinityMixing}};
\node (ISINGIVsJ) [startstop, right of=ISINGNBackIS, xshift=2cm,thick, fill=black!20] {\Cref{prop:CJVskNBM}};
\node (ISINGJVsK) [startstop, below of=ISINGIVsJ,thick, fill=black!20] {\Cref{thrm:L2InflRedux2KNBTM}};
\node (ISINGKVsH) [startstop, below of=ISINGNBackIS,thick, fill=black!20] {\Cref{lemma:eq:CLVsCJEntrywise}};
\draw [arrow] (ISINGKVsH) -- (ISINGJVsK);
\draw [arrow] (ISINGJVsK) -- (ISINGNBackIS);
\draw [arrow] (ISINGIVsJ) -- (ISINGJVsK);
\draw [arrow] (ISINGNBackIS) -- (ISINGNBack);
\end{tikzpicture}
		\caption{Proof structure for \Cref{thrm:Ising4SPRadiusHash} - Ising model and $\NBMatrix$.}
	\label{fig:StrctISINGNBacktracking}
\end{figure}

\begin{figure}
 \centering
\begin{tikzpicture}[node distance=2cm]
\node (ISINGAdjc) [startstop,thick, fill=black!20] {\Cref{thrm:Ising4SPRadiusAdj}};
\node (ISINGAdjcSI) [startstop, right of=ISINGAdjc, xshift=2cm,thick, fill=black!20] {\Cref{thrm:AdjacencyInfinityMixing}};
\draw [arrow] (ISINGAdjcSI) -- (ISINGAdjc);
\end{tikzpicture}
		\caption{Proof structure for \Cref{thrm:Ising4SPRadiusAdj} - Ising model and $\Adjacency_G$.}
	\label{fig:StrctISINGAdjacency}
\end{figure}

\spreadpoint 
\section{Preliminaries for the Analysis \LastReviewG{2025-03-03}}\label{sec:Preliminaries}

\subsection{Glauber dynamics and mixing times.}\label{sec:MCMCIntroStaff}
Suppose that we are given a graph $G=(V,E)$ and  a Gibbs distribution $\mu$ on $\{\pm 1\}^V$. 
We use the discrete-time, (single site) {\em Glauber dynamics} 
$\Glauber$ to approximately sample from $\mu$. Glauber dynamics is a very simple to describe 
Markov chain. The state space is the support of $\mu$.
We assume that the chain starts from an arbitrary configuration $X_0$. For 
$\kt\geq 0$, the transition from the state $X_{\kt}$ to $X_{\kt+1}$ is according to the following steps: 

\begin{enumerate}[(a)]
\item Choose uniformly at random a vertex $\kv\in V$. 
\item For every vertex $\kw$ different than $\kv$, set $X_{\kt+1}(\kw)=X_{\kt}(\kw)$.
\item Set $X_{\kt+1}(\kv)$ according to the marginal of $\mu$ at $\kv$, conditional on 
the neighbours of $\kv$ having the configuration specified by $X_{t+1}$.
\end{enumerate}

For two distributions $\upnu$ and $\hat{\upnu}$ on the discrete set $\Omega$,
the {\em total variation distance} satisfies 
\begin{align}\nonumber 
\norm{\upnu-\hat{\upnu}}{\tv} &=(1/{2}) \ { \sum\nolimits_{\kx\in \Omega} }\abs{\upnu(\kx)-\hat{\upnu}(\kx)} \enspace .
\end{align}

Let $\UpP$ be the transition matrix of an {\em ergodic} Markov chain $\Glauber$ on a finite state space
$\Omega$ with stationary distribution $\mu$. Recall that $\UpP$ is called ergodic if it is irreducible and aperiodic. 
For $\kt\geq 0$ and $\sigma\in \Omega$, we let $\UpP^{\kt}(\sigma, \cdot)$  be the distribution of $X_{\kt}$ when $X_0=\sigma$. Then, the {\em mixing time} of $\UpP$ is defined by 
\begin{align}\nonumber 
{\rm T}_{\rm mix}(\UpP)=\min\{\kt\geq 0\ :\ \forall \sigma\in \Omega \ \ \norm{\UpP^{\kt}(\sigma, \cdot)-\mu(\cdot)}{\tv}\leq 1/4\} \enspace.
\end{align}
We use the mixing time as a measure of the rate of convergence to equilibrium for Glauber dynamics.
For the cases we consider here, Glauber dynamics is trivially ergodic.

\subsection{Spectral Independence}\label{sec:SIPrelimIntro}
In this section, we introduce the  basics of the Spectral Independence method. Specifically, we start by formally
introducing the notion of the {\em pairwise influence matrix} $\infmatrix^{\Lambda,\tau}_{G}$. Subsequently, we 
present  some rapid mixing results  we can obtain from the Spectral Independence method. 

\subsubsection*{Influence Matrix:}
We have seen the definition of the influence matrix before. However, since it is such an important  notion for 
our  results, we state it once more.

Let graph $G=(V,E)$ and a Gibbs distribution $\mu$ on $\{\pm 1\}^V$.  For  $\Lambda\subset V$ and a configuration 
$\tau$ at $\Lambda$,  the {\em pairwise influence matrix}  $\infmatrix^{\Lambda,\tau}_{G}$ is   indexed by the vertices 
in $V\setminus \Lambda $  such that, for $\ku,\kw\in V\setminus \Lambda$, we have
\begin{align}\label{def:InfluenceMatrix}
\infmatrix^{\Lambda,\tau}_{G}(\kw,\ku) &
= \mu_{\ku}(+1 \ |\ (\Lambda, \tau), (\kw, +1))- \mu_{\ku}(+1 \ |\ (\Lambda, \tau), (\kw, -1)) \enspace, 
\end{align} 
where $\mu_{\ku}(+1 \ |\ (\Lambda, \tau), (\kw, +1))$ is the probability of the event that vertex $\ku$  has configuration
 $+1$, conditional on  that the configuration at $\Lambda$ is $\tau$ and the  configuration at $\kw$ is $+1$. We have 
 the analogous for $\mu_{\ku}(+1 \ |\ (\Lambda, \tau), (\kw, -1))$.

When we consider Gibbs distributions with hard constraints, e.g., the Hard-core model, it is possible  that for a 
vertex  $\kv\notin \Lambda$ to have $\mu^{\Lambda,\tau}_{\kv}(+1)\in \{0,1\}$.  For vertices $\kw,\ku$ such that 
at least one of $\mu^{\Lambda,\tau}_{\kw}(+1),\mu^{\Lambda,\tau}_{\ku}(+1)$ is in $\{0,1\}$, we have  
$\infmatrix^{\Lambda,\tau}_{G}(\kw,\ku)=0$.

In the analysis, we use the following folklore result which is standard to prove (e.g. see \Cref{sec:claim:InfSymmetrisation}). 

\begin{restatable}{claim}{Isymsym}\label{claim:InfSymmetrisation}
For graph $G=(V,E)$ and a Gibbs distribution $\mu$ on $\{\pm 1\}^V$, 
for $\Lambda\subset V$ and $\tau\in \{\pm 1\}^V$, let $\UpM$ be the $(V\setminus \Lambda)\times (V\setminus \Lambda)$
diagonal matrix such that 
\begin{align}\nonumber
\UpM(\kv,\kv) &\textstyle =\sqrt{\mu^{\Lambda,\tau}_{\kv}(+1)\cdot \mu^{\Lambda,\tau}_{\kv}(-1)} &\forall \kv\in V\setminus \Lambda \enspace.
\end{align}
Then, for the influence matrix $\infmatrix^{\Lambda,\tau}_G$ induced by $\mu$, the following is true:
if $\UpM$ is non-singular, then matrix $\UpM\cdot \infmatrix^{\Lambda,\tau}_G \cdot \UpM^{-1}$ is symmetric. 
\end{restatable}

\subsubsection*{Spectral Independence and Mixing Times:}
The main focus of the Spectral Independence method is on  the maximum eigenvalue of the influence
matrix $\infmatrix^{\Lambda,\tau}_{G}$, denoted as $\eigenval_{1} (\infmatrix^{\Lambda,\tau}_{G})$.

\begin{definition}[Spectral Independence]\label{Def:SpInMu}
For a real $\eta>0$, the Gibbs distribution $\mu$ on $G=(V,E)$ is $\eta$-spectrally
independent, if for every $0\leq \kk\leq |V|-2$, $\Lambda\subset V$ of size $\kk$ and $\tau\in \{\pm 1\}^\Lambda$
we have that $\eigenval_{1}(\cI^{\Lambda,\tau}_{G})\leq 1+\eta$. 
\end{definition}

\noindent
The notion of $\eta$-spectral independence for $\mu$ can be utilised to bound the mixing time of the corresponding 
Glauber dynamics.

\begin{theorem}[\cite{OptMCMCIS}]\label{thrm:SPCT-INDClosed}
For $\eta>0$, there is a constant $C\geq 0$ such that if $\mu$ is an $\eta$-spectrally independent distribution, 
then the Glauber dynamics for sampling from $\mu$ has {\em mixing time} which is at most $Cn^{2+\eta}$. 
\end{theorem}

Note that \Cref{thrm:SPCT-INDClosed} is quite general and implies that for bounded $\eta$, the mixing time is 
polynomial in $n$. However, this polynomial can be very large. There have been improvements on 
\Cref{thrm:SPCT-INDClosed} since its introduction in \cite{OptMCMCIS}, e.g., see 
\cite{VigodaSpectralIndB,FastMCMCLocalisation,WeimingCoUbboundedDelta}.

We use \Cref{thrm:SPINLOGN}, proved in \cite{VigodaSpectralIndB}, to turn our spectral independence results into 
optimum mixing ones. 
We need to introduce a few useful concepts from \cite{VigodaSpectralIndB}, to formally state the result we are using.

Let graph $G=(V,E)$. For a non-empty set  $S\subset V$, let  ${\tt H}_S$ be the graph whose vertices correspond to 
 the configurations in $\{\pm 1\}^S$.  Two configurations in $\{\pm 1\}^S$ are considered adjacent 
 if and only if they differ at the assignment  of a single vertex.
 Similarly, any subset $\Omega_0\subseteq \{\pm 1\}^S$ is considered to be connected if the subgraph
 induced by $\Omega_0$ is connected. 
 
 A distribution $\mu$ on $\{\pm 1\}^V$ is considered to be {\em totally
 connected} if for every nonempty $\Lambda \subset V$ and every boundary condition $\tau$ at $\Lambda$,
 the set of configurations in the support of $\mu(\cdot \ |\ \Lambda, \tau)$ is connected. 
We remark here that all Gibbs distributions with soft-constraints such as the Ising model are 
totally connected in a trivial way. The same holds for the Hard-core model and this follows from standard arguments.

The following result is a part of Theorem 1.9 from \cite{VigodaSpectralIndB} (arxiv-version).

\begin{theorem}[\cite{VigodaSpectralIndB}]\label{thrm:SPINLOGN} Let the integer $\maxDeg\geq 3$ and $b,\eta\in \mathbb{R}_{>0}$. 
Consider graph $G=(V,E)$ on $n$ vertices and maximum degree $\maxDeg$. Also, let $\mu$ be a totally connected Gibbs
distribution on $\{\pm 1\}^V$. 

If $\mu$ is both $b$-marginally bounded and $\eta$-spectrally independent, then there are constants $C_1, C_2>0$ such
the Glauber dynamics for $\mu$ exhibits mixing time
\begin{align}\nonumber 
{\rm T}_{\rm mix} &\leq C_1\times \left( {\maxDeg}/{b} \right)^{C_2\left ( \frac{\eta-1}{\kb^2}+1 \right)} \times \left( n \log n \right) \enspace.
\end{align} 
\end{theorem}
 \Cref{thrm:SPINLOGN} implies  $O(n\log n)$  mixing, provided that $\maxDeg$, $\eta$ and $\kb^{-1}$ are bounded numbers.

\subsection{(Local) Connective Constant:}\label{sec:Intro2LCC}
We start by introducing a few standard notions related to the (standard) {\em connective constant}.
We consider an infinite, locally finite graph $G$. 
Locally finite means that all vertices have finite degree. For a vertex $\kv$ of graph $G$, we denote by $\pi(\kv,\kk)$ the 
number of {\em self-avoiding walks} of length $\kk$ starting at $\kv$. Recall that a walk is called self-avoiding if it does not repeat vertices. 
We let 
\begin{align}\nonumber
c_{\kk}=\sup\{ \pi(\kv,\kk)\ |\ \kv\in V\}\enspace. 
\end{align}

\begin{definition}[connective constant]\label{def:CCInfinite}
Let $G=(V,E)$ be a locally finite, infinite graph. The connective constant $\cconnective=\cconnective(G)$ is defined by
$\cconnective = \lim_{\kk \to \infty}c^{1/\kk}_{\kk} .$
\end{definition}
It is well-known that the above limit always exists. 

For our algorithmic applications, we use  the notion of {\em radius-$\kk$ connective constant}. 
\begin{definition}\label{def:LocalCC}
For integer $\kk>0$ and graph $G=(V,E)$, the radius-$\kk$ connective constant $\cconnective_{\kk}=\cconnective_{\kk}(G)$ is defined by
$\cconnective_{\kk}= c^{1/\kk}_{\kk}$.
\end{definition}

The above definitions imply that $\cconnective=\lim_{\kk \to\infty} \cconnective_{\kk}$. 

The use of infinite graph $G$ was only to formally introduce  the notion of the connective constant. 
From now on,  unless otherwise specified,  we assume that the graph $G$ is always {\em finite}.

\subsection{The adjacency matrix $\Adjacency_G$}
For graph $G=(V,E)$, the {\em adjacency matrix} $\Adjacency_G$ is a zero-one, $V\times V$ matrix 
such that, for any $\ku,\kw\in V$, we have
\begin{align}\nonumber 
\Adjacency_G(\kw,\ku)= \Ind\{\textrm{$\ku, \kw$ are adjacent in $G$}\} \enspace.
\end{align}
Recall that a {\em walk} in  graph $G$ is any sequence of vertices $\kw_0, \ldots, \kw_{\kell}$ 
such that each consecutive pair $(\kw_{\ki-1}, \kw_{\ki})$ is an edge.
The length of the walk is equal to the number of consecutive pairs $(\kw_{\ki-1}, \kw_{\ki})$.

The adjacency matrix has the property that for any $\ku,\kw\in V$ and $\kell\geq 1$,
we have 
\begin{align}\label{eq:NoOfWalksVSA2L}
  \Adjacency^{\kell}_G (\kw,\ku)&=\textrm{$\#$ length $\kell$ walks\ from}\ \kw\ \textrm{to}\ \ku \enspace.
\end{align}
For undirected $G$, matrix $\Adjacency^{\kell}_G$ is symmetric for any $\kell\geq 0$. Hence, $\Adjacency^{\kell}_G$ has real eigenvalues, while 
the eigenvectors corresponding to distinct eigenvalues are orthogonal with each other. 
We denote with $\eigenv_{\kj} \in \mathbb{R}^V$ the eigenvector of $\Adjacency_G$ 
that corresponds to the eigenvalue $\eigenval_{\kj}(\Adjacency_G)$.
Unless otherwise specified, we always take $\eigenv_{\kj}$ such that $ \norm{ \eigenv_{\kj}}{2}=1$.

Since $\Adjacency_G$ is symmetric, the {\em spectral radius} satisfies that $\spradius(\Adjacency_G)=\norm{\Adjacency_G}{2}$. 
Recall that $\norm{\Adjacency_G}{2}=\max_{{\bf x}\in\mathbb{R}^V:{\bf x}\neq {\bf 0}}\frac{ \norm{ \Adjacency_G \cdot {\bf x}}{2}}{\norm{ {\bf x}}{2}}$.

When $G$ is connected, then $\Adjacency_G$ is irreducible. Then, the Perron-Frobenius Theorem 
implies that
\begin{align}\label{eq:DefOfBSMatrix}
\spradius(\Adjacency_G)&=\eigenval_1(\Adjacency_G)& \textrm{and} && \maxeigenv(\ku) &>0 \qquad \forall \ku\in V \enspace.
\end{align}

It is an easy exercise to show that  $\sqrt{\maxDeg} \leq \spradius(\Adjacency_G)\leq \maxDeg$.

 \begin{figure}
 \begin{minipage}{.6\textwidth}
 \centering
		\includegraphics[width=.7\textwidth]{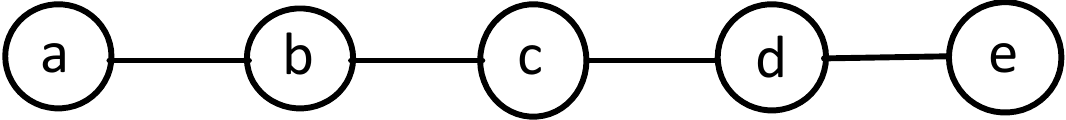}
		\caption{Example graph $G$ for $\NBMatrix$}
	\label{fig:OrientedEdges}
\end{minipage}
\end{figure}

\subsection{Matrices on oriented paths \& PT-Invariance}\label{sec:MatricesOnPaths}
For graph $G=(V,E)$ and integer $\kk>0$, we let $\ExtV_{\kk}$ be the set of self-avoiding walks of length $\kk$. 

For  $\kP \in \ExtV_{\kk}$, we let $\kP^{-1}$ denote the walk with the opposite direction. E.g., 
suppose $\kP=\ku_0,\ku_1, \ldots, \ku_{\kk-1}, \ku_\kk$, then $\kP^{-1}=\ku_{\kk},\ku_{\kk-1},\ldots, \ku_{1}, \ku_0$.

In this paper, we work with matrices defined on (subsets of) $\ExtV_{\kk}$. Typically, the matrices we consider
are not symmetric. However, they exhibit other symmetries that we utilise in the analysis. The most important one is 
the so-called {\em PT-invariance}, where PT stands for ``parity-time". 

Let $L\subseteq \ExtV_{\kk}$ be such that if $\kP \in L$, then $\kP^{-1}\in L$, too. 
For ${\bf \kx}\in \mathbb{R}^{L}$, let the vector $\PTV{\bf \kx}$ be such that
\begin{align}
\PTV{\bf \kx}(\kP) &={\bf \kx}(\kP^{-1}) &\forall \kP\in L \enspace. 
\end{align}

\noindent
Also, let $\Invol_L$ be the involution on $\mathbb{R}^{L}$ such that
\begin{align}\label{def:OfInvolutation}
\Invol_L \cdot {\bf \kx}=\PTV{\bf \kx} \enspace. 
\end{align}
The $L\times L$ matrix $\UpZ$ is PT-invariant if we have 
\begin{align}\label{eq:DefPTInvarianceFormal}
\UpZ\cdot \Invol_L &=\Invol_L \cdot \MTR{\UpZ}\enspace, 
\end{align}
where $\MTR{\UpZ}$ is the  matrix transpose of $\UpZ$. The above implies 
that $\UpZ\cdot \Invol_L$ is symmetric, while we also have
\begin{align}\label{def:PTInvariance}
\UpZ(\kP,\kQ) &=\UpZ(\kQ^{-1},\kP^{-1}) \enspace, & \forall \kP,\kQ \in L \enspace. 
\end{align}

\subsubsection{The $\kk$-non-backtracking matrix $\NBMatrix$}
 For graph $G=(V,E)$ and integer $\kk>0$, recall that $\ExtV_{\kk}$ is the set of self-avoiding walks of length $\kk$.

 The (order) $\kk$-non-backtracking matrix $\NBMatrix$ is an $\ExtV_{\kk} \times \ExtV_{\kk}$, zero-one matrix 
 such that for any $\kP=\kx_0,\ldots, \kx_{\kk}$ and $\kQ=\kz_0, \ldots, \kz_{\kk}$ in $\ExtV_{\kk}$, we have 
\begin{align}\label{eq:DefOfKNBMatrix}
\NBMatrix(\kP, \kQ) &=\Ind\{\kx_0 \neq \kz_{\kk} \}\times \prod\nolimits_{0\leq \ki <\kk}\Ind\{\kz_{\ki} = \kx_{\ki+1}\} \enspace.
\end{align}
That is, $\NBMatrix(\kP, \kQ)$ is equal to $1$ if the walk $\kx_0, \ldots, \kx_{\kk}, \kz_{\kk}$  does not repeat vertices, i.e.,
it is self-avoiding. Otherwise, $\NBMatrix(\kP, \kQ)$ is zero. 
When $\kk=1$, $\NBMatrix$  to the well-known {\em Hashimoto non-backtracking matrix} \cite{Hash89}.

For the graph $G$ in \Cref{fig:OrientedEdges}, consider $\NBMatrixE_{G,3}$.
Then, for $\kP=\ka\kb\kc$, $\kQ=\kb\kc \kd$ and $\kR=\kc \kd \ke$ 
we have $\NBMatrixE_{G,3}(\kP, \kQ)=1$, also, we have $\NBMatrixE_{G,3}(\kQ,\kR)=1$, while note that we have 
$\NBMatrixE_{G,3}(\kP, \kR)=0$.

For any integer $\kr>\kk$, we say that walk $\kR=\kw_{0},\ldots, \kw_{\kr}$ is a {\em $\kk$-non-backtracking 
walk} in $G$, if every $(\kk+2)$-tuple of consecutive vertices $\kw_{\ki}, \ldots, \kw_{\ki+\kk+1}$ 
in $\kR$ is a self-avoiding walk.
Furthermore, for any integer $\kell>0$ and any $\kP,\kQ\in \ExtV_{\kk}$, we have 
\begin{align}
\NBMatrix^{\kell}(\kP,\kQ) &=\textrm{ $\#$ $\kk$-non-backtracking walks from $\kP$ to $\kQ$ having length $\kell$} \enspace.
\end{align}
The above entry counts all the $\kk$-non-backtracking walks $\kR=\kw_0, \ldots, \kw_{\kell+\kk}$ such that the first $\kk+1$ vertices
in $\kR$ correspond to  walk $\kP$, while the last $\kk+1$ vertices correspond to walk $\kQ$. 
Note that $\kR$ is a walk with $\kell+\kk+1$ vertices.

\subsubsection*{Norms \& Singular Values of $\NBMatrix$:}
In the general case, $\NBMatrix$ is not normal, i.e., $\NBMatrix\cdot \MTR{\NBMatrixE}_{G,\kk} \neq \MTR{\NBMatrixE}_{G,\kk}\cdot {\NBMatrix}$.

It is elementary to show that it is PT-invariant. More specifically, for any integer $\kell\geq 0$, we have 
\begin{align}\label{eq:DefPTInvarianceFormalB}
\NBMatrix^{\kell} \cdot \Invol &= \Invol \cdot\MTR{\NBMatrixE}^{\kell}_{G,\kk}\enspace,
\end{align}
where $\Invol $ is the involution on $\mathbb{R}^{\ExtV_{\kk}}$ defined as in \eqref{def:OfInvolutation}. 
The PT-invariance emerges from the simple observation that for any two $\kP,\kQ \in \ExtV_{k}$, 
the number of $\kk$-non-backtracking walks of length $\kell$ from $\kP$ to $\kQ$ in graph $G$ is the
same as the number of $\kk$-non-backtracking walks of length $\kell$ from $\kQ^{-1}$ to $\kP^{-1}$.

Notably \eqref{eq:DefPTInvarianceFormalB} implies that 
matrix $\NBMatrix^{\kell} \cdot \Invol$ is symmetric, i.e., the r.h.s. of \eqref{eq:DefPTInvarianceFormalB} is
the transpose of the l.h.s. since we have $\MTR{\Invol}={\Invol}$.

For integers $\kk,\kell>0$, let $\hsingular_{\kk,\kell}$ be the maximum {\em singular value} of matrix 
$(\NBMatrix)^{\kell}$, i.e., $\norm{ (\NBMatrix)^{\kell}}{2}=\hsingular_{\kk,\kell}$. 
It is standard that 
\begin{align}
\norm{(\NBMatrix)^{\kell}}{2}=\norm{ (\NBMatrix)^{\kell} \cdot \Invol}{2}=\hsingular_{\kk,\kell}\enspace,
\end{align}
Since $(\NBMatrix)^{\kell}\cdot \Invol$ is symmetric, 
 $\sigma_{\kk,\kell}$ is equal to the spectral radius of matrix $(\NBMatrix)^{\kell}\cdot \Invol$. 

For the analysis, we use of the following two results about $\hsingular_{\kk,\kell}$, which are standard to prove.

\begin{restatable}{claim}{hsingpathbound}\label{claim:SigmaLVsPathNumber}
For integers $\kell,\kk>0$ and any $\kP,\kQ \in \ExtV_{\kk}$, we have 
$\NBMatrix^{\kell}(\kP,\kQ) \leq \sigma_{\kk,\kell}$.
\end{restatable}

The proof of \Cref{claim:SigmaLVsPathNumber} appears in \Cref{sec:claim:SigmaLVsPathNumber}.

\begin{restatable}{claim}{hsingseqnondec}\label{lemma:SingSequenConv}
For any integers $\kk,N>0$ and $\kz\in (0,1)$, there is $\ell_0=\ell_0(N,\kz,\maxDeg)$ such that 
for any integer $\kell\geq \ell_{0}$, we have
$\left( \hsingular_{\kk,\kell} \right)^{1/\kell} \leq (1+\kz) (\hsingular_{\kk,N})^{\frac{1}{N}}\enspace. $
\end{restatable}

The proof of \Cref{lemma:SingSequenConv} appears in \Cref{sec:lemma:SingSequenConv}.

\subsection{$\Tsaw$ Construction for $\infmatrix^{\Lambda,\tau}_{G}$} 
\label{sec:TsawConstruction}

 \begin{figure}
 \begin{minipage}{.45\textwidth}
 \centering
		\includegraphics[width=.41\textwidth]{./G4Tsaw.pdf}
		\caption{Initial graph $G$}
	\label{fig:Base4TSAW}
\end{minipage}
 \begin{minipage}{.4\textwidth}
 \centering
		\includegraphics[height=3.1cm]{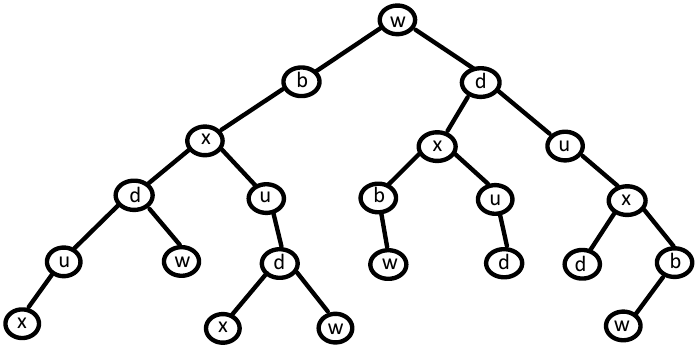}
		\caption{$\Tsaw(G,\kw)$}
	\label{fig:TSAWExample}
\end{minipage}
\end{figure}

The exposition here relies on results from \cite{OptMCMCIS,VigodaSpectralInd}.
Let graph $G=(V,E)$ and a Gibbs distribution $\mu_G$ on $\{\pm 1\}^V$. 
Assume w.l.o.g. that there is a {\em total ordering} of the vertices in $G$.

We start by introducing the notion of the {\em tree of self-avoiding walks} in $G$. 
For each vertex $\kw$ in $G$, we define $T=\Tsaw(G,\kw)$, the tree of self-avoiding walks 
starting from $\kw$, as follows:
Consider the set consisting of every walk $\kv_0, \ldots, \kv_{\kr}$ in graph $G$ that emanates 
from vertex $\kw$, i.e., $\kv_0=\kw$, while one of the following two holds
\begin{enumerate}
\item $\kv_0, \ldots, \kv_{\kr}$ is a self-avoiding walk,
\item $\kv_0, \ldots, \kv_{\kr-1}$ is a self-avoiding walk, while there is $j\leq \kr-3$ such that $\kv_{\kr}=\kv_{j}$.
\end{enumerate}
Each of these walks  corresponds to a vertex in $T$. Two vertices in $T$ are adjacent 
if the corresponding walks are adjacent, i.e.,  one walk extends the other by one vertex.

The vertex in $T$ that corresponds to the walk $\kv_0, \ldots, \kv_{\kr}$, 
is called  ``copy of vertex $\kv_{\kr}$", i.e., $\kv_{\kr}$ is the last vertex in the path. 
For each vertex $\kv$ in $G$, we let $\cp(\kv)$ be the set of its copies in $T$. 
Note that vertex $\kv$ may more than one  copies in $T$. 

For an example of the above construction, consider the graph $G$ in \Cref{fig:Base4TSAW}.
In \Cref{fig:TSAWExample}, we have the tree of self-avoiding walks that starts from vertex $\kw$ in $G$.
The vertices with label $\kv$ in the tree  correspond to the copies of  vertex $\kv$ of the initial 
graph $G$, e.g. the root if a copy of vertex $\kw$.

For $\Lambda\subset V$ and $\tau\in \{\pm \}^{\Lambda}$, let 
the influence matrix $\infmatrix^{\Lambda,\tau}_{G}$ induced by $\mu^{\Lambda,\tau}_G$. 
We describe how the entry 
$\infmatrix^{\Lambda,\tau}_{G}(\kw,\ku)$ can be expressed using an appropriately defined 
spin-system on $T=\Tsaw(G,\kw)$.

Let $\mu_T$ be a Gibbs distribution on $T$ which has the same specification as $\mu_G$.
Each $\kz \in \cp(\kx)$ in the tree $T$, such 
that $\kx \in \Lambda$, is pinned to   configuration  $\tau(\kx)$. Furthermore, if  
 vertex $\kz$ in $T$  corresponds  to walk  $\kv_0, \ldots, \kv_{\kell}$ in $G$ such that 
$\kv_{\kell}=\kv_j$, for $0\leq j \leq \kell-3$, then  we set a pinning at vertex $\kz$, as well. This pinning 
depends on the  total ordering of the vertices. Particularly, we set at $\kz$
\begin{enumerate}[(a)]	
\item $-1$ if $\kv_{\kj+1}>\kv_{\kell-1}$, 
\item $+1$ otherwise.
\end{enumerate}
Let $\Gamma=\Gamma(G,\Lambda)$ be the set of vertices in $T$ which have been pinned 
in the above construction, while let $\sigma=\sigma(G,\tau)$ be the pinning  at $\Gamma$.

For each edge $\ke$ in $T$, we specify weight $\infweight(\ke)$ as follows: letting $\ke=\{\kx,\kz\}$ be such 
that $\kx$ is the parent of $\kz$ in $T$, we set 
\begin{align}\label{def:OfInfluenceWeights}
\infweight(\ke)=\left \{ 
\begin{array}{lcl}
0 & \quad&\textrm{if there is a pinning at either $\kx$ or $\kz$},\\
\textstyle \dlogtrecur\left(\log \gratio^{\Gamma, \sigma}_{\kz}\right) & & \textrm{otherwise}.
\end{array}
\right .
\end{align}
The function $\dlogtrecur(\cdot)$ is from \eqref{eq:DerivOfLogRatio}, while $\gratio^{\Gamma, \sigma}_{\kz}$ 
is a ratio of Gibbs marginals at $\kz$ (see definitions in Section \ref{sec:RecursionVsSpectralIneq}). 
Then, we have the following proposition.  
\begin{proposition}[\cite{OptMCMCIS,VigodaSpectralInd}]\label{prop:Inf2TreeRedaux}
For every $\ku,\kw\in V\setminus \Lambda$ we have that 
\begin{align}\label{eq:InfluenceEntryAsWeightedSum}
\infmatrix^{\Lambda,\tau}_{G}(\kw,\ku) &=\sum\nolimits_{\kP}\prod\nolimits_{\ke\in \kP}\infweight(\ke) \enspace,
\end{align}
where $\kP$ varies over all paths from the root of $\Tsaw(G,\kw)$ to the set of vertices in $\cp(\ku)$. 
\end{proposition}

\spreadpoint

\section{Proof of \Cref{thrm:AdjacencyInfinityMixing} - SI with $\delta$-contraction \& $\Adjacency_G$ \LastReviewG{2025-03-1}}\label{sec:thrm:AdjacencyInfinityMixing}

Recalling that $\delta=(1-\varepsilon)/\aspradius$, let $\powadj$ be the $V\times V$ matrix defined by
\begin{align}\label{eq:DefOfPowerAdj}
 \powadj&= \sum\nolimits_{0\leq \kell \leq n} \left( \delta \cdot \Adjacency_G\right)^{\kell} \enspace.
\end{align}
Since the adjacency matrix $\Adjacency_G$ is symmetric,  $\powadj$ is symmetric, as well. We have 
\begin{align}\nonumber
\spradius\left (\powadj \right ) & = \norm{ \powadj}{2} = \norm{ \sum\nolimits_{\kell\geq 0}(\delta\cdot \Adjacency_G)^{\kell}}{2} 
\leq \sum\nolimits_{\kell\geq 0} \delta^{\kell} \cdot \norm{ (\Adjacency_G)^{\kell}}{2}  \leq \sum\nolimits_{\kell\geq 0}(1-\varepsilon)^{\kell} =\varepsilon^{-1} \enspace.
\end{align}
Let $\powadj_{V\setminus \Lambda}$ be the principal submatrix of $\powadj$,  which is obtained by 
removing rows and columns that correspond to the vertices in $\Lambda$. 
Cauchy’s interlacing theorem, e.g., see \cite{MatrixAnalysis}, implies that
\begin{align}
\spradius\left (\powadj_{V\setminus \Lambda} \right )&\leq \spradius\left (\powadj \right ) \leq \varepsilon^{-1}\enspace.
\end{align}
\Cref{thrm:AdjacencyInfinityMixing} follows by showing that $\spradius( \infmatrix^{\Lambda,\tau}_{G} ) \leq \spradius\left( \powadj_{V\setminus \Lambda} \right)$.
To this end,  it suffices to show  
\begin{align}\label{eq:Target4thrm:AdjacencyInfinityMixing}
 \abs{\infmatrix^{\Lambda,\tau}_{G} (\kw,\ku)}&\leq \powadj_{V\setminus \Lambda}(\kw,\ku) & \forall \ku,\kw\in V\setminus \Lambda \enspace. 
\end{align}
To see why the above implies $\spradius( \infmatrix^{\Lambda,\tau}_{G} ) \leq \spradius\left( \powadj_{V\setminus \Lambda} \right)$ see  \Cref{lemma:MonotoneVsSRad} in the Appendix,

It is immediate that for any $\ku\in V\setminus \Lambda$, we have $\abs{\infmatrix^{\Lambda,\tau}_{G}(\ku,\ku) } =1$.
Also, we have $\powadj(\ku,\ku)\geq 1$, as the summand in \eqref{eq:DefOfPowerAdj} for $\kell=0$ corresponds to 
the identity matrix. This proves that \eqref{eq:Target4thrm:AdjacencyInfinityMixing} is true when $\ku=\kw$.

Fix $\kw,\ku\in V\setminus \Lambda$ such that $\kw\neq \ku$.  Let $T=\Tsaw(G,\kw)$ and let $\{ \infweight(\ke)\}$ be the weights 
on the edges of $T$ specified as in \eqref{def:OfInfluenceWeights} with respect to $\mu^{\Lambda,\tau}_G$.   \Cref{prop:Inf2TreeRedaux} implies
\begin{align}\label{eq:AdjacencyInfinityMixingTsawConstruction}
 \infmatrix^{\Lambda,\tau}_{G}(\kw,\ku) &=\sum\nolimits_{\kell\geq 0}\sum\nolimits_{\kP\in \cP(\ku,\kell)}\prod\nolimits_{\ke\in \kP}\infweight(\ke)\enspace, 
\end{align}
where set $\cP(\ku,\kell)$ consists of the paths of length $\kell$ from the root of $T$ to $\cp(\ku)$. 

The  $\delta$-contraction assumption implies  that $\abs{\infweight(\ke)} \leq \delta$,
for each $\ke\in T$. Then,  
from \eqref{eq:AdjacencyInfinityMixingTsawConstruction},  we have 
\begin{align}\label{eq:AdjacencyInfinityMixingInlEntryBoundDelta}
\abs{\infmatrix^{\Lambda,\tau}_{G} (\kw,\ku)}&\leq \sum\nolimits_{\kell\geq 0} \abs{\cP(\ku,\kell)} \cdot \delta^{\kell} \enspace. 
\end{align}
Let $\saw_{\kell}(\kw,\ku)$ be the set of walks of length $\kell$ from $\kw$ to $\ku$, in graph $G$, that correspond to the elements in 
$\cp(\ku)$. Let ${\rm Walks}_{\kell}(\kw,\ku)$ be the set of walks of length $\kell$ from $\kw$ to $\ku$, in graph $G$.  We have 
$\saw_{\kell}(\kw,\ku)\subseteq {\rm Walks}_{\kell}(\kw,\ku)$  
i.e., since each element in $\saw_{\kell}(\kw,\ku)$ is also a walk in graph $G$ from $\kw$ to $\ku$. We have 
\begin{align}\nonumber 
\abs{\cP(\ku,\kell)} &= \abs{\saw_{\kell}(\kw,\ku)} \leq \abs{{\rm Walks}_{\kell}(\kw,\ku)} = \Adjacency^{\kell}_G(\kw,\ku) \enspace.
\end{align}
The last equality is from \eqref{eq:NoOfWalksVSA2L}.
Plugging the above bound into \eqref{eq:AdjacencyInfinityMixingInlEntryBoundDelta}, we get 
\begin{align}\nonumber 
\abs{\infmatrix^{\Lambda,\tau}_{G} (\kw,\ku)} &\leq \sum\nolimits_{\kell\geq 0} \delta^{\kell} \cdot \Adjacency^{\kell}_G(\kw,\ku)=\powadj_{V\setminus \Lambda}(\kw,\ku) \enspace. 
\end{align}
Hence, \eqref{eq:Target4thrm:AdjacencyInfinityMixing} is also true when $\kw\neq \ku$.
\Cref{thrm:AdjacencyInfinityMixing} follows. \hfill $\square$

\spreadpoint

\section{Proof of \Cref{thrm:AdjacencyPotentialSpIn} - SI with Potentials \& $\Adjacency_G$ \LastReviewG{2025-03-11}}\label{sec:thrm:AdjacencyPotentialSpIn}

We let $\ExtV$ be the set of ordered pairs of adjacent vertices in $V$. Consider the zero-one
$\ExtV\times V$ matrix $\EdgeToV$ such that, for any $\kw\kx\in \ExtV$
and any $\ku\in V$, we have
\begin{align}\label{def:SingleEdgeM}
\EdgeToV(\kw\kx,\ku)&=\Ind\{\kx=\ku\}\enspace. 
\end{align}

\begin{theorem}\label{thrm:InflNormBound4GeneralD}
Let $\pfs\geq 1$, $\delta,c >0$ and the integer $\maxDeg>1$. Also, let 
$\beta,\gamma, \lambda\in \mathbb{R}$ be such that $\gamma >0$, $0\leq \beta\leq \gamma$ and $\lambda>0$. 

Consider graph $G=(V,E)$ of maximum degree $\maxDeg$, while let 
$\mu$ the Gibbs distribution on $G$ specified by the parameters $(\beta, \gamma, \lambda)$.
Suppose  there is a $(\pfs,\delta,c)$-potential function $\potF$ with respect to $(\beta,\gamma,\lambda)$. 

For any $\Lambda\subset V$ and $\tau\in \{\pm 1\}^{\Lambda}$, for any diagonal, non-negative, 
non-singular $\UpD$, such that $\UpD$ and $\infmatrix^{\Lambda, \tau}_G$ are conformable for multiplication
the following is true: for any $\kw \in V\setminus \Lambda$,  we have 
\begin{align}
\lefteqn{
 \sum\nolimits_{\ku \in V\setminus \Lambda } \abs{ ( \UpD^{-1}\cdot \infmatrix^{\Lambda, \tau}_G\cdot \UpD ) (\kw,\ku) }
} \hspace{.5cm}  \label{eq:Target4thrm:InflNormBound4GeneralD} \\
&\leq 1+  
 \frac{c}{\UpD(\kw,\kw) } \left( \sum\nolimits_{\kz \in N(\kw)\setminus \Lambda} \UpD(\kz,\kz)+ \sum\nolimits_{\kell \geq 2} \left( \delta^{\kell-1} \sum\nolimits_{\ku \in V\setminus \Lambda}
\left(\EdgeToV\cdot \Adjacency^{\kell-1}_G\right)(\kw\kz, \ku) 
\cdot \UpD^{\pfs}(\ku,\ku)
\right)^{{1}/{\pfs}} \right)
\enspace,  \nonumber 
\end{align}
where $\Adjacency_G$ is the adjacency matrix of $G$ and matrix $\EdgeToV$ is defined in \eqref{def:SingleEdgeM}.
\end{theorem} 
The proof of \Cref{thrm:InflNormBound4GeneralD} appears in \Cref{sec:thrm:InflNormBound4GeneralD}.

\newcommand{\Pweight}{{\UpW}}

\begin{proof}[Proof of \Cref{thrm:AdjacencyPotentialSpIn}]

Let $\maxeigenv$ be the principal eigenvector of $\Adjacency_G$, while let $\eigenval_1$ be the 
corresponding eigenvalue. Since $G$ is assumed to be connected, matrix $\Adjacency_G$ is 
non-negative and irreducible.  The Perron-Frobenius theorem implies 
\begin{align}\label{eq:PFT4thrm:AdjacencyPotentialSpIn}
\aspradius &=\eigenval_1& \textrm{and} && \maxeigenv(\ku)>0 \qquad \forall \ku\in V\enspace. 
\end{align}
Let $\SpGMatrix$ be the $(V\setminus \Lambda)\times (V\setminus \Lambda)$ diagonal matrix such 
that for any  $u\in V\setminus \Lambda$, we have 
\begin{align}\label{def:AdjMatrixEigenvals}
\SpGMatrix(\ku,\ku)=\maxeigenv(\ku) \enspace.
\end{align} 
 \Cref{eq:PFT4thrm:AdjacencyPotentialSpIn} implies that $\SpGMatrix$ is {\em non-singular}.

We prove \Cref{thrm:AdjacencyPotentialSpIn} by applying \Cref{thrm:InflNormBound4GeneralD}, where 
$\UpD=\SpGMatrix^{1/\pfs}$. Specifically, we show that
\begin{align} \label{eq:Target4thrm:MySpectralMatrixNorm} 
\norm{ \SpGMatrix^{-1/\pfs} \cdot \infmatrix^{\Lambda, \tau}_G \cdot \SpGMatrix^{1/\pfs}}{\infty} &\leq 
1+c \cdot \maxDeg^{1-(1/\pfs)}\cdot
\aspradius^{{1}/{\pfs}} \cdot \sum\nolimits_{ \kell \geq 0} (\delta \cdot \aspradius)^{{\kell}/{\pfs}} \enspace.
\end{align}
Then, \Cref{thrm:AdjacencyPotentialSpIn} follows by substituting $c=\frac{\zeta}{\aspradius}$ and 
$\delta=\frac{1-\varepsilon}{\aspradius}$  and noting that the norm on the l.h.s. is an upper bound for 
$\spradius( \infmatrix^{\Lambda,\tau}_{G})$.  It remains to show that \eqref{eq:Target4thrm:MySpectralMatrixNorm} 
is true.

For $\kw\in V\setminus \Lambda$, we let $\DBounded(\kw)$ be equal to the absolute row sum for the row of matrix 
$\SpGMatrix^{-1/\pfs} \cdot \infmatrix^{\Lambda, \tau}_G \cdot \SpGMatrix^{1/\pfs}$ that corresponds to vertex $\kw$. 
\Cref{thrm:InflNormBound4GeneralD} and \eqref{def:AdjMatrixEigenvals} imply  
\begin{align} 
\DBounded(\kw)
 & \leq 1+ \frac{c}{\left(\maxeigenv(\kw)\right)^{1/\pfs}} \cdot \left( 
\sum\nolimits_{\kz \in N(\kw)}\left(\maxeigenv(\kz)\right)^{1/\pfs}+ \sum\nolimits_{\kell \geq 2}
\left( \delta^{\kell-1} \cdot \left(\EdgeToV\cdot\Adjacency^{\kell-1}_G\cdot \maxeigenv\right) (\kw\kz) 
\right)^{1/\pfs} 
\right) 
\enspace. \label{eq:Base4Proofthrm:PowerMatrixTopologicalNorm}
\end{align}
Recalling that $\maxeigenv$ is the principal eigenvector of $\Adjacency_G$, \eqref{eq:PFT4thrm:AdjacencyPotentialSpIn} implies
\begin{align}
\EdgeToV\cdot\Adjacency^{\kell-1}_G\cdot \maxeigenv &= \aspradius^{\kell-1} \cdot \EdgeToV\cdot \maxeigenv\enspace. 
\end{align}
Plugging the above into \eqref{eq:Base4Proofthrm:PowerMatrixTopologicalNorm} and rearranging, we get 
\begin{align}\label{eq:fthrm:PowerMatrixTopologicalNormTempA}
\DBounded(\kw) & =1+ c \cdot \sum\nolimits_{\kell\geq 1} \left(\delta \cdot \aspradius \right)^{{(\kell-1)}/\pfs} 
\cdot \sum\nolimits_{\kz\in N_G(\kw)}\left( {\textstyle \frac{\maxeigenv(\kz)}{\maxeigenv(\kw)} } \right)^{1/\pfs} \enspace.
\end{align}
We need to bound the rightmost sum in the equation above.
Letting $d=\abs{N_G(\kw)}$ and applying H\"older's inequality,  we get 
\begin{align}\label{eq:fthrm:PowerMatrixTopologicalNormTempB}
\sum\nolimits_{\kz\in N_G(\kw)}\left( {\textstyle\frac{\maxeigenv(\kz)}{ \maxeigenv(\kw)} }\right)^{1/\pfs} & \leq
\left( \sum\nolimits _{\kz\in N_G(\kw)} {\textstyle \frac{\maxeigenv(\kz)}{ \maxeigenv(\kw)}} \right)^{1/\pfs} d^{1-(1/\pfs)}\ = \ 
\maxDeg^{1-(1/\pfs)} \cdot \aspradius^{1/\pfs}\enspace,
\end{align}
where the last equality follows from the observations that $\sum_{\kz\in N_G(\kw)}\maxeigenv(\kz)=\aspradius\cdot \maxeigenv(\kw)$
and $d^{1-(1/\pfs)}\leq \maxDeg^{1-(1/\pfs)}$. 
Plugging \eqref{eq:fthrm:PowerMatrixTopologicalNormTempB} into \eqref{eq:fthrm:PowerMatrixTopologicalNormTempA} 
we get 
\begin{align}
\DBounded(\kw) &\leq
 1+ c \cdot \maxDeg^{1-(1/\pfs)} \cdot \aspradius^{1/\pfs} \cdot
 \sum\nolimits_{\kell \geq 0} \left(\delta \cdot \aspradius \right)^{\kell/\pfs} \enspace . \nonumber 
\end{align}
The above holds for any $\kw\in V\setminus \Lambda$,  implying that \eqref{eq:Target4thrm:MySpectralMatrixNorm} is true. 
\Cref{thrm:AdjacencyPotentialSpIn} follows.
\end{proof}

\subsection{Proof of \Cref{thrm:InflNormBound4GeneralD}}\label{sec:thrm:InflNormBound4GeneralD}

We abbreviate $\infmatrix^{\Lambda, \tau}_G$ to $\infmatrix$ and the diagonal element $\UpD (\ku,\ku)$ to $\UpD (\ku)$.

Let $T=\Tsaw(G,\kw)$, while consider the weights $\{\infweight(e)\}$ over the edges of $T$ obtain 
by applying the $\Tsaw$-construction to the Gibbs distribution $\mu^{\Lambda,\tau}_G$.

For every path $\kP$ of length $\kell\geq 0$ that starts from the root of tree $T$, i.e., $\kP=\ke_1, \ldots \ke_{\kell}$, let 
\begin{align}\label{eq:WeightPVsInfls}
 {\tt weight}(\kP) & = \prod\nolimits_{1\leq \ki \leq \kell }\infweight(e_{\ki}) \enspace.
\end{align}
Note that we have written path $\kP$ using its edges. 

The following result is useful to involve the potential function $\potF$ in our derivations.

\begin{claim}\label{claim:AddPotentialA}
For integer $\kell>1$ and any path $\kP=\ke_1, \ldots \ke_{\kell}$ such that ${\tt weight}(\kP)\neq 0$,  we have 
\begin{align} \nonumber
\abs{ {\tt weight}(\kP)} &=\chiofh(\ke_{\kell}) \cdot \frac{\abs{\infweight(\ke_1)}}{\chiofh(\ke_1)} \cdot\prod^{\kell}_{\ki=2}\frac{\chiofh(\ke_{\ki-1})}{\chiofh(\ke_{\ki})} \cdot \abs{\infweight(\ke_{\ki})} \enspace, 
\end{align}
where $\chiofh(\ke)=\xdpotF(\infweight(\ke))$ and $\xdpotF(\cdot )$ is the derivative of the potential function $\potF$, i.e., $\xdpotF=\potF'$. 
\end{claim}

\begin{proof}
Our assumptions for the potential function $\potF$ imply that, for every $\ke_{\ki}\in \kP$, we have $\chiofh(\ke_{\ki})\neq 0$. 

Using a simple telescopic trick, we get 
\begin{align}
\abs{{\tt weight}(\kP)}&= \prod^{\kell}_{\ki=1}\frac{\chiofh(\ke_{\ki})}{\chiofh(\ke_{\ki})} \cdot \abs{\infweight(\ke_{\ki})} \ 
 = \ \chiofh(\ke_{\kell})\cdot \frac{\abs{\infweight(\ke_1)}}{\chiofh(e_1)}\cdot \prod^{\kell}_{\ki=2}\frac{\chiofh(\ke_{\ki-1})}{\chiofh(\ke_{\ki})}\cdot \abs{\infweight(\ke_{\ki})} \enspace. \nonumber
\end{align}
The claim follows.
\end{proof}

For $\ku\in V\setminus \Lambda$ and $\kell\geq 0$, let $\cP(\ku,\kell)$ be the set of paths $\kP$ from the root of $T$ to 
the set of copies of  vertex $\ku$ which are at level  $\kell$ of the tree such that  ${\tt weight}(\kP)\neq 0$. Also,  let 
$\cp(\ku,\kell)$ be the set of copies of $\ku$ at level $\kell$ in $T$ that are connected to the root with a path $\kP\in \cP(\ku,\kell)$.

Let
\begin{align}\label{eq:DefOfCalHL}
\inflPotB_{\kell} & =
 \sum\nolimits_{\ku \in V\setminus \Lambda} \sum\nolimits_{\kP \in \cP(\ku, \kell) }
\abs{ {\tt weight}(\kP ) } \cdot \UpD(\ku) \enspace.
\end{align}
Recalling that \Cref{prop:Inf2TreeRedaux} implies 
\begin{align}\label{eq:DescOfbCru}
\textstyle \left( \UpD^{-1}\cdot \infmatrix\cdot \UpD \right) (\kw,\ku) & =\frac{1}{\UpD(\kw)} \cdot \sum\nolimits_{\kell \geq 0} \sum\nolimits_{\kP \in \cP(\ku,\kell)}{\tt weight}(\kP) \cdot \UpD(\ku)\enspace,
\end{align}
we have
\begin{align}\label{eq:bCVsbCELL}
 \sum\nolimits_{\ku\in V\setminus \Lambda } \abs{\left( \UpD^{-1}\cdot \infmatrix\cdot \UpD \right) (\kw,\ku)} & 
 \leq \frac{1}{\UpD(\kw)} \sum\nolimits_{\kell\geq 0} \inflPotB_{\kell} 
\enspace .
\end{align}
Fix $\kell> 1$. For a vertex $\kv$ at level $\kh$ of $T$, where $\kh=0, \ldots, \kell$, 
let the quantity $\subcont_{\kv}$ be defined as follows: 
For $\kh=\kell$, we let
\begin{align} \label{eq:SRecurBaseWeights}
\subcont_{\kv} & = \sum\nolimits_{\ku \in V\setminus \Lambda} \Ind\{\kv \in \cp (\ku,\kell)\} \cdot \UpD(\ku) \enspace.
\end{align}
For a vertex $\kv$, which is at level $0< \kh <\kell$ such that vertices $\kv_1, \ldots, \kv_d$ are its children, we let
\begin{align} \label{eq:SRecurStepHC1821}
\subcont_{\kv}&= 
\chiofh(\kb_{\kv})\cdot \sum\nolimits_{\kj}
\frac{\abs{\infweight(\kb_{\kj})}}{\chiofh(\kb_{\kj})} \cdot \subcont_{\kv_{\kj}} \enspace,
\end{align}
where $\kb_{\kv}$ is the edge that connects $\kv$ to its parent, while $\kb_{\kj}$ is the edge that connects $\kv$
to its child $\kv_{\kj}$. Since we assume that $\kh>0$, $\kv$ has a parent,  hence $\kb_{\kv}$ is well-defined.
Also, recall $\chiofh(\cdot)$ from \Cref{claim:AddPotentialA}.

Finally, for $\kh=0$, i.e., $\kv$ and the root of $T$ are identical, we let 
\begin{align}\label{eq:SRecurTopWeights}
\subcont_{\kv}\ = \ \subcont_{\rm root} & =  
\max_{\bar{\kb},\hat{\kb}\in T} \left\{ {\textstyle \chiofh\left(\bar{\kb}\right) \cdot \frac{\abs{ \infweight\left(\hat{\kb}\right)}}{\chiofh\left(\hat{\kb}\right)}} \right\} 
\cdot \sum\nolimits_{\kz_{\kj}} \subcont_{\kz_{\kj}} \enspace,
\end{align}
where $\kz_1, \ldots, \kz_{\kr}$ are the children of the root of $T$. 

\begin{claim}\label{claim:CEllVsCalD}
For any $\kell> 1$, we have $\inflPotB_{\kell} \leq \subcont_{\rm root}$. 
\end{claim}
The proof of \Cref{claim:CEllVsCalD} appears immediately after the proof of \Cref{thrm:InflNormBound4GeneralD}.

The Boundedness property of the $(\pfs,\delta,c)$-potential $\potF$, together with \eqref{eq:SRecurTopWeights} 
imply
\begin{align}\label{eq:SRecurTopWeightsB}
\subcont_{\rm root} & \leq 
c \cdot \sum\nolimits_{\kz_j} \subcont_{\kz_j} \enspace.
\end{align}

\noindent
The Contraction property  of the $(\pfs,\delta,c)$-potential $\potF$ and \eqref{eq:SRecurStepHC1821}
imply that, for every vertex $\kv$ at level $0<\kh<\kell$, we have 
\begin{align}\label{eq:SubContVsDeltaKWeights}
\left( \subcont_{\kv} \right) ^{\pfs} & \leq 
\delta \cdot \sum\nolimits_{\kv_j} \left( \subcont_{\kv_j} \right)^{\pfs} \enspace,
\end{align}
where recall that $\kv_1,\ldots, \kv_d$ are the children of $\kv$. 
Then, \eqref{eq:SubContVsDeltaKWeights} and \eqref{eq:SRecurBaseWeights}, imply 
\begin{align}\label{eq:SubContVsDeltaKWeights1821}
 \left ( \subcont_{\kv} \right)^{\pfs} &
 \leq \delta^{\kell - \kh }\cdot \sum\nolimits_{\ku\in V\setminus \Lambda}  \abs{\cp(\ku,\kell)\cap T_{\kv}} \cdot 
\left ( \UpD(\ku)\right)^{\pfs} \enspace ,
\end{align}
where, recall that $\cp(\ku,\kell)\subseteq \cp(\ku)$ is the set of copies of vertex $\ku$ in $T$
that are at level $\kell$. $T_{\kv}$ is the subtree of $T$ that includes $\kv$ and its descendants.

Applying \eqref{eq:SubContVsDeltaKWeights1821} for the children of the root and plugging the result into \eqref{eq:SRecurTopWeightsB}, we get
\begin{align}\label{eq:FinalBoundRroorAdj}
 \subcont_{\rm root} & \leq 
 c\cdot \sum\nolimits_{\kz_{\kj}} \left( \delta^{\kell-1} \cdot \sum\nolimits_{\ku\in V\setminus \Lambda} \abs{\cp(\ku,\kell)\cap T_{\kz_{\kj} }} \cdot 
\left ( \UpD(\ku)\right )^{\pfs} \right )^{1/{\pfs}} \enspace,
\end{align}
where note that each $\kz_{\kj}$ is at level $\kh=1$ of $T$.

Suppose that $\kz_{\kj}$ is a copy of vertex $\kz$ in $G$. Then, since $\kell>1$,  we have 
\begin{align}\label{eq:CPUELLTzjVsMAlwzu}
\abs{\cp(\ku,\kell)\cap T_{\kz_j}}   \leq \left( \EdgeToV\cdot \Adjacency^{\kell-1}_G\right)(\kw\kz,\ku)\enspace, 
\end{align}
where recall matrix $\EdgeToV$ from \eqref{def:SingleEdgeM}, while $\Adjacency_G$ is the 
adjacency matrix of $G$. 
It is not hard to check that  both quantities in \eqref{eq:CPUELLTzjVsMAlwzu} count walks of length $\kell$ in $G$ starting
from the oriented edge $\kw\kz$ and ending at vertex $\ku$. To see that \eqref{eq:CPUELLTzjVsMAlwzu} is true, 
we note that the l.h.s. considers only self-avoiding walks, while the
 r.h.s. considers simple walks., i.e., including the self-avoiding ones.

Combining \eqref{eq:CPUELLTzjVsMAlwzu} with \eqref{eq:FinalBoundRroorAdj} and \Cref{claim:CEllVsCalD}, 
for $\kell>1$,  we have
\begin{align}\label{eq:FinalBoundCJLAdj}
\inflPotB_{\kell} & \leq 
c \cdot \sum\nolimits_{\kz \in N_G(\kw)} \left( \delta^{\kell-1} \cdot 
\sum\nolimits_{\ku\in V\setminus \Lambda} \left( \EdgeToV\cdot \Adjacency^{\kell-1}_G\right)(\kw\kz,\ku)
 \cdot 
\left ( \UpD(\ku)\right )^{\pfs} \right )^{1/\pfs} \enspace.
\end{align}
As far as $\inflPotB_1$ is concerned, note that for any edge $\ke\in T$, we have 
\begin{align} \label{eq:BeBoundFromPotBoundedness}
\abs{\infweight(\ke)}= \frac{\chiofh(\ke)}{\chiofh(\ke)} \cdot \abs{\infweight(\ke)} \leq \max_{\bar{\ke}, \hat{\ke} \in T} 
\left\{ \chiofh(\bar{\ke}) \cdot \frac{\abs{\infweight(\hat{\ke})}}{\chiofh(\hat{\ke})} \right\} \leq c \enspace,
\end{align}
where the last inequality follows from the Boundedness condition of the $(\pfs,\delta,c)$-potential function $\potF$.
Applying the definition of $\inflPotB_{\kell}$ for $\kell=1$, we get 
\begin{align}\label{eq:FinalBoundCJ1Adj}
\inflPotB_1
&=\sum\nolimits_{\kz \in N_G(\kw)}\UpD(\kz)\cdot \abs{\infweight(\ke_{\kz})} \leq  c \cdot \sum\nolimits_{\kz\in N_G(\kw)}\UpD(\kz)\enspace,
\end{align}
where  $\infweight(\ke_{\kz})$ corresponds to the weight in $T$ for the edge $\ke_{\kz}$ that connects the root 
of $T$ with the copy of vertex $\kz$.  The last inequality uses \eqref{eq:BeBoundFromPotBoundedness}.

In light of all the above, we get \eqref{eq:Target4thrm:InflNormBound4GeneralD} by plugging \eqref{eq:FinalBoundCJLAdj} 
and \eqref{eq:FinalBoundCJ1Adj} into  \eqref{eq:bCVsbCELL} and noting that $\inflPotB_0=\UpD(\kw)$.
\Cref{thrm:InflNormBound4GeneralD} follows. \hfill $\square$.

\begin{proof}[Proof of \Cref{claim:CEllVsCalD}]
For brevity, we let 
\begin{align}\nonumber 
{\tt Max}& =\max_{\bar{\ke}, \hat{\ke} \in T}
\left\{ {\textstyle \chiofh(\bar{\ke}) \cdot \frac{\abs{\infweight(\hat{\ke})} }{\chiofh(\hat{\ke})}} \right\} \enspace.
\end{align}
Recall that $\kell>1$. From the definition of $\inflPotB_{\kell} $ in \eqref{eq:DefOfCalHL} and  \Cref{claim:AddPotentialA} 
we have 
\begin{align}
\inflPotB_{\kell} &=  \sum\nolimits_{\ku\in V\setminus \Lambda} \sum\nolimits_{\kP=(\ke_1,\ldots, \ke_\kell) \in \cP(\ku, \kell) }
 \left( \chiofh(\ke_{\kell})\cdot \frac{\abs{\infweight(\ke_1)} }{\chiofh(\ke_1)} \right)\cdot \prod\nolimits_{2\leq \ki\leq \kell}\frac{\chiofh(\ke_{\ki-1})}{\chiofh(\ke_{\ki})}\cdot \abs{\infweight(\ke_{\ki})} \cdot \UpD(\ku) \enspace.
 \nonumber
\end{align}
For $\ku\in V\setminus \Lambda$ and $\kell\geq 0$, recall that  $\cP(\ku,\kell)$ is the set of paths $\kP$ from the root of $T$ to 
the set of copies of  vertex $\ku$ which are at level  $\kell$ of the tree such that  ${\tt weight}(\kP)\neq 0$. Also,  recall that 
$\cp(\ku,\kell)$ is the set of copies of $\ku$ at level $\kell$ in $T$ that are connected to the root with a path $\kP\in \cP(\ku,\kell)$.

Noting that $\left( \chiofh(\ke_{\kell})\cdot \frac{\abs{\infweight(\ke_1)}}{\chiofh(\ke_1)} \right)\leq {\tt Max}$, we get 
\begin{align}
\inflPotB_{\kell} 
 &\leq   {\tt Max}\cdot \sum\nolimits_{\ku\in V\setminus \Lambda} \sum\nolimits_{\kP=(\ke_1,\ldots, \ke_\kell) \in \cP(\ku, \kell) } 
\prod\nolimits_{2\leq \ki\leq \kell}\frac{\chiofh(\ke_{\ki-1})}{\chiofh(\ke_{\ki})}\cdot \abs{\infweight(\ke_{\ki})}\cdot \UpD(\ku) \enspace. \label{eq:basis4HlAtHlVsRroot}
\end{align}
As far as $\subcont_{\rm root}$ is concerned, recalling \eqref{eq:SRecurTopWeights}, we have
\begin{align}\label{eq:SRecurTopWeightsRepeat4HlVsRroot}
\subcont_{\rm root} &= {\tt Max}\cdot \sum\nolimits_{1\leq \kj\leq\kr} \subcont_{\kz_{\kj}} \enspace,
\end{align}
where $\kz_1, \ldots, \kz_{\kr}$ are the children of the root, for some integer $\kr>0$.

We show that for any vertex $\kv$ at level $\kh=1,\ldots, \kell-1$ of the tree $T$, we have
\begin{align}\label{eq:Target4InductionRvVsHl}
\subcont_{\kv} &=\sum\nolimits_{\ku\in V\setminus \Lambda} 
\sum\nolimits_{\kP=(\ke_1,\ldots, \ke_{\kell-\kh+1})\in \cP(\kv,\ku, \kell)} 
 \prod\nolimits_{2\leq \ki\leq \kell-\kh+1}
\frac{\chiofh(\ke_{\ki-1}) }{\chiofh(\ke_{\ki})} \cdot \abs{\infweight(\ke_{\ki})} \cdot \UpD(\ku) \enspace,
\end{align}
where $\cP(\kv,\ku, \kell)$ is the set of paths of length $\kell-\kh+1$ from the parent of vertex $\kv$ to 
$\cp(\ku,\kell)\cap T_{\kv}$.

Before proving that \eqref{eq:Target4InductionRvVsHl} is true, let us show how it implies \Cref{claim:CEllVsCalD}.
For every child $\kz_{\kj}$ of the root, \eqref{eq:Target4InductionRvVsHl} implies 
\begin{align}
\subcont_{\kz_{\kj}} & = \sum\nolimits_{\ku\in V\setminus \Lambda} \sum\nolimits_{\kP=(\ke_1,\ldots, \ke_{\kell}) \in \cP(\kz_{\kj}, \ku, \kell) } 
\prod\nolimits_{2\leq \ki\leq \kell} \frac{\chiofh(\ke_{\ki-1})}{\chiofh(\ke_{\ki})}\cdot \abs{\infweight(\ke_{\ki})} \cdot \UpD(\ku) \enspace. 
\end{align}
Summing the above over all the children of the root, we get
\begin{align}\label{eq:GoodExpression4Rzj}
\sum\nolimits_{\kz_{\kj}} \subcont_{\kz_{\kj}} & = \sum\nolimits_{\ku\in V\setminus \Lambda} \sum\nolimits_{\kP=(\ke_1,\ldots \ke_{\kell}) \in \cP(\ku, \kell) } 
\prod\nolimits_{2\leq \ki\leq \kell} \frac{\chiofh(\ke_{\ki-1})}{\chiofh(\ke_{\ki})}\cdot \abs{\infweight(\ke_{\ki})} \cdot \UpD(\ku) \enspace. 
\end{align}
For the above, we use the observation that the sets $\cP(\kz_j,\ku,\kell)$'s partition $\cP(\ku,\kell)$, i.e., we have
 $\cup_{\kz_{\kj}}\cP(\kz_{\kj},\ku,\kell)=\cP(\ku,\kell)$, while for  distinct ${\kj}, \ki$ the sets 
$\cP(\kz_{\kj},\ku,\kell)$ and $\cP(\kz_{\ki},\ku,\kell)$ do not intersect. 

 Plugging \eqref{eq:GoodExpression4Rzj} into \eqref{eq:SRecurTopWeightsRepeat4HlVsRroot}, we see that
 $\subcont_{\rm root}$ is equal to the r.h.s. of the inequality in \eqref{eq:basis4HlAtHlVsRroot}. This implies
 that $\inflPotB_{\kell} \leq \subcont_{\rm root}$, which also proves the claim.

We conclude the proof of \Cref{claim:CEllVsCalD} by showing that \eqref{eq:Target4InductionRvVsHl} is true. 
To this end, we use induction on the level $\kh$ of vertex $\kv$ in \eqref{eq:Target4InductionRvVsHl}.
The basis of the induction corresponds to taking $\kh=\kell-1$. For the inductive step, we assume that there is 
$\kh\leq \kell-1$ such that that \eqref{eq:Target4InductionRvVsHl} is true. We then prove that \eqref{eq:Target4InductionRvVsHl} 
is also true for any vertex $\kv$ at level $\kh-1$ of tree $T$.

We start with the basis and show \eqref{eq:Target4InductionRvVsHl} for a vertex $\kv$
at level $\kell-1$. 
As per standard notation, let   $\kv_1, \ldots, \kv_d$ be the children of  $\kv$. These children are vertices at level $\kell$.
From \eqref{eq:SRecurBaseWeights} and \eqref{eq:SRecurStepHC1821}, we have 
\begin{align}
\subcont_{\kv} &= \sum\nolimits_{\kj}\frac{\chiofh(\kb_{\kv} )}{\chiofh(\kb_{\kj})}\cdot \abs{ \infweight(\kb_{\kj})} \cdot \subcont_{\kv_{\kj}} \nonumber\\
&=\sum\nolimits_{\ku \in V\setminus \Lambda} \sum\nolimits_{\kj}\frac{\chiofh(\kb_{\kv} )}{\chiofh(\kb_{\kj})}\cdot \abs{ \infweight(\kb_{\kj})} \cdot 
\Ind\{\kv_{\kj} \in \cp (\ku,\kell )\}\cdot \UpD(\ku) \enspace, \label{eq:FinalofBaseHlVsRvInter}
\end{align}
where we let $\kb_{\kv}$ be the edge that connects $\kv$ to its parent at $T$, while edge $\kb_{\kj}$ connects
$\kv$ to its child $\kv_{\kj}$. For the second equality, we substitute $\subcont_{\kv_{\kj}}$ by using \eqref{eq:SRecurBaseWeights}.

For each $\ku \in V\setminus \Lambda$ in \eqref{eq:FinalofBaseHlVsRvInter}, 
the summand $\frac{\chiofh(\kb_{\kv} )}{\chiofh(\kb_{\kj})}\cdot \abs{\infweight(\kb_{\kj})} \cdot 
\Ind\{\kv_{\kj} \in \cp (\ku,\kell )\}$ is non-zero only if the edges $\kb_{\kv}$ and $\kb_{\kj}$ 
 form a path which belongs to set $\cP(\kv,\ku,\kell)$, i.e., we have
$(\kb_{\kv},\kb_{\kj})\in \cP(\kv,\ku,\kell)$.

We use this observation and write \eqref{eq:FinalofBaseHlVsRvInter} as follows:
\begin{align}
\subcont_{\kv}&=\sum\nolimits_{\ku \in V\setminus \Lambda} \sum\nolimits_{\kP=(\ke_1,\ke_2) \in \cP(\kv, \ku, \kell)}
\frac{\chiofh(\ke_{1}) }{\chiofh(\ke_2)} \cdot \abs{\infweight(\ke_2)}  \cdot \UpD(\ku) \enspace. 
\label{eq:FinalofBaseHlVsRv}
\end{align}
It is not hard to see that  \eqref{eq:FinalofBaseHlVsRv}  corresponds to \eqref{eq:Target4InductionRvVsHl} 
for $\kh=\kell-1$. This concludes the proof of the  base of the induction. 

The induction hypothesis is that there exists $1<\kh\leq \kell-1$, such that \eqref{eq:Target4InductionRvVsHl} 
is true for any vertex  at level $\kh$ of tree $T$. We show that the hypothesis implies that 
\eqref{eq:Target4InductionRvVsHl}  is true for any vertex $\kv$ at level $\kh-1$. 
As per standard notation, we let  $\kv_1, \ldots, \kv_d$  be the children of $\kv$, 
while edge $\kb_{\kv}$ connects $\kv$ to its parent and  edge $\kb_{\kj}$ connects
$\kv$ to its child $\kv_{\kj}$.  From \eqref{eq:SRecurStepHC1821} we have
\begin{align}
\subcont_{\kv} &= \sum\nolimits_{\kj} \frac{\chiofh(\kb_{\kv} )}{\chiofh(\kb_{\kj})}\cdot \abs{\infweight(\kb_{\kj})} \cdot \subcont_{\kv_{\kj}}  \enspace.  \nonumber
\end{align}
The children of $\kv$ are at level $\kh$. 
Using the  hypothesis, we substitute each $\subcont_{\kv_{\kj}}$  using \eqref{eq:Target4InductionRvVsHl} and get
\begin{align}\nonumber
\subcont_{\kv}&= \sum\nolimits_{\ku\in V\setminus \Lambda} 
\sum\nolimits_{\kj} \frac{\chiofh(\kb_{\kv} )}{\chiofh(\kb_{\kj})}\cdot \abs{\infweight(\kb_{\kj})}\cdot 
\sum\nolimits_{\kP=(\ke_1,\ldots, \ke_{\kell-\kh+1})\in \cP(\kv_{\kj}, \ku, \kell)}
 \prod\nolimits_{2\leq \ki\leq \kell-h+1}
\frac{\chiofh(\ke_{\ki-1}) }{\chiofh(\ke_{\ki})} \cdot \abs{\infweight(\ke_{\ki})} \cdot \UpD(\ku) \enspace.
\end{align}
It is elementary  that  for any $1\leq \kj\leq d$ and for  every path $\kP=(\ke_1,\ldots, \ke_{\kell-\kh+1})\in \cP(\kv_{\kj}, \ku, \kell)$, 
there is exactly one path  $\kQ=(\overline{\ke}_1, \ldots, \overline{\ke}_{\kell-\kh+2})\in \cP(\kv,\ku,\kell)$,  with $\overline{\ke}_2=\{\kv,\kv_{\kj} \}$, 
such that 
\begin{align} \nonumber
 \prod\nolimits_{2\leq \ki\leq \kell-h+2}
\frac{\chiofh(\overline{\ke}_{\ki-1}) }{\chiofh(\overline{\ke}_{\ki})} \cdot \abs{\infweight(\overline{\ke}_{\ki})} &=
\frac{\chiofh(\kb_{\kv} )}{\chiofh(\kb_{\kj})}\cdot \abs{\infweight(\kb_{\kj})} \cdot \prod\nolimits_{2\leq \ki\leq \kell-\kh+1}
\frac{\chiofh(\ke_{\ki-1}) }{\chiofh(\ke_{\ki})} \cdot \abs{\infweight(\ke_{\ki})} \enspace
\end{align}
and vice versa.  Combining the two relations above, we get 
\begin{align}
\subcont_{\kv}&= \sum\nolimits_{\ku\in V\setminus \Lambda} 
\sum\nolimits_{\kQ=(\overline{\ke}_1,\ldots, \overline{\ke}_{\kell-\kh+2})\in \cP(\kv, \ku, \kell)}
 \prod\nolimits_{2\leq \ki\leq \kell-\kh+2}
\frac{\chiofh(\overline{\ke}_{\ki-1}) }{\chiofh(\overline{\ke}_{\ki})} \cdot \abs{\infweight(\overline{\ke}_{\ki})} \cdot \UpD(\ku) \enspace, \nonumber 
\end{align}
which    proves the inductive steps and concludes the inductive proof. 
\Cref{claim:CEllVsCalD} follows. 
\end{proof}

\spreadpoint

 \begin{figure}
 \begin{minipage}{.45\textwidth}
 \centering
		\includegraphics[width=.515\textwidth]{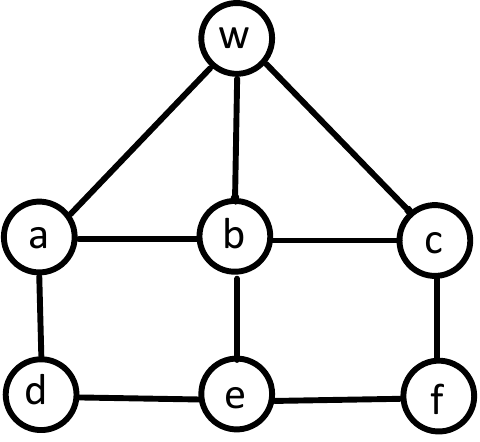}
		\caption{Initial graph $G$}
	\label{fig:GBeforeExt}
\end{minipage}
 \begin{minipage}{.5\textwidth}
 \centering
		\includegraphics[width=.51\textwidth]{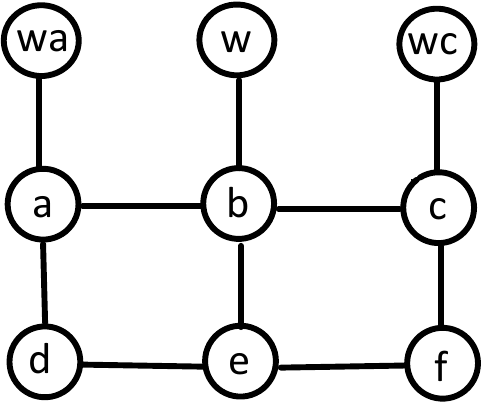}
		\caption{Graph $\gext{G}{\kP}$, for $\kP=\kw,\kb$}
	\label{fig:GAtRoot}
\end{minipage}
\end{figure}

\section{Taking Extensions \LastReviewG{2025-03-05}}\label{sec:Extensions}

In this section,  we introduce the notion of extension.  An extension is an operation that we typically apply to a 
graph $G$ or a Gibbs distribution on this graph.  Subsequently, we use extensions to define the extended influence 
matrix $\ExtdInfMatrixF$,  in \Cref{sec:ExtInfluenceMatrixNew}.

\subsection{The $\{\kP,\kQ \}$-extension of $G$:}\label{sec:TopologicalG}
We consider graph $G=(V,E)$ and a self-avoiding walk $\kP$ of length $\kk \geq 1$ in this graph. We introduce the notion of 
$\kP$-extension of $G$, denoted as $\gext{G}{\kP}$. This is a graph obtained after a set of 
operations on $G$. We use induction on the length $\kk$ of walk $\kP$ to define $\gext{G}{\kP}$.

First, consider $\kP$ of length $1$, while let $\kP=\ku, \kv$. Let $\sbsplit$ be the set of 
all neighbours of vertex $\ku$ in graph $G$ apart from vertex $\kv$, i.e., $\sbsplit=N_G(\ku)\setminus\{\kv\}$. 
 Graph $\gext{G}{\kP}$ is obtained by working as follows: 
\begin{enumerate}[(a)]
\item remove each edge of $G$ that connects $\ku$ to a vertex $\kz\in \sbsplit$,
\item  for each $\kz\in \sbsplit$, insert the new vertex $\ku\kz$ in the graph which is only connected to $\kz$.
\end{enumerate}
 \Cref{fig:GBeforeExt,fig:GAtRoot} show an example of the above operations. 
 \Cref{fig:GBeforeExt} shows the initial graph $G$. 
 \Cref{fig:GAtRoot} shows the $\kP$-extension of $G$, for walk $\kP=\kw,\kb$.

We call {\em split-vertices} all vertices introduced in step (b). 
We let $\ssplit_P$ denote the set of split-vertices in $\gext{G}{\kP}$. 
In the example shown in \Cref{fig:GBeforeExt,fig:GAtRoot}, we have $\ssplit_{\kP}=\{ \kw \ka, \kw\kc\}$.

Suppose, now, that $\kP \kx$ is a self-avoiding walk in $G$, which is of length $k>1$. 
With $\kP\kx$, we indicate that vertex $\kx$ is the last vertex of the walk, while $\kP$ corresponds 
to the first $\kk$ vertices.
 $\kP$ is a self-avoiding walk of length $\kk-1$ in $G$. We show how to obtain $\gext{G}{\kP\kx}$ using $\gext{G}{\kP}$.

It is easy to see that $\kP\kx$ is also a self-avoiding walk in $\gext{G}{\kP}$. Let $\ku$ be the last vertex in $\kP$, i.e., vertex $\ku$ 
is just before $\kx$ in $\kP\kx$. Let $\sbsplit$ contain all neighbours of $\ku$ in $\gext{G}{\kP}$ 
apart from $\kx$, the vertex in $\kP$ prior to $\ku$ and the vertices in $\ssplit_{\kP}$. 
If $\sbsplit\neq \emptyset$, then $\gext{G}{\kP\kx}$ is the result of the following operations on $\gext{G}{\kP}$:
\begin{enumerate}[(a)] \setcounter{enumi}{2}
\item remove each edge in $\gext{G}{\kP}$ that connects $\ku$ to a vertex $\kz\in \sbsplit$,
\item for each $\kz\in \sbsplit$, we insert the split-vertex $\ku\kz$, which is only connected to $\kz$.
\end{enumerate}
If $\sbsplit= \emptyset$, graphs $\gext{G}{\kP}$ and $\gext{G}{\kP\kx}$ are identical.

For graph $\gext{G}{\kP\kx}$, the set of split vertices $\ssplit_{\kP\kx}$ corresponds to the union of 
the split-vertices $\ssplit_{\kP}$ and the newly introduced split-vertices at step (d). 
We follow the convention to denote with $\ku\kz$ the split-vertex  obtained from 
vertex $\ku$ of the initial graph $G$ and is connected to vertex $\kz$.

\Cref{fig:GP-PWBE,fig:GP-PWBEF} provide further examples of extensions of
the graph shown in \Cref{fig:GBeforeExt}. Specifically, \Cref{fig:GP-PWBE} illustrates 
graph $\gext{G}{\kP}$ for the self-avoiding walk $\kP=\kw,\kb,\ke$. In \Cref{fig:GP-PWBEF}, we have 
graph $\gext{G}{\kP}$ for walk $\kP=\kw,\kb,\ke,\kf$.

We view graph $\gext{G}{\kP}$ as consisting of the original set of vertices, 
i.e., those from the initial graph $G$,  and $\ssplit_{\kP}$. In the example of \Cref{fig:GP-PWBE}, graph 
$\gext{G}{\kP}$ consists of the set of vertices of the initial graph $G$ from \Cref{fig:GBeforeExt} and the 
set of split-vertices $\ssplit_{\kP}=\{\kw\ka, \kb\ka, \kb\kc, \kw\kc\}$.

Having defined the $\kP$-extension of $G$, we introduce the related 
notion of the $\{\kP,\kQ\}$-extension of $G$.
For integer $\kk\geq 1$, let $\kP$ and $\kQ$ be two self-avoiding walks in $G$, each of length $\kk$. 
Let $\kw$ and $\ku$ be the starting vertices of $\kP$ and $\kQ$, respectively. We say that $\kP$ and 
$\kQ$ are {\em compatible} when $\dist_G(\ku,\kw)\geq 2\kk$.

For compatible $\kP$ and $ \kQ$, the $\{\kP,\kQ \}$-extension of graph $G$, denoted as $\gext{G}{\kP,\kQ}$, is obtained 
by taking the $\kQ$-extension of graph $\gext{G}{\kP}$.

 The order in which we consider $\kP$ and $\kQ$ to obtain the $\{\kP,\kQ\}$-extension of graph $G$ 
 does not matter. It is easy to check that, as long as $\kP$ and $\kQ$ are compatible, one obtains the same 
 graph $\gext{G}{\kP,\kQ}$ by either taking the $\kP$-extension of $\gext{G}{\kQ}$ or the $\kQ$-extension of $\gext{G}{\kP}$.

The notion of $\{\kP,\kQ\}$-extension is well-defined since the compatibility of $\kP,\kQ$ ensures that 
$\kQ$ is a self-avoiding walk of $\gext{G}{\kP}$ and $\kP$ is a self-avoiding walk of $\gext{G}{\kQ}$.
For $\kP$ and $\kQ$ that are not compatible, we do not define $\gext{G}{\kP,\kQ}$.

 \begin{figure}
 \begin{minipage}{.47\textwidth}
 \centering
		\includegraphics[width=.515\textwidth]{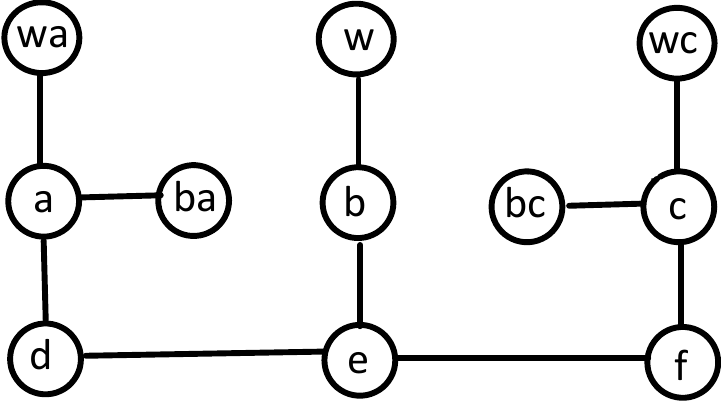}
		\caption{$\gext{G}{\kP}$, for $\kP=\kw,\kb,\ke$}
	\label{fig:GP-PWBE}
\end{minipage}
 \begin{minipage}{.5\textwidth}
 \centering
		\includegraphics[width=.51\textwidth]{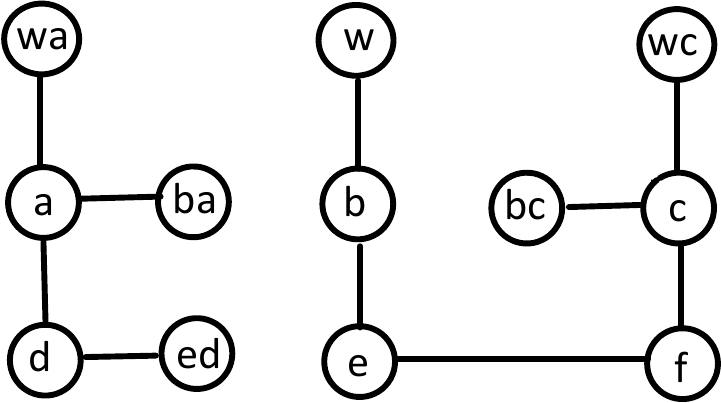}
		\caption{$\gext{G}{\kP}$, for $\kP=\kw,\kb,\ke,\kf$}
	\label{fig:GP-PWBEF}
\end{minipage}
\end{figure}

\subsection{The $\{\kP,\kQ\}$-extension of $\mu^{\Lambda,\tau}_G$:}\label{sec:TopologicalMu}
Consider graph $G=(V,E)$ and a Gibbs distribution $\mu_G$ for this graph.
Without loss of generality, assume a {\em total ordering} for the vertices in graph $G$.

Let $\Lambda\subset V$ and $\tau\in \{\pm 1\}^{\Lambda}$. 
Also, let the self-avoiding walk $\kP$ in graph $G$, which is of length $\kk \geq 1$. 
We take $\kP$ such that it does not intersect with set $\Lambda$. 
Then, consider $\gext{G}{\kP}$, i.e., the $\kP$-extension of $G$. 
Recall  that we follow the convention to call $\ku\kz$ the split-vertex in $\ssplit_{\kP}$ obtained from 
vertex $\ku$ of the initial graph $G$ and is connected with vertex $\kz$.

The $\kP$-extension of the Gibbs distribution $\mu^{\Lambda,\tau}_G$
 is a new Gibbs distribution with underlying graph $\gext{G}{\kP}$ having the same specifications as
$\mu^{\Lambda,\tau}_G$.

More specifically, we let $\mu^{M,\sigma}_{\gext{G}{\kP}}$ be the $\kP$-extension of $\mu^{\Lambda,\tau}_G$
such that the set of pinned vertices $M=\Lambda\cup \ssplit_{\kP}$ while pinning $\sigma$ satisfies the
following:
 for $\ks \in \Lambda\subseteq M$, we have $\sigma(\ks)=\tau(\ks)$, while 
for $\kq\kx\in \ssplit_{\kP}$, we have
\begin{align}\label{eq:DefOfTau_x}
\sigma(\kq\kx) &= \left\{
\begin{array}{lcl}
+1& \quad & \textrm{if $\kz>\kx$}\enspace,\\ 
-1& \quad &\textrm{if $\kz<\kx$} \enspace,
\end{array}
\right. 
\end{align}
where $\kz$ is the vertex after $\kq$ in path $\kP$. The comparison between $\kz$ and $\kx$ is with respect to the total ordering 
of the vertices in $G$. 

Having defined the $\kP$-extension of $\mu^{\Lambda,\tau}_G$, 
we introduce the related notion of the $\{\kP,\kQ\}$-extension of $\mu^{\Lambda,\tau}_G$.
For integer $\kk\geq 1$, let $\kP$ and $\kQ$ be two self-avoiding walks in $G$, each of 
length $\kk$. Let $\kw$ and $\ku$ be the starting vertices of $\kP$ and $\kQ$, respectively. 
 We assume that $\kP$ and $\kQ$ are compatible, i.e., $\dist_{G}(\ku,\kw)\geq 2\kk$.

We obtain the $\{\kP,\kQ \}$-extension of $\mu^{\Lambda,\tau}_G$ by first taking $\mu^{M,\sigma}_{\gext{G}{\kP}}$, i.e., 
 the $\kP$-extension of $\mu^{\Lambda,\tau}_G$, and then taking the $\kQ$-extension of $\mu^{M,\sigma}_{\gext{G}{\kP}}$. 
The $\{\kP,\kQ\}$-extension of $\mu^{\Lambda,\tau}_G$ is a Gibbs distribution on graph $\gext{G}{\kP,\kQ}$.

As in the case of graphs, for $\kP$ and $\kQ$ which are compatible, the order in which we consider them to obtain the 
$\{\kP,\kQ\}$-extension of $\mu^{\Lambda,\tau}_G$ does not make any difference to the outcome. 
For $\kP,\kQ$ that are not compatible, we do not define the $\{\kP,\kQ\}$-extension of $\mu^{\Lambda,\tau}_G$.

\newcommand{\Tw}{\magenta{T}}
\newcommand{\STw}{\capri{t}}
\newcommand{\XTPw}{\blue{T_{\kP}}}

\subsection{Properties of the extensions}\label{sec:PropPextG}\label{sec:PropPextMu}
\label{sec:PropsOfExtensions}

The results in this section are technical and are used later for  the proofs of 
\Cref{thrm:InflVsExtInfl,thrm:InfNormBoundHConj}.
The reader may skip this section on the first read of this paper.

We consider two different settings for the extensions and we study their properties.

\subsubsection{Two simple properties:}\label{sec:BasicPropExt}
The results presented  in \Cref{sec:BasicPropExt} are used to prove \Cref{thrm:InflVsExtInfl}. 

Let integer $\kk>0$, graph $G=(V,E)$, set $\Lambda\subset V$ and $\kw\in V\setminus \Lambda$. 
We let $\ExtV_{\Lambda, \kk}$ be the set of all self-avoiding walks in $G$ of length $\kk$ and do not intersect 
with $\Lambda$.

Consider $\Tw=\Tsaw(G,\kw)$. For each walk $\kP=\kx_0,\ldots,\kx_k$ in $\ExtV_{\Lambda,\kk}$ that
emanates from $\kw$,   there is a path $R_{\kP}=\ku_0,\ldots, \ku_{\kk}$ in tree $\Tw$, which emanates from the root, 
while $\ku_{\ki}$ is a copy of $\kx_{\ki}$, for all $0\leq \ki \leq \kk $. 

Let $\STw_{\kP}$ be the subtree of $\Tw$ induced by the following set of vertices:
\begin{enumerate}[(i)]
\item the vertices in path $R_{\kP}$,
\item all neighbours of $R_{\kP}$ which are copies of a vertex in $\kP$,
\item  all the descendants of vertex $\ku_{\kk}$ in $\Tw$, recall that $\ku_{\kk}$ is the last vertex in $R_{\kP}$.
\end{enumerate}
Also, let $\XTPw=\Tsaw(\gext{G}{\kP},\kw)$. Tree $\XTPw$ is well-defined since $\kw$ is also a vertex in $\gext{G}{\kP}$. 

In what follows, we show that we can {\em identify} $\STw_{\kP}$ with $\XTPw$. This means that there is 
 a bijection ${\uph}$ from the set of vertices in $\STw_{\kP}$ and those in $\XTPw$,
such that the following holds:
\begin{enumerate}[(a)]
\item  $\uph$ preserves adjacencies, i.e., if vertices $\kx$ and $\kz$ are adjacent in $\STw_{\kP}$, then $\uph(\kx)$ and $\uph(\kz)$ are adjacent in $\XTPw$,
\item for a vertex $\ku$ in $\STw_{\kP}$, which is a copy of vertex $\kv$ in $G$, we have that $\uph(\ku)$ is a copy of vertex $\kv$, or a 
copy of a split-vertex in $\gext{G}{\kP}$, generated by $\kv$. 
\end{enumerate}

\begin{lemma}\label{lemma:TsawPVsTsaw}
For walk $\kP\in \ExtV_{\Lambda,\kk}$ that emanates from vertex $\kw$, we  identify $\STw_{\kP}$ with $\XTPw$, where the two trees are defined above.
\end{lemma}

The proof of \Cref{lemma:TsawPVsTsaw} appears in \Cref{sec:lemma:TsawPVsTsaw}.

For $\Lambda\subset V$ and $\tau\in \{\pm 1\}^{\Lambda}$, consider the Gibbs distribution $\mu^{\Lambda,\tau}_G$. 
Also, consider the collection of weights $\{ \infweight(e)\}$ over the edges of $\Tw=\Tsaw(G,\kw)$ we obtain as described 
in \eqref{def:OfInfluenceWeights} with respect to $\mu^{\Lambda,\tau}_G$.

Fix $\kP\in \ExtV_{\Lambda,k}(w)$. Using \Cref{lemma:TsawPVsTsaw},  we identify $\XTPw$ as a subtree of $\Tw$. 
We endow $\XTPw$ with the edge weights $\{\infweight(e)\}_{e\in \XTPw}$.

Applying the $\Tsaw$-construction to the $\kP$-extension of $\mu^{\Lambda,\tau}_G$ gives rise to the weights $\{\infweight_{\kP} (e)\}$ 
over the edges of $\XTPw$. Now we have two sets of weights for the edges of $\XTPw$, i.e., $\{\infweight(e)\}_{e\in \XTPw}$ and 
$\{\infweight_{\kP}(e)\}_{e\in \XTPw}$.

The following result establishes a connection between these two sets of weights on the edges of $\XTPw$.

\begin{lemma}\label{lemma:GibbsWeightsTpVsP}
For walk $\kP\in \ExtV_{\Lambda,\kk}$ that emanates from vertex $\kw$, consider $\XTPw=\Tsaw(\gext{G}{\kP},w)$ and the weights $\{ \infweight(e)\}$ and $\{\infweight_{\kP}(e)\}$ 
over the edges of $\XTPw$, as defined above.
For all edges $\ke$ in $\XTPw$, which are at distance $\kell \geq \kk$ from the root of the tree, we have 
$\infweight(\ke)=\infweight_{\kP}(\ke)$. 
\end{lemma}

The proof of \Cref{lemma:GibbsWeightsTpVsP} appears in \Cref{sec:lemma:GibbsWeightsTpVsP}.

\newcommand{\SMTu}{ \applegreen{T}}
\newcommand{\SMTPu}{\overline{\cadmiumgreen{T}}}

\newcommand{\cpT}{\byzantine{\cp}_{\kP}}
\newcommand{\cpTP}{\magenta{\overline{\cp}}_{\kP}}

\subsubsection{Perturbation of weights:}\label{sec:NonSoBasicExtensions}

The properties of the extension presented in this section, i.e., \Cref{sec:NonSoBasicExtensions}, are used in the proof
of \Cref{thrm:InfNormBoundHConj}. The setting we consider here is different from the one in \Cref{sec:BasicPropExt}.

For integer $\kk>0$, consider $G=(V,E)$, set $\Lambda\subset V$ and $\tau\in \{\pm 1\}^{\Lambda}$. For 
$\kw, \ku\in V\setminus \Lambda$, such that $\dist_{G}(\kw,\ku)\geq 2\kk$ , let walk $\kP\in \ExtV_{\Lambda, \kk}$  
that emanates from $\kw$. 

Consider the Gibbs distribution $\mu^{\Lambda,\tau}_G$.  Using the $\Tsaw$-construction for $\mu^{\Lambda,\tau}_G$, 
we obtain the Gibbs distribution $\mu^{M,\sigma}_{\SMTu}$ for tree $\SMTu=\Tsaw(G,\ku)$. Also,  let $\{ \infweight(e)\}$ 
be the collection of weights over the edges of $\SMTu$ that arise  from $\mu^{M,\sigma}_{\SMTu}$.

We use the $\Tsaw$-construction for the $\kP$-extension of  $\mu^{\Lambda,\tau}_G$ and get the Gibbs distribution 
$\mu^{\kL,\eta}_{\SMTPu}$ for tree $\SMTPu=\Tsaw(\gext{G}{\kP},\ku)$.  Let $\{\infweight_{\kP}(\ke)\}$ be the collection 
of weights over the edges of $\SMTPu$ that arise  from $\mu^{\kL,\eta}_{\SMTPu}$.

As opposed to the setting in \Cref{sec:BasicPropExt}, in this section we do {\em not} start the trees of self-avoiding walks 
from  vertex $\kw$, i.e., the first vertex in $\kP$; we rather start the trees from a distant vertex $\ku$.

\begin{lemma}\label{lemma:Subtree4TSawU}
We can identify $\SMTPu=\Tsaw(\gext{G}{\kP}, \ku)$ as a subtree of $\SMTu=\Tsaw(G, \ku)$ such that both trees have 
the same root.

Each vertex in $\SMTu \setminus \SMTPu$ corresponds to a descendant of a vertex in $\cpT$,
where set $\cpT$ consist of each vertex $\kv$ in $\SMTu$ which is a copy of a vertex in $\kP$.
\end{lemma}

The proof of \Cref{lemma:Subtree4TSawU} appears in \Cref{sec:lemma:Subtree4TSawU}.

We use \Cref{lemma:Subtree4TSawU} and identify $\SMTPu$ as a subtree of $\SMTu$.
We endow $\SMTPu$ with the collection of weights
$\{ \infweight(\ke)\}_{\ke\in \SMTPu}$. 
Now, we have two sets of weights for $\SMTPu$, i.e., $\{\infweight(\ke)\}_{\ke\in \SMTPu}$ and
$\{\infweight_{\kP} (\ke)\}_{\ke\in \SMTPu}$.

 In the following result,  we consider a path $\kW$ from the root of $\SMTPu$ to a copy of vertex $\kw$
and describe a relation between the weights of this path that are induced by $\{\infweight_{\kP}(\ke)\}$ and 
$\{\infweight(\ke)\}$, respectively.

\begin{theorem}\label{thrm:ProdWideHatBetaVsBeta}
Let $\maxDeg > 1$, $\kk\geq 1$ and $N>1$ be integers, $c>0$, $\SingBound> 1$, $\pfs\geq 1$, while 
let $b,\varepsilon\in (0,1)$. Also, let $\beta,\gamma, \lambda\in \mathbb{R}$ be such that $\gamma >0$,
$0\leq \beta\leq \gamma$ and $\lambda>0$.

Consider graph $G=(V,E)$ of maximum degree $\maxDeg$ such that $\nnorm{ (\NBMatrix)^N}{1/N}{2}= \SingBound$. 
Assume that $\mu$, the Gibbs distribution on $G$ specified by the parameters $(\beta, \gamma, \lambda)$,
is $b$-marginally bounded. For $\delta= \frac{1-\varepsilon}{\SingBound}$, suppose that there
is a $(\pfs,\delta,c)$-potential function $\potF$ with respect to $(\beta,\gamma,\lambda)$.

There are constants $\ell_0> 2\kk$, $\widehat{\kc}_1\geq 1$ such that,  for $\Lambda\subset V$ and 
$\tau\in \{\pm 1\}^{\Lambda}$, for any $\ku,\kw\in V\setminus \Lambda$ with  $\dist(\ku,\kw)\geq 2\kk$ 
and walk $\kP\in \ExtV_{\Lambda,\kk}$ that emanates from $\kw$, the following is true:

For $\kell\geq \ell_0$, for path $\kW=\ke_1, \ldots, \ke_{\kell}$ from the root of $\SMTPu=\Tsaw(\gext{G}{\kP}, \ku)$ 
to a copy of vertex $\kw$, we have
\begin{align}\nonumber
\prod\nolimits_{1\leq \ki \leq \kell} \abs{\infweight_{\kP}(\ke_{\ki})} & \leq \widehat{\kc}_1 \cdot (1+\varepsilon/3)^{\kell/\pfs} \cdot \prod\nolimits_{1\leq \ki\leq \kell-\kk} \abs{\infweight (\ke_{\ki})} \enspace. 
\end{align}
\end{theorem}
The proof of \Cref{thrm:ProdWideHatBetaVsBeta} appears in \Cref{sec:thrm:ProdWideHatBetaVsBeta}.

\spreadpoint
\section{The extended influence matrix $\ExtdInfMatrixF$ \LastReviewG{2025-03-05}}
\label{sec:ExtInfluenceMatrixNew}
The reader needs to be familiar with the notion of extensions we introduce in \Cref{sec:TopologicalG,sec:TopologicalMu} to appreciate the content of this section.

Let integer $k\geq 1$. We consider graph $G=(V,E)$ and a Gibbs distribution $\mu_G$, defined as in \eqref{def:GibbDistr}. Furthermore, 
we have $\Lambda\subset V$ and $\tau\in \{\pm 1\}^{\Lambda}$. 

Recall that $\ExtV_{\kk}$ is the set of self-avoiding walks of length $\kk$ in $G$. 
Also, set $\ExtV_{\Lambda, \kk}\subseteq \ExtV_{\kk}$ consists of all walks $\kw_0, \ldots, \kw_{\kk} \in \ExtV_{\kk}$ such that 
$\kw_{\ki}\notin \Lambda$, for $\ki=0,\ldots, \kk$.

The extended influence matrix $\ExtdInfMatrix=\ExtdInfMatrixF$ is an $\spset_{\Lambda, \kk}\times \spset_{\Lambda, \kk}$ 
matrix such that for $\kP, \kQ \in \ExtV_{\Lambda, \kk}$, we define the entry $\ExtdInfMatrix(\kP,\kQ) $ as follows:
We have 
\begin{align}\label{def:EdgeInflMatrixA}
\ExtdInfMatrix(\kP, \kQ) &=0 \enspace, 
\end{align}
if $\kP$ and $\kQ$ are not compatible, i.e., the starting vertices of $\kP$ and $\kQ$ are at distance $< 2\kk$.

Consider $\kP$ and $\kQ$ which are compatible, while let $\kw$ and $\ku$ be their starting vertices, 
respectively.  Letting $\upzeta^{\kP,\kQ}$ be the $\{\kP, \kQ\}$-extension of $\mu^{\Lambda,\tau}_G$
we have
\begin{align}\label{def:EdgeInflMatrix}
\ExtdInfMatrix(\kP, \kQ) &=
\upzeta^{\kP,\kQ}_{\ku}(+1\ |\ (\kw,+1)) - \upzeta^{\kP,\kQ}_{\ku}(+1 \ |\ (\kw,-1)) \enspace,
\end{align}
where $\upzeta^{\kP,\kQ}_{\ku}$ is the marginal of $\upzeta^{\kP,\kQ}$ at vertex $\ku$.

Letting $ \infmatrix^{\kP,\kQ}$ be the influence matrix induced by $\upzeta^{\kP,\kQ}$, 
\eqref{def:EdgeInflMatrix} implies 
\begin{align}\label{def:EdgeInflMatrixEquivalent}
\ExtdInfMatrix(\kP, \kQ) &=\infmatrix^{\kP,\kQ}(\kw,\ku)\enspace, 
\end{align}
where, as mentioned above,   $\kw$ and $\ku$ are the starting vertices of $\kP$ and $\kQ$, respectively.

Concluding the basic definition of the extended influence matrix, it is useful to note that the entries of matrix 
$\ExtdInfMatrix$ do not refer to a single Gibbs distribution.  As $\kP$ and $\kQ$ vary, the corresponding 
$\{\kP,\kQ\}$-extensions of $\mu^{\Lambda,\tau}_G$ are  Gibbs distributions which are different from each 
other. This makes it  challenging to estimate matrix norms for $\ExtdInfMatrix$.

\subsection*{Connections between $\ExtdInfMatrixF$ and $\infmatrix^{\Lambda,\tau}_G$:}
In the following result, we show that
$\ExtdInfMatrixF$ enjoys a natural relation with the (standard) influence matrix $ \infmatrix^{\Lambda,\tau}_G$,
which is induced by the Gibbs distribution $\mu^{\Lambda,\tau}_G$.

For integer $\kk>0$,  let $\infmatrix^{\Lambda,\tau}_{G, <\kk}$ be the $(V\setminus \Lambda)\times (V\setminus \Lambda)$ matrix 
such that, for any $\ku,\kw\in V\setminus \Lambda$, we have
\begin{align}\label{eq:DefOfCISmallerK}
 \infmatrix^{\Lambda,\tau}_{G,<\kk}(\ku,\kw) &= \Ind\{\dist_{G}(\ku,\kw)< \kk\} \cdot \infmatrix^{\Lambda,\tau}_{G}(\ku,\kw) \enspace. 
\end{align}
Let $\VToEdge_{\kk}$ be the $(V\setminus \Lambda)\times \spset_{\Lambda,\kk}$, zero-one matrix such that for 
$\kr \in V\setminus \Lambda$ and $\kP \in \spset_{\Lambda,\kk}$, we have 
\begin{align} \label{def:Vertex2EdgeMatrixEntryBB}
\VToEdge_{\kk}(\kr, \kP) &= \Ind\{ \textrm{$\kr$ is the starting vertex of walk $\kP$} \} \enspace.
\end{align}
Also, let $\EdgeToV_{\kk}$ be the $\ExtV_{\Lambda, \kk}\times (V\setminus \Lambda)$, zero-one matrix, such that for 
 $\kr \in V\setminus \Lambda$ and $\kP \in \spset_{\Lambda,\kk}$, we have
\begin{align} \label{def:Edge2VertexMatrixEntryBB}
\EdgeToV_{\kk} (\kP, \kr) &= \Ind\{ \textrm{$\kr$ is the last vertex of walk $\kP$}\} \enspace.
\end{align}

\begin{theorem}\label{thrm:InflVsExtInfl}
Let $\maxDeg> 1$ and $\kk> 0$, let $\beta,\gamma, \lambda, b \in \mathbb{R}$ be such that $\gamma >0$,
$0\leq \beta\leq \gamma, \lambda>0$ and $b>0$. 

Let $G=(V,E)$ be of maximum degree $\maxDeg$. Consider $\mu_G$ the Gibbs distribution on $G$ 
specified by the parameters $(\beta, \gamma, \lambda)$, 
while assume that $\mu_G$ is $b$-marginally
bounded.

There is a constant $\widehat{\kC}>0$ such that, for any $\Lambda\subset V$ and $\tau\in \{\pm 1\}^{\Lambda}$, 
there is an $\ExtV_{\Lambda,k}\times \ExtV_{\Lambda, k }$ matrix $\SymWeightM$ over $\mathbb{R}$, 
for which we have
\begin{align}\label{eq:thrm:InflVsExtInfl}
 \infmatrix^{\Lambda,\tau}_G& =\VToEdge_{\kk} \cdot \left( \ExtdInfMatrixF\circ \SymWeightM \right )\cdot \EdgeToV_{\kk}
+ \infmatrix^{\Lambda,\tau}_{G, <2\kk} \enspace, 
\end{align}
while for any $\kP,\kQ\in \ExtV_{\Lambda,\kk}$ we also have $\abs{\SymWeightM(\kP,\kQ)}  \leq \widehat{\kC}$.

Recall that $\ExtdInfMatrixF\circ \SymWeightM$ is the Hadamard product of the two matrices. 
Matrix $\infmatrix^{\Lambda,\tau}_{G, <2\kk}$ is defined in \eqref{eq:DefOfCISmallerK}. 
\end{theorem}

The proof of \Cref{{thrm:InflVsExtInfl}} appears in \Cref{sec:thrm:InflVsExtInfl}.

\begin{proposition}\label{proposition:CINorm2VsCLNorm2}
Let $\maxDeg> 1$ and $k> 0$, let $\beta,\gamma, \lambda, b \in \mathbb{R}$ be such that $\gamma >0$,
$0\leq \beta\leq \gamma, \lambda>0$ and $b>0$. 

Let $G=(V,E)$ be of maximum degree $\maxDeg$. Consider $\mu_G$ the Gibbs distribution on $G$ 
specified by the parameters $(\beta, \gamma, \lambda)$, 
while assume that $\mu_G$ is $b$-marginally
bounded.

There are constants $C_1, C_2>0$ such that, for any $\Lambda\subset V$ and any $\tau\in \{\pm 1\}^{\Lambda}$, 
we have
\begin{align}\nonumber
\norm{ \infmatrix^{\Lambda,\tau}_G}{2}  \leq C_1+ C_2 \cdot \norm{ \abs{\ExtdInfMatrixF}}{2} \enspace.
\end{align}
\end{proposition}

Recall that, for matrix $\ExtdInfMatrixF$, we let $\abs{\ExtdInfMatrixF}$ denote the matrix 
having entries $\abs{\ExtdInfMatrixF(\kP,\kQ)}$.
The proof of \Cref{proposition:CINorm2VsCLNorm2} appears in \Cref{sec:proposition:CINorm2VsCLNorm2}.


\spreadpoint

\section{Proof of \Cref{thrm:NonBacktrackingInfinityMixing} - SI with $\delta$-contraction \& $\NBMatrix$ \LastReviewG{2025-03-12}}
\label{sec:thrm:NonBacktrackingInfinityMixing}

In this section, we assume that the reader is familiar with the notion of extended influence matrix $\ExtdInfMatrixF$
that was introduced in \Cref{sec:ExtInfluenceMatrixNew}.

We prove \Cref{thrm:NonBacktrackingInfinityMixing} by using the following result. 

\begin{theorem}\label{thrm:L2InflRedux2KNBTM}
Let $\delta, b \in \mathbb{R}_{>0}$ and integers $\maxDeg, \kk\geq 1$. 
Also, let $\beta,\gamma, \lambda\in \mathbb{R}$ be such that $0\leq \beta\leq \gamma$, 
$\gamma >0$ and $\lambda>0$. 

Consider graph $G=(V,E)$ of maximum degree $\maxDeg$ while 
assume that $\mu_G$, the Gibbs distribution on $G$ specified by the parameters $(\beta, \gamma, \lambda)$,
is $b$-marginally bounded.

Suppose  the set of functions $\{ \logtrecur_d\}_{1\leq d< \maxDeg}$ 
specified by $(\beta,\gamma,\lambda)$ exhibits $\delta$-contraction.
There are constants $C_1,C_2>0$ such that, for any $\Lambda\subset V$ and $\tau\in \{\pm 1\}^{\Lambda}$,
 the influence matrix $\infmatrix^{\Lambda,\tau}_{G}$, induced by $\mu^{\Lambda,\tau}_G$, satisfies 
\begin{align}\nonumber
 \norm{ \infmatrix^{\Lambda,\tau}_{G}}{2}  &\leq 
 C_1+C_2 \cdot \sum\nolimits_{\kell\geq \kk}  \delta^{\kell} \cdot \norm{ ( \NBMatrix)^{\kell}}{2}  \enspace .
\end{align}
\end{theorem}
The proof of \Cref{thrm:L2InflRedux2KNBTM} appears in \Cref{sec:thrm:L2InflRedux2KNBTM}.

\begin{proof}[Proof of \Cref{thrm:NonBacktrackingInfinityMixing}]
\Cref{thrm:L2InflRedux2KNBTM} implies that for  $\delta =\frac{1-\varepsilon}{\SingBound}$,  there are constants $C_1,C_2>0$ such that
\begin{align}\label{eq:ResultFromthrm:L2InflRedux2KNBTM}
 \norm{ \infmatrix^{\Lambda,\tau}_{G}}{2}  &\leq C_1+C_2 \cdot \sum\nolimits_{\kell \geq \kk} 
\delta^{\kell} \cdot \norm{ (\NBMatrix)^{\kell}}{2} 
 \enspace,
\end{align}
where recall that $\SingBound=\nnorm{ (\NBMatrix)^N}{1/N}{2}$. 

 \Cref{lemma:SingSequenConv} implies that there is a bounded number $\ell_0=\ell_0(N,\varepsilon, \maxDeg)>k$ such that, 
for any $\kell \geq \ell_0$, we have 
\begin{align}\label{eq:KVsHigherSingVals}
\nnorm{ (\NBMatrix)^{\kell}}{1/\kell}{2}
& \leq   (1+{\varepsilon}/{2}) \SingBound\enspace.
\end{align}
From \eqref{eq:ResultFromthrm:L2InflRedux2KNBTM}, we  have
\begin{align}\label{eq:Base4thrm:NonBacktrackingInfinityMixing}
 \norm{ \infmatrix^{\Lambda,\tau}_{G}}{2}  &\leq C_1+C_2 \cdot \sum\nolimits_{\kk \leq \kell< \ell_0} 
\delta^{\kell} \cdot \norm{ ( \NBMatrix)^{\kell}}{2} +
C_2 \cdot \sum\nolimits_{\kell \geq \ell_0} 
\delta^{\kell} \cdot \norm{( \NBMatrix)^{\kell}}{2}
 \enspace. 
\end{align}
Since $\ell_0$ and $C_1,C_2$ are bounded numbers, there is ${C}_0>1$, which is a bounded number, too, such that
\begin{align}\label{eq:Sum4BelowL0}
C_1+C_2 \cdot \sum\nolimits_{\kk \leq \kell < \ell_0} 
\delta^{\kell} \cdot \norm{ ( \NBMatrix)^{\kell}}{2} \leq {C}_0 \enspace.
\end{align}
For the above, we use the observation that, for any $0<\kell<\ell_0$, we have that 
$\norm{( \NBMatrix)^{\kell}}{2}  \leq (\maxDeg-1)^{\kell}$, while
$\maxDeg$ is also a bounded number.
Furthermore, we have 
\begin{align}
C_2 \cdot \sum\nolimits_{\kell \geq \ell_0} \delta^{\kell} \cdot \norm{( \NBMatrix)^{\kell}}{2} 
& \leq  C_2 \cdot \sum\nolimits_{\kell \geq \ell_0} 
\delta^{\kell} \cdot \left((1+\varepsilon/2)\SingBound\right)^{\kell} & \mbox{[from \eqref{eq:KVsHigherSingVals}]} \nonumber \\
&\leq C_2 \cdot \sum\nolimits_{\kell \geq \ell_0} 
\left( 1-\varepsilon \right)^{\kell} \cdot \left(1+\varepsilon/2\right )^{\kell} & \mbox{[since $\delta=\frac{1-\varepsilon}{\SingBound}$]} \nonumber\\
&\leq 2C_2 \cdot \varepsilon^{-1} \enspace. \label{eq:Sum4AboveL0}
\end{align}

\noindent
\Cref{thrm:NonBacktrackingInfinityMixing} follows by plugging \eqref{eq:Sum4BelowL0} and \eqref{eq:Sum4AboveL0} into
 \eqref{eq:Base4thrm:NonBacktrackingInfinityMixing} and setting $\widehat{C}=2C_2+{C}_0$.
\end{proof}

\subsection{Proof of \Cref{thrm:L2InflRedux2KNBTM}}\label{sec:thrm:L2InflRedux2KNBTM}
Recall that set $\ExtV_{\Lambda,\kk}$ consists of the self-avoiding walks of length $\kk$ in $G$
which do not intersect with set $\Lambda$.

Let $\infmatrixB$ be the $\ExtV_{\Lambda,k}\times \ExtV_{\Lambda,k}$ matrix 
such that, for any $\kP,\kQ\in \ExtV_{\Lambda, \kk}$, the entry $\infmatrixB(\kP,\kQ)$ is defined as follows:
Let $\kw$ and $\ku$ be the starting vertices of $\kP$ and $\kQ$, respectively.
If $\kP$ and $\kQ$ are not compatible, i.e., $\dist_G(\ku,\kw)< 2\kk$, we let $\infmatrixB(\kP,\kQ)=0$.

For $\kP$ and $\kQ$ which are compatible, we consider graph $\gext{G}{\kP,\kQ}$, the $(\kP,\kQ)$-extension of $G$,
and let $T_{\kP,\kQ}=\Tsaw(\gext{G}{\kP,\kQ},\kw)$. We let
\begin{align}\nonumber
\infmatrixB(\kP,\kQ)&= \sum\nolimits_{\kell \geq 2\kk}\delta^{\kell}\times \abs{ \scp(\ku,\kell)} \enspace,
\end{align}
where   we let set $\scp(\ku,\kell)$ consist of all copies of vertex  $\ku$ in $T_{\kP,\kQ}$ which correspond to 
self-avoiding walks of length  $\kell$ in $\gext{G}{\kP,\kQ}$.

\begin{lemma} \label{lemma:eq:CLVsCJEntrywise}
We have $\abs{\ExtdInfMatrixF} \leq \infmatrixB$,  where the inequality is entry-wise. 
\end{lemma}

The proof of \Cref{lemma:eq:CLVsCJEntrywise} appears in \Cref{sec:lemma:eq:CLVsCJEntrywise}.

As per standard notation, $\abs{\ExtdInfMatrixF}$ in \Cref{lemma:eq:CLVsCJEntrywise} corresponds 
to the matrix having entries $\abs{\ExtdInfMatrixF(\kP,\kQ)}$. 
\Cref{proposition:CINorm2VsCLNorm2} implies that there are constants $c_1,c_2>0$ such that
\begin{align}
\norm{ \infmatrix^{\Lambda,\tau}_G}{2}  &\leq c_1+ c_2 \cdot \norm{\abs{\ExtdInfMatrixF}}{2}  
\ \leq \   c_1+ c_2 \cdot \norm{ \infmatrixB}{2}  \enspace. \label{eq:IVsJ2Norm}
\end{align}
The second inequality follows from \Cref{lemma:eq:CLVsCJEntrywise} and the observation that 
$\abs{\ExtdInfMatrixF}$ and $\infmatrixB$ have non-negative entries.

\begin{proposition}\label{prop:CJVskNBM} 
We have 
\begin{align}\label{eq:thrm:DomOfInflB}
\norm{ \infmatrixB}{2}  &\leq 
\sum \nolimits_{\kell\geq \kk} \delta^{\kell+\kk} \cdot \norm{ ( \NBMatrix )^{\kell}}{2}  \enspace.
\end{align}
\end{proposition}

The proof of \Cref{prop:CJVskNBM} appears in \Cref{sec:prop:CJVskNBM}.

\Cref{thrm:L2InflRedux2KNBTM} follows by plugging \eqref{eq:thrm:DomOfInflB} into
 \eqref{eq:IVsJ2Norm} and setting $C_1=c_1$ and $C_2=c_2\cdot \delta^{\kk}$. 
\hfill $\square$

\subsection{Proof of \Cref{lemma:eq:CLVsCJEntrywise}}\label{sec:lemma:eq:CLVsCJEntrywise}
 
 When there is no danger of confusion, in what follows we abbreviate matrix $\ExtdInfMatrixF$ to $\ExtdInfMatrix$. 
 
 Let $\kP,\kQ \in \ExtV_{\Lambda, \kk}$, while let $\kw$ and $\ku$ be their starting vertices, respectively.
 For non-compatible $\kP, \kQ$, the definitions of matrices $\ExtdInfMatrix$ and
$ \infmatrixB$ imply that $\ExtdInfMatrix(\kP,\kQ)= \infmatrixB(\kP,\kQ)=0$. 
We consider compatible $\kP$ and $\kQ$. It suffices to show that $\abs{\ExtdInfMatrix(\kP,\kQ)} \leq \infmatrixB(\kP,\kQ)$.

We let $\upzeta^{\kP,\kQ}$ be the $\{\kP,\kQ\}$-extension of $\mu^{\Lambda,\tau}_G$,  while let $\infmatrix^{\kP,\kQ}$ be 
the influence matrix induced by $\upzeta^{\kP,\kQ}$.
Recall from \eqref{def:EdgeInflMatrixEquivalent} that 
\begin{align}\label{eq:IPQVsCalLPQ11-3}
\ExtdInfMatrix(\kP, \kQ)=\infmatrix^{\kP,\kQ}(\kw,\ku)\enspace,
\end{align}
where, as mentioned above,   $\kw$ and $\ku$ are the starting vertices of $\kP$ and $\kQ$,  respectively.

Let $T_{\kP,\kQ}=\Tsaw(\gext{G}{\kP,\kQ}, \kw)$. 
Also, let $\{\infweight(\ke)\}$ be the collection of weights over the edges of $T_{\kP,\kQ}$ we obtain as 
described in \eqref{def:OfInfluenceWeights} with respect to the Gibbs distribution $\upzeta^{\kP,\kQ}$. 
We have 
\begin{align}\label{eq:IPQVsEdgeWeights11-3}
\infmatrix^{\kP,\kQ}(\kw,\ku)&=\sum\nolimits_{\kell \geq 2\kk}\sum\nolimits_{\kW\in T_{\kP,\kQ} : \abs{\kW}=\kell} \prod\nolimits_{\ke\in \kW} \infweight(\ke)\enspace, 
\end{align}
where $\kW$ varies over the paths of length $\kell$ in $T_{\kP,\kQ}$ from the root to the copies of vertex $\ku$
corresponding to self-avoiding walks in $\gext{G}{\kP,\kQ}$. Any copies of $\ku$ in $T_{\kP,\kQ}$ that do not correspond
to self-avoiding walk have a pinning. The influence of the root to  vertices with pinning is zero.

The assumption about $\delta$-contraction implies that $\upzeta^{\kP,\kQ}$ also exhibits $\delta$-contraction, i.e., 
 it has the same specifications as $\mu^{\Lambda,\tau}_G$. Hence, we have 
$\abs{ \prod\nolimits_{\ke\in \kW } \infweight(\ke) } \leq \delta^{\kell}$
for each $\kW$ of length $\kell$.
Combining this observation with \eqref{eq:IPQVsCalLPQ11-3} and \eqref{eq:IPQVsEdgeWeights11-3}, we obtain 
\begin{align}\nonumber
\abs{\ExtdInfMatrix(\kP, \kQ)}=\abs{\infmatrix^{\kP,\kQ}(\kw,\ku)}&\leq \sum\nolimits_{\kell\geq 2\kk}\sum\nolimits_{\kW\in T_{\kP,\kQ} : \abs{\kW}=\kell} \delta^{\kell}
=\sum\nolimits_{\kell\geq 2\kk} \delta^{\kell} \cdot \abs{\scp(\ku,\kell)} = \infmatrixB(\kP,\kQ) \enspace.
\end{align}
The above proves that indeed $\abs{\ExtdInfMatrix(\kP,\kQ)} \leq \infmatrixB(\kP,\kQ)$ for $\kP$ and $\kQ$ which are compatible. 

\Cref{lemma:eq:CLVsCJEntrywise} follows.  \hfill $\square$

\subsection{Proof of \Cref{prop:CJVskNBM}}\label{sec:prop:CJVskNBM}
First, we show that the following entry-wise inequality
\begin{align}\label{eq:Final4prop:CJVskNBM}
\infmatrixB \leq \sum\nolimits_{\kell\geq \kk}\delta^{\kell+\kk}\cdot  ( ( \NBMatrix )^{\kell}\cdot \Invol )_{\ExtV_{\Lambda,\kk}}\enspace,
\end{align}
where $\Invol$ is the involution defined in \eqref{def:OfInvolutation}, while $\left( ( \NBMatrix )^{\kell}\cdot \Invol \right)_{\ExtV_{\Lambda,k}}$ is the principal submatrix of $( \NBMatrix )^{\kell}\cdot \Invol$
 induced by the elements in $\ExtV_{\Lambda,k}$.

Consider $\kP,\kQ\in \ExtV_{\Lambda,k}$, which are compatible. Let $\kw$ and $\ku$
be the starting vertices of $\kP$ and $\kQ$, respectively. 
Consider $\gext{G}{\kP,\kQ}$, the $(\kP,\kQ)$-extension of graph $G$ and let $T_{\kP,\kQ}=\Tsaw(\gext{G}{\kP,\kQ},\kw)$.

From the definition of matrix $\infmatrixB$, we have
\begin{align}\label{eq:RedefInfmatrixBP-Q}
\infmatrixB(\kP,\kQ)&=\sum\nolimits_{\kell\geq 2\kk}\delta^{\kell}\cdot \abs{\scp(\ku,\kell)} \enspace,
\end{align}
where recall that $\scp(\ku,\kell)$ is the set of copies of vertex $\ku$ in $T_{\kP,\kQ}$ which correspond to 
self-avoiding walks of length $\kell$ in $\gext{G}{\kP,\kQ}$.  
The vertex $\ku$ is the starting vertex of $\kQ$. Note that the compatibility assumption for $\kP$  and $\kQ$ implies 
$\dist_{G}(\kw,\ku)\geq 2\kk$. Hence the  copies of $\ku$ in $T$ are at distance $\geq 2\kk$ from the root.

For $\kell\geq 2\kk$,  we obtain a bound on $\abs{\scp(\ku,\kell)}$ in terms of the $\kk$-non-backtracking 
matrix $\NBMatrix$. 
W.l.o.g. let $\kP=\kv_0, \ldots, \kv_{\kk}$ and $\kQ=\kx_0,\ldots, \kx_{\kk}$. The cardinality of set $\scp(\ku,\kell)$ is 
equal to the number of self-avoiding walks $\kR=\kz_0,\ldots, \kz_{\kell}$ in $\gext{G}{\kP,\kQ}$, satisfying the following 
condition: for $0\leq i \leq \kk$, we have
\begin{itemize}
\item $\kz_{\ki}$ is a copy of vertex $\kv_{\ki}$ and $\kz_{\kell-{\ki}}$ is a copy of $\kx_{\ki}$. 
\end{itemize}
Recall that the additional vertices that $\gext{G}{\kP,\kQ}$ has compared to $G$ are split-vertices, each of which 
is of degree $1$. No split-vertex can be a part of a self-avoiding walk from $\kw$ to $\ku$.
Hence, every self-avoiding walk $\kR$ in $\gext{G}{\kP,\kQ}$ that satisfies the above condition also appears in $G$.

Furthermore, since every self-avoiding walk of length $\kell$ from $\kw$ to $\ku$ in $G$ can also be considered as a 
$\kk$-non-backtracking walk of length $\kell-\kk$ from $\kP$ to $\kQ^{-1}$, we conclude that 
\begin{align}\nonumber
\abs{\scp(\ku,\kell)} &\leq  ( \NBMatrix )^{\kell-\kk}(\kP,\kQ^{-1}) = ( ( \NBMatrix)^{\kell-\kk} \cdot \Invol ) (\kP,\kQ) \enspace.
\end{align}
Walk $\kQ^{-1}$ corresponds to $\kQ$ reversed, i.e., $\kQ^{-1}=\kx_{\kk},\kx_{\kk-1}, \ldots, \kx_0$.
Plugging the above inequality into \eqref{eq:RedefInfmatrixBP-Q}, we get that \eqref{eq:Final4prop:CJVskNBM} is 
true for compatible  walks $\kP$ and $\kQ$.

For non-compatible $\kP,\kQ$, the definition of $\infmatrixB$ specifies  that $\infmatrixB(\kP,\kQ)=0$. This implies that 
\eqref{eq:Final4prop:CJVskNBM} is true even for non-compatible  $\kP$ and $\kQ$ since we always have
$( ( \NBMatrix)^{\kell-\kk} \cdot \Invol ) (\kP,\kQ)\geq 0$.

All the above establish that \eqref{eq:Final4prop:CJVskNBM} is true. Furthermore, \eqref{eq:Final4prop:CJVskNBM}, 
the fact that both $\infmatrixB$  and $\left( \NBMatrix\right)^{\kell-\kk}\cdot \Invol$ have non-negative entries and the 
triangle inequality imply
\begin{align}\nonumber
\norm{ \infmatrixB}{2} &
\leq \sum\nolimits_{\kell\geq \kk}\delta^{\kell+\kk}\cdot \norm{ (( \NBMatrix)^{\kell} \cdot \Invol )_{\ExtV_{\Lambda, \kk}}}2 
\leq \sum\nolimits_{\kell\geq \kk}\delta^{\kell+\kk}\cdot \norm{ ( \NBMatrix)^{\kell} \cdot \Invol }{2} 
\leq \sum\nolimits_{\kell\geq \kk}\delta^{\kell+\kk}\cdot \norm{ ( \NBMatrix)^{\kell} }{2} \enspace. \nonumber
\end{align}
The second inequality follows from Cauchy’s  interlacing theorem. Note that   matrix $(\NBMatrix)^{\kell} \cdot \Invol$ is symmetric 
due to \eqref{eq:DefPTInvarianceFormalB}.
 The third inequality follows since $\norm{ ( \NBMatrix)^{\kell} \cdot \Invol }{2}\leq \norm{ ( \NBMatrix)^{\kell}}{2}\cdot \norm{ \Invol }{2}$ while, for the involutionary matrix $\Invol$,  
we have $\norm{ \Invol}{2}=1$.

\Cref{prop:CJVskNBM} follows.  \hfill $\square$

\spreadpoint

\section{Proof of \Cref{thrm:NonBacktrackingPotentialSpIn} - SI with Potentials \& $\NBMatrix$ \LastReviewG{2025-04-02}}
\label{sec:thrm:NonBacktrackingPotentialSpIn}

In this section, we assume that the reader is familiar with the notion of the extended influence matrix $\ExtdInfMatrixF$
 introduced in \Cref{sec:ExtInfluenceMatrixNew}.

In what follows, for a matrix $\UpD \in \mathbb{R}^{ N \times N}$, we let $\abs{\UpD}$ denote the matrix having entries $\abs{\UpD_{i,j}}$.

\begin{theorem}\label{thrm:InfNormBoundHConj}
Let $\maxDeg > 1$ and $\kk,N\geq 1$ be integers, $c>0$, $\SingBound> 1$, $\pfs\geq 1$ while 
let $b,\varepsilon\in (0,1)$. Also, let $\beta,\gamma, \lambda\in \mathbb{R}$ be such that $\gamma >0$,
$0\leq \beta\leq \gamma$ and $\lambda>0$. 

Consider graph $G=(V,E)$ of maximum degree $\maxDeg$ such that 
$\nnorm{ (\NBMatrix)^N}{1/N}{2}= \SingBound$.
Assume that $\mu_G$, the Gibbs distribution on $G$ specified by the parameters $(\beta, \gamma, \lambda)$,
is $b$-marginally bounded.
For $\delta= \frac{1-\varepsilon}{\SingBound}$, suppose that there
is a $(\pfs,\delta,c)$-potential function $\potF$ with respect to $(\beta,\gamma,\lambda)$.

There exists a fixed constant $\widehat{c} \geq 1$ such that 
for any $\Lambda\subset V$ and $\tau\in \{\pm 1\}^{\Lambda}$, the extended influence matrix 
$\ExtdInfMatrixF$, that is induced by $\mu^{\Lambda,\tau}_G$, satisfies that
\begin{align}\nonumber
\norm{  \abs{ \ExtdInfMatrixF}}{2\pfs} &\leq \widehat{c}\cdot \sum\nolimits_{\kell\geq \kk} 
\left(
 {\textstyle \left(\frac{1-\varepsilon/2}{\SingBound}\right)^{\kell} } \cdot %
\norm{ ( \NBMatrix)^{\kell}}{2} 
 \right)^{1/\pfs}
\enspace.
\end{align}
\end{theorem}

The proof of \Cref{thrm:InfNormBoundHConj} appears in \Cref{sec:thrm:InfNormBoundHConj}.

\begin{proof}[Proof of \Cref{thrm:NonBacktrackingPotentialSpIn}]
When there is no danger of confusion,  in this proof, we abbreviate 
$\infmatrix^{\Lambda,\tau}_{G}$ and $\ExtdInfMatrixF$ to $\infmatrix$ and $\ExtdInfMatrix$, respectively.

From \Cref{proposition:CINorm2VsCLNorm2}, there are constants $c_1, c_2>0$ such that 
\begin{align}\label{eq:L2BoundInflNoNuMatrix}
\norm{ \infmatrix}{2}  \leq c_1+c_2\cdot \norm{ \abs{ \ExtdInfMatrix }}{2}  \enspace.
\end{align} 

\Cref{thrm:NonBacktrackingPotentialSpIn} follows by bounding appropriately $\norm{ \abs{ \ExtdInfMatrix_{\kk} }}{2}$. 
To this end, we use the following claim.

\begin{claim}\label{claim:Ll2VsLl2s}
There is a constant $c_3>0$ such that $\norm{ \abs{ \ExtdInfMatrixF }}{2}  \leq c_3\cdot \norm{ \abs{ \ExtdInfMatrixF }}{2\pfs}$.
\end{claim}

\Cref{thrm:InfNormBoundHConj} and \Cref{claim:Ll2VsLl2s} imply that  there exists a fixed number ${c}_4\geq 1$, such that 
\begin{align}\label{eq:ExtInfN2sBoundWithDistL}
\norm{ \abs{ \ExtdInfMatrix }}{2}  &\leq 
{c}_4\cdot \sum\nolimits_{\kell\geq \kk} 
\left(
 {\textstyle \left(\frac{1-\varepsilon/2}{\SingBound}\right)^{\kell}} \cdot %
\norm{ ( \NBMatrix)^{\kell}}{2}  \right)^{1/\pfs}\enspace.
\end{align}
where recall that $\nnorm{ (\NBMatrix)^N}{1/N}{2}= \SingBound$.

 \Cref{lemma:SingSequenConv} implies that there is a bounded integer $\ell_0=\ell_0(N,\varepsilon, \maxDeg)>0$ such that, for any $\kell\geq \ell_0$, 
 we have 
\begin{align}\label{eq:CorFromlemma:SingSequenConv}
\nnorm{ (\NBMatrix)^{\kell}}{1/\kell}{2}
\leq (1+\varepsilon/4) \SingBound\enspace. 
\end{align}
Since $\ell_0$ is bounded, there exists a constant $c_5>0$ such that 
\begin{align}\label{eq:CorFromlemma:SingSequenConvBB}
{c}_4\cdot \sum\nolimits_{k \leq \kell < \ell_0} 
\left(
 {\textstyle \left(\frac{1-\varepsilon/2}{\SingBound}\right)^{\kell}} \cdot 
\norm{ (\NBMatrix)^{\kell}}{2}   \right)^{1/\pfs} \leq c_5\enspace. 
\end{align}
For the above, we use that $\norm{ (\NBMatrix)^{\kell}}{2} \leq \nnorm{\NBMatrix}{\kell}{2}\leq (\maxDeg-1)^{\kell}$. 
Plugging \eqref{eq:CorFromlemma:SingSequenConvBB} and \eqref{eq:CorFromlemma:SingSequenConv} into \eqref{eq:ExtInfN2sBoundWithDistL}, we obtain
\begin{align}
\norm{ \abs{ \ExtdInfMatrix } }{2}  &\leq 
c_5+ {c}_4\cdot \sum\nolimits_{\kell\geq \ell_0} 
\left( {\textstyle (1+\varepsilon/4)(1-\varepsilon/2) } 
 \right)^{\kell/\pfs} 
 \ \leq \ c_5+\frac{c_4}{1-(1-\varepsilon/4)^{1/\pfs}}
 \enspace. \nonumber
\end{align}
\Cref{thrm:NonBacktrackingPotentialSpIn} follows by plugging the above into \eqref{eq:L2BoundInflNoNuMatrix} and 
setting $C=c_1+c_2\cdot (c_4+c_5)$. 
\end{proof}

\begin{proof}[Proof of \Cref{claim:Ll2VsLl2s}]
For what follows, we abbreviate $\ExtdInfMatrixF$ to $\ExtdInfMatrix$. 

Consider $\kP, \kQ\in \ExtV_{\Lambda,\kk}$ which are compatible. 
Let $\kw$ and $\ku$ be the starting vertices of $\kP$ and $\kQ$, respectively. 
Then, as already noticed in \eqref{def:EdgeInflMatrixEquivalent}, we have 
\begin{align}\label{eq::InfNormHStartVsH-HVsI}
\ExtdInfMatrix(\kP,\kQ)&=\infmatrix^{\kP,\kQ}(\kw,\ku), 
& \textrm{and} && 
\ExtdInfMatrix(\kQ,\kP)&=\infmatrix^{\kP,\kQ}(\ku, \kw)\enspace, 
\end{align}
where the influence matrix $\infmatrix^{\kP,\kQ}$ is with respect to $\upzeta^{\kP,\kQ}$, the 
$\{\kP, \kQ\}$-extension of $\mu^{\Lambda,\tau}_G$. Note that $\infmatrix^{\kP,\kQ}$ and
$\infmatrix^{\kQ,\kP}$ correspond to the same influence matrix. 

A useful observation is that \Cref{claim:InfSymmetrisation} and \eqref{eq::InfNormHStartVsH-HVsI} imply that
if $\upzeta^{\kP,\kQ}_{\kw}(+1), \upzeta^{\kP,\kQ}_{\ku}(+1)\notin \{0,1\}$, then we have 
\begin{align}\label{eq::InfNormHStartVsH-HVsIBB}
\sqrt{
\frac{\upzeta^{\kP,\kQ}_{\kw}(+1) \cdot \upzeta^{\kP,\kQ}_{\kw}(-1)} 
{\upzeta^{\kP,\kQ}_{\ku}(+1) \cdot \upzeta^{\kP,\kQ}_{\ku}(-1) }
}\cdot
\ExtdInfMatrix(\kP,\kQ)&= 
\sqrt{\frac{\upzeta^{\kP,\kQ}_{\ku}(+1) \cdot \upzeta^{\kP,\kQ}_{\ku}(-1)}
{\upzeta^{\kP,\kQ}_{\kw}(+1) \cdot \upzeta^{\kP,\kQ}_{\kw}(-1)} }
\cdot 
\ExtdInfMatrix(\kQ,\kP)\enspace.
\end{align}
If at least one of $\upzeta^{\kP,\kQ}_{\kw}(+1), \upzeta^{\kP,\kQ}_{\ku}(+1)$ is in $\{0,1\}$, 
then we have $\ExtdInfMatrix(\kP,\kQ)=\ExtdInfMatrix(\kQ,\kP)=0$.

Let the $\ExtV_{\Lambda,\kk}\times \ExtV_{\Lambda,\kk}$ matrix $\cM$ be defined as follows:
For  $\kP \in \ExtV_{\Lambda,\kk}$ with starting vertex $\kw$ and 
$\kQ \in \ExtV_{\Lambda,\kk}$ with starting vertex $\ku$, such that $\kP$ and $\kQ$
are compatible while 
 $\upzeta^{\kP,\kQ}_{\kw}(+1), \upzeta^{\kP,\kQ}_{\ku}(+1)\notin \{0,1\}$, we let 
\begin{align}\label{eq:DefOfHatL}
\cM(\kP,\kQ)&=
\sqrt{
\frac{\upzeta^{\kP,\kQ}_{\kw}(+1) \cdot \upzeta^{\kP,\kQ}_{\kw}(-1)} 
{\upzeta^{\kP,\kQ}_{\ku}(+1) \cdot \upzeta^{\kP,\kQ}_{\ku}(-1) }
}
\cdot
\abs{\ExtdInfMatrix(\kP, \kQ)} \enspace. 
\end{align}
On the other hand, if at least one of $\upzeta^{\kP,\kQ}_{\kw}(+1), \upzeta^{\kP,\kQ}_{\ku}(+1)$ is in $\{0,1\}$, 
 then we let $\cM(\kP,\kQ)=0$. Also, we let  $\cM(\kP,\kQ)=0$ when $\kP$ and $\kQ$ are not compatible. 

$\cM$ has non-negative entries, while \eqref{eq::InfNormHStartVsH-HVsIBB} implies that it is symmetric. 
Furthermore, we have 
\begin{align}\label{eq:L2BoundInflNoNuMatrixNoMaxB}
\norm{ \abs{\ExtdInfMatrix}}{2}   &\leq b^{-2} \cdot \norm{  \cM }{2} 
 \leq b^{-2} \cdot \norm{ \cM}{2\pfs} 
 \leq b^{-4} \cdot \norm{ \abs{\ExtdInfMatrix}}{2\pfs} \enspace. 
\end{align}
The first inequality follows by noting that   \eqref{eq:DefOfHatL} and the $b$-marginal boundedness assumption imply 
that $\abs{\ExtdInfMatrix} \leq b^{-2} \cdot \cM$, where the comparison is entry-wise.  The second inequality follows since 
$\cM$ is symmetric. Similarly to the first inequality, the third inequality follows by noting that  \eqref{eq:DefOfHatL} and the $b$-marginal
boundedness assumption imply that $\cM$ has non-negative entries, while  $\cM \leq b^{-2}\cdot \abs{\ExtdInfMatrix}$, 
where the comparison is entry-wise. 

\Cref{claim:Ll2VsLl2s} follows for $c_3=b^{-4}$.
\end{proof}

\spreadpoint

\newcommand{\ExtdInfMatrixL}{\mathcal{D}} 

\section{Proof of \Cref{thrm:InfNormBoundHConj} \LastReviewG{2025-04-03} } \label{sec:thrm:InfNormBoundHConj}

\begin{figure}
 \centering
\begin{tikzpicture}[node distance=1.8cm]
\node (HCNBack) [startstop,thick, fill=black!20] {\Cref{thrm:InfNormBoundHConj}};
\node (HCNBackIS) [startstop, right of=HCNBack, xshift=2cm,thick, fill=black!20] {\Cref{thrm:HC-VEntryBound}}; 
\node (HCNBackNormRecurence) [startstop, right of=HCNBackIS, xshift=2cm,thick, fill=black!20] {\Cref{claim:CQLVsCKL}};
\node (HCNBackMultiRecur) [startstop, right of=HCNBackNormRecurence, xshift=2cm, thick, fill=black!20] {\Cref{thrm:ProdWideHatBetaVsBeta}};
\draw [arrow] (HCNBackIS) -- (HCNBack);
\draw [arrow] (HCNBackNormRecurence) -- (HCNBackIS);
\draw [arrow] (HCNBackMultiRecur) -- (HCNBackNormRecurence);
\end{tikzpicture}
		\caption{Proof structure for \Cref{thrm:InfNormBoundHConj}}
	\label{fig:Strct4thrm:InfNormBoundHConj}
\end{figure}

For the proof of \Cref{thrm:InfNormBoundHConj},  we use the properties of the extensions shown in  \Cref{sec:NonSoBasicExtensions}
and in particular we use \Cref{thrm:ProdWideHatBetaVsBeta}.
\Cref{fig:Strct4thrm:InfNormBoundHConj} shows the structure of the proof of \Cref{thrm:InfNormBoundHConj}.

\newcommand{\TPQ}{ \cadmiumgreen{T_{\kP,\kQ}} }
\newcommand{\TPP}{ \byzantine{\overline{T}}}
\newcommand{\TAdj}{\red{\UpZ}}

\newcommand{\vecF}{\upf}
\newcommand{\vecG}{\upg}
\newcommand{\vecK}{\upk}
\newcommand{\vtest}{\iris{{\bf q}}}

\subsection{Proof of \Cref{thrm:InfNormBoundHConj}}\label{sec:thrm:InfNormBoundHConj}
For a vector ${\bf v}\in \mathbb{R}^{N}_{\geq 0}$,  and $\kr \in \mathbb{R}$, we let vector ${\bf v}^{\kr}$ be such 
that ${\bf v}^{\kr}(\ki)=({\bf v}(\ki))^{\kr}$, for $1\leq \ki \leq  N$.  Also, when there is no danger of confusion, we abbreviate $\ExtdInfMatrixF$ to $\ExtdInfMatrix$.

For integer $\kell\geq 2\kk$, we let $\ExtdInfMatrixL_{\kell}$ be the $\ExtV_{\Lambda,\kk}\times \ExtV_{\Lambda,\kk}$ matrix 
such that for $\kP,\kQ \in \ExtV_{\Lambda,\kk}$,  the entry $\ExtdInfMatrixL_{\kell}(\kP,\kQ)$ is defined as follows:

Let $\kw$ and $\ku$ be the starting vertices of $\kP$ and $\kQ$, respectively.
Assume, first, that $\kP$ and $\kQ$ are compatible, i.e., $\dist_G(\kw,\ku)\geq 2\kk$.  Let $\gext{G}{\kP,\kQ}$ 
and $\upzeta^{\kP,\kQ}$ be the $\{\kP,\kQ\}$-extensions of $G$ and  $\mu^{\Lambda,\tau}_G$, respectively.  
Also, let $\TPQ=\Tsaw(G_{\kP,\kQ}, \kw)$ and the weights $\{ \infweight_{\kP,\kQ}(\ke)\}$ are obtained as  specified  
in \eqref{def:OfInfluenceWeights} with respect to $\upzeta^{\kP,\kQ}$.  We let
\begin{align}\nonumber
\ExtdInfMatrixL_{\kell}(\kP, \kQ) &=\sum\nolimits_{\kW\in \TPQ: \abs{\kW}=\kell} \prod\nolimits_{\ke\in \kW} \abs{\infweight_{\kP,\kQ}(\ke)} \enspace, 
\end{align}
where $\kW$ varies over all the paths of length $\kell$ from the root to the copies of $\ku$ in $\TPQ$ that correspond to
self-avoiding walks.   For non-compatible walks $\kP$ and $\kQ$, we let  $\ExtdInfMatrixL_{\kell}(\kP, \kQ) =0$.

From the definitions of the extended influence matrix $\ExtdInfMatrix$ and $ \ExtdInfMatrixL_{\kell}$, it is not hard to 
verify the entry-wise inequality   $\abs{\ExtdInfMatrix} \leq \sum\nolimits_{\kell\geq 2\kk} \ExtdInfMatrixL_{\kell}$, which
implies 
\begin{align}\nonumber 
\norm{ \abs{\ExtdInfMatrix}}{2\pfs} &
\leq \sum\nolimits_{\kell\geq 2\kk} \norm{ \ExtdInfMatrixL_{\kell}}{2\pfs}\enspace.
\end{align}
In light of the above, \Cref{thrm:InfNormBoundHConj} follows by showing that there exists a constant 
$\widehat{c}\geq 1$ such that
\begin{align}\label{eq:Target4thrm:InfNormBoundHConjLL}
\norm{ \ExtdInfMatrixL_{\kell}}{2\pfs} &\leq \widehat{c} \cdot \left(
 {\textstyle \left(\frac{1-\varepsilon/2}{\SingBound}\right)^{\kell-\kk}} \cdot 
\norm{ (\NBMatrix)^{{\kell-\kk}}}{2}
 \right)^{1/\pfs} & \forall \kell\geq 2\kk \enspace.
\end{align}
Before showing that \eqref{eq:Target4thrm:InfNormBoundHConjLL} is true,
we introduce a few concepts and establish some useful technical results.

For $\kell\geq 2\kk$, we let $\TAdj^{(\kell)}$ be the $\ExtV_{\Lambda,\kk}\times \ExtV_{\Lambda,\kk}$ 
matrix such that for $\kP,\kQ \in \ExtV_{\Lambda,\kk}$, the entry $\TAdj^{(\kell)}(\kP,\kQ)$ is defined 
as follows:

Let $\kw$ and $\ku$ be the starting vertices of $\kP$ and $\kQ$, respectively. For $\kP$ and 
$\kQ$ which are compatible,   we consider $\TPP=\Tsaw(\gext{G}{\kP}, \kw)$ and let 
\begin{align}\nonumber
\TAdj^{(\kell)}(\kP, \kQ) &=\# \textrm{copies of walk $\kQ$ in $\TPP$ which is at distance $\kell$ from the root} \enspace. 
\end{align}
For vertex $\kx$ in $\TPP$, which is connected to the root of the tree with path  $\kW=\kx_0, \ldots, \kx_{\kell}$, i.e., $\kx_0$ 
is the root and $\kx_{\kell}$ is $\kx$; we say that $\kx$ is  a copy of walk $\kQ=\kz_0,\ldots, \kz_{\kk}$, if $\kx_{\kell-\ki}$ is a copy of 
$\kz_{\ki}$, for all $\ki=0,\ldots, \kk$. 
 For non-compatible $\kP$ and $\kQ$, we have $\TAdj^{(\kell)}(\kP, \kQ) =0$.

For integer $\kell\geq 2\kk$ and $\vtest \in \mathbb{R}^{\ExtV_{\Lambda,\kk}}_{\geq 0}$, 
let vectors $\vecF_{\vtest, \kell},\vecG_{\vtest, \kell}, \vecK_{\vtest, \kell} \in \mathbb{R}^{\ExtV_{\Lambda,\kk}}_{\geq 0}$ be 
defined as follows: We let 
\begin{align}\label{eq:DefOfFvec}
\vecF_{\vtest, \kell}=\ExtdInfMatrixL_{\kell}\cdot \vtest \enspace.
\end{align}
We also let 
\begin{align} \label{eq:DefOfGvec}
\vecG_{\vtest, \kell} &= \left( {\textstyle \frac{1-\varepsilon/2}{\SingBound}} \right)^{\kell/\pfs} \cdot 
\left( \TAdj^{(\kell)}\cdot \vtest^{\pfs}\right)^{1/\pfs}
\enspace.
\end{align}
Similarly, we let
\begin{align} \label{eq:DefOfKvec}
\vecK_{\vtest, \kell}&= 
 {\textstyle \left(\frac{1-\varepsilon/2}{\SingBound}\right)^{\kell/\pfs} } \cdot 
\left( ( ( \NBMatrix)^{\kell-\kk}\cdot \Invol)_{\ExtV_{\Lambda,\kk}} \cdot \vtest^{\pfs}\right)^{1/\pfs}
 \enspace,
\end{align}
where $\Invol$ is the $\ExtV_{\kk}\times \ExtV_{\kk}$ involutionary matrix defined in \eqref{def:OfInvolutation}.
Also,  $(( \NBMatrix)^{{\kell-\kk}}\cdot \Invol)_{\ExtV_{\Lambda,\kk}}$ corresponds to the principle submatrix 
of $( \NBMatrix)^{\kell-\kk}\cdot \Invol$ induced by the elements in $\ExtV_{\Lambda, \kk}\subseteq \ExtV_{\kk}$.

\begin{proposition}\label{thrm:HC-VEntryBound}
For $\kk\geq 1$, there exists a constant $\widehat{c}_0\geq 1$ such that, for any $\kell\geq 2\kk$
 and any $\vtest \in \mathbb{R}^{\ExtV_{\Lambda,\kk}}_{\geq 0}$
\begin{align}\label{eq:thrm:HC-VEntryBoundBB}
 \vecF^{\pfs}_{\vtest, \kell} &\leq \widehat{c}_0 \cdot 
\vecG^{\pfs}_{\vtest, \kell} \enspace, 
\end{align}
where the inequality between the two vectors is entry-wise.
\end{proposition}
The proof of \Cref{thrm:HC-VEntryBound} appears in \Cref{sec:thrm:HC-VEntryBound}.

Letting ${\bf 0}\in \mathbb{R}^{\ExtV_{\Lambda,\kk}}$ be the all-zeros vector, 
for any $\vtest \in \mathbb{R}^{\ExtV_{\Lambda,\kk}}_{\geq 0}$, for any $\kell\geq 2\kk$, we have 
\begin{align}\label{eq:GKVecDom}
{\bf 0}\leq \vecF_{\vtest, \kell} \leq \widehat{c}_0\cdot \vecG_{\vtest, \kell} \leq \widehat{c}_0\cdot \vecK_{\vtest, \kell}\enspace,
\end{align}
where the above inequalities are entry-wise, while $\widehat{c}_0$ is from \Cref{thrm:HC-VEntryBound}.

To see why \eqref{eq:GKVecDom} is true, note the following: The first inequality implies that the entries of 
$f_{\vtest,\kell}$, which is defined in \eqref{eq:DefOfFvec}, are non-negative, which is true since both
$\ExtdInfMatrixL_{\kell}$ and $\vtest$ have non-negative entries. 
The second inequality follows from the fact that the entries of $\vecG_{\vtest,\kell}$ are non-negative and
 \Cref{thrm:HC-VEntryBound}. 
The third inequality follows from the observation that each entry $\TAdj^{(\kell)}(\kP, \kQ)$ is upper bounded by the 
number of $\kk$-non-backtracking walks of length $\kell-\kk$ from $\kP$ to $\kQ$.

Let $\vtest^{\star}=\vtest^{\star}(\kell) \in \mathbb{R}^{\ExtV_{\Lambda,\kk}}_{\geq 0}$ be 
 such that 
\begin{align}\label{eq:Norm2sYLVsKStar}
\norm{ \vtest^{\star}}{2\pfs}&=1& \textrm{and} &&
\norm{\ExtdInfMatrixL_{\kell}}{2\pfs}&=\norm{ \ExtdInfMatrixL_{\kell}\cdot \vtest^{\star}}{2\pfs} \enspace.
\end{align}
We let $\vecF^{\star}_{\kell}=\vecF_{\vtest^{\star}, \kell}$, while let $\vecG^{\star}_{\kell}=\vecG_{\vtest^{\star}, \kell}$ and 
 $\vecK^{\star}_{\kell}=\vecK_{\vtest^{\star}, \kell}$. 
Noting that $\vtest^{\star}$ has non-negative entries because the entries of $\ExtdInfMatrixL_{\kell}$ are non-negative, 
from \eqref{eq:GKVecDom}, we get
\begin{align}\label{eq:KSNorm4L1}
\norm{ \ExtdInfMatrixL_{\kell}}{2\pfs}&=\norm{ \vecF^{\star}_{\kell}}{2\pfs} \leq 
\widehat{c}_0 \cdot \norm{ \vecG^{\star}_{\kell}}{2\pfs} \leq \widehat{c}_0 \cdot \norm{ \vecK^{\star}_{\kell}}{2\pfs} \enspace. 
\end{align}
Let vector $\vecK^{\pfs,\star}_{\kell}$ be such that $\vecK^{\pfs,\star}_{\kell}(\kP)=({\vecK}^{\star}_{\kell}(\kP))^{\pfs}$, 
for all $\kP\in \ExtV_{\Lambda,\kk}$. Similarly, define $\vtest^{\pfs,\star}$ such that $\vtest^{\pfs,\star}(\kP)=({\vtest}^{\star}(\kP))^{\pfs}$. 
From the definition of the corresponding vectors, we have  
\begin{align}\label{eq:KstarL2sVsBarKStarL2}
\nnorm{ \vecK^{\star}_{\kell}}{2\pfs}{2\pfs}&=\nnorm{ \vecK^{\pfs,\star}_{\kell}}{2}{2}&
\textrm{and} && 
\norm{ \vtest^{\pfs,\star}}{2}=\norm{ \vtest^{\star}}{2\pfs}=1
\enspace. 
\end{align}
Using the above observation,  we have
\begin{align}
\nnorm{ \vecK^{\star}_{\kell}}{2\pfs}{2\pfs}=\nnorm{ \vecK^{\pfs,\star}_{\kell}}{2}{2} &=
 {\textstyle \left(\frac{1-\varepsilon/2}{\SingBound}\right)^{2\kell}} \cdot %
\nnorm{   ( ( \NBMatrix)^{{\kell-\kk}}\cdot \Invol )_{\ExtV_{\Lambda,\kk}} \cdot \vtest^{\pfs, \star}  }{2}{2} 
 &\textrm{[from \eqref{eq:DefOfKvec}]}\nonumber \\
 &\leq {\textstyle \left(\frac{1-\varepsilon/2}{\SingBound}\right)^{2\kell}} \cdot %
 \nnorm{ {\textstyle ( ( \NBMatrix)^{{\kell-\kk}}\cdot \Invol )_{\ExtV_{\Lambda, \kk}}}}{2}{2} &
 \textrm{[since $\norm{ \vtest^{\pfs,\star}}{2}=1$]}\nonumber \\
&\leq {\textstyle \left(\frac{1-\varepsilon/2}{\SingBound}\right)^{2 \kell}} \cdot %
\nnorm{ ( \NBMatrix)^{{\kell-\kk}}\cdot \Invol }{2}{2} & \textrm{[since $(\NBMatrix)^{\kell-\kk}\cdot \Invol$ is symmetric]} \nonumber \\
&= {\textstyle \left(\frac{1-\varepsilon/2}{\SingBound}\right)^{2\kell}} \cdot
\nnorm{ ( \NBMatrix)^{{\kell-\kk}}}{2}{2}
\enspace. \label{eq:VevKStarNormBound}
\end{align}
The third inequality follows from Cauchy’s interlacing theorem, e.g., see \cite{MatrixAnalysis}.
Note that   matrix $(\NBMatrix)^{\kell-\kk} \cdot \Invol$ is symmetric  due to \eqref{eq:DefPTInvarianceFormalB}.
The last derivation holds since for the involution $ \Invol$, we have $\norm{ \Invol}{2}=1$. 

We get \eqref{eq:Target4thrm:InfNormBoundHConjLL} by plugging \eqref{eq:VevKStarNormBound} into \eqref{eq:KSNorm4L1}
and setting $\widehat{c}=\widehat{c}_0\cdot \left(\frac{1-\varepsilon/2}{\SingBound}\right)^{\kk/\pfs}$. 
 \Cref{thrm:InfNormBoundHConj} follows. 
 \hfill $\square$

\newcommand{\SimpleExtdInfMatrixL}{\cC}
\newcommand{\SimpleFvec}{\upz}

\subsection{Proof of \Cref{thrm:HC-VEntryBound}}\label{sec:thrm:HC-VEntryBound}

For integer $\kell\geq 2\kk$, let $\SimpleExtdInfMatrixL_{\kell}$ be the $\ExtV_{\Lambda,\kk}\times \ExtV_{\Lambda,\kk}$ matrix 
such that for $\kP,\kQ \in \ExtV_{\Lambda,\kk}$, entry $\SimpleExtdInfMatrixL_{\kell}(\kP,\kQ)$ is defined as follows:

Let  $\kw$ and $\ku$ be the starting vertices of $\kP$ and $\kQ$, respectively. Assume, first,  that $\kP$ and $\kQ$ are compatible.
Let $\gext{G}{\kP}$ and $\upzeta^{\kP}$ be the $\kP$-extensions of $G$ and  $\mu^{\Lambda,\tau}_G$, respectively. 
Also,  let $\TPP=\Tsaw(\gext{G}{\kP}, \kw)$ while consider the edge-weights 
$\{\infweight_{\kP}(\ke)\}$ obtained as specified in \eqref{def:OfInfluenceWeights} with respect to $\upzeta^{\kP}$. 
Let 
\begin{align}\nonumber
\SimpleExtdInfMatrixL_{\kell}(\kP, \kQ) &=\sum\nolimits_{\kW=(\ke_1, \ldots, \ke_{\kell})\in\TPP} \prod\nolimits_{1\leq \ki \leq \kell-\kk} \abs{ \infweight_{\kP}(\ke_{\ki})} \enspace, 
\end{align}
where $\kW$ varies over all the paths of length $\kell$ from the root of $\TPP$ to set $\cp(\kQ,\kell)$. Set $\cp(\kQ,\kell)$ consists of the 
copies of $\kQ$ in $\TPP$  which are at distance $\kell$ from the root, while for each path $\kR=\ke_1,\ldots, \ke_{\kell}$ from the root to 
 such copy, we have $\prod\nolimits_{1\leq \ki \leq \kell-\kk} \abs{ \infweight_{\kP}(\ke_{\ki})}\neq 0$.

For non compatible walks $\kP$ and $\kQ$, i.e., $\dist_{G}(\ku,\kw)<2\kk$,  we let $\SimpleExtdInfMatrixL_{\kell}(\kP, \kQ) =0$. 

For vertex $\kx$ in $\TPP$, which is connected to the root with path 
$M=\kx_0, \ldots, \kx_{\kell}$, i.e., $\kx_0$ is the root and $\kx_{\kell}$ is vertex $\kx$, 
recall that we say $\kx$ is a copy of walk $\kQ=\kz_0,\ldots, \kz_{\kk}$ in $\ExtV_{\Lambda,\kk}$,
if $\kx_{\kell-\ki}$ is a copy of $\kz_{\ki}$, for $0\leq \ki\leq \kk$.

For $\vtest \in \mathbb{R}^{\ExtV_{\Lambda,\kk}}_{\geq 0}$, let
\begin{align}\label{eq:DefOfZQELL}
\SimpleFvec_{\vtest,\kell} &= \SimpleExtdInfMatrixL_{\kell} \cdot \vtest \enspace. 
\end{align}

\begin{claim}\label{claim:CQLVsCKL}
There exists fixed constant $\widehat{c}_1>0$ such that 
for any $\vtest \in \mathbb{R}^{\ExtV_{\Lambda,\kk}}_{\geq 0}$, we have
\begin{align}\label{eq:CQLVsCKL}
\vecF^{\pfs}_{\vtest, \kell} \leq \widehat{c}_1\cdot (1+\varepsilon/3)^{\kell} \cdot \SimpleFvec^{\pfs}_{\vtest,\kell} \enspace,
\end{align}
where the inequality between the two vectors is entry-wise.
\end{claim}

In light of \Cref{claim:CQLVsCKL}, \Cref{thrm:HC-VEntryBound} follows by showing that there exists constant 
$\widehat{c}_2>0$ such that, for any $\kell\geq 2\kk$ and any $\kP\in \ExtV_{\Lambda,\kk}$,  we have 
\begin{align}\label{eq:Target4thrm:InfNormBoundHConj}
 \SimpleFvec^{\pfs}_{\vtest,\kell}(\kP) &\leq \widehat{c}_2 \cdot (1-\varepsilon/3)^{\kell} \cdot \vecG^{\pfs}_{\vtest, \kell}(\kP)
\enspace.
\end{align}
Specifically, we get \eqref{eq:thrm:HC-VEntryBoundBB} by plugging \eqref{eq:Target4thrm:InfNormBoundHConj} into 
 \eqref{eq:CQLVsCKL} and setting $\widehat{c}_0=\widehat{c}_1\cdot \widehat{c}_2$.
It remains to show that \eqref{eq:Target4thrm:InfNormBoundHConj} is true. 
The high-level approach to obtain \eqref{eq:Target4thrm:InfNormBoundHConj} 
is not too different from that we use to prove \Cref{thrm:InflNormBound4GeneralD}.

For what follows, we fix $\kell\geq 2\kk$ and vector $\vtest\in \mathbb{R}^{\ExtV_{\Lambda,\kk}}_{\geq 0}$. Also, fix  $\kP \in \ExtV_{\Lambda,\kk}$
while let $\kw$ be the starting vertices of $\kP$.

\begin{claim}\label{claim:AddPotential}
For $\kell\geq 2\kk$ and any path $\kW=\ke_1, \ldots, \ke_{\kell}$  in $\TPP=\Tsaw(\gext{G}{\kP}, \kw)$ such that 
$\prod^{\kell-\kk}_{\ki=1} \abs{\infweight_{\kP}(\ke_{\ki})} \neq 0$, 
we have
\begin{align} \nonumber
\prod^{\kell-\kk}_{\ki=1} \abs{\infweight_{\kP}(\ke_{\ki})} &=\chiofh(\ke_{\kell-\kk}) \cdot \frac{\abs{\infweight_{\kP}(\ke_1)}}{\chiofh(\ke_1)} 
\cdot \prod^{\kell-\kk}_{\ki=2}\frac{\chiofh(\ke_{\ki-1})}{\chiofh(\ke_{\ki})} \cdot \abs{\infweight_{\kP}(\ke_{\ki})} 
\enspace, 
\end{align}
where $\chiofh(\ke)=\xdpotF(\infweight_{\kP}(\ke))$ and $\xdpotF=\potF'$. 
\end{claim} 
\Cref{claim:AddPotential} is very similar to  \Cref{claim:AddPotentialA}. For this reason, we omit its proof.

Recall that we fix $\vtest \in \mathbb{R}^{\ExtV_{\Lambda,\kk}}_{\geq 0}$.  For each vertex $\kv$ at level $\kh$ of $\TPP=\Tsaw(\gext{G}{\kP}, \kw)$, where 
$0\leq \kh\leq  \kell$, we let $\subcont_{\kv}$ be defined as follows: 
For $\kh=\kell$, we let 
\begin{align}\label{eq:BaseRvUzDef}
\subcont_{\kv} &= \sum\nolimits_{\kQ\in \ExtV_{\Lambda,\kk}: \kP, \kQ \textrm{ compatible}} 
\vtest(\kQ) \cdot \Ind\{\kv \in \cp(\kQ, \kell) \}\enspace.
\end{align}
For $\kell-\kk \leq \kh \leq \kell-1$, letting $\kv_1,\ldots, \kv_{\kd}$ be the children of $\kv$, we let 
\begin{align}\label{eq:1STEPRvUzDef}
\subcont_{\kv} &= \sum\nolimits_{\kv_i} \subcont_{\kv_i}\enspace.
\end{align}
For a vertex $\kv$ at level $0<\kh<\kell-\kk$ of the tree $\TPP$, we have
\begin{align}\label{eq:RecurRUZ}
\subcont_{\kv} & = \chiofh(\kb_{\kv}) \cdot \sum\nolimits_{\ki} \frac{\abs{\infweight_{\kP}(\kb_{\ki})}}{\chiofh(\kb_{\ki})} \cdot \subcont_{\kv_{\ki}}\enspace,
\end{align}
where $\kb_{\ki}$ is the edge that connects $\kv$ to its child $\kv_{\ki}$ and edge $\kb_{\kv}$ connects $\kv$ to its parent in tree $\TPP$. 

Finally, for $\kh=0$, i.e., vertex $\kv$ is identical to the root of $\TPP$, we let 
\begin{align}\label{eq:RecurRUZRoot}
\subcont_\kv=\subcont_{\rm root} &= \max_{\ke_a,\ke_b\in \TPP} \left\{ \chiofh(\ke_a) \cdot \frac{ \abs{\infweight_{\kP}(\ke_b)}}{\chiofh(\ke_b)} \right\} \cdot \subcont_{\pi} \enspace,
\end{align}
where vertex $\pi$ is the single child of the root of $\TPP$.

\begin{claim}\label{claim:ZvecQLVsCalRV}
We have $\SimpleFvec_{\vtest,\kell}(\kP)\leq \subcont_{\rm root}$. 
\end{claim}

The proof of \Cref{claim:ZvecQLVsCalRV} is tedious but very similar to that of \Cref{claim:CEllVsCalD} in the proof of \Cref{thrm:InflNormBound4GeneralD}.
For this reason we omit it.

H\"older's inequality and \eqref{eq:1STEPRvUzDef} imply,  for a vertex $\kv$ at level $\kell-\kk\leq \kh\leq \kell-1$ in $\TPP$, that we have
\begin{align}\label{eq:CaseVHEqLMinusK}
\left( \subcont_{\kv} \right)^{\pfs} &\leq \maxDeg^{\pfs-1} \cdot \sum\nolimits_{\kv_{\ki}}\left(\subcont_{\kv_{\ki}} \right)^{\pfs} \enspace.
\end{align}
The contraction property of $(\pfs,\delta,c)$-potential $\potF$, together with \eqref{eq:RecurRUZ},
imply that for a vertex $\kv$ at level $0<\kh<\kell-\kk$, we have 
\begin{align}\label{eq:CaseVHEqLMinusKSmaller}
\left( \subcont_{\kv} \right)^{\pfs} \leq \delta \cdot \sum\nolimits_{\kv_{\ki}} \left( \subcont_{\kv_{\ki}} \right)^{\pfs} \enspace.
\end{align} 
In light of the above two inequalities and \eqref{eq:BaseRvUzDef}, 
for a vertex $\kv$ at level $0<\kh < \kell-\kk$, we have
\begin{align}\nonumber 
\left( \subcont_{\kv} \right)^{\pfs} &\leq \maxDeg^{\kk\cdot(\pfs-1)}\cdot 
\sum\nolimits_{\kQ \in \ExtV_{\Lambda,\kk}: \kP,\kQ \textrm{ compatible}} \vtest^{\pfs}(\kQ) \cdot \delta^{\kell-\kh-\kk} \cdot \abs{\cp(\kQ,\kell)\cap \TPP_{\kv}} \enspace.
\end{align}
For the single child of the root of $\TPP$, vertex $\pi$, the above implies 
\begin{align}
\left( \subcont_{\pi} \right)^{\pfs} &\leq \maxDeg^{\kk\cdot(\pfs-1)} \cdot \delta^{\kell-\kk-1}\cdot\sum\nolimits_{\kQ \in \ExtV_{\Lambda,\kk}:\kP,\kQ \textrm{ compatible}} 
\vtest^{\pfs}(\kQ) \cdot \abs{\cp(\kQ,\kell)}
 \label{eq:AQLVSSubtree} \\ &
= \maxDeg^{\kk\cdot(\pfs-1)} \cdot \left( {\textstyle \frac{1-\varepsilon}{\SingBound} }\right)^{\kell-\kk-1 }\cdot  ( \TAdj^{(\kell)} \cdot \vtest^{\pfs} ) (\kP) 
\label{eq:cQc2SBound} \\ &
\leq \left( {\textstyle \frac{1-\varepsilon}{\SingBound} }\right)^{-\kk-1}\cdot \maxDeg^{\kk\cdot(\pfs-1)} \cdot (1-\varepsilon/3)^{\kell}\cdot \vecG^{\pfs}_{\vtest, \kell}(\kP)
\enspace. \label{eq:FinalIineqWithVegG}
\end{align}
For \eqref{eq:AQLVSSubtree} we use the observation that since the root of $\TPP$ has only one child, i.e., vertex $\pi$,  
and $\kell\geq 2\kk$, we  have $\cp(\kQ,\kell)\cap \TPP_{\pi}=\cp(\kQ,\kell)$. For \eqref{eq:cQc2SBound} we use the definition 
of matrix $\TAdj^{(\kell)}$ and recall that $\delta=\frac{1-\varepsilon}{\SingBound}$. 
The last inequality follows from the definition of $\vecG_{\vtest, \kell}$ in \eqref{eq:DefOfGvec}.

The boundedness property of the $(\pfs,\delta,c)$-potential and \eqref{eq:RecurRUZRoot} imply 
\begin{align}\label{eq:CQRootVsCRcBoundedness}
\left( \subcont_{\rm root}\right)^{\pfs} &\leq c^{\pfs} \cdot \left( \subcont_{\pi}\right)^{\pfs} 
\leq c^{\pfs} \left( {\textstyle \frac{1-\varepsilon}{\SingBound} }\right)^{-\kk-1 } \cdot \maxDeg^{\kk\cdot(\pfs-1)} \cdot (1-\varepsilon/3)^{\kell}\cdot \vecG^{\pfs}_{\vtest, \kell}(\kP)\enspace, 
\end{align}
where for the second equality, we use \eqref{eq:FinalIineqWithVegG}. 
We get \eqref{eq:Target4thrm:InfNormBoundHConj} by combining \eqref{eq:CQRootVsCRcBoundedness}
and \Cref{claim:ZvecQLVsCalRV} while setting  
$\widehat{c}_2=c^{\pfs}\cdot \maxDeg^{\kk\cdot(\pfs-1)}\cdot \left( {\textstyle \frac{1-\varepsilon}{\SingBound} }\right)^{-\kk-1}$.

All the above conclude the proof of \Cref{thrm:HC-VEntryBound}.
\hfill $\square$

\begin{proof}[Proof \Cref{claim:CQLVsCKL}]
The definition of $\vecF_{\vtest, \kell}$ implies that, for any $\kP \in \ExtV_{\Lambda,\kk}$, we have 
\begin{align}\label{eq:Base4claim:CQLVsCKL}
\vecF_{\vtest, \kell}(\kP)&=\sum\nolimits_{\kQ\in \ExtV_{\Lambda,\kk}: \kP,\kQ \textrm{ compatible}}\vtest(\kQ)\cdot \sum\nolimits_{\kW=(\ke_1, \ldots, \ke_{\kell})}
 \prod\nolimits_{1\leq \ki\leq \kell} \abs{\infweight_{\kP,\kQ}(\ke_{\ki})} \enspace,
\end{align}
where $\kW$ varies over the paths from the root to the vertices in $\cp(\kQ,\kell)$ in tree $\TPQ$.

The edge weights of path $\kW$  are obtained using $\zeta^{\kP,\kQ}$,   the $\{\kP,\kQ\}$-extension of
the Gibbs distribution $\mu^{\Lambda,\tau}_G$. 
To obtain our result, we need to express $ \prod\nolimits_{1\leq \ki\leq \kell} \abs{\infweight_{\kP,\kQ}(\ke_{\ki})}$
using the weights $\{\infweight_{\kP}(\ke)\}$, i.e., using the edge-weights obtained 
using $\zeta^{\kP}$, i.e.,   the $\kP$-extension of the Gibbs distribution $\mu^{\Lambda,\tau}_G$.

To this end,  we  use    \Cref{thrm:ProdWideHatBetaVsBeta}.  This theorem  implies that,  there is $\widehat{c}_1>0$ such that 
for $\kW=\ke_1, \ldots, \ke_{\kell}$ and $\kell\geq 2\kk$ we have
\begin{align}\nonumber
\prod\nolimits_{1\leq \kk\leq \kell} \abs{\infweight_{\kP,\kQ}(\ke_{\ki})}\leq \widehat{c}^{\ 1/\pfs}_1 \cdot (1+\varepsilon/3)^{\kell/\pfs} 
\cdot    \prod\nolimits_{1\leq \ki\leq \kell-\kk} \abs{ \infweight_{\kP}(\ke_{\ki})}\enspace.
\end{align}
Plugging the above into \eqref{eq:Base4claim:CQLVsCKL}, we obtain
\begin{align}\nonumber  
\vecF_{\vtest, \kell}(\kP) &\leq \widehat{c}^{\ 1/\pfs}_1 \cdot (1+\varepsilon/3)^{\kell/\pfs} 
\cdot \sum\nolimits_{\kQ\in\ExtV_{\Lambda,\kk}: \kP,\kQ\textrm{ compatible}}\vtest(\kQ)\cdot \sum\nolimits_{\kW=(\ke_1,\ldots, \ke_{\kell})}
 \prod\nolimits_{1\leq \ki\leq \kell-\kk} \abs{ \infweight_{\kP}(\ke_{\ki})}
 \\&
 = \widehat{c}^{\ 1/\pfs}_1 \cdot (1+\varepsilon/3)^{\kell/\pfs} 
\cdot \sum\nolimits_{\kQ\in\ExtV_{\Lambda,\kk}: \kP,\kQ\textrm{ compatible}}
\vtest(\kQ)\cdot \SimpleExtdInfMatrixL_{\kell}(\kP, \kQ) \nonumber
 \\ &
 =\widehat{c}^{\ 1/\pfs}_1\cdot (1+\varepsilon/3)^{\kell/\pfs} \cdot \SimpleFvec_{\vtest,\kell}(\kP) \nonumber
\enspace.
\end{align}
For the second equality we use the definition of matrix $\SimpleExtdInfMatrixL_{\kell}$.
The third equality follows from the definition of  $\SimpleFvec_{\vtest,\kell}$ in \eqref{eq:DefOfZQELL} and recalling  that for $\kP$ and $\kQ$ which
are not compatible we have $\SimpleExtdInfMatrixL_{\kell}(\kP, \kQ)=0$. 

\Cref{claim:CQLVsCKL} follows.
\end{proof}

\spreadpoint

\section{Proof of \Cref{thrm:ProdWideHatBetaVsBeta} \LastReviewG{2025-03-04}}\label{sec:thrm:ProdWideHatBetaVsBeta}

Let us briefly recall the setting for \Cref{thrm:ProdWideHatBetaVsBeta}.  We have $\kw, \ku\in V\setminus \Lambda$, 
such that $\dist(\kw,\ku)\geq 2\kk$ and  walk $\kP\in \ExtV_{\Lambda,\kk}$ that emanates from $\kw$. Note that our assumptions imply
that  $\ku\notin \kP$.

We apply the $\Tsaw$-construction on $\mu^{\Lambda,\tau}_G$ and  obtain the Gibbs distribution $\mu^{M,\sigma}_{\SMTu}$ 
for tree $\SMTu=\Tsaw(G,\ku)$.   We also use the $\Tsaw$-construction for the $\kP$-extension of $\mu^{\Lambda,\tau}_G$ 
and get  the Gibbs distribution $\mu^{\kL,\eta}_{\SMTPu}$ for tree $\SMTPu=\Tsaw(\gext{G}{\kP},\ku)$.

We let set $\cpT$ consist of each vertex in $\SMTu$ which is a copy of a vertex $\kv\in P$. Also, let set $\cpTP$ 
consist of each vertex in $\SMTPu$ which is either a copy of a vertex $\kv\in \kP$, or  a copy of a split-vertex in $\kv\in \ssplit_{\kP}$.

We use \Cref{lemma:Subtree4TSawU} and identify $\SMTPu$ as a subtree of $\SMTu$. This allows us to consider  $(\kL, \eta)$ as a pinning 
at the vertices $\SMTu$ and compare it with the pinning $(M,\sigma)$.  Among other results, this is carried out in  the following lemma.

\begin{lemma}\label{lemma:BoundaryTVsBoundaryTPFromU}
We identify $\SMTPu=\Tsaw(\gext{G}{\kP},\ku)$ as a subtree of $\SMTu=\Tsaw(G,\ku)$ and we consider $\cpTP$ as a subset of vertices of $\SMTu$
and 
$(\kL, \eta)$   as a pinning of vertices in $\SMTu$. 

Then, the following is true:
\begin{enumerate}[(a)]
\item the pinnings $\sigma$ and $\eta$ agree on the assignment of the vertices in $\kL\cap M$, \label{statement:lemma:BoundaryTVsBoundaryTPFromUA}
\item we have $\cpT \cap \SMTPu= \cpTP$, \label{stat:B-lemma:BoundaryTVsBoundaryTPFromU} \label{statement:lemma:BoundaryTVsBoundaryTPFromUB}
\item we, also, have $\kL\setminus M \subseteq \cpT$ and $M\cap \SMTPu \subseteq \kL$, \label{stat:C-lemma:BoundaryTVsBoundaryTPFromU}
\item letting $\kK=\kL\cup \cpTP$ and $\widehat{\eta}\in \{\pm 1\}^{\kK}$, the distribution 
$\mu^{\kK, \widehat{\eta}}_{\SMTPu}$ and the marginal of $\mu^{\kK, \widehat{\eta}}_{\SMTu}$ at the subtree $\SMTPu$ are identical. 
\label{statement:lemma:BoundaryTVsBoundaryTPFromUD}
\end{enumerate}
\end{lemma}
The proof of \Cref{lemma:BoundaryTVsBoundaryTPFromU} appears in \Cref{sec:lemma:BoundaryTVsBoundaryTPFromU}. 
%

Let $\{\infweight(\ke)\}$ be the collection of weights over the edges of $\SMTu$ that arise from $\mu^{M,\sigma}_{\SMTu}$, i.e., 
as specified in \eqref{def:OfInfluenceWeights}.
Also, let $\{\infweight_{\kP}(\ke)\}$ be the collection of weights over the edges of $\SMTPu$ that arise from $\mu^{\kL,\eta}_{\SMTPu}$.
We identify $\SMTPu$ as a subtree of $\SMTu$ and we endow $\SMTPu$ with the collection of weights $\{\infweight(\ke)\}_{\ke\in \SMTPu}$. 
 Now, we have two sets of weights for the edges of $\SMTPu$, i.e., $\{\infweight(\ke)\}_{\ke\in \SMTPu}$ and $\{\infweight_{\kP}(\ke)\}_{\ke\in \SMTPu}$.

The following proposition implies that, provided that we have an appropriately contracting potential function $\potF$, 
 for any edge $\kf\in \SMTPu$ which is far from $\cpTP$, weights $\infweight(\kf)$ and $\infweight_{\kP}(\kf)$ are 
 close to each other. Recall that set $\cpTP$ consists of each vertex in $\SMTPu$ which is either a copy of a vertex 
 $\kv\in \kP$ or  a copy of a split-vertex in  $\ssplit_{\kP}$.

\begin{proposition}\label{prop:StableSSM}
Let $\maxDeg > 1$, $\kk\geq 1$ and $N>1$ be integers, $c>0$, $\SingBound> 1$, $\pfs \geq 1$, while 
let $b,\varepsilon\in (0,1)$. Also, let $\beta,\gamma, \lambda\in \mathbb{R}$ be such that $\gamma >0$,
$0\leq \beta\leq \gamma$ and $\lambda>0$. 

Consider graph $G=(V,E)$ of maximum degree $\maxDeg$ such that 
$\nnorm{ (\NBMatrix)^N}{1/N}{2}= \SingBound$.
Assume that $\mu$, the Gibbs distribution on $G$ specified by the parameters $(\beta, \gamma, \lambda)$,
is $b$-marginally bounded.
For $\delta= \frac{1-\varepsilon}{\SingBound}$, suppose that there
is a $(\pfs,\delta,c)$-potential function $\potF$ with respect to $(\beta,\gamma,\lambda)$. 

There exists a bounded number $\ell_0>2\kk$ such that, 
for any $\Lambda\subset V$ and $\tau\in \{\pm 1\}^{\Lambda}$,
for any $\ku,\kw\in V\setminus \Lambda$ with $\dist(\ku,\kw)\geq 2\kk$ and for $\kP\in \ExtV_{\Lambda,\kk}$ 
that emanates from  $\kw$  the following is true:

For any edge $\kf$ in $\SMTPu=\Tsaw(\gext{G}{\kP}, \ku)$ which is at distance $\kell\geq \ell_0$ from $\cpTP$, we have 
\begin{align}\label{eq:prop:StableSSM}
\infweight_{\kP}(\kf)&\leq (1+\varepsilon/3)^{1/\pfs} \cdot \infweight(\kf) \enspace, 
\end{align}
where the set $\cpTP$ and the weights $\{ \infweight_{\kP}(\ke)\}$, $\{\infweight(\ke)\}$ are specified above.
\end{proposition}
The proof of \Cref{prop:StableSSM} appears in \Cref{sec:prop:StableSSM}.

\begin{proof}[Proof of \Cref{thrm:ProdWideHatBetaVsBeta}]
We consider path $\kW=\ke_1,\ldots, \ke_{\kell}$ from the root of $\SMTPu=\Tsaw(\gext{G}{\kP}, \ku)$ to a copy of vertex $\kw$.
Note that we specify $\kW$ it terms of its edges.

Recall  the two sets of edge-weights for $\SMTPu$, i.e., set  $\{\infweight(\ke)\}$ is induced by the Gibbs distribution $\mu^{M,\sigma}_{\SMTu}$, 
 and set $\{ \infweight_{\kP}(\ke)\}$ is induced by the Gibbs distribution $\mu^{\kL,\eta}_{\SMTPu}$.

Also, recall that set $\cpTP$ consists of each vertex in $\SMTPu$  which is either a copy of a vertex 
$\kv\in \kP$, or a copy of a split-vertex of $\kv\in \kP$.

Let $\infweight_{\rm min}=\min_{\ke}\{ \abs{\infweight_{\kP} (\ke)}, \abs{{\infweight}(\ke)} \}$,
where $\ke$ varies over the set of edges in $\SMTPu$ such that $ \infweight_{\kP} (\ke), \infweight(\ke)\neq 0$. 
Similarly, we define $\infweight_{\rm max}=\max_{\ke}\{\abs{\infweight_{\kP}(\ke)}, \abs{\infweight(\ke)}\}$
where $\ke$ varies  over the same set of edges. 
It is an easy exercise to show that our assumption about $b$-marginal boundedness, for fixed $b>0$, 
implies that $\infweight_{\rm min}$ and $\infweight_{\rm max}$ are constants bounded away from zero. 

Consider $\ell_0$ as specified in \Cref{prop:StableSSM}. 
Let set $H$ consist of the edges  $\ke_{\ki} \in\kW$ such that $\ki>\kell-\kk$.  Also, let set $\kK$ consist of the edges $\ke_{\ki}\in \kW$ such that 
$1\leq \ki \leq \kell-\kk$ and $\dist(\ke_{\ki}, \cpTP )<\ell_0$. Similarly, let set $R$ consist of the edges $\ke_{\ki}\in \kW$ such 
that $1\leq \ki \leq  \kell-\kk$ and $\dist(\ke_{\ki}, \cpTP )\geq \ell_0$.
We have 
\begin{align}
\prod\nolimits_{1\leq i\leq \kell}\abs{\infweight_{\kP}(\ke_{\ki})} &= \prod\nolimits_{\ke_{\ki}\in H}\abs{\infweight_{\kP}(\ke_{\ki})} \cdot
\prod\nolimits_{\ke_{\ki} \in \kK}\abs{\infweight_{\kP}(\ke_{\ki})}
\cdot
\prod\nolimits_{\ke_{\ki} \in R}\abs{\infweight_{\kP}(\ke_{\ki})} \enspace. \label{eq:splitProdHatBetas}
\end{align}

For our subsequent arguments, we use the following result. 

\begin{claim}\label{claim:SupportOfWeiS}
For $\kW=\ke_1,\ldots, \ke_{\kell}$ and $\{\infweight_{\kP}(\ke_{\ki})\}$, $\{\infweight(\ke_{\ki})\}$, 
the following is true:  if $\infweight_{\kP}(\ke_{\ki})\neq 0$, then we have $\infweight(\ke_{\ki})\neq 0$, for each $\ke_{\ki}\in \kK\cup R$.
\end{claim}

For $\infweight_{\kP}(\ke_i)$ such that  $\ke_{\ki}\in H$,   we use the standard bound 
\begin{align}\label{eq:WeightEllUBound}
\abs{\infweight_{\kP}(\ke_{\ki})} &\leq 1\enspace. 
\end{align}
For each edge $\ke_{\ki}\in \kK$, 
we have 
\begin{align}\label{eq:WeightPVsWeightClose3LM}
\abs{ \infweight_{\kP}(\ke_{\ki}) } \leq \frac{\infweight_{\rm max}}{\infweight_{\rm min}} \cdot \abs{\infweight(\ke_{\ki})}
\ \leq \ (1+\varepsilon/3)^{{1}/{\pfs}} \cdot \frac{\infweight_{\rm max}}{\infweight_{\rm min}} \cdot \abs{\infweight(\ke_{\ki})} \enspace.
\end{align}
We use  \Cref{claim:SupportOfWeiS} for the first inequality. The second inequality 
uses that $1<(1+\varepsilon/2)^{{1}/{\pfs}}$ for $\varepsilon, \pfs>0$.

 \Cref{prop:StableSSM} implies that, for every $\ke_{\ki} \in R$, we have
\begin{align}\label{eq:WeightPVsWeightFarfromLM}
\abs{ \infweight_{\kP}(\ke_{\ki}) }  \leq (1+\varepsilon/3)^{{1}/{\pfs}}\cdot \abs{ \infweight(\ke_{\ki})} \enspace.
\end{align}
Plugging \eqref{eq:WeightEllUBound}, \eqref{eq:WeightPVsWeightClose3LM} and \eqref{eq:WeightPVsWeightFarfromLM}
into \eqref{eq:splitProdHatBetas}, we get
\begin{align}\label{eq:Final4thrm:ProdWideHatBetaVsBeta}
\prod^{\kell}_{\ki=1} \abs{\infweight_{\kP}(\ke_{\ki})}
&\leq \left( \frac{\infweight_{\rm max}}{\infweight_{\rm min}} \right)^{\abs{\kK}}
\cdot 
\prod^{\kell-\kk}_{\ki=1} (1+\varepsilon/3)^{{1}/{s}} \cdot \abs{\infweight(\ke_{\ki})}\enspace.
\end{align}
For  the cardinality of set $\kK$, we use the following claim. 
\begin{claim}\label{cliam:NoOfEdgeClose2LM}
There are at most $(\kk+1)\cdot \maxDeg^{\ell_0+1}$ edges in path $\kW$ that are at distance $<\ell_0$ from set $\cpTP$. 
\end{claim}
The above implies that $\abs{\kK}\leq (\kk+1) \cdot \maxDeg^{\ell_0+1}$. 
Hence, \Cref{thrm:ProdWideHatBetaVsBeta} follows from \eqref{eq:Final4thrm:ProdWideHatBetaVsBeta} by setting 
$\widehat{c}_1=\left( \frac{\infweight_{\rm max}}{\infweight_{\rm min}} \right)^{(\kk+1)\cdot \maxDeg^{\ell_0+1}}$.
\end{proof}

\begin{proof}[Proof of \Cref{claim:SupportOfWeiS}]

We use \Cref{lemma:Subtree4TSawU} and identify $\SMTPu=\Tsaw(\gext{G}{\kP},\ku)$ as a subtree of $\SMTu=\Tsaw(G,\ku)$
such that both trees have  the same root. Then, $\kW$ is also a path in $\SMTu$ starting from the root.

Having $\SMTPu$ as a subtree of $\SMTu$, we endow the subset of edges of  $\SMTu$ that correspond to those of $\SMTPu$
with the edge weights $\{\infweight_{\kP}(\ke)\}$. Note that  we do not necessarily  have $\infweight_{\kP}(\ke)$ for all edges $\ke\in \SMTu$.
Also,  recall that the set of weights $\{\infweight(\ke)\}$  and $\{\infweight_{\kP}(\ke)\}$ are induced by 
$\mu^{M,\sigma}_{\SMTu}$ and  $\mu^{\kL,\eta}_{\SMTPu}$, respectively. 

For $\ke_{\ki}\in \kW$, let  $\ke_{\ki}=\{\kx_{\ki},\kz_{\ki}\}$ 
while let $\uppi_{\ki}$ be the set  that consists of $\kx_{\ki}, \kz_{\ki}$ and their 
children in $\SMTu$.
Furthermore, we have the following observations:
 \begin{enumerate}[(i)]
\item \Cref{lemma:BoundaryTVsBoundaryTPFromU} implies that $\cpT\cap \SMTPu=\cpTP$, \label{obs:APBarTVsBarAp}
\item \Cref{lemma:Subtree4TSawU} implies that each vertex in $\SMTu \setminus \SMTPu$ corresponds to a descendant of a vertex in $\cpT$, \label{obs:TminusBarTVsAp}
\item \Cref{lemma:BoundaryTVsBoundaryTPFromU} implies that $M \cap \SMTPu \subseteq \kL$, \label{obs:PinMVsBarTVsPinL}
\item \Cref{lemma:BoundaryTVsBoundaryTPFromU} implies that the pinnings $\sigma$ and $\eta$ agree on the assignment  in $\kL\cap M$. \label{obs:PinMVsPinLIntesection}
\end{enumerate}
Since we assume that our Gibbs distributions are $b$-marginally bounded,  we can have $\infweight(\ke_{\ki})=0$ (resp. $\infweight_{\kP}(\ke_{\ki})=0$) only if at least one vertex in
$\uppi_{\ki}$ has a fixed configuration under $\mu^{M,\sigma}_{\SMTu}$ (resp. $\mu^{L,\eta}_{\SMTPu}$).
Note that for Gibbs  distributions with parameter $\beta=0$, it is possible that  a pinning on a child of $\kx_{\ki}$ or   $\kz_{\ki}$ 
determines  the configuration of its parent with probability 1, implying that  $\infweight(\ke_{\ki})=0$ (resp. $\infweight_{\kP}(\ke_{\ki})=0$).

Consider $1\leq \ki \leq \kell-\kk$. If there is $\kv\in \uppi_{\ki}$ which  is not in $\SMTPu$, this means that at least one of 
$\kx_{\ki}$ and $\kz_{\ki}$ is  in $\cpT$, i.e., due to observation \ref{obs:TminusBarTVsAp}. 
But then, since $\kx_{\ki},\kz_{\ki}\in \SMTPu$,  observation \ref{obs:APBarTVsBarAp} implies  that the parent of $\kv$, i.e., one of  $\kx_{\ki}, \kz_{\ki}$, 
is in $\cpT\cap \SMTPu=\cpTP$. 
Since $i\leq \kell-\kk$,  if any of $\kx_{\ki},\kz_{\ki}$  is in  $\cpTP$, then it should be  a split-vertex in $\ssplit_{\kP}$, which means that is of degree $1$. 
This can only happen if $\ki=1$, however in this case we  have  $\kx_{\ki},\kz_{\ki}\notin \cpTP$ because the root is a copy of vertex $\ku$ which is
at distance at least $\kk$ from any copy of a vertex $\kv\in \cpTP$. We then conclude that this event cannot happen. 

Hence,  all elements in $\uppi_{\ki}$ are vertices in $\SMTPu$. The fact that $\ke_{\ki}\in \SMTPu$ and 
observations \ref{obs:PinMVsBarTVsPinL} and \ref{obs:PinMVsPinLIntesection} imply that  if any vertex in $\uppi_{\ki}$ is pinned, then the
pinning is the same under both $\sigma,\eta$. Hence,  if  $\infweight_{\kP}(\ke_{\ki})\neq 0$, then we also have that $\infweight(\ke_{\ki})\neq 0$. 

\Cref{claim:SupportOfWeiS} follows. 
\end{proof}

\begin{proof}[Proof of \Cref{cliam:NoOfEdgeClose2LM}]
It is not hard to verify that the number of edges at distance $\kell$ from the vertices in $\kP\cup \ssplit_{\kP}$ in $\gext{G}{\kP}$ is 
 the same  as the number of edges at distance $\kell$ from the vertices  in $\kP$ in graph $G$.

Furthermore, it is not hard to check that there are at most $\maxDeg^{\kell+1}$ edges in $G$ which are at distance 
$\kell$ a from vertex $\kv\in \kP$. A simple geometric series calculation implies that the total number of edges at 
distance $<\ell_0$ from all  vertices $\kv\in \kP$ is at most $(\kk+1)\cdot \maxDeg^{\ell_0+1}$. Recall that there 
are in total $\kk+1$ vertices in $\kP$.

Hence, we conclude that the number of edges at distance $<\ell_0$ from the vertices in $\kP\cup \ssplit_{\kP}$ in 
$\gext{G}{\kP}$ is  at most $(\kk+1)\cdot \maxDeg^{\ell_0+1}$.  The claim follows by noting that each one of these, 
at most $(\kk+1)\cdot \maxDeg^{\ell_0+1}$ edges can have at  most one copy in path $\kW$. 
\end{proof}

\newcommand{\ratioT}{\blue{\gratio}}
\newcommand{\ratioTP}{\magenta{\overline{\gratio}}}

\subsection{Proof of \Cref{prop:StableSSM}}\label{sec:prop:StableSSM}

We use \Cref{lemma:Subtree4TSawU} and identify $\SMTPu=\Tsaw(\gext{G}{\kP}, \ku)$ as a subtree of 
$\SMTu=\Tsaw(G,\ku)$.

Recall that the edge-weights $\{ \infweight(e)\}$ in $\SMTu$ arise from $\mu^{M,\sigma}_{\SMTu}$. Also, the
edge-weights   $\{\infweight_{\kP}(\ke)\}$ in $\SMTPu$ arise from $\mu^{\kL,\eta}_{\SMTPu}$. Having identified 
$\SMTPu$ as a subtree of $\SMTu$, we endow $\SMTPu$ with the edge-weights $\{ \infweight(\ke)\}_{\ke\in \SMTPu}$. 
 Hence now we have two sets of weights for $\SMTPu$.

Also, recall that set $\cpTP$ consists of each vertex in $\SMTPu$ which is either a copy of vertex $\kv\in P$, or 
a copy of a split-vertex in $\ssplit_{\kP}$. 

We start with the following claim.
\begin{claim}\label{claim:AlphaPZero}
For any edge $\kf$ in $\SMTPu=\Tsaw(\gext{G}{\kP},\ku)$, at distance $\kell >10$ from $\cpTP$, 
 if $\infweight(\kf)=0$, then $\infweight_{\kP}(\kf)=0$. 
\end{claim}

\Cref{claim:AlphaPZero} implies that \eqref{eq:prop:StableSSM} is true for any edge $\kf$ in $\SMTPu$ at 
distance $\kell>10$ from $\cpTP$  such that $\infweight(\kf)=0$. For the rest of the proof, we consider edge 
$\kf$ in $\SMTPu$ such that $\infweight(\kf)\neq 0$ which is at distance $\kell \geq \ell_0$ from $\cpTP$, where 
the quantity  $\ell_0>10$ is specified later in the proof.

Let edge $\kf=\{\kv,\kx\}$ while w.l.o.g. assume that $\kv$ is a child of $\kx$. Also, let $\kv_1, \ldots, \kv_d$ 
be the children of  $\kv$ in $\SMTPu$. We take $\kf$ at distance $\kell>10$ from $\cpTP$. This implies that 
vertex $\kv$ has the same set of children in both trees $\SMTu$ and $\SMTPu$. If $\kv$ had children in 
$\SMTu$ which are not in $\SMTPu$, then according to \Cref{lemma:Subtree4TSawU}, we would have had
$\kv\in \cpT$. But since $\kv\in \SMTPu$, \Cref{lemma:BoundaryTVsBoundaryTPFromU} (statement 
\ref{stat:B-lemma:BoundaryTVsBoundaryTPFromU})  implies that $\kv\in \cpT\cap \SMTPu=\cpTP$.  
This cannot happen since we assume that $\kf$ is at distance at least $10$ from $\cpTP$.  Hence, $\kv$ 
indeed has the same set of children in both trees $\SMTu$ and $\SMTPu$.

Let function $\Phi_d:\mathbb{R}^d\to \mathbb{R}$ be such that 
\begin{align}\label{eq:DefOfPhiSSM}
{\bf \kx} &
\mapsto \frac{(\beta \gamma-1) \cdot F_d({\bf \kx})} {(\beta \cdot F_d({\bf \kx})+1)(F_d({\bf \kx})+\gamma)}\enspace, 
\end{align}
while $F_d({\bf \kx})=F_d({\bf \kx}_1, \ldots, {\bf \kx}_d)$ is defined in \eqref{eq:BPRecursion}. 
Recall that $d$ is the number of children of vertex $\kv$.

For each child $\kv_{\ki}$,  consider the ratios of marginals $\ratioT^{M, \sigma}_{\kv_{\ki}}$ and 
$\ratioTP^{L, \eta}_{\kv_{\ki}}$ defined as described in \Cref{sec:RecursionVsSpectralIneq}. 
Specifically, the ratio $\ratioT^{M, \sigma}_{\kv_{\ki}}$ is defined with respect to the Gibbs distribution
$\mu^{M,\sigma}_{\SMTu}$, while the ratio $\ratioTP^{L, \eta}_{\kv_{\ki}}$ is defined with respect to the 
Gibbs distribution $\mu^{\kL,\eta}_{\SMTPu}$.

Using \eqref{def:OfInfluenceWeights}, it is not hard to verify that we have 
$\infweight(\kf)=\Phi_d(\ratioT^{M, \sigma}_{\kv_1}, \ldots, \ratioT^{M, \sigma}_{\kv_d})$
and similarly $\infweight_{\kP}(\kf)=\Phi_d(\ratioTP^{\kL,\eta}_{\kv_1},\ldots, \ratioTP^{\kL, \eta}_{\kv_d})$.
From the {\em mean value theorem}, there exists ${\bf \kz}\in \mathbb{R}^d$ such that 
\begin{align}\label{eq:MeanValue4Alphas}
\infweight_{\kP}(\kf)&=\infweight(\kf) +
\sum^{d}_{\ki=1} \left . \frac{\partial }{\partial {\bf \kx}_{\ki}}\Phi_d({\bf \kx}) \right |_{{\bf \kx}={\bf \kz}} \cdot 
\left (\ratioTP^{\kL, \eta}_{\kv_{\ki}}-\ratioT^{M, \sigma}_{\kv_{\ki}} \right )\enspace.
\end{align}
Recall that we take $\kf$ such that $\infweight(\kf) \neq 0$.  Our  $b$-boundedness assumption, for fixed 
$b>0$, and having  $ \infweight(\kf) \neq 0$, imply that $\abs{\infweight (\kf)}$ is a constant  bounded away 
from zero. Hence, \eqref{eq:prop:StableSSM} is true if the sum on the r.h.s. in \eqref{eq:MeanValue4Alphas} 
is a  sufficiently small constant.

A standard calculation implies that for $\ki =1,\ldots, d$, we have 
\begin{align}
\frac{\partial }{\partial {\bf \kx}_{\ki}}\Phi_d({\bf \kx}) &= \Phi_d({\bf \kx}) \cdot {\textstyle \frac{\beta\gamma-1}{(\beta {\bf \kx}_{\ki}+1)({\bf \kx}_{\ki}+\gamma)}}
\cdot {\textstyle \left(1-\frac{1+\beta\gamma+2F_d({\bf \kx})}{\beta\gamma-1} \cdot \Phi_d({\bf \kx}) \right) }\enspace. 
\end{align}
The above implies that $\frac{\partial }{\partial {\bf \kx}_{\ki}}\Phi_d({\bf \kx}) $ is bounded only
if  both $\Phi_d({\bf \kx})$ and $F_d({\bf \kx})$ are bounded for any ${\bf \kx}\in [0,\infty]^d$ and 
$d\in [\maxDeg]$.
From the definition of the functions $\Phi_d({\bf \kx})$ and $F_d({\bf \kx})$, it follows that they are indeed  bounded,  implying that 
$\frac{\partial }{\partial {\bf \kx}_{\ki}}\Phi_d({\bf \kx}) $ is also bounded  for $\ki =1,\ldots, d$.

In light of all the above, \eqref{eq:prop:StableSSM} follows by showing that for any fixed number 
$\upkappa>0$,  there is $\ell_0=\ell_0(\upkappa)$ such that for any edge 
$\kf$ at distance $\kell\geq \ell_0$ from set $\cpTP$, we have
\begin{align}\label{eq:Target4prop:StableSSM}
\abs{ \ratioTP^{\kL, \eta}_{\kv_{\ki}}-\ratioT^{M, \sigma}_{\kv_{\ki}} } &\leq \upkappa & \textrm{for } 1\leq \ki\leq d\enspace. 
\end{align} 
Note that assuming that  $\kf$ is at distance $\kell$ from $\cpTP$ implies that the ball of radius  
$\kell-1$ from each vertex $\kv_{\ki}$ is inside  tree  $\SMTPu$.  This follows from 
\Cref{lemma:Subtree4TSawU,lemma:BoundaryTVsBoundaryTPFromU}.

Furthermore, statement \ref{stat:C-lemma:BoundaryTVsBoundaryTPFromU} in \Cref{lemma:BoundaryTVsBoundaryTPFromU} implies that there is a pair of
configurations $\xi^{(\ki)}_1$ and $\xi^{(\ki)}_2$ at $N=\kL\cup \cpTP$, which only disagree at $\cpTP$, such that 
\begin{align} \label{eq:RealRatiosVsTPRatios}
\abs{ \ratioTP^{\kL, \eta}_{\kv_{\ki}}-\ratioT^{M, \sigma}_{\kv_{\ki}} } &\leq 
\abs{ \ratioTP^{N, \xi^{(\ki)}_1}_{\kv_{\ki}}-\ratioTP^{N, \xi^{(\ki)}_2}_{\kv_{\ki}} } & \textrm{for } 1\leq \ki\leq d
\enspace,
\end{align}
where $\ratioTP^{N, \xi^{(\ki)}_1}_{\kz}$ and $\ratioTP^{N, \xi^{(\ki)}_2}_{\kz}$ are the ratios of marginals at vertex $\kz$ with respect
to the Gibbs distributions $\mu^{N, \xi^{(\ki)}_1}_{\SMTPu}$ and $\mu^{N, \xi^{(\ki)}_2}_{\SMTPu}$, respectively.
We bound the r.h.s. of \eqref{eq:RealRatiosVsTPRatios} by using the following result.

 \begin{claim} \label{claim:Target4prop:StableSSM}
For $\delta= \frac{1-\varepsilon}{\SingBound}$, suppose that there
is a $(\pfs,\delta,c)$-potential function $\potF$ with respect to $(\beta,\gamma,\lambda)$, 
where $\varepsilon, \beta,\gamma, \lambda,\SingBound$ are defined in \Cref{prop:StableSSM}. 

Then, there are constants $c_A, c_B>0$ and $t_0>0$, such that for any vertex $\kz$ in tree $\SMTPu$ which is at 
distance $\kell\geq \kt_0$ from $\cpTP$, the following is true:

For $N=\kL\cup \cpTP$ and for {\em any} two configurations $\eta_1, \eta_2\in \{\pm 1\}^{N }$
that differ only at $\cpTP$, we have 
\begin{align}\label{eq:ClaimTarget4prop:StableSSM}
\abs{ \ratioTP^{N, \eta_1}_{\kz} - \ratioTP^{N, \eta_2}_{\kz}} \leq c_A \cdot \exp\left( -c_B \cdot \kell \right)\enspace,
\end{align}
where $\ratioTP^{N, \eta_1}_{\kz}$ and $\ratioTP^{N, \eta_2}_{\kz}$ are the ratios of marginals at vertex $\kz$ with respect
to the Gibbs distributions $\mu^{N, \eta_1}_{\SMTPu}$ and $\mu^{N, \eta_2}_{\SMTPu}$, respectively.
 \end{claim}

We get \eqref{eq:Target4prop:StableSSM} from \eqref{eq:RealRatiosVsTPRatios} and \Cref{claim:Target4prop:StableSSM}. 
 \Cref{prop:StableSSM} follows. \hfill $\square$

\subsubsection{Proof of \Cref{claim:AlphaPZero}}
The proof is not too different from that of \Cref{claim:SupportOfWeiS}

In this proof, we identify $\SMTPu=\Tsaw(\gext{G}{\kP},\ku)$ as a subtree of $\SMTu=\Tsaw(G,\ku)$. 
W.l.o.g.,  let edge $\kf$ in $\SMTPu$ be such that $\kf=\{\kv,\kx\}$. Also, let $\uppi$ be the set of vertices that consists 
of $\kv,\kx$ and their children in $\SMTu$.

Arguing as in \Cref{claim:SupportOfWeiS}, the assumption of having $\kf$ at distance $>10$ from $\cpTP$ implies
that all elements in $\uppi$ are vertices in $\SMTPu$.

Our assumption for $b$-boundedness, for fixed $b>0$, implies that for edge $\kf$ in $\SMTPu$, we either 
have $\infweight(\kf)=0$, or $\abs{\infweight(\kf)}$ is bounded away from zero. The same holds for  
$\infweight_{\kP}(\kf)$. Furthermore, we have $\infweight(\kf)=0$ (resp. $\infweight_{\kP}(\kf)=0$) only if at 
least one of the vertices in $\uppi$ has a pinning under $\mu^{M,\sigma}_{\SMTu}$ (resp. $\mu^{L,\eta}_{\SMTPu}$).
Note that for the case where the parameter of the Gibbs distribution $\beta=0$, it is possible that by pinning one 
of the children of $\kv$ or $\kx$ we get $\infweight(\kf)=0$ (resp.  $\infweight_{\kP}(\kf)=0$).

Recall from \Cref{lemma:BoundaryTVsBoundaryTPFromU}, statement \ref{stat:C-lemma:BoundaryTVsBoundaryTPFromU}, 
that the pinnings imposed by  $\mu^{M,\sigma}_{\SMTu}$ and $\mu^{\kL,\eta}_{\SMTPu}$ on the vertices in
$\SMTPu$ can only differ on the vertices in $\kL\setminus M\subseteq \cpT$.   But since $\kL\setminus M\subseteq \SMTPu$,
we have that $\kL\setminus M\subseteq \cpTP$, i.e., since statement  \ref{stat:B-lemma:BoundaryTVsBoundaryTPFromU} in
\Cref{lemma:BoundaryTVsBoundaryTPFromU} implies that   $\cpT\cap \SMTPu=\cpTP$.

Taking $\kf$ at a sufficiently large distance from $\cpTP$, e.g., distance greater than $10$, we have that  $\uppi\not\subset \cpTP$. 
Hence,  if the two distributions  $\mu^{M,\sigma}_{\SMTu}$ and $\mu^{\kL,\eta}_{\SMTPu}$ impose  pinnings to  vertices in 
set $\uppi$, these pinning should be identical.

This leads us to conclude that if $\infweight(\kf)=0$, this is because $\mu^{M,\sigma}_{\SMTu}$ fixes the configuration
of a vertex in $\uppi$. But then, $\mu^{\kL,\eta}_{\SMTPu}$ imposes identical pinning to $\uppi$. Hence, we also have 
$\infweight_{\kP}(\kf)=0$. 

All the above conclude the proof of \Cref{claim:AlphaPZero}. \hfill $\square$

\subsubsection{Proof of \Cref{claim:Target4prop:StableSSM}}
In this proof, we identify $\SMTPu=\Tsaw(\gext{G}{\kP},\ku)$ as a subtree of $\SMTu=\Tsaw(G,\ku)$.

Firstly, we upper bound the number of paths in $\SMTPu$ from vertex $\kz$ to the set of vertices  $\cpTP$ as a function of 
$\kell$, i.e., the distance of $\kz$ from the set. If $\kz$ is a copy of vertex $\widehat{\kz}$ in $\gext{G}{\kP}$, then the paths from 
$\kz$ to  $\cpTP$ correspond to the self-avoiding walks from $\widehat{\kz}$ to the set of vertices $\kP\cup \ssplit_{\kP}$. 
Let $\kK(\kell)$ be the number of these walks. We assume that $\kell\geq \kt_0$  for a large constant $\kt_0$, e.g., $\kt_0>100\kk$. 

Note that $\kK(\kell)$ is upper bounded by the number of self-avoiding walks in $G$ from vertex $\widehat{\kz}$ to $\kP$. 
Furthermore, since each self-avoiding path of length $\kell>\kk$ is also a $\kk$-non-backtracking walk, we have 
\begin{align}\label{eq:KEllVsNBMGPEntry}
\kK(\kell) &\leq\sum\nolimits_{\kW,\kQ}\NBMatrix^{\kell-\kk}(\kW,\kQ)  
\enspace.
\end{align}
Variable $\kW$ varies over the walks in $\ExtV_{\kk}$ which emanate from vertex $\widehat{\kz}$.   Variable $\kQ$ varies over 
$\ExtV_{\kk}$ such that  its  last vertex is in $\kP$.  The exponent $\kell-\kk$ has to do with how we measure distances 
between $\kk$-non-backtracking walks with $\NBMatrix$.

A crude estimation implies that the number of summands in 
\eqref{eq:KEllVsNBMGPEntry} is $\leq (\kk+2)\cdot \maxDeg^{\kk}$. 
We have
\begin{align}\label{eq:KELLVMaxHGKWQ}
\kK(\kell) &\leq (\kk+2)\cdot \maxDeg^{\kk} \cdot \max\nolimits_{\kW,\kQ}\{ \NBMatrix^{\kell-\kk}(\kW,\kQ)\}
 \enspace.
\end{align}
Furthermore, \Cref{claim:SigmaLVsPathNumber,lemma:SingSequenConv} imply that for any $\zeta\in (0,1)$,
 there exists integer $\hat{\kr}_0=\hat{\kr}_0(N,\maxDeg, \kk, \zeta)>1$ such that for any $\kr \geq \hat{\kr}_0$ we have
 \begin{align}\label{eq:FromClaim5758}
\NBMatrix^{\kr}(\kW,\kQ) \leq \left( (1+\zeta)\cdot \SingBound \right)^{\kr}\enspace.
 \end{align}
Recall that $\nnorm{ (\NBMatrix)^N}{1/N}{2}=\SingBound$.
Combining \eqref{eq:KELLVMaxHGKWQ} and \eqref{eq:FromClaim5758}, for $\kell\geq \max\{\hat{\kr}_0,\kk\}$, we have 
\begin{align}
\kK(\kell) & \leq (\kk+2)\cdot \maxDeg^{\kk}\cdot \left( (1+\zeta)\cdot \SingBound \right)^{\kell-\kk} \enspace. 
\end{align}
We choose $\zeta$ such that $\zeta=(1+\varepsilon/2)(1-\varepsilon/2)$, where $\varepsilon$ is from the 
parameter $\delta$ of the $(\pfs,\delta,c)$-potential function $\potF$, i.e., recall that we have $\delta= \frac{1-\varepsilon}{\SingBound}$. 
We choose $\kt_0$ to be the maximum between the quantities $2\varepsilon^{-1}\log((\kk+2)\cdot \maxDeg^{\kk})$ 
and $\hat{\kr}_0$. 
With our choices of $\zeta$ and $\kt_0$, for any $\kell\geq \kt_0$, we have 
$\kK(\kell)\leq \left((1+\varepsilon/2)\cdot \SingBound\right)^{\kell}$.

With the above in mind, standard arguments imply that  for \eqref{eq:ClaimTarget4prop:StableSSM}, it suffices to show 
 that there exists a differentiable, increasing potential function  $\Upxi:[-\infty,+\infty]\to(-\infty,+\infty)$ such that 
\begin{align}\label{eq:SSMCondition4Stability}
\phi\left( F({\bf \kz})\right)\cdot \sum^d_{\ki=1} \frac{1}{\phi({\bf \kz}_{\ki})}\cdot \left. \frac{\partial}{\partial {\bf \kx}_{\ki}}F({\bf \kx})\right|_{{\bf \kx}={\bf \kz}} \cdot {\bf m}_{\ki}
&\leq \left( \frac{1-\varepsilon}{\SingBound}\right)^{1/\pfs} \cdot \norm{ {\bf m}}{\pfs}\enspace, 
\end{align}
where $\phi=\Upxi'$. 
Using the standard chain rule for derivatives, (see also  \cite{VigodaSpectralInd}) we conclude that \eqref{eq:SSMCondition4Stability} is true by using 
$\Upxi=\potF\circ\log$, where $\potF$ is the $(\pfs,\delta,c)$-potential function. \Cref{claim:Target4prop:StableSSM} follows. 
\hfill $\square$

\spreadpoint
 \section{Proof of rapid mixing results in \Cref{sec:Introduction} - Hard-core Model} \label{sec:MainResultsHC}

We start by introducing the potential function
$\potF$. We define $\potF$ in terms of $\xdpotF=\potF'$ by having
{
\begin{align}\label{def:XPotentialFunction}
\xdpotF &: \mathbb{R}_{>0} \to \mathbb{R} & \textrm{such that} && y& \mapsto {\textstyle{\sqrt{\frac{e^y}{1+e^y}} }}  \enspace,
\end{align}
}
while $\potF(0)=0$. 

The potential function $\potF$ was proposed in a more general form in \cite{VigodaSpectralInd}. 
 It is standard to show that $\potF$ is well-defined (see also \cite{VigodaSpectralInd}). 

For any given $\lambda>0$, we define, implicitly,  function $\dcritical(\lambda)$ to be the positive 
number $z>1$ such that $\frac{z^{z}}{(z-1)^{(z+1)}}=\lambda$. 
From its definition it is not hard to see that $\dcritical(\cdot)$ is the inverse map of $\lcritical(\cdot )$, i.e., we have that 
$\dcritical(x)=\lcritical^{-1}(x)$. In that respect, $\dcritical(x)$ is well-defined as $\lcritical(x)$ is monotonically 
decreasing in $x$.

\begin{theorem}\label{thrm:GoodPotentialHC}
For $\lambda>0$, let $\dcritical=\dcritical(\lambda)$. The function $\potF$ defined in \eqref{def:XPotentialFunction} is a 
$(\pfs, \delta, c)$-potential function such that 
{
\begin{align}
\pfs^{-1}&= 1-{\textstyle \frac{\dcritical-1}{2}}\log\left(1+\textstyle{\frac{1}{\dcritical-1}} \right) \enspace, 
&\delta& \leq {\textstyle \frac{1}{\dcritical}} &\textrm{and} && c \leq {\textstyle \frac{\lambda}{1+\lambda} }\enspace.
\end{align}
}
\end{theorem}
The proof of \Cref{thrm:GoodPotentialHC} appears in \Cref{sec:thrm:GoodPotentialHC}.
We also have the following claim.

\begin{restatable}{claim}{ElaborateGoodPotenHC}\label{claim:InterpretGoodPotentialHC}
For $\varepsilon\in (0,1)$, $\kR \geq 2$ and $0<\lambda<(1-\varepsilon)\lcritical(\kR)$ the following is true: 
There is $0<\kz<1$ , which only depend on $\varepsilon$, such that for $\dcritical=\dcritical(\lambda)$, we have 
{
\begin{align}\label{eq:claim:InterpretGoodPotentialHC}
{\textstyle \frac{1-z}{\kR}} & \geq {\textstyle\frac{1}{\dcritical} }&\textrm{and} && {\textstyle \frac{\lambda}{1+\lambda}} &<\textstyle{\frac{e^3}{\kR} }\enspace. 
\end{align}
}
\end{restatable}
\Cref{claim:InterpretGoodPotentialHC} follows from elementary calculations, 
for a proof see \Cref{sec:claim:InterpretGoodPotentialHC}.

\subsection{Proof of \Cref{thrm:HC4SPRadiusHash}}
Combining \Cref{thrm:GoodPotentialHC}, \Cref{thrm:NonBacktrackingPotentialSpIn} and \Cref{claim:InterpretGoodPotentialHC}, 
with our assumption that $\maxDeg$ and $\SingBound$ being $\Theta(1)$,
we have
\begin{align} \nonumber 
\spradius {\textstyle \left( \infmatrix^{\Lambda,\tau}_{G} \right)} &= O(1) \enspace, 
\end{align}
for any $\Lambda\subset V$ and $\tau\in \{\pm 1\}^{\Lambda}$.
Furthermore, our assumptions about $\lambda$ imply that $\mu$ is $b$-marginally bounded, 
for fixed $b>0$.

Considering all the above, \Cref{thrm:HC4SPRadiusHash} follows as a corollary from \Cref{thrm:SPINLOGN}.
\hfill $\square$

\subsection{Proof of \Cref{thrm:HC4SPRadiusAdj}} 
The proof of \Cref{thrm:HC4SPRadiusAdj} is almost identical to that of \Cref{thrm:HC4SPRadiusHash}.

Combining \Cref{thrm:GoodPotentialHC}, \Cref{thrm:AdjacencyPotentialSpIn} and \Cref{claim:InterpretGoodPotentialHC},
with our assumption about $\maxDeg$ and $\aspradius$ being $\Theta(1)$,
%
we have
\begin{align} \nonumber 
\spradius \left( \infmatrix^{\Lambda,\tau}_{G} \right) &=O(1)\enspace,
\end{align}
for any $\Lambda\subset V$ and $\tau\in \{\pm 1\}^{\Lambda}$.
Furthermore, our assumptions about $\lambda$ imply that $\mu$ is $b$-marginally bounded, 
for fixed $b>0$.

Considering all the above, \Cref{thrm:HC4SPRadiusAdj} follows as a corollary from \Cref{thrm:SPINLOGN}.
\hfill $\square$

\subsection{Proof of \Cref{thrm:HC4CCK}} 
 
To prove \Cref{thrm:HC4CCK}, we use the following proposition. 

\begin{restatable}{proposition}{CcLccVsHSPradius}\label{prop:HSPRadCCLCC}
For integers $\kk,N>0$ and graph $G=(V,E)$, we have that 
$\norm{(\NBMatrix)^N}{2} \leq \left( \cconnective_{\kk}\right)^{N+k}$.
\end{restatable}
The proof of \Cref{prop:HSPRadCCLCC} appears in \Cref{sec:prop:HSPRadCCLCC}.

Given the parameters $\varepsilon\in (0,1)$ and $\cconnective_{\kk}$,
 let $\kz=\kz(\varepsilon, \cconnective_{\kk})>0$ be  such that
\begin{align}\label{eq:Base4thrm:HC4CCK}
\lcritical((1+\kz)\cconnective_k)= (1-\varepsilon/3)\lcritical(\cconnective_k) \enspace.
\end{align}
Such $\kz$ exists and it is unique, since   $\lcritical(\kx)$ is monotonically decreasing in $\kx$. 

 \Cref{prop:HSPRadCCLCC} implies that there exists $M=M(\kz,\cconnective_{\kk})$ such that for any $N\geq M$,
we have 
\begin{align}\label{eq:Base4thrm:HC4CCKBB}
\nnorm{ (\NBMatrix)^N}{1/N}{2} \leq \left( \cconnective_{\kk}\right)^{\frac{N+k}{N}}\leq (1+\kz)\cconnective_{\kk}\enspace. 
\end{align}
Let $\SingBound=\nnorm{ (\NBMatrix)^M}{1/M}{2}$, where $M$ is defined above. From \eqref{eq:Base4thrm:HC4CCK} and \eqref{eq:Base4thrm:HC4CCKBB}
we have
\begin{align}
\lcritical(\SingBound) &\geq (1-\varepsilon/3)\lcritical(\cconnective_k) \enspace.
\end{align}
For the above, we use once more that    $\lcritical(x)$ is monotonically decreasing in $x$. 

\Cref{thrm:HC4CCK} follows by noting that the above implies that for $\lambda<(1-\varepsilon)\lcritical(\cconnective_k)$ we also have that 
$\lambda<(1-\varepsilon/3)\lcritical(\SingBound)$. For such $\lambda$,  \Cref{thrm:HC4SPRadiusHash} implies that 
 there is a constant $C>0$, such that Glauber dynamics on the Hard-core model on $G$ 
has mixing time at most $C n\log n$.

\hfill $\square$

\subsection{Proof of \Cref{prop:HSPRadCCLCC}}\label{sec:prop:HSPRadCCLCC}

%
%

Let integer $\kM>1$ which we will specify later. Let a length $\kk \cdot \kM$ walk 
$\kP=\kw_0, \ldots, \kw_{\kk \cdot \kM}$  in $G$.  We say that $\kP$ is $\kk$-wise self-avoiding, 
if the sub-walk $\kP_{\ki}=\kw_{\ki\cdot \kk}, w_{\ki\cdot \kk+1}, \ldots, w_{\ki\cdot \kk+\kk}$ is 
self-avoiding, for all $\ki=0,\ldots, \kM-1$. Note that $\kP$ is not necessarily a self-avoiding walk 
since the same vertex $\kv$ can appear in both $\kP_{\ki}$ and $\kP_{\kj}$, for $\ki\neq \kj$.

For each $\kv\in V$, let $\uppi_{\kk}(\kv, \kk \kM)$ be the set of $\kk$-wise self-avoiding walks of length 
$\kk \cdot \kM$ that emanate from $\kv$.  Let $\upchi_{\kk,\kM}=\sup\{ \uppi_{\kk}(\kv, \kk \kM)\ |\ \kv\in V\}$. 
It is an easy exercise to show that 
\begin{align}
\upchi_{\kk, \kM} & \leq (\cconnective_{\kk})^{\kk \kM}\enspace. 
\end{align}
Furthermore, let $\kQ=x_{0}, \ldots, x_{\kk \kM}$ be a $\kk$-non-backtracking walk in $G$.  $\kQ$ is also a $\kk$-wise 
self-avoiding walk. To see this,  note that for all $\kj=0,\ldots, \kk\cdot(\kM-1)$, the sub-walk $\kx_{\kj},\ldots, \kx_{\kj+\kk}$ 
in $\kQ$ is self-avoiding.   For $\kQ$ to be $\kk$-wise self-avoiding, we only need the sub-walks $\kx_{\kj},\ldots, \kx_{\kj+\kk}$ 
to be self-avoiding, for $\kj$ being a multiple of $\kk$. Hence,  $Q$ is a $\kk$-wise self-avoiding walk.

For $\kW, S\in \ExtV_{\kk}$, i.e.,  two self-avoiding walks of length $\kk$ in $G$, recall that $\NBMatrix^{\kk (M-1)}(\kW, S)$ is equal to the number of 
$\kk$-non-backtracking walks of length $\kk\cdot (\kM-1)$ that start from $\kW$ and end at $S$. 

Every $\kk$-non-backtracking walk of length $\kk\cdot (\kM-1)$ from $\kW$ to $S$ is also a $\kk$-wise self-avoiding 
walk of length $\kk\cdot \kM$ starting at $\kv$,  the first vertex in $\kW$. Hence, we have 
\begin{align}
\sum\nolimits _{S\in \ExtV_{\kk}}\NBMatrix^{\kk(\kM-1)}(\kW, S) \leq \uppi_{\kk}(\kv, \kk \kM) \leq \upchi_{\kk, \kM} \enspace.
\end{align}
%

From all the above, we obtain 
\begin{align}\nonumber
\norm{ (\NBMatrix)^{\kk(\kM-1)}\cdot \Invol }{\infty}=
\norm{ (\NBMatrix)^{\kk(\kM-1)}}{\infty} &\leq \upchi_{\kk,\kM}= (\cconnective_{\kk})^{\kk \kM}\enspace, 
\end{align}
where $\Invol$ is the involution defined in \eqref{def:OfInvolutation}.
For any integer $\kell>0$, we have that
\begin{align}\label{eq:HRLInfBound}
\norm{ ( (\NBMatrix)^{\kk(\kM-1)}\cdot \Invol )^{\kell}}{\infty} &\leq 
\nnorm{
(\NBMatrix)^{\kk(\kM-1)}\cdot \Invol}{\kell}{\infty} \leq (\cconnective_{\kk})^{\kk\cdot \kM\cdot \kell}
\enspace. 
\end{align}
Recall that $(\NBMatrix)^{\kk(\kM-1)}\cdot \Invol$ is a symmetric matrix, e.g., see \eqref{eq:DefPTInvarianceFormalB}. 
From  Gelfand's formula e.g., see \cite{SIAM-LAlg}, we get
\begin{align}\nonumber
\norm{ (\NBMatrix)^{\kk(\kM-1)}\cdot \Invol }{2} =
 \lim_{\kell\to\infty} 
  \nnorm{  {\textstyle ( (\NBMatrix)^{\kk(\kM-1)}\cdot \Invol )^{\kell} } }{1/\kell}{\infty} &
 \leq \lim_{\kell \to \infty} \left ((\cconnective_{\kk})^{\kk\cdot M\cdot \kell} \right)^{1/\kell}
= (\cconnective_{\kk})^{\kk\cdot \kM} \enspace. 
\end{align}
where the second derivation follows from \eqref{eq:HRLInfBound}. 
Since $\Invol$ is an involution, we have  $\norm{ (\NBMatrix)^{\kk(\kM-1)}\cdot \Invol}{2}=\norm{ (\NBMatrix)^{\kk(\kM-1)}}{2}$. 
From the above,  we conclude that $\norm{ (\NBMatrix)^{\kk(\kM-1)}}{2} \leq \left(\cconnective_{\kk}\right)^{\kk\cdot \kM}$.

\Cref{prop:HSPRadCCLCC} follows by setting $\kM=\frac{N}{k}+1$. 
\hfill $\square$

\spreadpoint
 \section{Proof of results in \Cref{sec:Introduction} - Ising Model}\label{sec:MainResultsIsing}
For $d>0$, consider the function $\logtrecur_d$ defined in \eqref{eq:DefOfH}.
For the Ising model, this correspond to the  function  $\logtrecur_d:[-\infty, +\infty]^d\to [-\infty, +\infty] $ such that
{\small
\begin{align}\label{eq:LogRatioIsing}
 (\kx_1, \ldots, \kx_d)\mapsto \log\lambda+ \sum\nolimits_{\ki\in [d]}
\log\left( \frac{\beta \exp(\kx_{\ki})+1}{\exp(\kx_{\ki})+\beta} \right) \enspace .
\end{align}
}

The following is a folklore result. 

\begin{restatable}{lemma}{IsingInfNormBound}\label{lemma:IsingInfNormBound}
For any $\varepsilon \in (0,1)$, $d>0$, $\kR > 1$ and $\beta\in \UnIsing(\kR,\varepsilon )$
we have
{\small 
\begin{align}\label{eq:lemma:IsingInfNormBound}
\sup\nolimits_{\bf y}\norm{\nabla \logtrecur_d ({\bf y})}{\infty} &\leq { \frac{1-\varepsilon}{\kR}} \enspace.
\end{align}
}
\end{restatable}
For the sake of our paper being self-contained, we present a proof of \Cref{lemma:IsingInfNormBound} in
\Cref{sec:lemma:IsingInfNormBound}.

\subsection{Proof of \Cref{thrm:Ising4SPRadiusHash}}\label{sec:thrm:Ising4SPRadiusHash}

For our choice of $\beta$ and $\lambda$,  \Cref{lemma:IsingInfNormBound} implies 
that the corresponding Gibbs distribution  $\mu^{\Lambda,\tau}_G$ exhibits $\left( \frac{1-\varepsilon}{\SingBound}\right)$-contraction
for any $\Lambda\subset V$ and $\tau\in \{\pm 1\}^{\Lambda}$. That is, 
\begin{align}\nonumber 
\sup\nolimits_{\bf y}\norm{\nabla \logtrecur_d ({\bf y})}{\infty} &\leq ({1-\varepsilon})/{\SingBound} & \forall d\in [\maxDeg] \enspace.
\end{align}
The above, together with \Cref{thrm:NonBacktrackingInfinityMixing},  imply that 
for any $\Lambda\subset V$ and $\tau\in \{\pm 1\}^{\Lambda}$, the pairwise influence 
matrix $\infmatrix^{\Lambda,\tau}_{G}$, induced by $\mu^{\Lambda,\tau}_G$, satisfies that 
\begin{align}\nonumber 
\spradius(\infmatrix^{\Lambda,\tau}_{G})\leq \widehat{C} \cdot \varepsilon^{-1} \enspace,
\end{align}
for fixed number $\widehat{C}>0$. Clearly, we have $\spradius(\infmatrix^{\Lambda,\tau}_{G})\in O(1)$.

Furthermore, our assumptions about $\beta$ and $\lambda$ imply that $\mu_G$ is $b$-marginally bounded, 
for fixed $b>0$.

Considering all the above, \Cref{thrm:Ising4SPRadiusHash} follows as a corollary from \Cref{thrm:SPINLOGN}.
\hfill $\square$

\subsection{Proof of \Cref{thrm:Ising4SPRadiusAdj}}\label{sec:thrm:Ising4SPRadiusAdj}
The proof of \Cref{thrm:Ising4SPRadiusAdj} is almost identical to  \Cref{thrm:Ising4SPRadiusHash}.

For our choice of $\beta$ and $\lambda$,  \Cref{lemma:IsingInfNormBound} implies 
that the corresponding Gibbs distribution  $\mu^{\Lambda,\tau}_G$ exhibits $\left( \frac{1-\varepsilon}{\aspradius}\right)$-contraction
for all $\Lambda\subset V$ and $\tau\in \{\pm 1\}^{\Lambda}$. That is, 
\begin{align}\nonumber
\sup\nolimits_{\bf y}\norm{\nabla \logtrecur_d ({\bf y})}{\infty} &\leq (1-\varepsilon)/\aspradius & \forall d\in [\maxDeg] \enspace.
\end{align}
The above, combined with \Cref{thrm:AdjacencyInfinityMixing} imply that 
for $\Lambda\subset V$ and $\tau\in \{\pm 1\}^{\Lambda}$, the pairwise influence 
matrix $\infmatrix^{\Lambda,\tau}_{G}$, induced by $\mu^{\Lambda,\tau}_G$, satisfies that 
\begin{align}\nonumber 
\spradius(\infmatrix^{\Lambda,\tau}_{G})\leq \varepsilon^{-1}=O(1) \enspace.
\end{align}
Furthermore, our assumptions about $\beta$ and $\lambda$ imply that $\mu$ is $b$-marginally bounded, 
for fixed $b>0$. 

Considering all the above, \Cref{thrm:Ising4SPRadiusAdj} follows as a corollary from \Cref{thrm:SPINLOGN}.
\hfill $\square$

\subsection{Proof of \Cref{thrm:Ising4CCK}}
The proof of \Cref{thrm:Ising4CCK} is conceptually similar to \Cref{thrm:HC4CCK}.
We are using \Cref{prop:HSPRadCCLCC} once more in this proof.

It is straightforward to show that for any $\kz<\hat{\kz}$, we have $\UnIsing(\hat{\kz},\delta)\subset \UnIsing(\kz,\delta)$. 
Similarly,  for any $\delta<\hat{\delta}$, we have $\UnIsing(\kz,\hat{\delta})\subset \UnIsing(\kz,\delta)$. 

Given the parameters $\varepsilon\in (0,1)$ and $\cconnective_{\kk}$,
 let $\kz=\kz(\varepsilon, \cconnective_{\kk})>0$ be  such that
\begin{align}\label{eq:Base4thrm:Ising4CCK}
\UnIsing(\cconnective_{\kk}, \varepsilon)= \UnIsing((1+\kz)\cconnective_{\kk}, \varepsilon/3)\enspace. 
\end{align}
Furthermore, \Cref{prop:HSPRadCCLCC} implies that there exists $M=M(z,\cconnective_{\kk})$ such that for any $N\geq M$,
we have 
\begin{align}\label{eq:Base4thrm:Ising4CCKBB}
\nnorm{(\NBMatrix)^N}{1/N}{2} \leq \left( \cconnective_{\kk}\right)^{\frac{N+k}{N}}\leq (1+\kz)\cconnective_{\kk}\enspace. 
\end{align}
Let $\SingBound=\nnorm{(\NBMatrix)^M}{1/M}{2}$, where $M$ is defined above. The above relations imply that
\begin{align}
\UnIsing(\cconnective_{\kk}, \varepsilon)\subseteq  \UnIsing(\SingBound, \varepsilon/3)\enspace. 
\end{align}
\Cref{thrm:Ising4CCK} follows by noting that the above relation implies that for $\beta\in \UnIsing(\cconnective_{\kk}, \varepsilon)$,
we also have $\beta\in \UnIsing(\SingBound, \varepsilon/3)$. For such $\beta$ and for $\lambda>0$,  \Cref{thrm:Ising4CCK} implies that there is 
a constant $C>0$, such that Glauber dynamics on the Ising model  on $G$ has mixing time at most $C n\log n$.
%
\hfill $\square$

\spreadpoint
\section{Proof of \Cref{thrm:GoodPotentialHC}} \label{sec:thrm:GoodPotentialHC}
Recall that the ratio of Gibbs marginals $\gratio^{\Lambda, \tau}_{\ku}$, defined in 
\Cref{sec:RecursionVsSpectralIneq}, is possible to be equal to zero, or $\infty$. Typically, this 
happens if  vertex $\ku$ with respect to which we consider the ratio is pinned, i.e., $\ku \in \Lambda$, 
or has a neighbour in $\Lambda$. When we deal with the Hard-core model, 
there is a standard way to avoid these infinities and zeros in our calculations and make the derivation
much simpler.

Suppose we have the Hard-core model with fugacity $\lambda>0$ on a tree $T$, while at the 
set of vertices $\Lambda$ we have a pinning $\tau$. Then, it is elementary that 
this distribution is identical to the Hard-core model with the same fugacity on the tree (or forest) $T'$ 
which is obtained from $T$ by removing every vertex $w$ which either belongs to $\Lambda$, or 
has a neighbour in $ \Lambda$ which is pinned to 1.

From now on, for the instances we consider, assume that we have applied the above steps and
removed any pinnings.

It is useful to write down the functions that arise from the tree recursions
in \Cref{sec:RecursionVsSpectralIneq}, specifically for 
the Hard-core model with fugacity $\lambda$. 
Recall that, in this case, we have $\beta=0$ and $\gamma=1$. 
In the following definitions, we assume that there are no pinnings. 

For integer $d\geq 1$, we have  $\trecur_d:\mathbb{R}^d_{>0}\to (0, \lambda)$ such that 
for any ${\bf \kx}=({\bf \kx}_1, \ldots, {\bf \kx}_d)$, we have  
{\small 
\begin{align}\label{eq:HC-BPRecursion}
{\bf \kx}\mapsto \lambda \prod\nolimits_{\ki\in [d]}\frac{1}{{\bf \kx}_{\ki}+1} \enspace.
\end{align}
}
We also define $\trecur_{d,{\rm sym}}:\mathbb{R}_{>0}\to (0, \lambda)$ the {\em symmetric} version 
of the above function, that is
\begin{align}\label{def:SymTreeHC}
\kx&\mapsto \trecur_d(\kx,\kx,\ldots, \kx) \enspace.
\end{align}
We also have $\logtrecur_d:\mathbb{R}^d\to \mathbb{R}$ such that for any 
${\bf \kx}=({\bf \kx}_1, \ldots, {\bf \kx}_d)$, we have  
{\small
\begin{align}\label{eq:HC-DefOfH}
{\bf \kx} \mapsto \log \lambda+\sum\nolimits_{\ki\in [d]}
\log\left( \frac{1}{\exp({\bf \kx}_{\ki})+1} \right) \enspace.
\end{align}
}
For $\frac{\partial }{\partial {\bf \kx}_{\ki}}\logtrecur_d({\bf x})=\dlogtrecur({\bf \kx}_{\ki})$, 
we have $\dlogtrecur:\mathbb{R} \to \mathbb{R}$ such that 
{\small
\begin{align}\label{def:GradientRecHC}
{\bf  \kx} \mapsto -\frac{e^{{\bf \kx}_{\ki}}}{e^{{\bf \kx}_{\ki}}+1} \enspace.
\end{align}
}
Finally, the set of log-ratios $\ratiorange$, defined in \eqref{eq:DefOfRatiorange}, 
satisfies that 
\begin{align}\label{eq:RangeJ4HC}
 \ratiorange=(-\infty, \log(\lambda)) \enspace.
\end{align} 
Set $S_{\potF}$, i.e., the image of $\potF$, satisfies $S_{\potF}=(-\infty, \infty)$.

\newcommand{\easyspace}{Q_{\potF}}
\newcommand{\easyspaceX}{L_{\potF}}

 \Cref{thrm:GoodPotentialHC}  follows by showing that $\potF$ satisfies the Contraction
and the Boundedness conditions of \Cref{def:GoodPotential}, for the parameters indicated in the statement.

We start with the {\em Contraction}. For any integer $d>0$, we let 
$\cR_{d}:\mathbb{R}^d\times \mathbb{R}^d \to \mathbb{R}$ be such that 
for ${\bf m}=({\bf m}_1, \ldots, {\bf m}_d)\in \mathbb{R}^{d}_{\geq 0}$, and 
${\bf \ky}=({\bf \ky}_1, \ldots, {\bf \ky}_d) \in \mathbb{R}^d$ we have  
{\small 
\begin{align}
\cR_{d}({\bf m}, {\bf \ky})={\textstyle \xdpotF(\logtrecur_{d}({\bf \ky})) }\sum^{d}_{j=1}
\frac{ \abs{ \dlogtrecur( {\bf \ky}_j ) }}{ \xdpotF( {\bf \ky}_j )} 
\times {\bf m}_j \enspace, \nonumber
\end{align}
}
where recall that $\xdpotF=\potF'$.

The following result implies that $\potF$ satisfies the contraction condition.
\begin{proposition}[Contraction]\label{prop:PotContracts}
For $\delta$ and $s$ as defined in \Cref{thrm:GoodPotentialHC}, the following is true: 
for any integer $d>0$, for ${\bf m}\in \mathbb{R}^{d}_{\geq 0}$, we have 
\begin{align}
\sup\nolimits_{{\bf \ky}\in (\easyspace)^d }\{ \cR_{d}({\bf m}, {\bf y}) \} \leq \delta^{\frac{1}{s}} \cdot \norm{ {\bf m}}{s} \enspace,
\end{align}
where $\easyspace \subseteq \mathbb{R}$ contains every $y \in \mathbb{R}$ such that there is 
$\tilde{y} \in S_{\potF} $ for which we have ${y}=\potF^{-1}(\tilde{y})$. 
\end{proposition}

The proof of \Cref{prop:PotContracts} appears in \Cref{sec:prop:PotContracts}.

 We now focus on establishing the {\em Boundedness} property of $\potF$.

\begin{lemma}[Boundedness]\label{lemma:PotBooundedness}
We have that $\textstyle \max_{\ky_1, \ky_2\in \ratiorange }\left\{ \xdpotF(\ky_2) \cdot \frac{|\dlogtrecur(\ky_1)|}{\xdpotF(\ky_1)} \right\} 
\leq \frac{\lambda}{1+\lambda}$.
\end{lemma}
\begin{proof}
Using the definitions of the functions $\xdpotF$ and $\dlogtrecur$ from \eqref{def:PotentialFunction} and 
\eqref{def:GradientRecHC}, respectively, we have
{\small 
\begin{align}
 \max_{\ky_1,\ky_2\in \ratiorange }\left\{ \xdpotF(\ky_2) \cdot \frac{|\dlogtrecur(\ky_1)|}{\xdpotF(\ky_1)} \right\} &=
 \max_{\ky_1,\ky_2\in \ratiorange} \left\{\sqrt{\dlogtrecur(\ky_1)\dlogtrecur(\ky_2) } \right\} 
 \ =\ \max_{\ky_1,\ky_2\in \ratiorange} \left\{\sqrt{ \frac{e^{\ky_1}}{1+e^{\ky_1}}\frac{e^{\ky_2}}{1+e^{\ky_2}} } \right\} 
\ =\ { \frac{\lambda}{1+\lambda}} \enspace. \nonumber 
\end{align}
}
The last inequality follows from the observation that the function $g(\kx)=\frac{e^{\kx}}{1+e^{\kx}}$ is increasing in $\kx$, 
while, from \eqref{eq:RangeJ4HC}, we have that $e^{\ky_1},e^{\ky_2}\leq \lambda$. 
The claim follows. 
\end{proof}

In light of \Cref{prop:PotContracts} and \Cref{lemma:PotBooundedness}, \Cref{thrm:GoodPotentialHC} follows. 
 \hfill $\square$

\newcommand{\HCFactor}{{\Xi}}

\subsection{Proof of \Cref{prop:PotContracts}}\label{sec:prop:PotContracts}

The proposition follows by using results from \cite{ConnectiveConst}. However, in order to apply these results, we need
to bring $\cR_{d}({\bf m}, {\bf y})$ into an appropriate form.

For any $d>0$, we let $\cJ_{d}:\mathbb{R}^d_{\geq 0}\times \mathbb{R}^d_{\geq 0}\to \mathbb{R}$ be such that 
for ${\bf m}=({\bf m}_1, \ldots, {\bf m}_d)\in \mathbb{R}^{d}_{\geq 0}$ and ${\bf \kz}=({\bf \kz}_1,\ldots, {\bf \kz}_d) \in \mathbb{R}^{d}_{\geq 0}$, we have
{\small 
\begin{align}
\cJ_{d}({\bf m}, {\bf \kz}) &= \xdpotF(\log \trecur_{d}({\bf \kz}) ) \sum\nolimits_{j\in [d]}
\frac{\abs{ \dlogtrecur(\log {\bf \kz}_j )}}{ \xdpotF( \log {\bf \kz}_j )} 
\times {\bf m}_j \enspace. \nonumber
\end{align}
}
Using the definitions in \eqref{eq:HC-BPRecursion} and \eqref{eq:HC-DefOfH}, it is elementary to verify that for any $d>0$, for any 
${\bf m}\in \mathbb{R}^d_{\geq 0}$, ${\bf z}\in \mathbb{R}^d_{>0}$ and ${\bf y}\in \mathbb{R}^{d}$ such that ${\bf z}_j=e^{{\bf y}_j}$,
we have that
\begin{align}\nonumber
\cJ_{d}({\bf m}, {\bf z})= \cR_{d}({\bf m}, {\bf y}) \enspace.
\end{align}
In light of the above, the proposition follows by showing that 
\begin{align}\label{prop:Target4PotContractsA}
\sup\nolimits_{{\bf z}\in \mathbb{R}^d_{>0}} \{ \cJ_{d}({\bf m}, {\bf z}) \} \leq \delta^{{1}/{\pfs}} \cdot \norm{ {\bf m}}{\pfs} \enspace,
\end{align}
for $\pfs$ and $\delta$ indicated in \Cref{thrm:GoodPotentialHC}. 

In order to prove \eqref{prop:Target4PotContractsA}, we let $\dpotF: \mathbb{R}_{>0}\to \mathbb{R}$ be such that
{\small 
\begin{align}\label{def:PotentialFunction}
 y& \mapsto \frac12\sqrt{\frac{1}{y(1+y)}} \enspace.
\end{align}
}

\begin{claim}\label{claim:OnePotVsOtherPot}
For any ${\bf m}=({\bf m}_1, \ldots, {\bf m}_d)\in \mathbb{R}^{d}_{ \geq 0}$ and 
${\bf \kz}=({\bf \kz}_1, \ldots, {\bf \kz}_d)\in \mathbb{R}^{d}_{ > 0}$,
we have 
{
\begin{align}\label{eq:SubCondVsPotentials}
\cJ_{d}({\bf m}, {\bf \kz})&= 
\dpotF (\trecur_{d}({\bf \kz})) \times \sum\nolimits_{i\in [d]} \frac{{\bf m}_i }{\dpotF({\bf \kz}_i)} \left . 
\left| \frac{\partial }{\partial {\bf \kx}_i}\trecur_{d}({\bf \kx}) \right|_{{\bf \kx} ={\bf \kz}}\right| \enspace, 
\end{align}
}
where $\trecur_{d}$ and $\dpotF$ are defined
in \eqref{eq:HC-BPRecursion} and \eqref{def:PotentialFunction}, respectively. 
\end{claim}

\begin{proof}
The claim follows by using simple rearrangements. We have  
\begin{align}
\cJ_{d}({\bf m}, {\bf \kz}) &= 
 \xdpotF\left(\log \trecur_{d}({\bf \kz}) \right) \sum\nolimits_{j\in [d]}
\frac{ \abs{ \dlogtrecur(\log {\bf \kz}_j )}}{ \xdpotF\left( \log {\bf \kz}_j \right)} 
\times {\bf m}_j \nonumber\\
&= \sqrt{\frac{\trecur_{d}({\bf \kz})}{1+\trecur_{d}(\bf \kz)}}
\sum^{d}_{j=1} \sqrt{\frac{{\bf \kz}_j}{1+{\bf \kz}_j}}\times {\bf m}_j \label{eq:AppOfFX}\\
&= \sqrt{\frac{1}{\trecur_{d}({\bf \kz})(1+\trecur_{d}(\bf \kz))}}
\sum^{d}_{j=1} \sqrt{{\bf \kz}_j(1+{\bf \kz}_j) }\times\frac{ \trecur_{d}(\bf \kz)}{1+{\bf \kz}_j} \times {\bf m}_j \enspace. \nonumber
\end{align}
In \eqref{eq:AppOfFX}, we substitute $\xdpotF$ and $\dlogtrecur$ according to 
\eqref{def:XPotentialFunction} and \eqref{def:GradientRecHC}, respectively.
Using the definition of $\dpotF $ from \eqref{def:PotentialFunction}, we get that 
\begin{align}
\cJ_{d}({\bf m}, {\bf \kz}) &= \dpotF(\trecur_{d}({\bf \kz})) 
\sum\nolimits_{j\in [d]}\frac{1}{ \dpotF({\bf \kz}_j)} \times \frac{ \trecur_{d}(\bf \kz)}{1+{\bf \kz}_i} \times {\bf m}_j \enspace. \nonumber
\end{align}
The above implies \eqref{eq:SubCondVsPotentials}, i.e.,  since 
$\left| \frac{\partial }{\partial {\bf \kx}_i}\trecur_{d}({\bf \kx}) \right|=\frac{\trecur_{d}({\bf \kx})}{1+{\bf \kx}_i}$, 
for any $i\in [d]$. 
The claim follows.
\end{proof}

In light of \Cref{claim:OnePotVsOtherPot}, \eqref{prop:Target4PotContractsA} follows by showing that
for any ${\bf m}=({\bf m}_1, \ldots, {\bf m}_d) \in \mathbb{R}^{d}_{ \geq 0}$ and 
${\bf \kz}=({\bf \kz}_1, \ldots, {\bf \kz}_d)\in \mathbb{R}^{d}_{>0}$ we have 
\begin{align}\label{prop:Target4PotContractsB}
\dpotF (\trecur_{d}({\bf \kz})) \times \sum\nolimits_{i\in [d]} \frac{{\bf m}_i }{\dpotF({\bf \kz}_i)} \left . 
\left| \frac{\partial }{\partial {\bf \kx}_i}\trecur_{d}({\bf \kx}) \right|_{{\bf \kx} ={\bf \kz}}\right|
& \leq \delta^{\frac{1}{s}} \cdot \norm{ {\bf m}}{\pfs} \enspace. 
\end{align}

\noindent 
From our choice of the potential function, the above follows by using standard results form \cite{ConnectiveConst}. 
For any $s\geq 1, d>0$ and $x\geq 0$, we let the function
\begin{align} \nonumber 
\HCFactor(s, d, x)=d^{-1}\cdot \left( \frac{\dpotF(\trecur_{d,{\rm sym}}(x))}{\dpotF(x)} \trecur'_{d,{\rm sym}}(x) \right)^s \enspace,
\end{align}
where the functions $\trecur_{d,{\rm sym}}$, $\dpotF$ are defined in \eqref{def:SymTreeHC} and \eqref{def:PotentialFunction}, 
respectively. 

\begin{lemma}[\cite{ConnectiveConst}]\label{lemma:InsteadOfHolder}
For any $\lambda>0$, for integer $d\geq 1$, for $s\geq 1$, for ${\bf \kx}=({\bf \kx}_1, \ldots, {\bf \kx}_d) \in \mathbb{R}^d_{>0}$ 
and ${\bf m}=({\bf m}_1, \ldots, {\bf m}_d)\in \mathbb{R}^d_{\geq 0}$, the following holds: 
there exists $\bar{\kx}>0$ and integer $0\leq \kk\leq d$ such that 
\begin{align}\nonumber
\dpotF (\trecur_d({\bf \kx})) \times \sum\nolimits_{\ki\in [d]} \frac{{\bf m}_i }{\dpotF({\bf \kx}_i)} 
\left| \left . \frac{\partial }{\partial \kz_i}\trecur_d({\bf \kz}) \right|_{{\bf \kz} ={\bf \kx}} \right|
& \leq \left( \HCFactor(s, k, \bar{x}) \right)^{1/s} \times \norm{ {\bf m}}{s} \enspace.
\end{align}
\end{lemma}
In light of the above lemma, our proposition follows as a corollary from the following result.

\begin{lemma}[\cite{ConnectiveConst}] \label{lemma:XiBoundAtFixPoint}
For $\lambda>0$, consider $\dcritical=\dcritical(\lambda)$ and $\trecur_{\dcritical, {\rm sym}}$ with fugacity $\lambda$.
Let $\pfs\geq 1$ be such that
{\small 
\begin{align}\nonumber
\pfs^{-1}&= 1-{\textstyle \frac{\dcritical-1}{2}}\cdot \log\left(1+{\textstyle \frac{1}{\dcritical-1}} \right) \enspace.
\end{align}
}
For any $x> 0$, $d> 0$, we have that $\HCFactor(\pfs,d,x) \leq (\dcritical)^{-1}$.
\end{lemma}

By combining \Cref{lemma:InsteadOfHolder,lemma:XiBoundAtFixPoint} we get 
\eqref{prop:Target4PotContractsB} for $\pfs$ as indicated in \Cref{lemma:XiBoundAtFixPoint} and  $\delta=(\dcritical)^{-1}$. This concludes the proof of \Cref{prop:PotContracts}. 
\hfill $\square$

\spreadpoint

\newcommand{\blambda}{\mathbold{\lambda}}

\spreadpoint
\section{Proof of \Cref{thrm:HarndessTransHCNBM,thrm:HarndessTransIsingNBM}}\label{sec:thrm:HarndessTrans}

The proofs of \Cref{thrm:HarndessTransHCNBM,thrm:HarndessTransIsingNBM} are almost identical. 
We only show how to prove \Cref{thrm:HarndessTransHCNBM} and we omit the proof of \Cref{thrm:HarndessTransIsingNBM}. 

For a graph $G=(V,E)$ and $\lambda>0$, consider the Hard-core model with fugacity $\lambda$.
We define $\PartFun(G,\lambda)$, the {\em partition function} of this distribution by
\begin{align}
\PartFun(G,\lambda)&=\sum\nolimits_{\sigma}\lambda^{|\sigma|}\enspace,
\end{align}
where $\sigma$ varies over all the independent sets of $G$. 
For the proof of \Cref{thrm:HarndessTransHCNBM} we use the following result. 

\begin{theorem}[\cite{GSV16Hardness,SS14}]\label{thrm:HarndessCounting}
Unless ${\rm NP} = {\rm RP}$, for the Hard-core model, for all $\maxDeg\geq 3$ for all 
$\lambda>\lcritical(\maxDeg-1)$ 
there does not exist an FPRAS for the partition function with fugacity $\lambda$ for graphs of 
maximum degree at most $\maxDeg$.
\end{theorem}

\Cref{thrm:HC4SPRadiusHash} yields (via \cite{JVV86,SVV09}) an FPRAS for estimating 
the partition function $\PartFun(G,\lambda)$ with running time $O^*(n^2)$, 
for any $\lambda\leq (1-\varepsilon)\lcritical(\SingBound)$. Recall that $\SingBound=\nnorm{ (\NBMatrix)^N}{1/N}{2}$.
Also, recall that $O^*()$ hides multiplicative $\log n$-factors.


The negative part, i.e., the non-existence of FPRAS for $\lambda\geq (1+\varepsilon)\lcritical(\lceil \SingBound\rceil)$, follows 
from \Cref{thrm:HarndessCounting}. For graph $G$, such that $\nnorm{ (\NBMatrix)^N}{1/N}{2}=\SingBound$, it is standard that 
$\maxDeg\geq \lceil \SingBound \rceil +1$. Then, 
it suffices to apply \Cref{thrm:HarndessCounting} for graphs $G$ with $\nnorm{ (\NBMatrix)^N}{1/N}{2}=\SingBound$ and
  maximum degree  at most $\maxDeg=\lceil \SingBound \rceil +1$.
 
 The above concludes the proof of \Cref{thrm:HarndessTransHCNBM}.

\spreadpoint

\newcommand{\EMTw}{\byzantine{T}}
\newcommand{\EMTPw}{\deeppink{T_{\kP}}}

\section{Proof of \Cref{thrm:InflVsExtInfl} \LastReviewG{2025-03-05}}

For the material of this section, the reader needs to be familiar  with the results in \Cref{sec:BasicPropExt}.

\subsection{Proof of \Cref{thrm:InflVsExtInfl} }\label{sec:thrm:InflVsExtInfl}
For $G=(V,E)$, the Gibbs distribution $\mu^{\Lambda,\tau}_G$ and integer $\kk\geq 1$, 
we introduce the $\ExtV_{\Lambda,\kk}\times (V\setminus \Lambda)$ matrix
 $\SemExtdInfMatrix^{\Lambda,\tau}_{G,\kk}$ such that, for any $\kP \in \ExtV_{\Lambda, \kk}$ 
 and $\ku\in V\setminus \Lambda$, the entry $\SemExtdInfMatrix^{\Lambda, \tau}_G(\kP, \ku) $ is 
 defined as follows:

Letting $\kw$ be the starting vertex of walk $\kP$, we have 
\begin{align}\label{def:EdgeInflMatrixA}
\SemExtdInfMatrix^{\Lambda, \tau}_{G,\kk}(\kP, \ku) &=0 \enspace, 
\end{align}
if vertices $\ku$ and $\kw$ are at distance $<2\kk$.

Consider $\ku$ and $\kw$ which are at distance $\geq 2\kk$. Let $\upzeta^{\kP}$ be the $\kP$-extension of $\mu^{\Lambda,\tau}_G$.
We have
\begin{align}\label{def:EdgeInflMatrixB}
\SemExtdInfMatrix^{\Lambda, \tau}_{G,\kk}(\kP, \ku) &=
\upzeta^{\kP}_{\ku}(+1\ |\ (\kw,+1)) - \upzeta^{\kP}_{\ku}(+1 \ |\ (\kw,-1)) \enspace, 
\end{align}
where $\upzeta^{\kP}_{\ku}$ is the marginal of $\upzeta^{\kP}$ at vertex $\ku$.

Letting $\infmatrix^{\kP}$ be the influence matrix  induced by $\upzeta^{\kP}$, \eqref{def:EdgeInflMatrixB}
implies 
\begin{align}\label{eq:SemExtIndVsInfM}
\SemExtdInfMatrix^{\Lambda, \tau}_{G,\kk}(\kP, \ku) &=\infmatrix^{\kP}(\kw,\ku)\enspace.
\end{align}
In light of all the above, \Cref{thrm:InflVsExtInfl} follows as a corollary from the following two results.

\begin{lemma}\label{lemma:FromCalI2CalF}
There is an $\ExtV_{\Lambda,\kk}\times (V\setminus \Lambda)$ matrix $\SymWeightMA$ such that 
\begin{align}\label{eq:lemma:FromCalI2CalF}
\infmatrix^{\Lambda,\tau}_G = \VToEdge_{\kk} \cdot\left(\SymWeightMA \circ \SemExtdInfMatrix^{\Lambda,\tau}_{G,\kk} \right)+
 \infmatrix^{\Lambda,\tau}_{G, <2\kk} \enspace.
\end{align} 
Furthermore, there is a fixed number $C_A>0$, which is independent of $\Lambda$ and $\tau$, such that, for any
 $\kP\in \ExtV_{\Lambda,\kk}$ and any $\kv\in V\setminus \Lambda$, we have $|\SymWeightMA(\kP,\kv)|\leq C_A$. 
\end{lemma}

In \eqref{eq:lemma:FromCalI2CalF}, the matrices $ \infmatrix^{\Lambda,\tau}_{G, <2{\kk}}$ and $\VToEdge_{\kk}$ are defined in 
\eqref{eq:DefOfCISmallerK} and \eqref{def:Vertex2EdgeMatrixEntryBB}, respectively. 
Also, $\SymWeightMA \circ \SemExtdInfMatrix^{\Lambda,\tau}_{G,k}$ is the Hadamard product of the two matrices.
The proof of \Cref{lemma:FromCalI2CalF} appears in \Cref{sec:lemma:FromCalI2CalF}.

\begin{lemma}\label{lemma:FromCalF2CalH}
There exists an $\spset_{\Lambda, k}\times \spset_{\Lambda, k}$ matrix $\SymWeightMB$  such that 
\begin{align}\label{eq:lemma:FromCalF2CalH}
\SemExtdInfMatrix^{\Lambda,\tau}_{G,\kk} =\left( \ExtdInfMatrixF \circ \SymWeightMB \right)\cdot \EdgeToV_{\kk}
\enspace.
\end{align}
Furthermore, there is a fixed number $C_B>0$, which is independent of $\Lambda$ and $\tau$,  such that 
 for any $\kQ,\kP\in \ExtV_{\Lambda,\kk}$, we have $\abs{\SymWeightMB(\kP,\kQ)} \leq C_B$. 
\end{lemma}

In \eqref{eq:lemma:FromCalF2CalH}, matrix $\EdgeToV_k$ is defined in \eqref{def:Edge2VertexMatrixEntryBB}.
The proof of \Cref{lemma:FromCalF2CalH} appears in \Cref{sec:lemma:FromCalF2CalH}.

 \Cref{thrm:InflVsExtInfl} follows as a corollary from \Cref{lemma:FromCalI2CalF,lemma:FromCalF2CalH}. 
Specifically, 
$\SymWeightM$ is given by the Hadamard product $\SymWeightM=\SymWeightM_C \circ\SymWeightMB$,
where $\SymWeightMB$ is from \Cref{lemma:FromCalF2CalH} and $\SymWeightM_C$ is the $\ExtV_{\Lambda,k}\times \ExtV_{\Lambda,\kk}$ matrix such that 
for any $\kP,\kQ\in \ExtV_{\Lambda,\kk}$ we have $\SymWeightM_C(\kP,\kQ)=\SymWeightMA(\kP,\ku)$,
where $\ku$ is the starting vertex of $\kQ$. Matrix $\SymWeightMA$ is from  \Cref{lemma:FromCalI2CalF}.

We conclude that \Cref{thrm:InflVsExtInfl} is true for $\widehat{\kC}=C_A\cdot C_B$. 
\hfill $\square$

\subsection{Proof of \Cref{lemma:FromCalI2CalF}}\label{sec:lemma:FromCalI2CalF}
In this proof,  we use the properties of extensions from \Cref{lemma:TsawPVsTsaw,lemma:GibbsWeightsTpVsP}. 
Also, we abbreviate $\VToEdge_{\kk}$, $\SemExtdInfMatrix^{\Lambda,\tau}_{G,k}$,
 $\infmatrix^{\Lambda,\tau}_{G, <2\kk}$ and $\infmatrix^{\Lambda,\tau}_{G}$ to 
 $ \VToEdge$, $\SemExtdInfMatrix$, $\infmatrix_{<2\kk}$ and $\infmatrix$, respectively. 

Firstly, we note that the matrices at the two sides of the equation in \eqref{eq:lemma:FromCalI2CalF} are of the same dimension, 
i.e., they are $(V\setminus \Lambda)\times (V\setminus \Lambda)$. 
We show that their corresponding entries are equal to each other.

For any $\kv\in V\setminus \Lambda$, we have
\begin{align}\label{eq:FirstEqInlemma:FromCalI2CalF}
\infmatrix_{<2{\kk}}(\kv,\kv)&=1& \textrm{and}&& \left(\VToEdge\cdot \left(\SymWeightMA \circ\SemExtdInfMatrix \right)\right)(\kv,\kv)&=0 \enspace. 
\end{align}
The equality on the right is true because $\left(\VToEdge\cdot \left(\SymWeightMA \circ\SemExtdInfMatrix \right)\right)(\kv,\kv)=\sum_{\kP}\VToEdge(\kv,\kP)\cdot\SymWeightMA(\kP,\kv)\cdot \SemExtdInfMatrix(\kP,\kv)$
and the variable $\kP\in \ExtV_{\Lambda,\kk}$ runs over walks that emanate from $\kv$. For such walks $\SemExtdInfMatrix(\kP,\kv)=0$. Hence, 
for any matrix $\SymWeightMA$ the above is true. 

From the  definition of $\infmatrix$ and \eqref{eq:FirstEqInlemma:FromCalI2CalF} we conclude that $\infmatrix$ and 
$ \infmatrix_{<2\kk}+\VToEdge\cdot\left(\SymWeightMA \circ \SemExtdInfMatrix \right)$
have ones at their diagonal. 
Now, we focus on the off-diagonal elements. Let 
\begin{align}\label{eq:DefOfILarge2k}
\infmatrix_{\geq 2\kk} &=\infmatrix-\infmatrix_{<2\kk}\enspace.
\end{align}
It suffices to show that there exists an $\ExtV_{\Lambda,\kk}\times V\setminus \Lambda$ matrix $\SymWeightMA$, 
with its entries bounded as specified in the statement of \Cref{lemma:FromCalI2CalF},   
such that for any $\ku,\kw\in V\setminus \Lambda$, which are different from each other, 
we have 
\begin{align}\label{eq:Target4lemma:FromCalI2CalF}
\infmatrix_{\geq 2\kk}(\kw,\ku) &= \sum\nolimits_{\kP} \SymWeightMA(\kP,\ku) \cdot \SemExtdInfMatrix(\kP,\ku)
\enspace,
\end{align}
where $\kP$ varies in $\ExtV_{\Lambda, \kk}$, such that its first vertex is $\kw$.

For vertices $\kw,\ku\in V\setminus \Lambda$ at distance $< 2\kk$,  The definition of  $\SemExtdInfMatrix$ implies that 
for  $\kP$ emanating from $\kw$, we have $ \SemExtdInfMatrix(\kP,\ku)=0$. Also, we have $\infmatrix_{\geq 2\kk}(\kw,\ku)=0$.
We conclude that for pairs of vertices  $\kw,\ku\in V\setminus \Lambda$ at distance $< 2\kk$,
 \eqref{eq:Target4lemma:FromCalI2CalF} is true for any choice of matrix $\SymWeightMA$.

We consider  $\kw$ and $\ku$ at distance $\geq 2\kk$ from each other. Let $\EMTw=T_{\saw}(G, \kw)$.  Also, let $\{\infweight(\ke)\}$ 
be the collection of weights over the edges of $\EMTw$ we obtain  by applying the $\Tsaw$-construction on the Gibbs 
distribution $\mu^{\Lambda,\tau}_G$.  Then, \Cref{prop:Inf2TreeRedaux} implies 
\begin{align}\label{eq:IfVsWeightscFVsVi}
\infmatrix_{\geq 2\kk}(\kw,\ku) &=\sum\nolimits_{\kell\geq 2\kk}
\sum\nolimits_{\kW=(\ke_1, \ldots, \ke_{\kell}) \in \EMTw}
\prod\nolimits_{1 \leq \ki\leq \kell}\infweight(\ke_{\ki}) \enspace,
\end{align}
where $\kW$ varies over the length-$\kell$ paths in $\EMTw$ from the root of the tree to the copies of vertex $\ku$.

For each $\kP\in \ExtV_{\Lambda,k}$, which emanates from vertex $\kw$, consider $\EMTPw=\Tsaw(\gext{G}{\kP},\kw)$. 
Using \Cref{lemma:TsawPVsTsaw},   we  identify each $\EMTPw$ as a subtree of $\EMTw$, while note that both trees
have the same root. Recall also the weights $\{\infweight(\ke)\}$ on the edges of $\EMTw$ that are induced by the 
$\Tsaw$-construction on $\mu^{\Lambda,\tau}_G$. Since $\EMTPw$ is a subtree of $\EMTw$, we endow 
$\EMTPw$ with the edge-weights $\{\infweight(\ke)\}_{\ke\in \EMTPw}$. 

Each path $\kW$ in $\EMTw$ that emanates from the root and is of length $\geq 2\kk$ belongs to exactly one of the subtrees 
$\EMTPw$ of $\EMTw$ we define above. Using this observation, we write \eqref{eq:IfVsWeightscFVsVi} as
follows:
\begin{align}\label{eq:IfVsWeightscFVsViBB}
\infmatrix_{\geq 2\kk}(\kw,\ku)&=\sum\nolimits_{\kP} \sum\nolimits_{\kell \geq 2\kk}
\sum\nolimits_{\kW=(\ke_1,\ldots, \ke_{\kell}) \in \EMTPw} 
\prod\nolimits_{1\leq i\leq \kell }\infweight(\ke_i) \enspace,
\end{align}
where variable $\kP$ varies over the walks in $\ExtV_{\Lambda,\kk}$ that emanate from vertex $\kw$ and
variable $\kW$ varies over the length-$\kell$ paths in $\EMTPw=\Tsaw(\gext{G}{\kP},\kw)$ from the root of the 
tree to the copies of vertex $\ku$.

For each $\kP\in \ExtV_{\Lambda, \kk}$ whose first vertex is $\kw$, 
and $\EMTPw=\Tsaw(\gext{G}{\kP},\kw)$, we let $\{\infweight_{\kP}(\ke)\}_{\ke\in \EMTPw}$ be the 
collection of weights over the edges of $\EMTPw$ specified with respect to the $\kP$-extension of $\mu^{\Lambda,\tau}_G$. 

The definition of matrix $\SemExtdInfMatrix$ in \eqref{def:EdgeInflMatrix} and \Cref{prop:Inf2TreeRedaux} imply 
\begin{align}\label{eq:FVsITs}
\SemExtdInfMatrix(\kP, \ku)&=\sum\nolimits_{\kell \geq 2\kk}\sum\nolimits_{\kW=(\ke_1,\ldots, \ke_{\kell}) \in \EMTPw} 
\prod\nolimits_{1\leq \ki\leq \kell} \infweight_{\kP}(\ke_{\ki}) \enspace,
\end{align}
where $\kW$ varies over the length-$\kell$ paths from the root to the copies of $\ku$ in $\EMTPw$.

Letting $\kP=\kx_0,\ldots,\kx_k$, 
there is a path  $R_{\kP}=\ku_0,\ldots, \ku_{\kk}$ in tree $\Tw$, which emanates from the root  such that 
$\ku_{\ki}$ is a copy of $\kx_{\ki}$, for all $0\leq \ki \leq \kk $. 
We have that $R_{\kP}$ is in $ \EMTPw$, too. Furthermore, 
all  paths $\kW\in \EMTPw$ of length $\kell\geq 2\kk$ that emanate from the root of $\EMTPw$ overlap with $R_{\kP}$ in  their first $\kk$ edges. 

The above observation implies that fixing $\kell\geq 2\kk$ and path $\kW=(\ke_1,\ldots, \ke_{\kell})$, the corresponding summand in \eqref{eq:FVsITs} trivially satisfies
\begin{align}\label{eq:ProdWeigtPTPBreak}
\prod\nolimits_{1\leq \ki\leq \kell} \infweight_{\kP}(\ke_{\ki}) &= \UpK \times \prod\nolimits_{\kk< \ki\leq \kell } \infweight_{\kP}(\ke_{\ki}) \enspace,
\end{align}
where $\UpK=\prod\nolimits_{\ke \in R_{\kP}} \infweight_{\kP}(\ke)$. 

In a previous step, we endowed each $\EMTPw$ with the edge weights $\{\infweight(\ke)\}_{\ke\in \EMTPw}$ since it is a subtree of
$\EMTw$. Overall now, for each $\EMTPw$ we have two sets of edges weights i.e., we have $\{\infweight(\ke)\}_{\ke\in \EMTPw}$ and
$\{\infweight_{\kP}(\ke)\}_{\ke\in \EMTPw}$.
\Cref{lemma:GibbsWeightsTpVsP} implies that for any edge $\ke\in \EMTPw$, which is at distance $\geq k$ from
the root of $\EMTPw$, we have
\begin{align}\label{eq:IfVsWeightscFVsViCC}
\infweight_{\kP}(\ke)&=\infweight(\ke) \enspace. 
\end{align}
Combining \eqref{eq:IfVsWeightscFVsViCC} and \eqref{eq:ProdWeigtPTPBreak} we get the following:
fixing $\kell\geq 2\kk$ and path $\kW=(\ke_1,\ldots, \ke_{\kell})$, the corresponding summand in \eqref{eq:FVsITs} satisfies
\begin{align}\label{eq:ProdWeigtPTPBreakNew}
\prod\nolimits_{1\leq \ki\leq \kell} \infweight_{\kP}(\ke_{\ki}) &= \UpK \times \prod\nolimits_{\kk< \ki\leq \kell } \infweight(\ke_{\ki}) \enspace.
\end{align}
We define the $\ExtV_{\Lambda,\kk}\times V\setminus \Lambda$ matrix $\SymWeightMA$ 
such that for each $\kP\in \ExtV_{\Lambda,\kk}$ and $\ku\in V\setminus \Lambda$, we have
\begin{align}
\SymWeightMA(\kP,\ku) = \left \{
\begin{array}{lcl} 
1 &\quad &\textrm{if distance from $\ku$ to the starting vertex of $\kP$ is $<2\kk$}\\
\prod_{\ke\in R_{\kP}}\frac{\infweight(\ke)}{\infweight_{\kP}(\ke)+\Ind\{\infweight_{\kP}(\ke)=0\}} &&\textrm{otherwise}\enspace,
\end{array}
\right . \nonumber
\end{align}
where $R_{\kP}$ is the  path in $\EMTPw$ defined earlier. 
The above definition of $\SymWeightMA$ implies 
\begin{align}\label{eq:IfVsWeightscFVsViDD}
 \SymWeightMA(\kP,\ku)\cdot \SemExtdInfMatrix(\kP, \ku)&= \sum\nolimits_{\kell\geq 2\kk} \sum\nolimits_{\kW=(\ke_1,\ldots, \ke_{\kell})\in \EMTPw} 
\prod\nolimits_{1\leq \ki\leq \kell }\infweight(\ke_{\ki}) \enspace,
\end{align}
where $\kW$ varies over the paths in $\EMTPw$ from the root to the copies of vertex $\ku$ at distance $\kell$
from the root.

Verifying the above equality is a matter of simple calculations when $\infweight_{\kP}(\ke)\neq 0$ 
for all $\ke\in R_{\kP}$.  That is, it follows from \eqref{eq:FVsITs}, \eqref{eq:ProdWeigtPTPBreakNew} 
and the definition of $\SymWeightMA$. 
When there is at least one $\ke\in R_{\kP}$ such that $\infweight_{\kP}(\ke)=0$, then it is not hard to 
see that we also have $\infweight(\ke)=0$; hence both sides in \eqref{eq:IfVsWeightscFVsViDD} are 
equal to zero.  All the above imply that \eqref{eq:IfVsWeightscFVsViDD} is true.

Then, plugging \eqref{eq:IfVsWeightscFVsViDD} into \eqref{eq:IfVsWeightscFVsViBB}, we obtain 
\eqref{eq:Target4lemma:FromCalI2CalF}. Hence, we conclude that \eqref{eq:lemma:FromCalI2CalF} is 
true.

The assumption that $\mu^{\Lambda,\tau}_G$ is $b$-marginally bounded implies that 
there is $\upc \in (0,1)$, bounded away from zero such that we either have $\infweight_{\kP}(\ke)=0$, or 
$\upc\leq \abs{\infweight_{\kP}(\ke)} \leq 1$ and similarly for $\infweight(\ke)$. 
Hence, for any $\kP\in \ExtV_{\Lambda,k}$ and any 
vertex $\ku\in V\setminus \Lambda$, we have $| \SymWeightMA(\kP, \ku)| \leq \max\{1,\upc^{-\kk}\}$. 

We conclude that \Cref{lemma:FromCalI2CalF} is true for $C_A=\max\{1,\upc^{- \kk}\}$.
\hfill $\square$

\subsection{Proof of \Cref{lemma:FromCalF2CalH}} \label{sec:lemma:FromCalF2CalH}
In what follows, we abbreviate $\SemExtdInfMatrix^{\Lambda,\tau}_{G,\kk}$ 
and $\ExtdInfMatrixF$ to $\SemExtdInfMatrix$ and $\ExtdInfMatrix$, respectively.

Firstly, we note that the matrices at the two sides of the equation in \eqref{eq:lemma:FromCalF2CalH} are 
of the same dimension, i.e., they are $\ExtV_{\Lambda, \kk}\times (V\setminus \Lambda)$.

\Cref{lemma:FromCalF2CalH} follows by showing that there is an $\ExtV_{\Lambda, \kk}\times \ExtV_{\Lambda,\kk}$ matrix 
$\SymWeightMB$, with entries bounded as specified in the statement of \Cref{lemma:FromCalF2CalH}, 
such that for any $\ku\in V\setminus \Lambda$ and $\kP\in \ExtV_{\Lambda, \kk}$, we have
\begin{align}\label{eq:Target4lemma:FromCalF2CalH}
\SemExtdInfMatrix(\kP, \ku) &= \sum\nolimits_{\kQ} 
\SymWeightMB(\kP, \kQ)\cdot \ExtdInfMatrix(\kP, \kQ)
\enspace,
\end{align}
where $\kQ$ varies over all walks in $ \ExtV_{\Lambda,\kk}$ that emanate from vertex $\ku$.

Consider vertices $\kw, \ku\in V\setminus \Lambda$ at distance $\geq 2k$
and $\kP\in \ExtV_{\Lambda,\kk}$ which emanates from $\kw$. We start by showing that 
for such $\kP$ and $\kw, \ku$, \eqref{eq:Target4lemma:FromCalF2CalH} is true.

Let $\upzeta^{\kP}$ be the $\kP$-extension of $\mu^{\Lambda,\tau}_{G}$. Let $\infmatrix^{\kP}$ 
be the influence matrix induced by $\upzeta^{\kP}$. 
From \eqref{eq:SemExtIndVsInfM}, we have 
\begin{align}\label{eq:CFReverse}
\SemExtdInfMatrix(\kP, \ku) &=\infmatrix^{\kP}(\kw,\ku) =\cW \cdot \infmatrix^{\kP}(\ku,\kw) \enspace, 
\end{align}
where $\cW=\cW(\kP, \ku)=\frac {\upzeta^{\kP}_{\ku}(+1)\cdot \upzeta^{\kP}_{\ku}(-1 )} {\upzeta^{\kP}_{\kw}(+1)\cdot \upzeta^{\kP}_{\kw}(-1)+\Ind\{\upzeta^{\kP}_{\kw}(+1)\in \{0,1\}\}}$.

To verify the second equality in \eqref{eq:CFReverse} note the following:
when $\upzeta^{\kP}_{\ku}(+1), \upzeta^{\kP}_{\kw}(+1)\notin \{0,1\}$, the second equality in \eqref{eq:CFReverse} 
is true due to \Cref{claim:InfSymmetrisation}.
When at least one of $\upzeta^{\kP}_{\ku}(1), \upzeta^{\kP}_{\kw}(1)$ is in $\{0,1\}$, 
the definitions of $\SemExtdInfMatrix$ and the influence matrix $\infmatrix^{\kP}$
imply 
$\SemExtdInfMatrix(\kP, \ku)=0$ and $\infmatrix^{\kP}(\ku,\kw)=0$. 
Hence, the second equality in \eqref{eq:CFReverse} is  true.

The $b$-marginal boundedness assumption implies 
\begin{align}\label{eq:Bound4CN1}
0\leq \cW(\kP,\ku) & \leq b^{-2} \enspace. 
\end{align}

\noindent
For $\kQ\in \ExtV_{\Lambda,\kk}$ that emanates from $\ku$, let $\upzeta^{\kP,\kQ}$ be the $\{\kP,\kQ\}$-extension of 
$\mu^{\Lambda,\tau}_{G}$. Note that $\kP$ and $\kQ$ are compatible because we assume that $\ku$ and $\kw$ 
are at distance $\geq 2\kk$; hence, $\upzeta^{\kP,\kQ}$ is well-defined. 

In what follows, denote $\infmatrix^{\kP,\kQ}$ the influence matrix induced by $\upzeta^{\kP,\kQ}$.
%
%
\begin{claim}\label{claim:CFReverseSplit}
There is an $\ExtV_{\Lambda,\kk}\times (V\setminus \Lambda)$ matrix $\widehat{\SymWeightM}_A$ such that 
for $\kw, \ku\in V\setminus \Lambda$ with $\dist_{G}(\kw,\ku)\geq 2\kk$ and any $\kP\in \ExtV_{\Lambda,\kk}$ that emanates from $\kw$, 
we have
\begin{align}\label{eq:CFReverseSplit}
\infmatrix^{\kP}(\ku,\kw) &=\sum\nolimits_{\kQ} 
\widehat{\SymWeightM}_A(\kQ, \kw) \cdot \infmatrix^{\kP,\kQ} (\ku,\kw)
\enspace,
\end{align}
where $\kQ$ varies over the walks in $\ExtV_{\Lambda,\kk}$, which emanate from $\ku$. 
Furthermore, we have
 \begin{align}\label{eq:CFReverseSplitBB}
\max\nolimits_{\kQ,\kz} \{ |\widehat{\SymWeightM}_A(\kQ,\kz)| \} & \leq C_A\enspace,
\end{align}
 where $C_A$ is specified in 
the statement of \Cref{lemma:FromCalI2CalF}.
\end{claim}

The definition of matrix $\ExtdInfMatrix$ and in particular \eqref{def:EdgeInflMatrixEquivalent}, imply\footnote{The definition of the $\{P,Q\}$-extension of 
$\mu^{\Lambda,\tau}_G$ imply that $\infmatrix^{\kQ,\kP}=\infmatrix^{\kP,\kQ}$.} that, 
for any $\kP, \kQ\in \ExtV_{\Lambda,\kk}$ that emanate from the vertices $\kw$ and $\ku$, respectively, we have
\begin{align}\label{eq:CalEllAsInfluence}
\infmatrix^{\kP,\kQ} (\ku, \kw)= \ExtdInfMatrix(\kQ,\kP)\enspace.
\end{align}
Combining \eqref{eq:CalEllAsInfluence} and \eqref{eq:CFReverseSplit}, we obtain 
\begin{align}\label{eq:CFReverseSplitAA}
\infmatrix^{\kP}(\ku,\kw) &=
\sum\nolimits_{\kQ} \widehat{\SymWeightM}_A(\kQ, \kw) \cdot \ExtdInfMatrix(\kQ,\kP)\enspace,
\end{align}
where recall that $\kQ$ varies over the walks in $\ExtV_{\Lambda,\kk}$ which emanate from $\ku$.

Working as in \eqref{eq:CFReverse}, we use \eqref{eq:CalEllAsInfluence} to obtain that 
for any $\kQ\in \ExtV_{\Lambda,\kk}$ that emanates from $\ku$, we have
\begin{align}\label{eq:CHReverse} 
\ExtdInfMatrix(\kQ, \kP) &=
\widehat{\cW} \cdot \ExtdInfMatrix(\kP,\kQ)\enspace, 
\end{align}
where $\widehat{\cW}=\widehat{W}(\kP,\kQ)=\frac{\upzeta^{\kP,\kQ}_{\kw}(+1)\cdot \upzeta^{\kP,\kQ}_{\kw}(-1)} 
{\upzeta^{\kP,\kQ}_{\ku}(+1)\cdot \upzeta^{\kP,\kQ}_{\ku}(-1 )+\Ind\{\upzeta^{\kP,\kQ}_{\ku}(+1)\in \{0,1\}\}}$.

The $b$-marginal boundedness assumption implies 
\begin{align}\label{eq:CalN2Bound}
0\leq \widehat{\cW}(\kP,\kQ) \leq b^{-2}\enspace.
\end{align}

We get \eqref{eq:Target4lemma:FromCalF2CalH} by 
plugging \eqref{eq:CFReverseSplitAA} and \eqref{eq:CHReverse} into \eqref{eq:CFReverse} while setting 
\begin{align}\label{eq:CLVsCIPQ}
\SymWeightMB(\kP, \kQ) &=\cW(\kP,\ku) \cdot \widehat{\SymWeightM}_A(\kQ,\kw) \cdot \widehat{\cW}(\kP, \kQ)\enspace.
\end{align}
Recall that $\ku$ is the first vertex of path $\kQ$. 
Plugging \eqref{eq:Bound4CN1}, \eqref{eq:CFReverseSplitBB} and \eqref{eq:CalN2Bound}
into \eqref{eq:CLVsCIPQ} we obtain that $|\SymWeightMB(\kP, \kQ)| \leq b^{-4}\cdot C_A$.

We consider \eqref{eq:Target4lemma:FromCalF2CalH} for $\ku,\kw\in V\setminus \Lambda$ which are at distance $< 2\kk$
and $\kP\in \ExtV_{\Lambda,k}$ that emanates from $\kw$.
For such $\ku,\kw$ and $\kP$, we have
\begin{align}\label{eq:LNCVsFNonCompatiblePQ}
\left(\left( \ExtdInfMatrix \circ \SymWeightMB \right)\cdot \EdgeToV_{\kk}\right)(\kP,\ku)&=0 &
\textrm{and} &&\SemExtdInfMatrix(\kP,\ku)&=0\enspace. 
\end{align}
The equality on the left is true because $\left(\left( \ExtdInfMatrix \circ \SymWeightMB \right)\cdot \EdgeToV_{\kk}\right)(\kP,\ku)=
\sum_{\kQ} \ExtdInfMatrix(\kP,\kQ) \cdot \SymWeightMB(\kP,\kQ) \cdot \EdgeToV_{\kk}(\kQ,\ku)$
and the variable $\kQ\in \ExtV_{\Lambda,k}$ varies over the walks that emanate from $\ku$. In this case, we have 
$\ExtdInfMatrix(\kP,\kQ)=0$ since  $\kP$ and $\kQ$ are not compatible. The equality on the right follows from the 
definition of matrix $\SemExtdInfMatrix$.  Hence, for any choice of matrix $\SymWeightMB$, the above is true. 

For matrix $\SymWeightMB$ to be well defined, for $\ku,\kw$ at distance $<2\kk$ and any $\kP, \kQ\in \ExtV_{\Lambda,k}$ 
that emanate from $\kw$ and $\ku$, respectively, 
we let $\SymWeightMB(\kP,\kQ)=1$.
This choice of $\SymWeightMB$ and \eqref{eq:LNCVsFNonCompatiblePQ} imply that \eqref{eq:Target4lemma:FromCalF2CalH} is also true
for any $\ku,\kw\in V\setminus \Lambda$ which are at distance $< 2\kk$ and $\kP$ that emanates from $\kw$.

Hence  \eqref{eq:Target4lemma:FromCalF2CalH} is true for any $\ku,\kw\in V\setminus \Lambda$. 
We conclude that  \Cref{lemma:FromCalF2CalH} is true for $C_B=\max\{1, b^{-4}\cdot C_A\}$.
\hfill $\square$

\begin{proof}[Proof of \Cref{claim:CFReverseSplit}]
The claim follows by utilising the equality in \eqref{eq:Target4lemma:FromCalI2CalF} from the proof of \Cref{lemma:FromCalI2CalF}.
We apply this equality with respect to the influence matrices $\infmatrix^{\kP}$, $\infmatrix^{\kP,\kQ}$, i.e., instead of matrices $\infmatrix^{\Lambda,\tau}_G$ 
and $\infmatrix^{\kP}$ that are considered there. 
More specifically, we have 
\begin{align}\label{eq:Obs4lemma:FromCalF2CalH}
\infmatrix^{\kP}(\ku,\kw) &=\sum\nolimits_{\kQ} \widehat{\SymWeightM}_A(\kQ, \kw) \cdot \infmatrix^{\kP,\kQ}(\ku,\kw) \enspace,
\end{align}
where $\kQ$ varies over the walks in $\ExtV_{\Lambda,k}$ that emanate from vertex $\ku$ and matrix
 $\widehat{\SymWeightM}_A$ is defined as in \Cref{lemma:FromCalI2CalF}, with respect to the Gibbs distributions
 $\upzeta^{\kP}$ and $\upzeta^{\kP,\kQ}$.

Applying the same arguments as in the proof of \Cref{lemma:FromCalI2CalF}, 
we obtain  $\max_{\kQ,\kz}\{ |\widehat{\SymWeightM}_A(\kQ,\kz)\}| \leq C_A$, where $C_A$ is specified in 
the statement of \Cref{lemma:FromCalI2CalF}. 
Note that the bound $C_A$  depends only on the parameters $\kk$ and $b$ for the
marginal boundedness. 
\Cref{claim:CFReverseSplit} follows. 
\end{proof}

\spreadpoint

\newcommand{\sawpaths}{\magenta{\Omega}}
\newcommand{\ABijection}{\cadmiumorange{\uph}}

\section{Remaining Proofs \LastReviewG{2025-03-11}}

\subsection{Proofs of \Cref{lemma:TsawPVsTsaw,lemma:GibbsWeightsTpVsP}}
\label{sec:lemma:TsawPVsTsaw}\label{sec:lemma:GibbsWeightsTpVsP}
We introduce a few useful notions  for the proofs of \Cref{lemma:TsawPVsTsaw,lemma:GibbsWeightsTpVsP}.
Recall   walk $\kP\in \ExtV_{\Lambda,\kk}$ that emanates from vertex $\kw$.

We let set $\sawpaths$ consist of all walks $\kW$ in $G$   that correspond to the vertices in  $\STw_{\kP}$. Similarly, we define set 
$\sawpaths_{\kP}$ that consists of all walks that correspond to the vertices in $\XTPw$.

Since each of $\STw_{\kP}$ and $\XTPw$ is a (part of)  tree of self-avoiding walks starting from vertex $\kw$, 
the $\Tsaw$-construction implies that each walk $\kW=\kv_0, \ldots, \kv_{\kr}$ in $\sawpaths$ or $\sawpaths_{\kP}$ 
is such that $\kv_0=\kw$, while it is one  of the following two types:
\begin{description}
\item[Type I] $\kv_0, \ldots, \kv_{\kr}$ is a self-avoiding walk, \label{TsawTypeAWalk}
\item[Type II] $\kv_0, \ldots, \kv_{\kr-1}$ is a self-avoiding walk, while there is $j\leq \kr-3$ such that $\kv_{\kr}=\kv_{j}$. \label{TsawTypeBWalk}
\end{description}
 
Let set $\sawpaths^{(a)}\subseteq \sawpaths$ consist of all the elements in $\sawpaths$ which are 
self-avoiding walks,  i.e., they are of Type I.  Let set $\sawpaths^{(b)}\subseteq \sawpaths$ consist of  all 
walks $\kW=\kv_0, \ldots, \kv_{\kr} \in \sawpaths$, which  are of Type II such that $\kv_{\kr}\in \kP$.  Finally, 
let set $\sawpaths^{(c)}\subseteq \sawpaths$ consist of all walks $\kW=\kv_0, \ldots, \kv_{\kr} \in \sawpaths$ 
which  is of Type II, while  $\kv_{\kr}\notin \kP$. It is immediate that $\sawpaths^{(a)}$, $\sawpaths^{(b)}$ 
and $\sawpaths^{(c)}$ partition  set $\sawpaths$.

Similarly, let set $\sawpaths^{(a)}_{\kP} \subseteq \sawpaths_{\kP}$ consist of all the elements in 
$\sawpaths_{\kP}$,  which are of Type I,  and do not use any of the split-vertices in $\ssplit_{\kP}$. Let 
set  $\sawpaths^{(b)}_{\kP}\subseteq \sawpaths_{\kP}$ consist of all the elements in $\sawpaths_{\kP}$, 
which  are of Type I, such that the last vertex is in $\ssplit_{\kP}$. Finally, let set 
$\sawpaths^{(c)}_{\kP}\subseteq \sawpaths_{\kP}$ consist of all walks $\kW$  of Type II.
Recall that the split-vertices in $\ssplit_{\kP}$ are of degree one; hence,  there can be no path  
$\kW=\kv_0, \ldots, \kv_{\kr} \in  \sawpaths^{(c)}_{\kP}$  such that $\kv_{\kr}\in \ssplit_{\kP}$.
It is immediate  that  $\sawpaths^{(a)}_{\kP}$, $\sawpaths^{(b)}_{\kP}$ and $\sawpaths^{(c)}_{\kP}$ 
partition set $\sawpaths_{\kP}$.

\begin{proof}[Proof of \Cref{lemma:TsawPVsTsaw}]
Our result follows by showing that there is an appropriate bijection $\ABijection$ from  the vertex set of 
$\STw_{\kP}$ to that  of $\XTPw$ that allows us to identify $\STw_{\kP}$ with $\XTPw$.

Since the vertices in $\STw_{\kP}$ correspond to the walks in $\sawpaths$ and the vertices in $\XTPw$ 
correspond to the  walks in $\sawpaths_{\kP}$, we equivalently establish the bijection $\ABijection$ from 
set $\sawpaths$ to $\sawpaths_{\kP}$.

For concreteness, let $\kP=\kx_0, \ldots, \kx_{\kk}$, while  $\kx_0=\kw$. Let set $\cO$ consist of each 
edge $\ke$ in $G$, which does not connect two consecutive vertices in $\kP$ but has at least one end in the set 
$\{\kx_{0}, \ldots, \kx_{\kk-1}\}$.  Similarly, let set  $\cO_{\kP}$ consist of each  edge in $\gext{G}{\kP}$ 
which is between a split-vertex in set $\ssplit_{\kP}$ and another  vertex of $\gext{G}{\kP}$.

The only edges in $G$ that do not appear in $\gext{G}{\kP}$ are exactly those in 
set $\cO$. Similarly, the only edges in $\gext{G}{\kP}$ that do not appear in $G$ are exactly those in set 
$\cO_{\kP}$. This implies that  the elements in $\sawpaths \cap \sawpaths_{\kP}$ correspond  to the walks 
that do not use edges from sets $\cO$ and $\cO_{\kP}$.

The definitions of the two trees $\STw_{\kP}$ and $\XTPw$ imply that  
\begin{align}\nonumber
\sawpaths^{(a)}&=\sawpaths^{(a)}_{\kP} & \textrm{and} &&\sawpaths^{(c)}&=\sawpaths^{(c)}_{\kP} \enspace.
\end{align}
To see this, recall that we have $\kP=\kx_0, \ldots, \kx_{\kk}$. Then each  $\kW=\kv_{0}, \ldots, \kv_{\kr} \in \sawpaths^{(a)}$ 
corresponds to a self-avoiding walk  that emanates from vertex $\kw$, while  $\kv_{\ki}=\kx_{\ki}$, for $\ki=0,\ldots, \min\{\kk,\kr \}$.
It is immediate to check that each such walk $\kW$ also belongs to $\sawpaths^{(a)}_{\kP}$ and vice versa. This implies that
$\sawpaths^{(a)}=\sawpaths^{(a)}_{\kP}$.   With a very similar argument,  we also get  $\sawpaths^{(c)}=\sawpaths^{(c)}_{\kP}$.

It is only the walks in $\sawpaths^{(b)}$ and $\sawpaths^{(b)}_{\kP}$ that use edges from $\cO$ and 
$\cO_{\kP}$,  respectively.  For each walk $\kW=\kv_0, \ldots, \kv_{\kr} \in \sawpaths^{(b)}$, there is a 
walk $\kW_{\kP}=\kz_0, \ldots, \kz_{\kr} \in \sawpaths^{(b)}_{\kP}$  such that $\kv_{\ki}=\kz_{\ki}$ for $\ki=0,\ldots, \kr-1$, 
while $z_{\kr}$ is a split-vertex that is generated by vertex $v_{\kr}$ and  vice versa.  This holds since 
the definition of  $\gext{G}{\kP}$ implies that when constructing $\gext{G}{\kP}$, we substitute  the edge  
$\{\kv_{\kr-1}, \kv_{\kr}\}\in\cO$ with the edge $\{\kz_{\kr-1}, \kz_{\kr}\}\in \cO_{\kP}$.

Considering all the above, we specify $\ABijection$ that
 \begin{align}\label{def:OfUPHBijection}
\ABijection(\kW) &= \left\{
 \begin{array}{lcl}
 \kW && \textrm{if } \kW\in \sawpaths^{(\alpha)}\cup \sawpaths^{(c)}\enspace, \\
 \kW_{\kP} && \textrm{if } \kW \in \sawpaths^{(b)} \enspace, \\
 \end{array}
 \right .
 \end{align}
where $\kW_{\kP}$ is as described in the  paragraph above \eqref{def:OfUPHBijection}.

It is easy to check that  for each $\kW_{\kP}\in \sawpaths^{(b)}_{\kP}$ there is a unique $\kW \in \sawpaths^{(b)}$
such that  $\ABijection(\kW)=\kW_{\kP}$.  This observation and \eqref{def:OfUPHBijection} imply 
that  $\ABijection$ is a bijection from $\sawpaths$ to $\sawpaths_{\kP}$. Using the correspondence 
between sets $\sawpaths$, $\sawpaths_{\kP}$ and the sets of vertices of trees  $\STw_{\kP}$ and $\XTPw$, 
respectively, we identify $\ABijection$  as a bijection from the set of vertices in $\STw_{\kP}$ 
to those in $\XTPw$.

For vertex $\ku$ in $\STw_{\kP}$ that corresponds to walk $\kW\in \sawpaths$,   we  have that  $\ABijection(\ku)$ is 
a vertex in   $\XTPw$ that  corresponds to  walk  $\ABijection(\kW)\in \sawpaths_{\kP}$. From \eqref{def:OfUPHBijection}, 
it is easy to verify that if $\ku$ is a copy of $\kv$ in $G$,  then $\ABijection(\ku)$ must be a copy of $\kv$ or a split-vertex in 
$\gext{G}{\kP}$ generated by $\kv$.

We now show that $\ABijection$ preserves adjacencies. That is, for two vertices $\kx$, $\kz$ in $\STw_{\kP}$
which are adjacent, we have that $\ABijection(\kx)$ and $\ABijection(\kz)$ are two  adjacent vertices in $\XTPw$. 
Suppose  $\kx$ and $\kz$ correspond to walks $\kW, \overline{\kW}\in \sawpaths$, respectively. W.l.o.g., 
suppose  walk $ \kW$ extends walk $\overline{\kW}$ by  one vertex. Then, the vertices  $\ABijection(\kx)$ 
and $\ABijection(\kz)$ correspond to walks $\ABijection(\kW)$ and $\ABijection(\overline{\kW})$, respectively. 
The definition of $\ABijection$ implies that $\ABijection(\kW)$ extends  $\ABijection(\overline{\kW})$ by one 
vertex.  To see this,  note that in this setting we have $\overline{\kW}\in \sawpaths^{(a)}$ while $\kW\in \sawpaths$, 
hence we have $\ABijection(\overline{\kW})=\overline{\kW}$. If $\kW\in \sawpaths^{(a)}\cup \sawpaths^{(c)}$,  then 
$\ABijection(\kW)=\kW$ implying that $\kW$ extends $\overline{\kW}$. If on the other hand
we have $\kW\in \sawpaths^{(b)}$ then  $\ABijection(\kW)=\kW_{\kP}$, where $\kW_{\kP}$ is defined as in
the paragraph above \eqref{def:OfUPHBijection}. It is easy to verify  that $\kW_{\kP}$ extends $\overline{\kW}$ in this case too. 
Hence, $\ABijection(\kx)$ and $\ABijection(\kz)$ are  two  adjacent vertices in  $\XTPw$. This implies that  
$\ABijection$ preserves adjacencies between the two trees,   $\STw_{\kP}$ and $\XTPw$.

From all the above, we conclude that $\ABijection$ has all the necessary properties to allow us to identify 
$\STw_{\kP}$ with $\XTPw$.   \Cref{lemma:TsawPVsTsaw} follows. 
\end{proof}

\begin{proof}[Proof of \Cref{lemma:GibbsWeightsTpVsP}]

Let $\mu^{M,\sigma}$ be the Gibbs distribution on $T=\Tsaw(G,\kw)$ induced by the $\Tsaw$-construction on 
$\mu^{\Lambda,\tau}_G$. Similarly, let the Gibbs distribution $\mu^{\kK,\eta}_P$ on $\XTPw=\Tsaw(\gext{G}{\kP},w)$ 
be induced by the same construction  with respect to the $\kP$-extension of $\mu^{\Lambda,\tau}_G$.

Consider the edge-weights $\{ \infweight(\ke)\}$ and $\{ \infweight_{\kP}(\ke)\}$ induced by the two Gibbs distributions 
 $\mu^{M,\sigma}$ and $\mu^{\kK,\eta}_{\kP}$ respectively. Note that the edge-weights $\{ \infweight(\ke)\}$ are on
 the edges of $T$, while the edge-weights $\{ \infweight_{\kP}(\ke)\}$ are on the edges of $\XTPw$. 
Using \Cref{lemma:TsawPVsTsaw}, we identify $\XTPw$ as a subtree of $T$,  and we endow the
edges of $\XTPw$ with the edge-weights $\{\infweight(\ke)\}_{\ke\in \XTPw}$.

Fix an edge $\ke=\{\kx,\kz\}$ in $\XTPw$ such that $\kz$ is the child of $\kx$.  Suppose  $\ke$ is at distance $\geq \kk$ 
from the root. \Cref{lemma:GibbsWeightsTpVsP} follows by showing that $\infweight(\ke)=\infweight_{\kP}(\ke)$.

The edge weights $\infweight(\ke)$ and $\infweight_{\kP}(\ke)$ only depend on the
parameters of the initial Gibbs distribution $\mu^{\Lambda,\tau}_G$ and the pinnings $(M,\sigma)$ and 
$(\kK,\eta)$ that the two Gibbs distributions $\mu^{M,\sigma}$, $\mu^{\kK,\eta}_{\kP}$ impose, respectively,
at the vertices of $T_{\kz}$.  As per standard notation, $T_{\kz}$ is the subtree of $\XTPw$ that consists 
of vertex $\kz$ and its descendants. Since $\kz$  is also a vertex in $T$, it is not hard to verify that 
the subtree of $T$ which consists of $\kz$ and its  descendants is exactly the same as $T_{\kz}$.
This is due to the assumption that edge $\ke$ is at distance $\geq k$ from the root of $T$.

Since, by definition, both $\mu^{M,\sigma}$ and $\mu^{\kK,\eta}_{\kP}$ have the same parameters,
it suffices to show  the pinnings $(M\cap \XTPw,\sigma)$ and 
$(\kK,\eta)$ are identical. Here,  we let $(M\cap \XTPw,\sigma)$ stand for the restriction of pinning 
$(M,\sigma)$ to the vertices in $\XTPw$.

Let set $\sawpaths^{\Lambda}\subseteq \sawpaths^{(\alpha)}$ consist of the self-avoiding walks in  $\sawpaths^{(\alpha)}$ 
whose last vertex is  in set $\Lambda$. Similarly, we define $\sawpaths^{\Lambda}_{\kP} \subseteq \sawpaths^{(\alpha)}_{\kP}$. 
The $\Tsaw$-construction implies that set $M\cap \XTPw$ consists of the vertices that  correspond to the walks in 
$ \sawpaths^{\Lambda} \cup \sawpaths^{(b)}\cup \sawpaths^{(c)}$. Similarly, set $\kK$ in $\mu^{K,\eta}_{\kP}$ consists of the 
vertices that correspond to the walks in   $ \sawpaths^{\Lambda}_{\kP}\cup \sawpaths^{(b)}_{\kP}\cup \sawpaths^{(c)}_{\kP}$.

We start by showing that $M\cap \XTPw \subseteq \kK$ while for every $\kv\in M\cap \XTPw$, we have $\sigma(\kv)=\eta(\kv)$. 
Since we identify  $\XTPw$ as a subtree of $T$, each vertex in $\XTPw$ corresponds to two walks, one from set 
$\sawpaths$ and one from set $\sawpaths_{\kP}$. Particularly, for vertex $\kv\in \XTPw$ which corresponds to  walk 
$\kW\in \sawpaths$, it also corresponds to walk $\ABijection(\kW)\in \sawpaths_{\kP}$, where $\ABijection$ is the  bijection between 
$\sawpaths$ and $\sawpaths_{\kP}$ described in \eqref{def:OfUPHBijection}, i.e.,  the proof of \Cref{lemma:TsawPVsTsaw}.

Each vertex $\kv\in M\cap \XTPw$, which corresponds to a walk $\kW\in \sawpaths^{\Lambda}\subseteq \sawpaths^{(\alpha)}$,  also corresponds to 
the same  walk $\kW\in \sawpaths^{\Lambda}_P$ since \eqref{def:OfUPHBijection} implies $\ABijection(\kW)=\kW\in \sawpaths^{\Lambda}_{\kP}$.
We then obtain that vertex $\kv\in \kK$. Furthermore,  the pinning rules of the $\Tsaw$-construction 
imply  $\eta(\kv)=\sigma(\kv)=\tau(\kv)$. Recall that $\tau$ is from the pinning of the initial Gibbs distribution $\mu^{\Lambda,\tau}_G$.

Arguing as above,  for each vertex $\kv\in M\cap \XTPw$, which corresponds to walk $\kW\in \sawpaths^{(c)}$, 
we have $\ABijection(\kW)=\kW\in \sawpaths^{(c)}_{\kP}$. Hence, we obtain that $\kv\in \kK$.  For such vertex $\kv$, 
the pinning rules of the $\Tsaw$-construction imply that $\eta(\kv)=\sigma(\kv)$.

Each vertex $\kv\in M\cap \XTPw$ which corresponds to walk $\kW=\kv_0, \ldots, \kv_{\kr} \in \sawpaths^{(b)}$  also 
corresponds to walk $\kW_{\kP}=\kz_0, \ldots, \kz_{\kr}$ such that $\kv_{\ki}=\kz_{\ki}$ for $0\leq \ki \leq \kr-1$,  while 
$\kz_{\kr}$ is a split-vertex in $\ssplit_{\kP} $generated by vertex $\kv_{\kr}$. That is, we have $\ABijection(\kW)=\kW_{\kP}\in \sawpaths^{(b)}$.
This  implies that $\kv\in \kK$.

For vertex $\kv$ in the paragraph above, the pinning $\sigma(\kv)$ is induced by the $\Tsaw$-construction on $\mu^{\Lambda,\tau}_G$, while
the pinning $\eta(\kv)$ is because we consider the $\kP$-extension of $\mu^{\Lambda,\tau}_G$,  which pins the split-vertices in $\ssplit_{\kP}$. 
From the description of the $\Tsaw$-construction and the definition of the $\kP$-extension of $\mu^{\Lambda,\tau}_{G}$ it not hard to verify that 
$\eta(\kv)=\sigma(\kv)$, e.g., see \eqref{eq:DefOfTau_x}.

All the above imply that $M\cap \XTPw \subseteq \kK$, while for each $\kv\in M\cap \XTPw$, we have $\sigma(\kv)=\eta(\kv)$. 
With a very similar line of arguments, we also deduce that $\kK \subseteq M\cap \XTPw$, while for each $\kv\in \kK$,
 we have  $\sigma(\kv)=\eta(\kv)$. We omit these arguments because they are almost identical to those we presented  above.

We conclude that the two pinnings $(M\cap \XTPw,\sigma)$ and $(\kK,\eta)$ are identical. Consequently, we have that
$\infweight(\ke)=\infweight_{\kP}(\ke)$ for any edge $\ke$ at distance $\geq k$ from the root of  $\XTPw$. 

\Cref{lemma:GibbsWeightsTpVsP} follows. 
\end{proof}

\newcommand{\BBijection}{\iris{\upphi}}

\subsection{Proof of \Cref{lemma:Subtree4TSawU}}\label{sec:lemma:Subtree4TSawU}

Let $\sawpaths$ be the set of walks in $G$ which correspond to the vertices in $\SMTu=\Tsaw(G,\ku)$. 
Also, let $\sawpaths_{\kP}$ be the set of walks in $\gext{G}{\kP}$ which correspond to the vertices in 
$\SMTPu=\Tsaw(\gext{G}{\kP},\ku)$.

For concreteness, let $\kP=\kx_0, \ldots, \kx_{\kk}$.  Recall that $\kx_0=\kw$ and $\ku\notin \kP$.
Let set $\cO$ consist of each edge $\ke$ in $G$ 
which does not connect two consecutive vertices in $\kP$ but has at least one end in the set of vertices 
$\{\kx_{0}, \ldots, \kx_{\kk-1}\}$. Let set $\cO_{\kP}$ consist of the set of edges in $\gext{G}{\kP}$ which 
are between a split-vertex in set  $\ssplit_{\kP}$ and another vertex of $\gext{G}{\kP}$.

The only edges in $G$ that do not appear in $\gext{G}{\kP}$ are exactly those in set $\cO$.
Similarly, the only edges in $\gext{G}{\kP}$ that do not appear in $G$ are exactly those in set $\cO_{\kP}$.
Hence, the elements in $\sawpaths \cap \sawpaths_{\kP}$ correspond to paths that do not use edges from
set $\cO$ and $\cO_{\kP}$.

We let the map $\BBijection:\sawpaths_{\kP}\to \sawpaths$ be defined as follows: 
\begin{enumerate}[(a)]
\item For each   $\kW=\kv_0, \ldots, \kv_{\kr}\in \sawpaths_{\kP}$ such that $\kv_{\kr}\notin \ssplit_{\kP}$
 we set $\BBijection(\kW)=\kW$. 
 \item For each $\kW=\kv_0, \ldots, \kv_{\kr}\in \sawpaths_{\kP}$ such that $\kv_{\kr}\in \ssplit_{\kP}$, we set
$\BBijection(\kW)=\kW_{\kP}$, where $\kW_{\kP}=\kz_0, \ldots, \kz_{\kr}$ is such that $\kz_{\ki}=\kv_{\ki}$ for $\ki=0, \ldots, \kr-1$,  
while the split-vertex $\kv_{\kr}$ is generated from $\kz_{\kr}$. 
\end{enumerate}
We proceed to show that $\BBijection$ is well-defined.

Since each vertex from $\ssplit_{\kP}$ is of degree one,  it can only appear in a walk in $\gext{G}{\kP}$ 
as a first or as a last vertex. Hence, each $\kW \in \sawpaths_{\kP}$ such that $\kv_{\kr}\notin \ssplit_{\kP}$ 
does not use edges from $\cO_{\kP}$,  implying $\kW\in \sawpaths$. Hence,  item (a)  for map
$\BBijection$  is well-defined.

For each walk $\kW=\kv_0, \ldots, \kv_{\kr} \in \sawpaths_{\kP}$ such that $\kv_{\kr}\in \ssplit_{\kP}$, there 
is walk $\kW_{\kP}=\kz_0, \ldots, \kz_{\kr} \in \sawpaths$  such that $\kv_{\ki}=\kz_{\ki}$ for $\ki=0,\ldots, \kr-1$, 
while $\kz_{\kr}$  is a split-vertex that is generated by vertex $\kv_{\kr}$.  This holds since the definition of  
$\gext{G}{\kP}$ implies that when constructing $\gext{G}{\kP}$, we substitute  the edge  $\{\kz_{\kr-1}, \kz_{\kr}\}\in \cO$ 
with the edge $\{\kv_{\kr-1}, \kv_{\kr}\}\in \cO_{\kP}$.  Hence,  item (b) for $\BBijection$ is well-defined.

We let $\widehat{\sawpaths} \subseteq \sawpaths$ be the image set of map $\BBijection$. The discussion in the 
previous paragraph  implies that for each $\kW\in \widehat{\sawpaths}$,  there is a (unique) well-defined walk 
$\BBijection^{-1}(\kW)\in \sawpaths_{\kP}$.  Furthermore, the definition of $\widehat{\sawpaths}$  implies that 
$\BBijection$ is a surjective map from $\sawpaths_{\kP}$ to $\widehat{\sawpaths}$.  We conclude that 
$\BBijection$ is a bijection from set $\sawpaths_{\kP}$ to set $\widehat{\sawpaths}$.

Recalling that $\SMTu=\Tsaw(G,\ku)$,  let $\widehat{T}$ be the subtree (perhaps forest) of $\SMTu$ 
induced by the vertices that  correspond to the walks in $\widehat{\sawpaths}$.  We extend the definition 
of $\BBijection$ and identify  it as a map from the set of vertices  in $\SMTPu$ to that in $\widehat{T}$. 
Particularly,   $\BBijection$  is bijection between the  sets of vertices of the two trees.

We show that $\BBijection$ preserves adjacencies. That is, for two vertices $\kx$, $\kz$ in $\SMTPu$
which are adjacent, we have that $\BBijection(\kx)$ and $\BBijection(\kz)$ are two adjacent vertices in 
$\widehat{T}$.  Suppose that $\kx$ and $\kz$ correspond to walks $\kW, \overline{\kW}\in \sawpaths_{\kP}$, 
respectively. W.l.o.g., suppose  walk $ \kW$ extends walk $\overline{\kW}$ by  one vertex. Then, the 
vertices  $\BBijection(\kx)$ and $\BBijection(\kz)$ correspond to walks $\BBijection(\kW)$ and 
$\BBijection(\overline{\kW})$, respectively. In what follows we show that  $\BBijection(\kW)$ 
extends $\BBijection(\overline{\kW})$.

Let $\overline{\kW}=\kv_0, \ldots, \kv_{\kr}$ and $\kW=\kv_0, \ldots, \kv_{\kr+1}$. First note that  we cannot  have 
$\kv_{\kr}\in \ssplit_{\kP}$ because this would imply that walk $\overline{\kW}$ cannot be extended as $\kv_{\kr}$  is of 
degree 1. Hence we always have  $\kv_{\kr}\notin \ssplit_{\kP}$. We consider cases for $\kv_{\kr+1}$. If 
$\kv_{\kr+1}\notin \ssplit_{\kP}$, then we have  $\BBijection(\overline{\kW})=\overline{\kW}$ and 
 $\BBijection(\kW)=\kW$.
Clearly,  in this case, $\BBijection(\kW)$ extends path  $\BBijection(\overline{\kW})$.  If $ \kv_{\kr+1}\in \ssplit_{\kP}$, 
then  we have $\BBijection(\kW)=\kW_{\kP}$, where  $\kW_{\kP}=\kz_0, \ldots, \kz_{\kr+1}$ 
is such that  $\kz_{\ki}=\kv_{\ki}$ for $\ki=0, \ldots, \kr$,  while the split-vertex $\kv_{\kr+1}$ is generated from 
$\kz_{\kr+1}$.  Clearly we have that  $\kW_{\kP}$   extends walk $\overline{\kW}$.  The above implies that 
$\BBijection(\kW)$  extends $\BBijection(\overline{\kW})$.

The conclusion of the previous paragraph implies that   $\BBijection(\kx)$ and $\BBijection(\kz)$ are  two adjacent 
vertices in $\widehat{T}$.   Hence,   $\BBijection$ preserves adjacencies between the two trees $\widehat{T}$ and 
$\SMTPu$.

The definition of $\BBijection$ also implies the following:  if vertex $\kv$ in $\SMTPu$ is a copy of vertex $\overline{\kv}$ in 
$\gext{G}{\kP}$ and $\overline{\kv}$ is not a split-vertex, then  $\BBijection(\kv)$ is also a copy of $\overline{\kv}$. If $\kv$ is a copy 
of the split-vertex  which is generated by vertex $\overline{\kv} \in \kP$, then  $\BBijection(\kv)$ is a copy of $\overline{\kv}$.

All the above imply that we can identify $\SMTPu$ with the tree $\widehat{T}$. Consequently, we can identify 
$\SMTPu$ as a subtree of $\SMTu$. The two trees $\SMTPu$ and $\SMTu$ have the same root, since they both 
consider the walks  starting from vertex $\ku$.

We now consider the set of vertices $\kv$ in $\SMTu\setminus \SMTPu$. 
The vertiex $\kv $corresponda to a walk  $\kW\in \sawpaths\setminus \widehat{\sawpaths}$. 
Recall that we set $\kP=\kx_0, \ldots, \kx_{\kk}$.  Then, for   walk $\kW=\kv_0, \ldots, \kv_{\kr}\in \sawpaths\setminus \widehat{\sawpaths}$
 there exists  $\kj<\kr$ such that $\kv_{\kj}$ is a  copy of a vertex in $\kP \setminus\{ \kx_{\kk}\}$ while $\kv_{j+1}$ is a copy of a vertex 
outside $\kP$.  This implies that  vertex $\kv$ in $\SMTu$ is a descendant of a copy of $\kv_{\kj}$ in   $\cpT$.  Recall
that set $\cpT$ consist of each vertex in $\SMTu$ which is a copy of a vertex in $\kP$.

All the above conclude the proof of \Cref{lemma:Subtree4TSawU}.
\hfill $\square$

\subsection{Proof of \Cref{lemma:BoundaryTVsBoundaryTPFromU}}\label{sec:lemma:BoundaryTVsBoundaryTPFromU}

The setting in this proof is the same as  in the proof of \Cref{lemma:Subtree4TSawU}.   Recall that we have vertices
$\kw,\ku$ in graph $G$ such that $\dist_{G}(\ku,\kw)\geq 2\kk$ and a self-avoiding path $\kP$ of length $\kk$ that 
emanates from vertex $\kw$; hence,  we have $\ku\notin \kP$.

Let $\sawpaths$ be the set of walks in  $G$ which correspond to the vertices in $\SMTu=\Tsaw(G,\ku)$.  Also, let 
$\sawpaths_{\kP}$ be the set of walks in $\gext{G}{\kP}$ which  correspond to the vertices in 
$\SMTPu=\Tsaw(\gext{G}{\kP},\ku)$.

We use the same map $\BBijection:\sawpaths_{\kP}\to \sawpaths$ as in the proof of \Cref{lemma:Subtree4TSawU}. 
That is, 
\begin{enumerate}[(a)]
\item For each $\kW=\kv_0, \ldots, \kv_{\kr}\in \sawpaths_{\kP}$ such that $\kv_{\kr}\notin \ssplit_{\kP}$
 we set $\BBijection(\kW)=\kW$. 
 \item For each $\kW=\kv_0, \ldots, \kv_{\kr}\in \sawpaths_{\kP}$ such that $\kv_{\kr}\in \ssplit_{\kP}$, we set
$\BBijection(\kW)=\kW_{\kP}$, where $\kW_{\kP}=\kz_0, \ldots, \kz_{\kr}$ is such that $\kz_{\ki}=\kv_{\ki}$ for $\ki=0, \ldots, \kr-1$,  
while the split-vertex $\kv_{\kr}$ is generated from $\kz_{\kr}$. 
\end{enumerate}
We let $\widehat{\sawpaths} \subseteq \sawpaths$ be the image set of  $\BBijection$.  Let $\widehat{T}$ be the 
subtree  of $\SMTu$ induced by the vertices that correspond to the walks in $\widehat{\sawpaths}$.  We showed in 
the proof of \Cref{lemma:Subtree4TSawU} that we can identify $\SMTPu$ as a subtree of $\SMTu$,  in particular,  we identify 
$\SMTPu$ as the subtree $\widehat{T}$ of $\SMTu$,  using  $\BBijection$.  Every vertex in $\SMTPu$ corresponds 
to two walks, one walk in $\gext{G}{\kP}$ and another  in $G$, e.g.,  the walk $\kW\in \sawpaths_{\kP}$ and the 
walk $\BBijection(\kW)\in \widehat{\sawpaths}\subseteq \sawpaths$.

Furthermore,  we apply the $\Tsaw$-construction on $\mu^{\Lambda,\tau}_G$ and obtain the Gibbs distribution $\mu^{M,\sigma}_{\SMTu}$ 
for tree $\SMTu=\Tsaw(G,\ku)$. We proceed to identify the walks in $\sawpaths$ that correspond to the vertices in $M$. 
Specifically, consider the following subsets of $\sawpaths$,
\begin{enumerate}[(i)]
\item set $\sawpaths^{\Lambda}\subseteq \sawpaths$ consists of all self-avoiding walks whose last vertex is in $\Lambda$,
\item set $\sawpaths^{(b)}\subseteq \sawpaths$ consists of all walks $\kW=\kv_0, \ldots, \kv_{\kr} \in \sawpaths$ such that $\kv_0, \ldots, \kv_{\kr-1}$ is self-avoiding and
$\kv_{\kr}$  repeats, while  $\kW\cap \kP\neq \emptyset$,  
\item set $\sawpaths^{(c)}\subseteq \sawpaths$ consists of all walks $\kW=\kv_0, \ldots, \kv_{\kr} \in \sawpaths$ such that $\kv_0, \ldots, \kv_{\kr-1}$ is self-avoiding and
$\kv_{\kr}$  repeats, while  $\kW\cap \kP= \emptyset$. 
\end{enumerate}
The $\Tsaw$-construction implies that set $M$ consists of the vertices that correspond to  the walks in the three aforementioned categories. 

Also,   we use the $\Tsaw$-construction for the $\kP$-extension of $\mu^{\Lambda,\tau}_G$  and get  the Gibbs distribution 
$\mu^{\kL,\eta}_{\SMTPu}$ for tree $\SMTPu=\Tsaw(\gext{G}{\kP},\ku)$. 
We proceed to identify the walks in $\sawpaths_{\kP}$ that correspond to the vertices in $\kL$. 
Consider the following subsets of $\sawpaths_{\kP}$, 
\begin{enumerate}[(i)]
\setcounter{enumi}{3}
\item set $\sawpaths^{\Lambda}_{\kP}\subseteq \sawpaths_{\kP}$ consist of all self-avoiding walks whose last vertex is in $\Lambda$,
\item set $\sawpaths^{(b)}_{\kP}\subseteq \sawpaths_{\kP}$ consists of all self-avoiding walks whose last vertex is in $\ssplit_{\kP}$, 
\item set $\sawpaths^{(c)}_{\kP}\subseteq \sawpaths_{\kP}$ consists of all walks $\kW=\kv_0, \ldots, \kv_{\kr} \in \sawpaths$ 
such that $\kv_0, \ldots, \kv_{\kr-1}$ is self-avoiding and
$\kv_{\kr}$  repeats, while  $\kW\cap \kP\neq \emptyset$. 
\end{enumerate}
It is not hard to verify that set $\kL$ consists of the vertices that correspond to the walks of the above three categories.

We use all the above to prove statements \ref{statement:lemma:BoundaryTVsBoundaryTPFromUA}-\ref{statement:lemma:BoundaryTVsBoundaryTPFromUD}
of \Cref{lemma:BoundaryTVsBoundaryTPFromU}.

We start with statement \ref{statement:lemma:BoundaryTVsBoundaryTPFromUA}.   For each vertex 
$\kv\in M\cap \kL$, there are two  alternatives. In the first one, vertex $\kv$ corresponds to a walk 
$\kW\in \sawpaths^{\Lambda}_{\kP}$ and $\BBijection(\kW)\in \sawpaths^{\Lambda}$. In the second 
alternative, vertex $\kv$ corresponds to a walk $\kW\in \sawpaths^{(c)}_{\kP}$ and  $\BBijection(\kW)\in \sawpaths^{(c)}$.

To see  the above,  note that for each walk $\kW\in \sawpaths^{\Lambda}_{\kP}\cup\sawpaths^{(c)}_{\kP}$, the definition
of  $\BBijection$ implies that we should have  $\BBijection(\kW)\in \sawpaths^{\Lambda}\cup \sawpaths^{(c)}$.  
On the other hand, for each walk $\kW\in \sawpaths^{(b)}_{\kP}$ we have 
$\BBijection(\kW)\notin (\sawpaths^{\Lambda}\cup \sawpaths^{(b)} \cup\sawpaths^{(c)})$.  This is because 
$\BBijection(\kW)$ is a self-avoiding walk in $G$ from $\ku$ to a vertex in $\kP$. Hence, every
vertex $\kv$ which corresponds  to walk $\kW\in \sawpaths^{(b)}_{\kP}$ does not belong to $\kM$. 
Also, for every $\kW\in \sawpaths^{(b)}$,   we have that $\BBijection^{-1}(\kW)$ is not defined, 
implying  that the vertex $\kv$ which corresponds  to walk $\kW$ does not belong to $\SMTPu$;
hence it does not belong  to $\kL$.

Statement \ref{statement:lemma:BoundaryTVsBoundaryTPFromUA} follows by noting that 
the definition of the $\Tsaw$-construction implies that all vertices in $M\cap \kL$ take on the same 
 pinning, i.e.,   there are  no copies of split-vertices in this intersection.

As far as statement \ref{statement:lemma:BoundaryTVsBoundaryTPFromUB} is concerned, note 
that each vertex $\kv\in \cpT \cap \SMTPu$ can only
correspond to a self-avoiding walk $\kW\in \sawpaths$ which ends at a vertex in $\kP$. The definition of 
$\BBijection$ implies that $\BBijection^{-1}(\kW)$ is well-defined and corresponds to  a self-avoiding walk in 
$\sawpaths_{\kP}$ whose last vertex is  either in $\kP$ or in $\ssplit_{\kP}$. Hence, it follows from the
definition of set $\cpTP$ that $\kv\in \cpTP$. This proves \ref{statement:lemma:BoundaryTVsBoundaryTPFromUB}.

We proceed to prove statement \ref{stat:C-lemma:BoundaryTVsBoundaryTPFromU}. The previous discussion
shows that  each vertex  $\kv\in \kL\setminus M$  corresponds to a walk $\kW\in \sawpaths^{(b)}_{\kP}$. 
Furthermore, the definition of $\BBijection$ implies that the last vertex of walk $\BBijection(\kW)$ is in $\kP$. Hence, we have 
$\kv\in \cpT$. We conclude that $\kL\setminus M\subseteq \cpT$.

If vertex $\kv\in M\cap \SMTPu$ corresponds to a walk $\kW\in \sawpaths^{\Lambda}$, then the definition 
of $\BBijection$ implies that $\BBijection^{-1}(\kW)\in \sawpaths^{\Lambda}_{\kP}$; hence $\kv\in \kL$. Similarly, if vertex
$\kv\in M\cap \SMTPu$ corresponds to a walk $\kW\in \sawpaths^{(c)}$, then it also corresponds to walk 
$\BBijection^{-1}(\kW)\in \sawpaths^{(c)}_{\kP}$; hence $\kv\in \kL$. 
We conclude that $M\cap \SMTPu\subseteq \kL$ 
by noting that there is no vertex $\kv\in M\cap \SMTPu$ which corresponds to a walk $\kW\in \sawpaths^{(b)}$
(this was also noted in  the proof of \ref{statement:lemma:BoundaryTVsBoundaryTPFromUA}).

The two conclusions that $\kL\setminus M\subseteq \cpT$ and that $M\cap \SMTPu\subseteq \kL$
imply that statement \ref{stat:C-lemma:BoundaryTVsBoundaryTPFromU} is true.

We proceed to prove statement \ref{statement:lemma:BoundaryTVsBoundaryTPFromUD}. This follows directly by recalling 
that all vertices in $T\setminus \SMTPu$ are descendants of the vertices in $\cpT$. Hence, after pinning the vertices in $\kL\cup \cpT$, 
the vertices in $T\setminus \SMTPu$ have no effect on the configurations of the vertices in $T\cap \SMTPu$.

All the above conclude the proof of \Cref{lemma:BoundaryTVsBoundaryTPFromU}. \hfill $\square$

\subsection{{Proof of \Cref{proposition:CINorm2VsCLNorm2}}}\label{sec:proposition:CINorm2VsCLNorm2}

Before proving \Cref{proposition:CINorm2VsCLNorm2}, we present two useful results.

\begin{claim}\label{claim:2NormBoundforVVsEdge-AAA}
We have $\norm{ \VToEdge_{\kk}}{2} \leq \maxDeg^{\kk/2}$ and $ \norm{ \EdgeToV_{\kk}}{2}\leq \maxDeg^{\kk/2}$.
\end{claim}

\begin{claim}\label{claim:2Norm4BoundedNoOfIL-AAA}
We have $\norm{\infmatrix^{\Lambda,\tau}_{G, <2\kk}}{2} \leq \frac{\maxDeg^{2\kk+1}-1}{\maxDeg-1}$.
\end{claim}

\noindent
In what follows, we abbreviate $\infmatrix^{\Lambda,\tau}_{G}$, $\infmatrix^{\Lambda,\tau}_{G, <2\kk}$, $\ExtdInfMatrixF$,
 $\VToEdge_{\kk}$ and $\EdgeToV_{\kk}$ to 
$\infmatrix$, $\infmatrix_{<2\kk}$, $\ExtdInfMatrix$, $\VToEdge$ and $\EdgeToV$, respectively.
From \Cref{thrm:InflVsExtInfl} we have that 
\begin{align}
\norm{\infmatrix }{2} & = \norm{\VToEdge \cdot \left( \ExtdInfMatrix\circ \SymWeightM \right )\cdot \EdgeToV+ \infmatrix_{<2\kk} }{2} \nonumber \\
&\leq \norm{ \VToEdge}{2} \cdot \norm{ \EdgeToV}{2} \cdot \norm{\ExtdInfMatrix\circ \SymWeightM }{2}
+ \norm{ \infmatrix_{<2\kk}}{2} \nonumber \\
&\leq \SymWeightM_{\rm max}\cdot \norm{ \VToEdge }{2} \cdot \norm{ \EdgeToV}{2} \cdot \norm{ \abs{ \ExtdInfMatrix } }{2}
+ \norm{ \infmatrix_{<2\kk}}{2} \label{eq:prop:IVsJ2NormStepAA-AAA} \enspace, 
\end{align}
where $\SymWeightM_{\rm max}=\max_{\kP,\kQ}\{\abs{\SymWeightM(\kP,\kQ )}\}$. \Cref{thrm:InflVsExtInfl} implies that $\SymWeightM_{\rm max}$
is a constant.

\Cref{proposition:CINorm2VsCLNorm2} follows from \eqref{eq:prop:IVsJ2NormStepAA-AAA} and
\Cref{claim:2NormBoundforVVsEdge-AAA,claim:2Norm4BoundedNoOfIL-AAA}, while we set 
$C_1=\frac{\maxDeg^{2\kk+1}-1}{\maxDeg-1}$
and $C_2=\SymWeightM_{\rm max}\cdot \maxDeg^{\kk}$.
\hfill $\square$

\begin{proof}[Proof of \Cref{claim:2NormBoundforVVsEdge-AAA}]
We start with $\norm{ \VToEdge_{\kk}}{2}$. 
Consider the product $\VToEdge_{\kk} \cdot \MTR{\VToEdge}_{\kk}$, where $\MTR{\VToEdge}_{\kk}$ is the matrix transpose of $\VToEdge_{\kk}$. 
Matrix $\VToEdge_{\kk} \cdot \MTR{\VToEdge}_{\kk}$ is $(V\setminus \Lambda)\times (V\setminus \Lambda)$, while for  any $\ku,\kw\in V\setminus \Lambda$, we have 
\begin{align}
\left (\VToEdge_{\kk} \cdot \MTR{\VToEdge}_{\kk} \right)(\ku,\kw) 
& = \sum\nolimits_{\kP \in \ExtV_{\Lambda,\kk}} \VToEdge_{\kk}(\ku, \kP)\times \MTR{\VToEdge}_{\kk}(\kP, \kw) \nonumber\\
& = \sum\nolimits_{\kP \in \ExtV_{\Lambda,\kk}} \Ind\{\textrm{$\ku$ is the starting vertex in $\kP$} \} \times \Ind\{\textrm{$\kw$ is the starting vertex in $\kP$} \} \nonumber\\
& = \Ind\{ \kw=\ku\} \times \sum\nolimits_{\kP\in \ExtV_{\Lambda,\kk}} \Ind\{\textrm{$\ku$ is the starting vertex in $\kP$} \} \nonumber\\
&\leq \Ind\{ \kw=\ku\} \times \maxDeg^{\kk} \enspace. \nonumber
\end{align}
We conclude that $\VToEdge_{\kk} \cdot \MTR{\VToEdge}_{\kk}$ is a non-negative and diagonal matrix. For $\kw\in V\setminus \Lambda$, we have 
$(\VToEdge_{\kk} \cdot \MTR{\VToEdge}_{\kk} )(\kw,\kw) \leq \maxDeg^{\kk}$. 
Hence, we have $ \norm{ \VToEdge_{\kk} \cdot \MTR{\VToEdge}_{\kk}}{2} \leq \maxDeg^{\kk}$, which implies that  $\norm{ \VToEdge_{\kk}}{2} \leq \maxDeg^{\kk/2}$.
Working similarly, we get $ \norm{ \EdgeToV_{\kk}}{2}\leq \maxDeg^{\kk/2}$.
 \Cref{claim:2NormBoundforVVsEdge-AAA} follows.
\end{proof}

\begin{proof}[Proof of \Cref{claim:2Norm4BoundedNoOfIL-AAA}]
Let $\Adjacency^{(2\kk)}$ be a  $(V\setminus \Lambda)\times (V\setminus \Lambda)$,  zero-one matrix  such that 
for any $\kw,\ku\in V\setminus \Lambda$,  we have
\begin{align}
\Adjacency^{(2\kk)}(\kw,\ku)&=\Ind\{\dist_{G}(\kw,\ku)\leq 2\kk\}\enspace. 
\end{align}
Noting that each entry in $\infmatrix^{\Lambda,\tau}_{G, <2\kk}$ is in the interval $[-1,1]$, we have
\begin{align}\label{eq:AbsInfLes2kVsAGPow2k}
\abs{\infmatrix^{\Lambda,\tau}_{G, <2\kk}} &\leq \Adjacency^{(2\kk)}\enspace,
\end{align}
where the inequality is entrywise. We have
\begin{align}\label{eq:FinalEq2Norm4BoundedNoOfIL-AAA}
\norm{\infmatrix^{\Lambda,\tau}_{G, <2\kk}}{2}&\leq \norm{ \abs{\infmatrix^{\Lambda,\tau}_{G, <2\kk}} }{2} \leq 
\norm{\Adjacency^{(2\kk)}}{2} \leq \norm{\Adjacency^{(2\kk)}}{\infty}\enspace.  
\end{align}
The first inequality above is elementary, e.g. see \Cref{lemma:AbsoluteVs2Norm}. The second inequality follows from 
\eqref{eq:AbsInfLes2kVsAGPow2k}.  The last inequality follows by noting that matrix $\Adjacency^{(2\kk)}$ is symmetric. 

Since $G$ is of maximum degree $\maxDeg$, we have that $\norm{\Adjacency^{(2\kk)}}{\infty}\leq \frac{\maxDeg^{2\kk+1}-1}{\maxDeg-1}$. 
The claim follows by plugging the bound on $\norm{\Adjacency^{(2\kk)}}{\infty}$ into \eqref{eq:FinalEq2Norm4BoundedNoOfIL-AAA}.
\end{proof}

\spreadpoint

\spreadpoint
\newpage
\appendix

\section{Some Standard Proofs}

\subsection{Proof of \Cref{claim:InfSymmetrisation}}\label{sec:claim:InfSymmetrisation}

Here we restate \Cref{claim:InfSymmetrisation} and provide its proof.

\Isymsym*

\begin{proof}
Let $\CovMatrix^{\Lambda,\tau}_G$ be the $(V\times \Lambda)\times (V\times \Lambda)$ covariance matrix defined with respect to $\mu^{\Lambda, \tau}$. That is, for any $u,v\in V\setminus \Lambda$ we have 
\begin{align}\nonumber
\CovMatrix^{\Lambda,\tau}_G(\kv,\ku) &=\mu^{\Lambda, \tau}_{(\kv,\ku)}((+1,+1))-\mu^{\Lambda, \tau}_{\kv}(+1)\cdot \mu^{\Lambda, \tau}_{\ku}(+1)\enspace. 
\end{align}
A straightforward observation is that $\CovMatrix^{\Lambda,\tau}_G$ is symmetric. W.l.o.g. in this proof, assume that 
the diagonal entries of $\CovMatrix^{\Lambda,\tau}_G$ are non-zero. 
Note that for any $\kv\in V\setminus \Lambda$ we have that
\begin{align}\label{eq:UpMVsCovMat}
\UpM(\kv,\kv)= \sqrt{\CovMatrix^{\Lambda,\tau}_G(\kv,\kv)}\enspace. 
\end{align}
Furthermore, it is standard to show, e.g. see \cite{VigodaSpectralInd}, that $\CovMatrix^{\Lambda,\tau}_G$ and 
$\infmatrix^{\Lambda,\tau}_G$ satisfy that 
$\infmatrix^{\Lambda,\tau}_G(\ku,\kv)= \frac{\CovMatrix^{\Lambda,\tau}_G(\ku,\kv)}{\CovMatrix^{\Lambda,\tau}_G(\ku,\ku)}.$ 
Since $\CovMatrix^{\Lambda,\tau}_G$ is symmetric, this relation implies that
\begin{align}\nonumber
\CovMatrix^{\Lambda,\tau}_G(\ku,\ku) \cdot \infmatrix^{\Lambda,\tau}_G(\ku,\kv) &
=\CovMatrix^{\Lambda,\tau}_G(\kv,\kv) \cdot \infmatrix^{\Lambda,\tau}_G(\kv,\ku) \enspace. 
\end{align}
Then, a simple rearrangement and \eqref{eq:UpMVsCovMat} imply that 
\begin{align}\nonumber
\textstyle \frac{\UpM(\ku,\ku)}{\UpM(\kv,\kv) } \cdot \infmatrix^{\Lambda,\tau}_G(\ku,\kv)& 
\textstyle = \frac{\UpM(\kv,\kv)}{\UpM(\ku,\ku) } \cdot \infmatrix^{\Lambda,\tau}_G(\kv,\ku) \enspace. 
\end{align}
The l.h.s. in the above equation corresponds to the entry 
$(\UpM\cdot \infmatrix^{\Lambda,\tau}_G \cdot \UpM^{-1})(\ku,\kv)$, while the r.h.s. corresponds to 
$(\UpM\cdot \infmatrix^{\Lambda,\tau}_G \cdot \UpM^{-1})(\kv,\ku)$.
The claim follows. 
\end{proof}

\subsection{Proof of \Cref{claim:SigmaLVsPathNumber}}\label{sec:claim:SigmaLVsPathNumber}

We restate \Cref{claim:SigmaLVsPathNumber} before presenting its proof.
\hsingpathbound*

\begin{proof}
In what follows, we abbreviate $\NBMatrix$ to $\NBMatrixE$.

Let the vector ${\bf \kx} \in \mathbb{R}^{\ExtV_k}$ be such that for any $\kR\in \ExtV_k$ we have 
\begin{align}
{\bf \kx}(\kR)&=\frac{1}{\sqrt{2}}\left( \Ind\{\kR=\kP\} + \Ind\{\kR=\kQ^{-1}\} \right)\enspace.  
\end{align}
We have that 
\begin{align} \nonumber
{\bf \kx}^T (\NBMatrixE^{\kell} \cdot \Invol){\bf \kx}&=\frac{1}{2}\left((\NBMatrixE^{\kell} \cdot \Invol)(\kP,\kQ^{-1}) 
+(\NBMatrixE^{\kell} \cdot \Invol)(\kQ^{-1}, \kQ^{-1})+(\NBMatrixE^{\kell} \cdot \Invol)(\kP,\kP)+
 (\NBMatrixE^{\kell} \cdot \Invol)(\kQ^{-1}, \kP) \right)\enspace. 
\end{align}
Recalling the definition of matrix $\NBMatrixE^{\kell}\cdot \Invol$, we have 
\begin{align} \nonumber
\NBMatrixE^{\kell}(\kP,\kQ) =(\NBMatrixE^{\kell}\cdot \Invol)(\kP,\kQ^{-1})=(\NBMatrixE^{\kell} \cdot \Invol)(\kQ^{-1},\kP) \enspace,
\end{align}
the last equation follows since $\NBMatrixE^{\kell} \cdot \Invol$ is symmetric. 
All the above imply that
\begin{align} \nonumber
\NBMatrixE^{\kell}(\kP,\kQ) \leq {\bf \kx}^T \cdot (\NBMatrixE^{\kell} \cdot \Invol) \cdot {\bf \kx} 
\leq \norm{ \NBMatrixE^{\kell} \cdot \Invol}{2}=\sigma_{\kk,\kell} \enspace.
\end{align}
The last inequality above follows from the fact that $\NBMatrixE^{\kell}\cdot \Invol$ is symmetric, with maximum
eigenvalue $\sigma_{\kk,\kell}$. 

The claim follows. 
\end{proof}

\subsection{Proof of \Cref{lemma:SingSequenConv}}\label{sec:lemma:SingSequenConv}

We restate \Cref{lemma:SingSequenConv} before presenting its proof.
\hsingseqnondec*

\begin{proof}
In what follows, we abbreviate $\NBMatrix$ to $\NBMatrixE$

Let the integers $t=\lfloor \frac{\kell}{N} \rfloor$ and $\kr\geq 0$ be such that $\kell= N\cdot t+\kr$.
We have 
\begin{align}\label{eq:Base4lemma:SingSequenConv}
\left( \hsingular_{\kk,\kell} \right)^{1/\kell}& = \nnorm{  \NBMatrixE^{\kell}}{\frac{1}{\kell}}{2} \leq 
\nnorm{  \NBMatrixE^{\kell-\kr}}{\frac{1}{\kell}}{2} \cdot \nnorm{  \NBMatrixE^{\kr}}{\frac{1}{\kell}}{2} \leq 
\nnorm{  \NBMatrixE^{\kell-\kr}}{\frac{1}{\kell-\kr}}{2} \cdot \nnorm{ \NBMatrixE^{\kr}}{\frac{1}{\kell}}{2} \enspace. 
\end{align}
Furthermore, noting that $\kell-\kr=N\cdot t$, we get 
\begin{align}
\nnorm{  \NBMatrixE^{\kell-\kr}}{\frac{1}{\kell-\kr}}{2} &= \nnorm{ \NBMatrixE^{N\cdot t}}{\frac{1}{N\cdot t}}{2} \leq 
\nnorm{ \NBMatrixE^N}{\frac{1}{N}}{2} = (\hsingular_{\kk,N})^{\frac{1}{N}}\enspace. 
\end{align}
Plugging the above into \eqref{eq:Base4lemma:SingSequenConv} we have
\begin{align}\nonumber
\left( \hsingular_{\kk,\kell} \right)^{1/\kell}& \leq (\hsingular_{\kk,N})^{\frac{1}{N}} 
\cdot \nnorm{  \NBMatrixE^{\kr}}{\frac{1}{\kell}}{2} \enspace.
\end{align}
Noting that $\nnorm{ \NBMatrixE^{\kr}}{\frac{1}{\kell}}{2}\leq (\maxDeg-1)^{\frac{\kr}{\kell}}$
and recalling that $\maxDeg$ is bounded, there is fixed $\ell_0$ such that 
for $\kell>\ell_0$ we have $\nnorm{ \NBMatrixE^r}{\frac{1}{\kell}}{2}\leq (1+\kz)$.
\Cref{lemma:SingSequenConv} follows.
\end{proof}

\subsection{{Proof of \Cref{lemma:IsingInfNormBound}}}\label{sec:lemma:IsingInfNormBound}
We restate \Cref{lemma:IsingInfNormBound} and provide its proof.

\IsingInfNormBound*

\begin{proof}
It suffices to show that any $d>0$ and any 
${\bf y}\in [-\infty, +\infty]^d$ we have 
\begin{align}\label{eq:Target4thrm:SI4Ising}
\norm{\nabla \logtrecur_d ({\bf y})}{\infty}
&\leq \frac{|\beta-1|}{\beta +1} \enspace.
\end{align}
Before showing that \eqref{eq:Target4thrm:SI4Ising} is true, let us show how it implies 
\eqref{eq:lemma:IsingInfNormBound}. That is, we show that for any $\beta \in \UnIsing(R,\varepsilon)$,
we have that $\frac{|\beta-1|}{\beta +1}\leq \frac{1-\zeta}{R}$.

Consider the function $f(\kx)=\frac{|\kx-1|}{\kx+1}$ defined on the closed interval 
$\left[\frac{R-1}{R+1}, \frac{R+1}{R-1} \right]$. 
Taking derivatives, it is elementary to verify that $f(\kx)$ is increasing in the interval 
$1< \kx < \frac{R+1}{R-1}$, while it is decreasing in the interval $\frac{R-1}{R+1}< \kx <1$. 
Furthermore, noting that $f(1)=0$, it is direct that 
{\small
\begin{align}\nonumber
\sup\nolimits_{\beta \in \UnIsing(R,\varepsilon)}f(\beta)= 
f\left(\frac{R-1+\varepsilon}{R+1-\varepsilon}\right)= f\left(\frac{R+1-\varepsilon}{R-1+\varepsilon}\right)
=\frac{1-\varepsilon}{R} \enspace.
\end{align}
}
We conclude that \eqref{eq:Target4thrm:SI4Ising} implies \eqref{eq:lemma:IsingInfNormBound}.
It remains to show that \eqref{eq:Target4thrm:SI4Ising} is true.

Since $\frac{\partial }{\partial \kx_{\ki}}\logtrecur_d(\kx_1, \ldots, \kx_d)=\dlogtrecur(\kx_{\ki})$, 
it suffices to show that for any $\kx\in [-\infty, +\infty]$ we have that
{\small 
\begin{align}\label{eq:TargetB4thrm:SI4Ising}
\abs{ \dlogtrecur(\kx)} &\leq \frac{\abs{1-\beta } }{1+\beta} \enspace.
\end{align}
}
For the distribution we consider here, the function $\dlogtrecur(\cdot)$ is given by 
\begin{align}
\dlogtrecur(\kx) & =-\frac{(1-\beta^2) \exp(\kx)}{(\beta \exp(\kx)+1)(\exp(\kx)+\beta)} \enspace.
\end{align}
From the above we get that
{\small
\begin{align}
\abs{\dlogtrecur(\kx)} & = \frac{\abs{1-\beta^2}\exp(\kx)}{(b \exp(\kx) +1)(b+ \exp(\kx))} \ =
\ \frac{\abs{1-\beta^2 }}{\beta^2+1 +\beta (\exp(-\kx)+\exp(\kx) ) } \enspace. \nonumber
\end{align}
}
It is straightforward to verify that $\phi(\kx)=e^{-\kx}+e^{\kx}$ is convex and for any $\kx\in [-\infty, +\infty]$
and attains its minimum at $\kx=0$, i.e., we have that $\phi(\kx)\geq 2$. Consequently, we get 
{\small 
\begin{align} \nonumber 
\abs{\dlogtrecur(\kx)} & \leq \frac{\abs{1-\beta^2}}{\beta^2+1 +2\beta } \ = \   \frac{\abs{1-\beta}}{1+\beta} \enspace,
\end{align}
}
for any $\kx\in [-\infty, +\infty]$.
The above proves that \eqref{eq:TargetB4thrm:SI4Ising} is true and concludes our proof.
\end{proof}

\subsection{{Proof of \Cref{claim:InterpretGoodPotentialHC}}} \label{sec:claim:InterpretGoodPotentialHC}
We restate \Cref{claim:InterpretGoodPotentialHC} and provide its proof.

\ElaborateGoodPotenHC*

\begin{proof}
It elementary to verify that the function $\dcritical(x)$ is decreasing in $x$. This implies that for any $\lambda\leq (1-\varepsilon)\lcritical(R)$, 
we have $\dcritical(\lambda)\geq \dcritical(\lcritical(R))=R$. Particularly, this implies that there is $0<z<1$, which only depends on 
$\varepsilon$ such that $\dcritical(\lambda) \geq \frac{R}{(1-z)}$.
This proves the leftmost inequality in \eqref{eq:claim:InterpretGoodPotentialHC}.

As far as the rightmost inequality is concerned, we have 
\begin{align}\label{eq:Base4Righteq:claim:InterpretGoodPotentialHC}
\textstyle \frac{\lambda}{1+\lambda} \leq \lambda < \lcritical(R) \enspace.
\end{align}
The first inequality follows since $\lambda>0$, while the second follows since $\lambda<\lcritical(R)$.
 From the definition of $\lcritical(\cdot)$, we have 
{\small 
\begin{align}
\lcritical(R)& =\frac{R^{R}}{(R-1)^{(R+1)}}\ = \ 
\frac{1}{R}\left( 1+\frac{1}{R-1}\right)^{R+1} \ \leq \ \frac{1}{R}\exp\left(\frac{R+1}{R-1}\right) 
\ \leq \ \frac{e^3}{R} \enspace. 
\end{align}
}
For the one before the last inequality we use that $1+x\leq e^{x}$. For the last inequality we note that $\frac{R+1}{R-1}$ is decreasing in 
$R$, hence, for $R\geq 2$, we have that $\frac{R+1}{R-1}\leq 3$. 

Plugging the above bound into \eqref{eq:Base4Righteq:claim:InterpretGoodPotentialHC}, gives the rightmost inequality in
\eqref{eq:claim:InterpretGoodPotentialHC}. 
The claim follows. 
\end{proof}

\section{Standard linear algebra}\label{sec:PerronFrobeniusThrm}

For the matrix $\UpL \in \mathbb{R}^{N\times N}$ we follow the convention to call it non-negative, if all its entries are non-negative numbers, i.e., every entry $\UpL_{i,j}\geq 0$.

For an $N\times N$ matrix $\UpL$, we let $\eigenval_i(\UpL)$, for $i=1,\ldots, N$, denote the eigenvalues of $\UpL$ such that
$\eigenval_1(\UpL)\geq \eigenval_2(\UpL) \geq \ldots \geq \eigenval_N(\UpL)$.
Also, we let $\spectrum(\UpL)$ denote the set of distinct eigenvalues of $\UpL$. We also refer to $\spectrum(\UpL)$ 
as the {\em spectrum} of $\UpL$.

We define the {\em spectral radius} of $\UpL$, denoted as $\spradius(\UpL)$, to be the real number
such that 
\begin{align}\nonumber 
\spradius(\UpL)=\max_{}\{ |\eigenval| \ : \ \eigenval\in \spectrum(\UpL) \} \enspace.
\end{align}

It is a well-known result that the spectral radius of $\UpL$ is the greatest lower bound for all 
of its matrix norms, e.g. see Theorem 6.5.9 in \cite{MatrixAnalysis}. 
That is, letting $\enorm{\cdot} $ be  a matrix norm on $N\times N$ matrices, we have that
\begin{align}\label{eq:SpRadiusVsMNorm}
\spradius(\UpL) \leq \enorm{ \UpL}  \enspace.
\end{align}
Furthermore, we let $\MTR{\UpL}$ be the standard matrix transpose of $\UpL$. 
It is useful to mention that when $\UpL$ is symmetric, i.e., $\MTR{\UpL}=\UpL$,
we have that $\spradius(\UpL) = \norm{\UpL}{2}$.

For $\UpD, \UpB, \UpC \in \mathbb{R}^{ N \times N}$, we let $|\UpD|$ denote the matrix having entries $|\UpD_{i,j}|$.
 For the matrices $\UpB, \UpC$ we define $\UpB\leq \UpC$ to mean that $\UpB_{i,j}\leq \UpC_{i,j}$ for each $i$ and $j$.
The following is a folklore result e.g. see \cite{SIAM-LAlg,MatrixAnalysis}.

\begin{lemma}\label{lemma:MonotoneVsSRad}
For integer $N>0$, let $\UpD, \UpB\in \mathbb{R}^{N\times N}$. If $\abs{\UpD}\leq \UpB$, then 
$\spradius(\UpD) \leq \spradius (\abs{\UpD}) \leq \spradius(\UpB)$.
\end{lemma}

\begin{lemma}\label{lemma:AbsoluteVs2Norm}
For integer $N>0$, let $\UpD\in \mathbb{R}^{N\times N}$. We have that $\norm {\UpD}{2}\leq \norm{ \abs{ \UpD}}{2}$.
\end{lemma}
\begin{proof}
If $\UpD$ is symmetric, then $\spradius(\UpD)=\norm{ \UpD}{2}$ and $\spradius(\abs{\UpD}) = \norm{ \abs{ \UpD}}{2}$.
Hence, \Cref{lemma:AbsoluteVs2Norm} follows as a corollary from \Cref{lemma:MonotoneVsSRad}. 

Let us consider the general case, i.e., where $\UpD$ is not necessarily symmetric.
Let $\overline{\UpD}$ be the standard matrix transpose of $\UpD$. We have that
\begin{align}
\nnorm{ \UpD}{2}{2} = \norm{ \UpD\cdot \overline{\UpD}}{2} &\leq \norm{ \abs{ \UpD\cdot \overline{\UpD}}}{2} \leq 
\norm{ \abs{\UpD} \cdot \abs{\overline{\UpD}}}{2}=\nnorm{ \abs{\UpD}}{2}{2}\enspace.
\end{align}
The first inequality follows by noting that $\UpD\cdot \overline{\UpD}$ is a symmetric matrix. The second inequality follows by noting that
$\abs{ \UpD\cdot \overline{\UpD}}\leq \abs{ \UpD} \cdot \abs{\overline{\UpD}}$. 

\Cref{lemma:AbsoluteVs2Norm} follows. 
\end{proof}

\begin{lemma}\label{lemma:MonotonicityVs2Norm}
For integer $N>0$, let $\UpD, \UpB\in \mathbb{R}^{N\times N}_{\geq 0}$. If $\UpD\leq \UpB$, then $\norm{ \UpD}{2}\leq \norm{ \UpB}{2}$.
\end{lemma}
\begin{proof}
We have that $\norm{ \UpD}2=\sup_{{\bf z}\in \mathbb{R}^N: \norm{ {\bf z}}{2}=1}{\norm{ \UpD {\bf z}}{2}}$.
Hence, we have that
\begin{align}
\norm{ \UpD}{2}=\sup\nolimits_{{\bf z}\in \mathbb{R}^N: \norm{ {\bf z}}{2}=1}\sqrt{\sum\nolimits_{\ki=1,\ldots,N} 
\sum\nolimits_{j=1,\ldots, N} \UpD(\ki,j)\cdot{\bf z}(j)} \enspace.
\end{align}
Since the entries of $\UpD$ are assumed to be non-negative, it is elementary to verify that supremum, above, is achieved for 
 ${\bf z}$ whose entries are non-negative. 

Let ${\bf x}\in \mathbb{R}^{N}$ be such that $\norm{ {\bf x}}{2}=1$ and $\norm{ \UpD \cdot {\bf x}}{2}=\norm{ \UpD}{2} $, while
let ${\bf v}=\UpD{\bf x}$ and ${\bf w}=\UpB{\bf x}$. Since the entries of $\UpB$ are assumed to be non-negative, while $\UpD\leq \UpB$, 
for any $i\in [N]$ we have that ${\bf w}(i)\geq {\bf v}(i)$. Clearly, this implies that
\begin{align}
\norm{ \UpD}{2}=\norm{ {\bf v}}{2} \leq \norm{ {\bf w}}{2} \leq \norm{ \UpB}{2}\enspace. 
\end{align}
The above concludes the proof of \Cref{lemma:MonotonicityVs2Norm}.
\end{proof}

\end{document}